\documentclass[11pt]{article}
\usepackage{graphicx} 
\usepackage{bbm}
\usepackage[margin=1in]{geometry}
\usepackage{color}
\usepackage{amsfonts}
\usepackage{algorithm}
\usepackage{algorithmic}
\usepackage{latexsym, amssymb, amsmath, amscd, amsthm, amsxtra}
\usepackage{mathtools}
\usepackage{enumerate}
\usepackage[all]{xy}
\usepackage{mathrsfs}
\usepackage{fancyhdr}
\usepackage{listings}
\usepackage{hyperref}
\usepackage{enumitem}
\usepackage{multirow}
\usepackage{tikz}

\def\P{\mathbb{P}}

\def\E{\mathbb{E}}

\def\11{\mathbbm{1}}

\def\ER{Erd\H{o}s-R\'enyi\ }

\def\Eb{\mathbb E}

\newcommand{\Pb}{\mathbb P}
\newcommand{\Qb}{\mathbb Q}

\newtheorem{thm}{Theorem}[section]

\newtheorem{proposition}[thm]{Proposition}

\newtheorem{lemma}[thm]{Lemma}
\newtheorem{cor}[thm]{Corollary}

\newtheorem{remark}[thm]{Remark}
\newtheorem{DEF}[thm]{Definition}

\numberwithin{equation}{section}

\makeatletter
\newenvironment{breakablealgorithm}
{
		\begin{center}
			\refstepcounter{algorithm}
			\hrule height.8pt depth0pt \kern2pt
			\renewcommand{\caption}[2][\relax]{
				{\raggedright\textbf{\ALG@name~\thealgorithm} ##2\par}%
				\ifx\relax##1\relax 
				\addcontentsline{loa}{algorithm}{\protect\numberline{\thealgorithm}##2}%
				\else 
				\addcontentsline{loa}{algorithm}{\protect\numberline{\thealgorithm}##1}%
				\fi
				\kern2pt\hrule\kern2pt
			}
		}{
		\kern2pt\hrule\relax
	\end{center}
}
\makeatother

\hypersetup{colorlinks=true,linkcolor=blue}

\title{Detecting Correlation Efficiently in Stochastic Block Models: Breaking Otter's Threshold in the Entire Supercritical Regime}

\usepackage{authblk}

\author[1]{Guanyi Chen}
\author[1]{Jian Ding\thanks{Partially supported by the National Key R$\&$D program of China (Project No. 2023YFA1010103), the NSFC Key Program (Project No. 12231002) and the New Cornerstone Science Foundation through the XPLORER PRIZE.}}
\author[1]{Shuyang Gong\thanks{Partially supported by the National Key R\&D program of China (Project No. 2020YFA0712902). Part of the work was carried out when S. Gong was visiting Duke University. S. Gong would like to thank Yihong Wu and Jiaming Xu for their kind help and warm guidance during the visit.}}
\author[1]{Zhangsong Li}

\affil[1]{School of Mathematical Sciences, Peking University}

\date{\today}
\begin{document}
\maketitle

\begin{abstract}
    Consider a pair of sparse correlated stochastic block models $\mathcal S(n,\tfrac{\lambda}{n},\epsilon;s)$ subsampled from a common parent stochastic block model with two symmetric communities, average degree $\lambda=O(1)$, divergence parameter $\epsilon\in (0,1)$ and subsampling probability $s$. For all $\epsilon\in(0,1)$ and $\Delta>0$, we construct a statistic based on the combination of two low-degree polynomials and show that there exists a sufficiently small constant $\delta=\delta(\epsilon,\lambda,\Delta)>0$ such that if $\epsilon^2 \lambda s>1+\Delta$ and $s>\sqrt{\alpha}-\delta$ where $\alpha\approx 0.338$ is Otter's constant, this statistic can distinguish this model and a pair of independent stochastic block models $\mathcal S(n,\tfrac{\lambda s}{n},\epsilon)$ with probability $1-o(1)$. We also provide an efficient algorithm that approximates this statistic in polynomial time. 
    
    The crux of our statistic's construction lies in a carefully curated family of multigraphs called \emph{decorated trees}, which enables effective aggregation of the community signal and graph correlation by leveraging the counts of the same decorated tree while suppressing the undesirable correlations among counts of different decorated trees. We believe such construction may be of independent interest.
\end{abstract}

\tableofcontents

\section{Introduction}{\label{sec:intro}}

Graph matching (also known as graph alignment) seeks a correspondence between the vertices of two graphs that maximizes the total number of common edges. When the two graphs are exactly isomorphic to each other, this problem reduces to the classical graph isomorphism problem, for which the best known algorithm runs in quasi-polynomial time \cite{Babai16}. In general, graph matching is an example of the \emph{quadratic assignment problem} \cite{BCPP98}, which was known to be NP-hard to solve or even to approximate \cite{MMS10}. 

Motivated by applications in various applied areas including computational biology \cite{SXB08}, social networks \cite{NS08, NS09}, and natural language processing \cite{HNM05}, a recent line of work is devoted to the study of statistical theory and efficient algorithms for graph matching under statistical models, assuming the two graphs are randomly generated with correlated edges under a hidden vertex correspondence. From a theoretical perspective, perhaps the most studied setting is the correlated \ER graph model \cite{PG11}, in which two correlated \ER graphs are observed, with edges correlated through latent vertex bijection $\pi_*$. Two important and entangling issues for this model, i.e., the information-theoretic threshold and the computational phase transition, have been extensively studied recently. Thanks to collective efforts from the community, we now have obtained rather satisfying understanding on information-theoretic thresholds for the problems of correlation detection and vertex matching \cite{CK16, CK17, HM23, GML21, WXY22, WXY23, DD23a, DD23b, Du25+}. On the side of computation, in extensive works including \cite{BCL+19, DMWX21, FMWX23a, FMWX23b, BSH19, DCK+19, MX20, GM20, GML20+, MRT21, MRT23, MWXY21+, GMS22+, MWXY23, DL22+, DL23+}, substantial progress on algorithms was achieved, which may be roughly summarized as follows: in the sparse regime, efficient matching algorithms were designed when the correlation exceeds the square root of Otter’s constant (the Otter's constant is approximately 0.338) \cite{GML20+, GMS22+, MWXY21+, MWXY23}; in the dense regime, efficient matching algorithms were proposed when the correlation exceeds an arbitrarily small constant \cite{DL22+, DL23+}. The separation between the sparse and dense regimes above, roughly speaking, depends on whether the average degree grows polynomially or sub-polynomially. In addition, another important direction is to establish lower bounds for computation complexity: while it is challenging to prove hardness for a typical instance of the problem even under the assumption of P$\neq$NP, in a recent work \cite{DDL23+} evidence based on the analysis of low-degree polynomials indicates that the aforementioned algorithms may essentially capture the correct computational thresholds. 

Since our current understanding for the information-computation transition of the correlated \ER model is more or less satisfactory assuming belief in the low-degree hardness conjecture, an important future direction is to understand the matching problem in other important correlated random graph models. This research direction is motivated by an important question that whether \emph{additional structural information} in pairs of graphs will enhance algorithmic performance and, consequently, affect the computation threshold. We believe this question holds substantial importance from both applied and theoretical perspectives. On the one hand, realistic networks are usually captured better by more structured graph models, such as the random geometric graph model \cite{WWXY22, GL24+}, the random growing graph model \cite{RS22} and the stochastic block model. On the other hand, the study of algorithmic threshold on structural models enables us to better characterize the \emph{algorithmic universality class} of the correlated \ER model--that is, to identify which models exhibit algorithmic behaviors similar to the correlated \ER model.

In this paper, we consider a pair of correlated sparse stochastic block models (SBMs) that are subsampled from a common parent stochastic block model with two symmetric communities, average degree $\lambda=O(1)$, divergence parameter $\epsilon$ and subsampling probability $s$. Specifically, denoting by $\operatorname{U}_n$ the set of unordered pairs $(i,j)$ where $1\le i<j\le n$, we can define this model as follows. 

\begin{DEF}[Stochastic block model] {\label{def-SBM}}
    Given an integer $n\ge 1$ and two parameters $\lambda>0, \epsilon \in (0,1)$, we define a random graph $G$ on $[n]=\{ 1,\ldots,n \}$ as follows. First, we select a labeling $\sigma_* \in \{ -1,+1 \}^{n}$ uniformly at random. For each distinct pair $(i,j) \in \operatorname{U}_n$, we independently add an edge $(i,j)$ with probability $\frac{(1+\epsilon)\lambda}{n}$ if $\sigma_*(i) = \sigma_*(j)$, and with probability $\frac{(1-\epsilon)\lambda}{n}$ if $\sigma_*(i) \neq \sigma_*(j)$. We denote $\mathcal{S}(n,\frac{\lambda}{n},\epsilon)$ to be the law of $G$.
\end{DEF}

\begin{DEF}[Correlated stochastic block model] {\label{def-correlated-SBM}}
    Given an integer $n\ge 1$ and three parameters $\lambda>0, \epsilon,s \in (0,1)$, let $J_{i,j}$ and $K_{i,j}$ be independent Bernoulli variables with parameter $s$ for $(i,j)\in \operatorname{U}_n$. In addition, let $\pi_*$ be an independent uniform permutation on $[n]=\{1,\dots,n\}$. Then, we define a triple of correlated random graphs $(G,A,B)$ such that $G$ is sampled from a stochastic block model $\mathcal{S}(n,\tfrac{\lambda}{n},\epsilon)$, and conditioned on $G$ (note that we identify a graph with its adjacency matrix)
    \begin{align*}
        A_{i,j} = G_{i,j}J_{i,j} \,, \quad B_{i,j} = G_{\pi_*^{-1}(i),\pi_*^{-1}(j)} K_{i,j} \,.
    \end{align*}
    We denote the joint probability of $(\pi_*,\sigma_*,G,A,B)$ by $\Pb_{*,n}:=\Pb_{*,n,\lambda,\epsilon,s}$, and the marginal probability of $(A,B)$ by $\Pb_n:=\Pb_{n,\lambda,\epsilon,s}$. We use $\mathcal{S}(n,\tfrac{\lambda}{n},\epsilon;s)$ to denote distribution of the pair $(A,B)$.
\end{DEF}

Given two graphs $(A,B)$ on the vertex set $[n]$, our goal is to study the following hypothesis testing problem: determine whether $(A,B)$ is sampled from $\mathbb P_{n}$ or $\mathbb Q_n$, where $\mathbb Q_n$ is the distribution of two independent stochastic block models $\mathcal S(n,\tfrac{\lambda s}{n},\epsilon)$. In our previous work \cite{CDGL24+} (where we considered symmetric SBMs with $k$ communities for $k\geq 2$), we considered the problem of testing $\Pb_n$ against $\widetilde{\Qb}_n$, where $\widetilde{\Qb}_n$ is the law of two independent \ER graphs $\mathcal G(n, \tfrac{\lambda s}{n})$. We showed that in the subcritical regime $\epsilon^2 \lambda s<1$ where community recovery in a single graph is (information-theoretically) impossible, there is evidence suggesting that no efficient algorithm can distinguish $\mathbb P_n$ and $\widetilde{\Qb}_n$ provided $s<\sqrt{\alpha}$, where $\alpha$ is Otter's threshold. In addition, \cite[Theorem~2]{MNS15} implies that when $k=2$ (as we focus on in this paper) $\widetilde{\Qb}_n$ and $\Qb_n$ are mutually contiguous in the subcritical regime (although their result focuses on a single graph, extending it to two \emph{independent} graphs is straightforward). Altogether, we see that there are evidences suggesting that no efficient algorithm can distinguish any two measures from $\{ \mathbb P_n,\mathbb Q_n,\widetilde{\mathbb Q}_n \}$ as long as $s<\sqrt{\alpha}$. This aligns with the lower bound established in \cite{DDL23+} for the problem of distinguishing a pair of correlated \ER graphs from a pair of independent \ER graphs. Thus, we see that in the subcritical regime, this model belongs to the same algorithmic universality class with the correlated \ER model. In stark contrast, in this paper, we focus on the supercritical regime $\epsilon^2 \lambda s >1$ where efficient community detection and (weak) recovery in each individual graph are possible. Our main result stated below suggests that the testing problem is strictly easier than that in the subcritical regime in the algorithmic sense.

\begin{thm}{\label{MAIN-THM}}
    Suppose $\epsilon\in(0,1)$ and $\Delta>0$ are two constants. There exist two positive constants $\delta=\delta(\epsilon , \Delta)$ and $C=C(\epsilon,\delta,\Delta)$ such that when $\epsilon^2\lambda s = 1 + \Delta$ and $s>\sqrt{\alpha}-\delta$ where $\alpha \approx 0.338$ is Otter's threshold, there exists a statistic $f=f(A,B)$ that can distinguish $\mathbb P_n$ and $\mathbb Q_n$ with probability tending to $1$ as $n\to\infty$. In addition, this statistic $f$ is a polynomial of degree $\omega(\log n)$ that can be approximately computed in time $O(n^{C})$, where the approximation can also distinguish $\mathbb P_n$ and $\mathbb Q_n$ with probability tending to $1$.
\end{thm}

\begin{remark}
    Our result suggests that when the signal for community structure of a single graph is strong enough such that (weak) community recovery in a single graph is possible on a positive fraction of vertices, the recovered community structure can be leveraged to enhance the task of testing network correlation. We would like to emphasize that this is by all means a highly nontrivial task, and we refer to Sections~\ref{subsec:innovations} and \ref{subsec:discussions} for elaborative discussions on the challenges. 
\end{remark}

\begin{remark}
    In \cite{CR24}, the authors seem to tend to feel that in the logarithmic degree regime the correlated SBMs belong to the same algorithmic universality class of the sparse correlated \ER model, and thus all inference tasks are computationally impossible when $s<\sqrt{\alpha}$. However, our result suggests a different possibility by demonstrating a new phase transition phenomenon in this model in the constant-degree regime.
\end{remark}

\subsection{Other related works}{\label{subsec:backgrounds}}

$\textbf{Community recovery in sparse stochastic block models.}$ The stochastic block model (SBM) is a fundamental probabilistic generative model for networks with community structure. Introduced in \cite{HLL83}, it has garnered significant attention in recent decades. In particular, it serves as a natural theoretical testbed for evaluating and comparing clustering algorithms on average-case networks (see, e.g., \cite{DF89, BCL+84, Boppana87}). The SBM provides a framework for understanding when community signal can be extracted from network data, as it exhibits sharp informational and computational phase transitions for various inference tasks. These phase transitions were first conjectured by \cite{DKMZ11} and were later rigorously established in \cite{Massoulie14, MNS18, BLM15, BDG+16, AS16, AS18, BMNN16, AS18, CS21+, MSS23, MSS24+}. 

Here, we focus on the SBM with two symmetric communities, arguably the simplest setting. In this setting, a community recovery algorithm takes as input the graph $G \sim \mathcal S(n,\tfrac{\lambda}{n},\epsilon)$ without knowledge of the underlying community labeling $\sigma_*$ (the ground truth), and outputs an estimated community labeling $\widehat{\sigma}$. The success of an algorithm is measured by the overlap between the estimated community labeling and the ground truth, defined as $\mathsf{Overlap}(\sigma_*,\widehat \sigma) := \frac{1}{n}| \sum_{i=1}^{n} \sigma_*(i) \widehat{\sigma}(i)|$ (we take the absolute value in this formula since the labelings $\sigma_*$ and $-\sigma_*$ represent the same partition of communities, and therefore it is only possible to recover $\sigma_*$ up to a global flip). We say that an estimator $\widehat{\sigma}$ achieves \emph{weak recovery} if $\mathsf{Overlap}(\sigma_*,\widehat \sigma)$ stays uniformly away from 0 as $n\to \infty$. For the stochastic block model with two symmetric communities, it is well known that in the constant-degree regime there is a sharp information-theoretic and computational threshold for weak recovery in the SBM \cite{MNS15, MNS18, Massoulie14} at $\epsilon^2 \lambda = 1$. This threshold is called the Kesten-Stigum (KS) threshold, which was first discovered in the context of the broadcasting process on trees \cite{KS66}. We also point out that when the model involves $k$ communities for $k\geq 2$, the situation becomes more sophisticated. For instance, for $k \geq 5$, the information-theoretic threshold drops below the KS threshold \cite{MSS23}, while evidence suggests that the computational threshold remains unchanged \cite{HS17}, indicating the emergence of a statistical-computational gap.

$\textbf{Correlated stochastic block models.}$ The goal of our work is to explore how the network correlation detection threshold at $\sqrt{\alpha}$ is influenced when input data consists of multiple \emph{correlated} SBMs. First studied in \cite{OGE16}, this model of correlated SBMs is a natural generalization of correlated \ER random graphs, where in particular each marginal \ER graph is now replaced by a stochastic block model $\mathcal S(n,\tfrac{\lambda s}{n},\epsilon)$. This model was then studied in \cite{RS21, GRS22, YSC23, CR24, CDGL24+} and in particular \cite{CDGL24+} provided evidence that in the subcritical regime $\epsilon^2 \lambda s<1$, the computational correlation detection threshold is still $\sqrt{\alpha}$.

However, as shown in Theorem~\ref{MAIN-THM}, we show that in the supercritical regime (i.e., when the ``community signal strength'' is strictly above the Kesten-Stigum threshold), the correlation detection threshold is indeed influenced by the community structure and thus is strictly below Otter's threshold. 
Intuitively, this effect arises because knowing the community labels of $A$ and $B$ restricts the latent matching $\pi_*$: it can only pair each vertex in $A$ with another vertex in $B$ sharing the same label.

$\textbf{Low-degree polynomial algorithms.}$ The main idea behind our algorithm is to construct a polynomial that is a \emph{combination} of two low-degree polynomials, one is from counting self-avoiding/non-backtracking paths (which was widely used in efficient community recovery \cite{MNS16, MNS18, HS17}) and the other is from counting trees (which was used in \cite{MWXY21+, MWXY23} for correlation detection/graph matching). Inspired by the sum-of-squares hierarchy, the low-degree polynomial method offers a promising framework for establishing computational lower bounds in high-dimensional inference problems \cite{BHK+19, HS17, HKP+17, Hopkins18}. Roughly speaking, to study the computational efficiency of the hypothesis testing problem between two probability measures $\Pb$ and $\Qb$, our idea is to construct a low-degree polynomial $f$ (usually with degree $O(\log n)$ where $n$ is the size of data) that separates $\Pb$ and $\Qb$ in a certain sense. The key quantity is the low-degree advantage
\begin{equation}{\label{eq-def-SNR}}
    \mathsf{Adv}(f) := \frac{\mathbb E_{\Pb}[f]}{ \sqrt{ \mathbb E_{\Qb}[f^2] } } \,.
\end{equation}
The low-degree polynomial method suggests that, if the low-degree advantage of $f$ is uniformly bounded for all $f$ with degree $O(\log n)$, then no polynomial-time algorithm can strongly distinguish $\Pb$ and $\Qb$. This approach is appealing partly because it has yielded tight hardness results for a wide range of problems. Prominent examples include detection and recovery problems such as planted clique, planted dense subgraph, community detection, sparse-PCA and tensor-PCA (see \cite{HS17, HKP+17, Hopkins18, KWB22, SW22, DMW23+, BKW19, MW21+, DKW22, BAH+22, MW22, DMW23+, KMW24+}), optimization problems such as maximal independent sets in sparse random graphs \cite{GJW20, Wein22}, and constraint satisfaction problems such as random $k$-SAT \cite{BH22}. 

Another appealing feature of low-degree polynomial method is its ability to provide valuable guidance for designing efficient algorithms. Specifically, if there exists a polynomial $\mathrm{deg}(f)=O(\log n)$ such that $\mathsf{Adv}(f)=\omega(1)$, then it strongly suggests the existence of an efficient algorithm that can distinguish $\Pb$ and $\Qb$ using $f$. Since this approach has proven successful in both community recovery \cite{HS17} and graph matching \cite{MWXY21+, MWXY23}, it is tempting to find a low-degree polynomial with a large low-degree advantage. However, our approach is more subtle as the polynomial we construct is a \emph{combination} of two low-degree polynomials: one for community recovery and the other for correlation detection. As a result, unlike in the standard low-degree framework, the polynomial we construct has a degree of $\omega(\log n)$ but can still be computed efficiently.

\subsection{Key challenges and algorithmic innovations}{\label{subsec:innovations}}

The major algorithmic innovation of this work is a new construction based on \emph{subgraph counts}. To this end, we first briefly explain how the threshold $\sqrt{\alpha}$ emerges in the detection problem in correlated \ER models. Suppose that marginally both $A$ and $B$ are \ER graphs $\mathcal G(n,\tfrac{\lambda s}{n})$ and their edge correlation is given by $\operatorname{Cov}(A_{\pi_*(i),\pi_*(j)},B_{i,j})=[1+o(1)] s$. A natural attempt is to count the (centered) graphs $\mathbf H$ in both $A$ and $B$ for each unlabeled graph $\mathbf H$, i.e., we consider the statistic
\begin{align*}
    g_{\mathbf H} = \Big( \sum_{ S \cong \mathbf H } \prod_{(i,j) \in E(S)} (A_{i,j}-\tfrac{\lambda s}{n}) \Big) \Big( \sum_{ K \cong \mathbf H } \prod_{(i,j) \in E(K)} (B_{i,j}-\tfrac{\lambda s}{n}) \Big) \,,
\end{align*}
where $S \cong \mathbf H$ denotes that we will sum over all subgraphs of the complete graph $\mathsf K_n$ that are isomorphic to $\mathbf H$, and $E(S)$ is the edge set of $S$. A standard calculation yields that (under the law of correlated \ER graphs) the expectation of $g_{\mathbf H}$ with suitable standardization is given by $s^{|E(\mathbf H)|}$. In addition, the signal contained in each $g_{\mathbf H}$ is approximately non-repeating, which implies that $g_{\mathbf H}$'s constitute a family of ``almost orthogonal'' statistics. Thus, if we take our statistic to be a suitable linear combination of $\{ g_{\mathbf H}: \mathbf H \in \mathcal H \}$ where $\mathcal H$ is a family of unlabeled graphs, the total signal-to-noise ratio is given by
\begin{align*}
    \sum_{ \mathbf H \in \mathcal H } s^{2|E(\mathbf H)|} \,.
\end{align*}
The key idea is that since marginally $A$ and $B$ are sparse graphs, we expect a ``proper'' choice of $\mathcal H$ should guarantee that all $\mathbf H \in \mathcal H$ has low edge density, and thus it can be shown that under this restriction, the cardinality of $|\mathcal H|$ should be approximately the same as the cardinality of unlabeled trees, which grows at rate $\alpha^{-1}$ \cite{Otter48}. This informally suggests that $s=\sqrt{\alpha}$ is the ``correct'' threshold of the detection problem in correlated \ER graphs.  

We now focus on our case where marginally $A$ and $B$ are supercritical SBMs such that weak community recovery in $A$ and $B$ is possible. The initial motivation of our result is the following (albeit simple) observation: assuming that all the community labelings in $A$ and $B$ (we denote them as $\sigma_A$ and $\sigma_B$) are known to us, we can count a larger class $\mathcal H'$ which consists of all unlabeled trees \emph{decorated with $\{ -1,+1 \}$} labels. This is because when counting subgraphs $S,K$ in $A,B$ respectively we can regard the community labels $\sigma_{A},\sigma_{B}$ as a ``decoration'' of $S,K$, and the information from the community labels guarantees that if $S$ and $K$ are isomorphic graphs but have non-isomorphic label decorations, the signal contained in counting $S$ and $K$ are nearly orthogonal. Thus, a naive attempt is to first run the weak community recovery algorithm in both $A$ and $B$ which produces an estimator $(\widehat{\sigma}_A, \widehat{\sigma}_B)$ of $(\sigma_A,\sigma_B)$, and then counting $\mathcal H'$ in $A$ and $B$ simply by viewing $(\widehat{\sigma}_A, \widehat{\sigma}_B)$ as its ``decorations''. However, there are several obstacles that impede the rigorous analysis of this naive approach. Firstly, the subgraph counting procedure on $\mathcal H'$ is highly fragile and thus we cannot simply ignore those vertices $i$ such that $\widehat{\sigma}_A(i) \neq \sigma_A(i)$ or $\widehat{\sigma}_B(i) \neq \sigma_B(i)$ (which constitutes a positive fraction of all vertices). Moreover, the community recovery step has a strong influence on later subgraph counts, which means that $(\widehat{\sigma}_A, \widehat{\sigma}_B)$ is strongly correlated with $g_{\mathbf H}(A,B)$. In fact, it seems that conditioning on $(\widehat{\sigma}_A, \widehat{\sigma}_B)$ forces the enumerations of certain types of non-isomorphic graphs to have correlations tending to $1$ with each other. 

To address these issues, rather than performing the two procedures in succession, we propose a unified approach that runs two procedures \emph{simultaneously}. The key conceptual innovation in our work is to construct a family of subgraphs that \emph{combines} the subgraph counts relevant to both tasks. Specifically, our approach involves counting a carefully chosen family of unlabeled multigraphs, which (informally speaking) is formed by attaching non-backtracking paths to an unlabeled tree; the flexibility to choose the place of attachment enriches the enumeration of such multigraphs and helps us to gain an extra exponential factor compared to the number of all unlabeled trees. Furthermore, we can approximate the count of such multigraphs with sufficient accuracy in polynomial time by leveraging the method of color coding \cite{AYZ95, HS17}. The major difficulty lies in analyzing the behavior of our statistic. While centering the adjacency matrices and counting signed subgraphs are helpful, we still have excessive correlations among different subgraph counts and these correlations are difficult to control. To address this challenge, a further innovation of this work is to count a special family $\mathcal H$ of all such multigraphs, which we call \emph{decorated trees}; see Definition~\ref{def-decorated-trees} for the formal definition. As discussed in Section~\ref{subsec:discussions}, this choice plays a crucial role in curbing the undesired correlation between different subgraph counts. Moreover, despite the fact that we pose several structural requirements on such decorated trees, by choosing the parameters appropriately, we can nevertheless ensure that the number of such decorated trees still grows exponentially faster than that of unlabeled trees. 

Finally, we point out that our work is significantly more complicated than the detection and matching algorithms for correlated \ER graphs \cite{MWXY21+, MWXY23} and their extension to correlated SBMs \cite{CR24} (which are based on counting a specific family of unlabeled trees) at a technical level. Indeed, a key simplification in \cite{MWXY21+, MWXY23} comes from the fact that the edges within each graph are (marginally) independent, which allows for straightforward cancellations in many parts of their analysis. In contrast, our setting involves edges that are correlated through latent community labels. Thus, we need to carefully analyze the correlation from the latent matching and the latent community labels simultaneously. The authors of~\cite{CR24} made initial progress on this challenge; however, their analysis focus on the logarithmic degree regime, where the strong community signal enables a de-correlation technique on tree counts. In our constant-degree regime, the community signal is inherently weaker, necessitating a far more delicate approach. In addition, as we hope to break the tree-counting threshold, our key innovation is to construct a family of unlabeled graphs that grow faster than unlabeled trees and allow us to leverage the community structure more effectively. Our choice of decorated trees is substantially more intricate than their choice of trees, and our decorated trees have size $\omega(\log n)$, which introduces a significantly larger number of potential overlapping structures between different decorated trees. Handling these overlapping structures poses a major technical challenge, as it requires highly delicate treatment to control their enumeration and  correlation. We provide a brief overview of how we control such overlapping structures in Section~\ref{subsec:discussions}.

\subsection{Discussions and perspectives}{\label{subsec:open-probs}}

Our work reiterates a number of future research directions as we discuss below.

\underline{Extension to general block models.} Although we believe that our results and many of our arguments can be extended to the more general setting of correlated stochastic block models with $k>2$ communities, certain parts of our proof rely crucially on the assumption of a two-community model. We expect that establishing analogous results in the general case would require substantial additional effort and thus we leave this direction for future work.

\underline{The partial recovery problem.} A natural next step is to achieve partial recovery below the Otter threshold by combining our approach with the methods of \cite{MWXY23} (which solves the partial recover problem in correlated \ER graphs by counting a specific family of unlabeled trees called chandeliers). A promising strategy would be to adapt these methods by counting decorated chandeliers in the correlated SBM setting. However, additional technical challenges arise in the analysis of this approach due to the more complicated overlapping structures, so we leave this goal for future work.

\underline{Settling the sharp computational threshold.} We conjecture that, at least for the correlated SBMs with two symmetric communities $\mathcal S(n,\tfrac{\lambda}{n};\epsilon,s)$, the exact computational threshold depends only on $\kappa$, where $\kappa=\kappa(\lambda,\epsilon,s)$ is the optimal fraction of vertices whose community labels can be recovered in a single SBM $\mathcal S(n,\tfrac{\lambda s}{n};\epsilon)$.\footnote{We note that the precise value of $\kappa$ was determined only for sufficiently large $\lambda$ \cite{MNS16}.} Proving this conjecture seems to require ideas beyond the scope of our current work, so we also leave it for future work.

\subsection{Notation and paper organization}{\label{subsec:notations}}

In this subsection, we record a list of notations that we shall use throughout the paper. Let $\mathrm{Ber}(p)$ be the Bernoulli distribution with parameter $p$, and let $\mathrm{Binom}(n,p)$ be the binomial distribution with $n$ trials and success probability $p$. Recall that $\Pb=\Pb_n,\Qb=\Qb_n$ are two probability measures on a pair of random graphs on $[n]=\{1,\ldots,n\}$. Denote $\mathfrak{S}_n$ the set of permutations in $[n]$ and $\mu=\mu_n$ the uniform distribution on $\mathfrak S_n$. In addition, denote by $\nu=\nu_n$ the uniform distribution on $\{-1,+1\}^n$. For two probability measures $\mu$ and $\nu$, we define $\mathsf{TV}(\mu,\nu)$ to be their total variation distance. We will use the following notation conventions for graphs.
\begin{itemize}
    \item {\em Labeled graphs}. Denote by $\mathsf K_n$ the complete graph with vertex set $[n]$ and edge set $\operatorname{U}_n$. For any graph $H$, let $V(H)$ denote the vertex set of $H$ and let $E(H)$ denote the edge set of $H$. For $v \in V(H)$, denote $\mathsf{Deg}_H(v)$ as the degree of $v$ in $H$, i.e. the number of edges connected to $v$ in $H$. We say $H$ is a subgraph of $G$, denoted by $H\subset G$, if $V(H) \subset V(G)$ and $E(H) \subset E(G)$. We say $\varphi:V(H) \to V(S)$ is an injection, if for all $(i,j) \in E(H)$ we have $(\varphi(i),\varphi(j)) \in E(S)$. For $H,S \subset \mathsf K_n$, denote by $H \cap S$ the graph with vertex set given by $V(H) \cap V(S)$ and edge set given by $E(H)\cap E(S)$, and denote by $S \cup H$ the graph with vertex set given by $V(H) \cup V(S)$ and edge set given by $E(H) \cup E(S)$. Let $H \triangle S$ be the graph induced by all the edges $E(H) \triangle E(S)$. For any graph $H$, denote the excess of $H$ by $\tau(H)=|E(H)|-|V(H)|$. Given $u \in V(H)$, define $\mathsf{Nei}_H(u)$ to be the set of neighbors of $u$ in $H$. For two vertices $u,v\in V(H)$, we define $\mathsf{Dist}_{H}(u,v)$ to be the graph distance. Denote the diameter of a connected graph $H$ by $\mathsf{Diam}(H)=\max_{u,v \in V(H)} \mathsf{Dist}_H(u,v)$. 
    \item {\em Graph isomorphisms and unlabeled graphs.} Two graphs $H$ and $H'$ are isomorphic, denoted by $H\cong H'$, if there exists a bijection $\pi:V(H) \to V(H')$ such that $(\pi(u),\pi(v)) \in E(H')$ if and only if $(u,v)\in E(H)$. Denote by $\mathcal H$ the isomorphism class of graphs; it is customary to refer to these isomorphic classes as unlabeled graphs. Let $\operatorname{Aut}(H)$ be the number of automorphisms of $H$ (graph isomorphisms to itself). For any graph $H$, define $\mathsf{Fix}(H)=\{ u \in V(H): \varphi(u)=u, \forall \varphi \in \operatorname{Aut}(H) \}$.
    \item {\em Vertex induced subgraphs}. For a graph $H=(V,E)$ and a subset $A \subset V$, define $H_A =(A,E_A)$ to be the vertex induced subgraph of $H$ in $A$, where $E_A = \{ (u,v) \in E: u,v \in A \}$. Also, define $H_{\setminus A} = (V,E_{\setminus A})$ to be the subgraph of $H$ obtained by deleting all edges with both endpoints in $A$. Note that $E_A \cup E_{\setminus A} = E$.
    \item {\em Isolated vertices}. For $u \in V(H)$, we say $u$ is an isolated vertex of $H$ if there is no edge in $E(H)$ incident to $u$. Denote $\mathcal I(H)$ as the set of isolated vertices of $H$.
    \item {\em Paths, self-avoiding paths and non-backtracking paths.} We say a subgraph $H \subset \mathsf K_n$ is a path with endpoints $u,v$ (possibly with $u=v$), if there exist $w_1, \ldots, w_m \in [n] \neq u,v$ such that $V(H)=\{ u,v,w_1,\ldots,w_m \}$ and $E(H)=\{ (u,w_1), (w_1,w_2) \ldots, (w_m,v) \}$ (we allow the occurrence of multiple vertices or edges). We say $H$ is a self-avoiding path if $w_0, w_1,\ldots,w_m,w_{m+1}$ are distinct (where we denote $w_0=u$ and $w_{m+1}=v$), and we say $H$ is a non-backtracking path if $w_{i+1} \neq w_{i-1}$ for $1 \leq i \leq m$. Denote $\operatorname{EndP}(P)$ as the set of endpoints of a path $P$. 
    \item {\em Cycles and independent cycles.} We say a subgraph $H$ is an $m$-cycle if $$V(H)=\{ v_1, \ldots, v_m \}\,,E(H)=\{ (v_1,v_2), \ldots, (v_{m-1},v_m), (v_m,v_1) \}.$$ 
    For a subgraph $K \subset \mathsf K_n$, we denote $\mathsf{Cycle}(K)$ as the set of cycles in $K$. For a subgraph $K\subset H$, we say $K$ is an independent cycle of $H$, if $K$ is a cycle and no edge in $E(H)\setminus E(K)$ is incident to $V(K)$. Denote $\mathsf{Cycle}_{\mathsf{ind}}(H)$ be the set of independent cycles in $H$.
    \item {\em Leaves.} A vertex $u \in V(H)$ is called a leaf of $H$, if the degree of $u$ in $H$ is $1$; denote $\mathcal L(H)$ as the set of leaves of $H$. 
    \item {\em Trees and unlabeled rooted trees}. We say a graph $\mathbf T=(V(\mathbf T),E(\mathbf T))$ is a tree, if $\mathbf T$ is connected and has no cycles. We say a pair $(\mathbf T,\mathfrak{R}(\mathbf T))$ is an unlabeled rooted tree with root $\mathfrak R(\mathbf T)$, if $\mathbf T$ is a tree and $\mathfrak R(\mathbf T) \in V(\mathbf T)$. For an unlabeled rooted tree $\mathbf T$ and $u\in V(\mathbf T)$, we define $\mathsf{Dep}_\mathbf T(u)=\mathsf{Dist}_\mathbf T(\mathfrak R(\mathbf T),u)$ to be the depth of $u$ in $\mathbf T$, and define $\mathsf{Depth}(\mathbf T) = \max_{u \in V(\mathbf T)} \mathsf{Dep}_\mathbf T(u)$. For $u,v \in V(\mathbf T)$, denote by $u \rightarrow v$ (or equivalently, $v \leftarrow u$) if $v$ is a child of $u$, and denote by $u \hookrightarrow v$ (or equivalently, $v \hookleftarrow u$) if $v$ is a descendant of $u$. In addition, denote $\mathfrak p_{\mathbf T}(u,v)$ as the shortest path between $u$ and $v$ in $\mathbf T$. For an unlabeled rooted tree $(T,\mathfrak R(\mathbf T))$ and $u \in V(\mathbf T)$, define $\mathsf{Branch}_\mathbf T(u) = \mathsf{Deg}_\mathbf T(u) - \mathbf{1}_{\{ u \neq \mathfrak R(\mathbf T)\}}$. For convenience, denote $\mathcal L'(\mathbf T) = \mathcal L(\mathbf T) \setminus \{\mathfrak R(\mathbf T)\}$. 
    \item {\em Descendant tree.} For an unlabeled rooted tree $(\mathbf T,\mathfrak R(\mathbf T))$ and $u \in V(\mathbf T)$, denote $\mathsf{Des}_\mathbf T(u)$ the descendant tree of $u$ in $\mathbf T$, that is, the unlabeled rooted tree with root $u$ and with vertices given by the descendants of $u$ in the unlabeled rooted tree $\mathbf T$. We call $\mathsf{Des}_\mathbf T(u)$ the descendant tree of $u$ in $\mathbf T$. For unlabeled rooted trees $\mathbf T,\mathbf T'$ we write $\mathbf T \hookrightarrow \mathbf T'$ (or equivalently, $\mathbf T' \hookleftarrow \mathbf T$) if $\mathbf T'$ is a descendant tree of $\mathbf T$.
    \item {\em Subtrees of the root.} For an unlabeled rooted tree $(\mathbf T,\mathfrak R(\mathbf T))$, denote 
    \begin{align}\label{def-sub-on-root}
         \mathsf{Subs}(\mathbf T) = \big\{ \mathsf{Des}_\mathbf T(u): u \in \mathsf{Nei}_\mathbf T(\mathfrak R(\mathbf T)) \big\} \,.
    \end{align}
    as the set of subtrees of the root $\mathfrak R(\mathbf T)$. Also denote
    \begin{align}\label{def-sub-on-root-size-i}
         \mathsf{Subs}_i(\mathbf T) = \big\{ \mathbf T' \in \mathsf{Subs}(\mathbf T): |V(\mathbf T')| = i \big\} \,.
    \end{align}
    \item {\em Arm-paths and arm-trees.} For an unlabeled rooted tree $(\mathbf T,\mathfrak R(\mathbf T))$, we say a self-avoiding path $P=(u_1,\ldots,u_m)$ with $m \geq 2$ is an arm-path of $\mathbf T$ starting from $u_1$ with length $m-1$, if $P \subset \mathbf T$, $\mathfrak R(\mathbf T) \not\in \{ u_2 , \ldots , u_m \}$ and $\mathsf{Deg}_\mathbf T(u_i) = 2 - \mathbf{1}_{i\in\{1, m\}}$ for $1 \leq i \leq m$. For $N \geq 1$, let 
    \begin{equation}
        \label{eq-AN}
         \mathbf A_N =\mbox{the unlabeled rooted tree isomorphic to a path with } N-1 \mbox{ edges} \,.
    \end{equation}
    Given an unlabeled rooted tree $\mathbf T$ and $\mathbf u \in V(\mathbf T)$, denote by $\mathbf T \oplus \mathcal L_{\mathbf u}^{(x)}$ the graph obtained by sticking an arm-path with length $x$ starting from $\mathbf u$. Also, for an unlabeled rooted tree $(\mathbf T,\mathfrak R(\mathbf T))$ define
    \begin{align}
        &\mathsf{ASubs}(\mathbf T) = \big\{ \mathbf T' \in \mathsf{Subs}(\mathbf T): \mathbf T' \mbox{ is an arm-path of } \mathbf T \big\} \,, \label{def-armsub-on-root} \\
        &\mathsf{ASubs}_i(\mathbf T) = \big\{ \mathbf T' \in \mathsf{ASubs}(\mathbf T): |V(\mathbf T')| = i \big\} \,. \label{def-armsub-on-root-size-i}
    \end{align}
    Also, we define
    \begin{equation}
        \mathsf{Subs}^*(\mathbf T) = \mathsf{Subs}(\mathbf T) \setminus \mathsf{ASubs}(\mathbf T)\,, \quad \mathsf{Subs}_i^*(\mathbf T) = \mathsf{Subs}_i(\mathbf T) \setminus \mathsf{ASubs}_i(\mathbf T)\,.\label{def-sub-except-arm-on-rooted-trees}
    \end{equation}
    \item {\em Arm candidates.} For any unlabeled rooted tree $(\mathbf T,\mathfrak R(\mathbf T))$ and $k \geq 0$, define the arm-candidate vertex set of $\mathbf T$ with intensity $k$ by
    \begin{align}\label{eq-def-frakAc}
        \mathfrak{Ac}_k(\mathbf T) = \{ v \in V(\mathbf T): |\mathsf{ASubs}(\mathsf{Des}_v(\mathbf T))| = k\} \,,
    \end{align}
    and for any $v \in V(\mathbf T)$ define $\mathfrak k(v)=\mathfrak k_\mathbf T(v)$ to be the unique $k$ such that $v \in \mathfrak{Ac}_k(\mathbf T)$. Also, for any $v \in V(\mathbf T)$ define the essential degree of $v$ in $\mathbf T$ by $\mathsf{Deg}^{\text{ess}}_\mathbf T(v)= \mathsf{Branch}_\mathbf T(v) - \mathfrak k(v)$.
    \item {\em Chained pairs.} For an unlabeled rooted tree $(\mathbf T,\mathfrak R(\mathbf T))$, $u,v \in V(\mathbf T)$, $\mathtt m > 0$ and $a,b \geq 0$, we say $(u,v)$ is an (unordered) $\mathtt m$-chained pair of $\mathbf T$ with order $(a,b)$, if $\mathsf{Dist}_{\mathbf T}(u,v)=\mathtt m$, $\{\mathfrak k(u),\mathfrak k(v)\} = \{a,b\}$ as unordered tuples, $\mathsf{Branch}_\mathbf T(w)=1$ for all $w \in V(\mathfrak p_\mathbf T(u,v)) \setminus \operatorname{EndP}(\mathfrak p_\mathbf T(u,v))$, and $\mathsf{Branch}_{\mathbf T}(u) \wedge \mathsf{Branch}_{\mathbf T}(v) \geq 2$. Given an unlabeled rooted tree $(\mathbf T,\mathfrak R(\mathbf T))$, define $\mathsf{Ch}^{(a,b)}_{\mathtt m}(\mathbf T)$ the set of $\mathtt m$-chained pairs $(u,v)$ of $\mathbf T$ with order $(a,b)$ with $\mathsf{Branch}_\mathbf T(u)\wedge \mathsf{Branch}_\mathbf T(v) \geq 2$, and define
    \begin{align}{\label{eq-def-mathsf-Ch-a,b,*}}
        \mathsf{Ch}^{(a,b;*)}_{\mathtt m}(\mathbf T) = \big\{ (u,v) \in \mathsf{Ch}^{(a,b)}_{\mathtt m}(\mathbf T) : \mathsf{Branch}_\mathbf T(u) = \mathsf{Branch}_\mathbf T(v) = 2 \big\} \,.
    \end{align}
    \item {\em Essential root.} For an unlabeled rooted tree $(\mathbf T,\mathfrak R(\mathbf T))$, define 
    \begin{equation}{\label{eq-def-mathfrak-m-T}}
        \mathfrak m(\mathbf T) = \inf_{v \in V(\mathbf T),|\mathsf{Subs}(\mathsf{Des}_{\mathbf T}(v))| \geq 2} \big\{ \mathsf{Dep}_\mathbf T(v) \big\}
    \end{equation}
    to be the depth at which $\mathbf T$ ``starts to branch", and define 
    \begin{equation}\label{eq-def-mathfrak-v-T}
        \mathfrak v(\mathbf T) = \arg \min_{v \in V(\mathbf T),|\mathsf{Subs}(\mathsf{Des}_\mathbf T(v))| \geq 2} \big\{ \mathsf{Dep}_\mathbf T(v) \big\}
    \end{equation}
    to be the essential root of $\mathbf T$.
    \item {\em Multigraphs.} We say $S$ is a multigraph, if $S=\big( V(S), E(S), \{ E(S)_{i,j}: (i,j) \in E(S) \} \big)$ where $E(S)_{i,j} \in \mathbb N$ denotes the multiplicity of the edge $(i,j)$. For a multigraph $S$, denote $\widetilde{S} \subset \mathsf K_n$ with $E(\widetilde S)=E(S)$ and $V(\widetilde S)=V(S)$ as the simple graph (ignore the multiplicity) corresponding to $S$. For two multigraphs $H,S$, we say $H \subset S$ if $V(H) \subset V(S)$ and $E(H)_{i,j} \leq E(S)_{i,j}$ for all $i,j \in V(H)$. Also denote $S \cup H$ to be the multigraph with 
    \begin{align*}
        & V(S \cup H) = V(S) \cup V(H) \,, \ E(S \cup H) = E(S) \cup E(H) \\
        \mbox{ and } & E(S \cup H)_{i,j} = E(S)_{i,j} + E(H)_{i,j} \,,
    \end{align*} 
    and denote $S \setminus H$ to be the multigraph with 
    \begin{align*}
        & V(S \setminus H) \subset V(S) \cup V(H) \,, \ E(S \setminus H) \subset E(S) \cup E(H) \\
        \mbox{ and } & E(S \setminus H)_{i,j} = (E(S)_{i,j} - E(H)_{i,j})\vee 0 \,.
    \end{align*} 
    In addition, for a multigraph $S$, denote $\mathcal L(S)$ to be the set of leaves that have degree $1$ in $S$ (counting multiplicity).
\end{itemize}

For two real numbers $a$ and $b$, we let $a \vee b = \max \{ a,b \}$ and $a \wedge b = \min \{ a,b \}$. For two sets $A$ and $B$, we define $A\sqcup B$ to be the disjoint union of $A$ and $B$ (so the notation $\sqcup$ only applies when $A, B$ are disjoint). We use standard asymptotic notations: for two sequences $a_n$ and $b_n$ of positive numbers, we write $a_n = O(b_n)$, if $a_n<Cb_n$ for an absolute constant $C$ and for all $n$ (similarly we use the notation $O_h$ if the constant $C$ is not absolute but depends only on $h$); we write $a_n = \Omega(b_n)$, if $b_n = O(a_n)$; we write $a_n = \Theta(b_n)$, if $a_n =O(b_n)$ and $a_n = \Omega(b_n)$; we write $a_n = o(b_n)$ or $b_n = \omega(a_n)$, if $a_n/b_n \to 0$ as $n \to \infty$. In addition, we write $a_n \circeq b_n$ if $a_n = [1+o(1)] b_n$, and we write $a_n \equiv b_n$ if $a_n = [1+O(n^{-1/2})] b_n$. For a set $\mathsf A$, we will use both $\# \mathsf A$ and $|\mathsf A|$ to denote its cardinality. We denote $\mathsf{id} \in \mathfrak S_n$ as the identity map.

The rest of this paper is organized as follows: in Section~\ref{sec:construction-stat} we describe the construction of decorated trees and the statistic $f(A,B)$, and state our main result (as presented in Theorem~\ref{THM-validity-of-statistics}). In Section~\ref{sec:bound-1st-moment} we bound the first moment of the statistic, and in Section~\ref{sec:bound-var} we control the variance of the statistic. In Section~\ref{sec:approx-via-color-coding} we describe an efficient algorithm to approximate the statistic in polynomial running time. Several auxiliary results are moved to the appendix to ensure a smooth flow of presentation.

\section{Main results and discussions}{\label{sec:construction-stat}}

\subsection{On the assumptions of parameters}{\label{subsec:choice-parameters}}

Before presenting the statistic $f(A,B)$, we first show how to choose $\delta$ in Theorem~\ref{MAIN-THM} and make several assumptions on various auxiliary parameters that will be used throughout the paper. Firstly, we may always assume that $\Delta<0.1$ is a small fixed constant. To this end, we first need the following results regarding unlabeled trees which were established in \cite{Otter48}.
\begin{lemma}{\label{lem-num-trees}}
    Denote $\mathcal V_N$ as the set of unlabeled trees on $N$ vertices and denote $\mathcal R_N$ the set of unlabeled rooted trees on $N$ vertices. Then there exists a deterministic constant $1>C_0>0.01$ and a deterministic constant $1>C_1 \geq 0.81$ such that as $N \to\infty$
    \begin{align}
        |\mathcal V_N| \circeq C_0 \alpha^{-N}/N^{2.5} \,, \quad |\mathcal R_N| \circeq C_0 C_1 \alpha^{-N}/N^{1.5} \,. 
    \end{align}
    where (again and throughout the paper) $\alpha\approx 0.338$ is the Otter's constant.
\end{lemma}
Provided with Lemma~\ref{lem-num-trees}, we can choose a sufficiently large constant $M\geq 100$ such that the following holds: for all $k \geq M$
\begin{align}
    \alpha^{-k}/k^{2.5} \geq |\mathcal V_k| \geq \alpha^{-k}/100k^{2.5} \mbox{ and } \alpha^{-k}/k^{1.5} \geq |\mathcal R_k| \geq \alpha^{-k}/200k^{1.5} \,. \label{eq-choice-N-2}  
\end{align}
Also we choose a sufficiently large constant $\mathfrak c$ such that 
\begin{align}\label{eq-choice-mathfrak-c}
    \mathfrak c>10^{10}\mbox{ and }|\mathcal R_N| \leq\tfrac{\mathfrak c-10}{10} \cdot 3^N \mbox{ for all } N\geq 1 
\end{align}
and a sufficiently large positive integer $\mathtt m$ such that
\begin{align}{\label{eq-choice-mathtt-m}}
    \mathtt m > (\mathfrak c\lambda s^{-1})^{1000}\cdot\big(\epsilon(1-\epsilon)\Delta\big)^{-1000} \mbox{ and } \epsilon^2\lambda s(1-\epsilon^{\mathtt m}) >1+\tfrac{\Delta}{2} \,.
\end{align}
In addition, we choose a sufficiently small positive rational $\iota=\iota(M,\epsilon)$ such that 
\begin{align}{\label{eq-choice-iota}}
    \log\log\log\log(\iota^{-1}) > \max\big\{ M,\mathtt m, \tfrac{1}{\epsilon}, \tfrac{1}{\log(\epsilon^{-1})}, \Delta^{-10}  \big\}  \,.  
\end{align}
Also, we choose a sufficiently small $\delta=\delta(M,\epsilon,\iota)$ such that 
\begin{align}{\label{eq-choice-delta-1}}
    \alpha^{-1} (\sqrt{\alpha}-\delta)^2 e^{0.9\iota\log(\iota^{-1})} \big( 1-\tfrac{1}{\epsilon^2 \lambda s(1-\epsilon^{\mathtt m})} \big)^{4\iota} > 1\,,
\end{align}
We now assume that throughout the rest part of the paper
\begin{equation}{\label{eq-assumption-s,lambda}}
     \epsilon^2\lambda s>1+\Delta \mbox{ and } \sqrt{\alpha}-\delta<s \,.
\end{equation}
Since $\epsilon^2\lambda s>1$ according to \eqref{eq-assumption-s,lambda}, we can choose $\ell = \ell_n \leq 100(\Delta^{-1}+1)\log n$ and choose $\aleph=\aleph_n$ such that $\iota\aleph ,\tfrac{\aleph-1}{2}\in\mathbb N$ and
\begin{equation}{\label{eq-choice-aleph}}
    (\epsilon^2 \lambda s)^{\ell} > n^4 \mbox{ and } \omega(1) = \aleph_n = o(\log\log\log n)  \,.
\end{equation}

\subsection{Construction of decorated trees}{\label{subsec:decorated-trees}}

We first introduce the notion of decorated trees and then discuss the specific selection of trees. We refer the reader to Section~\ref{subsec:discussions} for underlying intuitions behind these constructions.

\begin{DEF}{\label{def-pairing}}
    For a tree $T$ with $\aleph$ vertices, let $\xi=\{ (u_1,v_1), \ldots, (u_{\iota\aleph},v_{\iota\aleph}) \}$ such that $u_1,\ldots,u_{\iota\aleph}$, $v_1,\ldots,v_{\iota\aleph}$ are distinct vertices in $V(T)$. We say $\xi$ is a pairing of $T$ and we define $\mathsf{Vert}(\xi) = \{ u_1,\ldots,u_{\iota\aleph},v_1,\ldots,v_{\iota\aleph} \}$.
\end{DEF}

We can now introduce decorated trees.

\begin{DEF}{\label{def-decorated-trees}}
    For $\mathcal T_{\aleph} \subset \mathcal V_{\aleph}$, define $\mathcal S(\mathcal T_{\aleph})=\{ \mathcal S(\mathbf T): \mathbf T \in \mathcal T_{\aleph} \}$ where 
    \begin{align*}
        \mathcal S(\mathbf T) = \big\{ \mathtt W_{\mathtt 1}, \ldots, \mathtt W_{\mathtt M} \big\} \mbox{ with } \mathtt M\in\mathbb N \mbox{ and } \mathtt W_i \mbox{'s are pairings of } \mathbf T \,. 
    \end{align*}
    We define \emph{the family of decorated trees} associated with $\mathcal T_{\aleph}$ and $\mathcal S(\mathcal T_{\aleph})$, denoted by $\mathcal H=\mathcal H(\mathcal T_{\aleph},\mathcal S(\mathcal T_{\aleph}))$, to be the collection of all $\mathbf H = \big( \mathsf{T}(\mathbf H), \mathsf P(\mathbf H) \big)$ satisfying the following conditions:
    \begin{enumerate}
        \item[(1)] $\mathsf{T}(\mathbf H)$ is an unlabeled $\aleph$-tree in $\mathcal T_{\aleph}$;
        \item[(2)] $\mathsf P (\mathbf H)=\{ (u_1,v_1), \ldots, (u_{\iota\aleph},v_{\iota\aleph}) \} \in \mathcal S(\mathsf T(\mathbf H))$.
    \end{enumerate}
    In addition, we say a multigraph $S \Vdash \mathbf H$, if (counting multiplicity of edges) $S$ can be decomposed into a tree $\mathsf T(S)$ and $\iota\aleph$ self-avoiding paths $\mathsf L_1(S),\ldots,\mathsf L_{\iota\aleph}(S)$ each with length $\ell$ such that the following conditions hold (see Figure~\ref{fig:decorated-trees} for an illustration):
    \begin{enumerate}
        \item[(1)] $\operatorname{EndP}(\mathsf L_i(S))=\{ u_i,v_i \}$ where $u_i,v_i \in V(\mathsf T(S))$;
        \item[(2)] Denoting $\mathsf P(S)=\{ (u_1,v_1),\ldots,(u_{\iota\aleph},v_{\iota\aleph}) \}$, there exists a bijection $\varphi: V(\mathsf T(S)) \to V(\mathsf T(\mathbf H))$ such that $\varphi$ maps $\mathsf T(S)$ to $\mathsf T(\mathbf H)$ and maps $\mathsf P(S)$ to $\mathsf P(\mathbf H)$.
    \end{enumerate}
    \begin{figure}[!ht]
    \centering
    \vspace{0cm}
    \includegraphics[height=4.8cm,width=9.5cm]{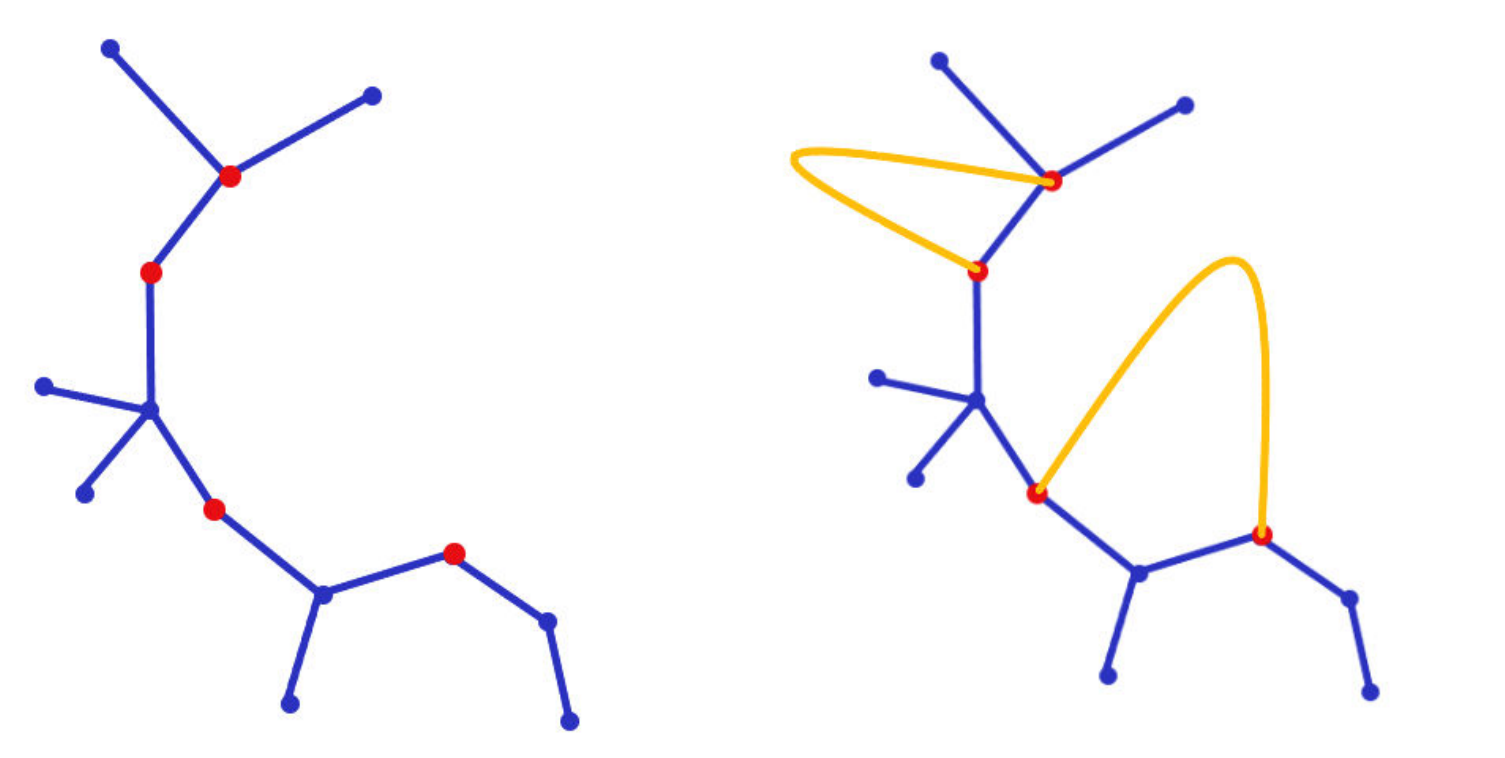}
    \caption{\noindent An example of a decorated tree. Left: example of an element $\mathbf H\in\mathcal H$, where the blue part is the tree $\mathsf T(\mathbf H)$ and the red vertices constitute the pairing $\mathsf P(\mathbf H)$. Right: example of a multigraph $S\Vdash_{A} \mathbf H$ where the yellow parts are the self-avoiding paths attached to $\mathsf T(S)$. }
    \label{fig:decorated-trees}
\end{figure}
\end{DEF}

Given any multigraph $S\subset \mathsf K_n$ (including simple graphs), we define
\begin{equation}{\label{eq-def-beta-S}}
    \beta_S(X)=\prod_{(i,j) \in E(S)} \Bigg( \frac{ X_{i,j}-\tfrac{\lambda s}{n} }{ \sqrt{ \tfrac{\lambda s}{n} (1-\tfrac{\lambda s}{n}) } } \Bigg)^{E(S)_{i,j}}\,.
\end{equation}
Then, given $S \Vdash \mathbf H$ where $\mathbf H$ is a decorated tree, we define a centered version of $\beta_{\mathsf L_k(S)}(X)$ by
\begin{equation}{\label{eq-def-psi-S}}
    \psi_{\mathsf{L}_k(S)}(X)=\beta_{\mathsf L_k(S)}(X)- \epsilon^{\mathtt m} \big( \tfrac{\epsilon^2 \lambda s}{n} \big)^{\ell/2} = \prod_{ (i,j) \in E(\mathsf{L}_k(S)) } \frac{ (X_{i,j}-\tfrac{\lambda s}{n})  }{ \sqrt{\tfrac{\lambda s}{n} (1-\tfrac{\lambda s}{n})} }- \epsilon^{\mathtt m} \big( \tfrac{\epsilon^2 \lambda s}{n} \big)^{\ell/2}  \,,
\end{equation}
where $\epsilon^{\mathtt m} \big( \tfrac{\epsilon^2 \lambda s}{n} \big)^{\ell/2}$ is the centering we add to the path $\mathsf L_k(S)$, and $\mathtt m$ is the distance of the endpoints of this path.\footnote{For two vertices $u,v$ with $\mathsf{Dist}_{\mathsf T(S)}(u,v)=\mathtt m$, a self-avoiding path $\mathsf L_k(S)$ of length $\ell$ connecting $u,v$, and a shortest path $\mathfrak p_{\mathsf T(S)}(u,v)$ between them on $\mathsf{T}(S)$, it can be shown that 
\[
    \mathbb E\Big[ \beta_{\mathsf L_k(S)}(A)\beta^2_{\mathfrak p_{\mathsf T}(u,v)}(A) \mid \sigma(u),\sigma(v) \Big] \overset{\circ}{=} \epsilon^\mathtt m \big( \tfrac{\epsilon^2\lambda s}{n} \big)^{\ell/2} \,.
\]
Later we will let $\mathsf{Dist}_{\mathsf T(S)}(u_i,v_i)=\mathtt m$ and that's how the centering arises.} 
For $1 \leq k \leq \iota\aleph$, and define
\begin{equation}{\label{eq-def-f-S}}
    \phi_S(X) = \beta_{\mathsf T(S)}(X) \cdot \prod_{ k=1 }^{ \iota \aleph } \psi_{\mathsf{L}_k(S)}(X) \,.
\end{equation}
Finally, given $S_1 \Vdash \mathbf H, S_2 \Vdash \mathbf H$ where $\mathbf H$ is a decorated tree, we write $\phi_{S_1,S_2}(A,B)=\phi_{S_1}(A) \phi_{S_2}(B)$ and define our statistic as 
\begin{equation}{\label{eq-def-f}}
    f=f(A,B) = \sum_{ \mathbf H \in \mathcal H } \frac{ s^{\aleph-1} \operatorname{Aut}(\mathsf{T}(\mathbf H)) (\epsilon^2 \lambda s)^{\ell\iota\aleph} }{ n^{\aleph+ \ell\iota\aleph} } \sum_{ S_1,S_2 \Vdash \mathbf H } \phi_{S_1,S_2}(A,B) \,.
\end{equation}

\subsection{Choices of trees and pairings}{\label{subsec:choice-trees-sets}}

Although it might seem natural to choose $\mathcal T_{\aleph}=\mathcal V_{\aleph}$ and $\mathcal S(\mathbf T)$ as the set of all pairings of $\mathbf T$, this approach poses several technical challenges when analyzing $f(A,B)$. To address these difficulties, we restrict our choices of $\mathcal T_{\aleph}$ and $\mathcal S(\mathbf T)$ to unlabeled trees and pairings with some desired properties. Before precisely describing our choices of trees and pairings, we introduce some additional notations. Recall that we use $\mathcal R_{\aleph}$ to denote the set of all unlabeled rooted trees with $\aleph$ vertices.
\begin{DEF}
{\label{def-similar-relation}}
    For any unlabeled rooted trees $\mathbf T,\mathbf T'$, we denote $\mathbf T \sim \mathbf T'$ if there exist two integers $k \leq \log^{-2}(\iota^{-1}) \cdot |V(\mathbf T)|, k'\leq \log^{-2}(\iota^{-1}) \cdot |V(\mathbf T')|$ and
    \begin{align*}
        & \{ \mathbf u_1,\ldots,\mathbf u_k \}\subset V(\mathbf T) \setminus \mathcal L(\mathbf T),\{ \mathbf u_1',\ldots,\mathbf u_{k'}'\} \subset V(\mathbf T') \setminus \mathcal L(\mathbf T') \,, \\
        &3 \leq \mathsf{Deg}_{\mathbf T}(\mathbf u_1), \ldots ,\mathsf{Deg}_{\mathbf T}(\mathbf u_k), \mathsf{Deg}_{\mathbf T'}(\mathbf u'_1), \ldots ,  \mathsf{Deg}_{\mathbf T'}(\mathbf u'_{k'}) \leq 4\,, \\
        & 0 \leq x_1,\ldots,x_k, x_1',\ldots,x_{k'}' \leq \log^{2}(\iota^{-1})
    \end{align*}
    such that (recall the definition of arm-path in Section~\ref{subsec:notations}) 
    \begin{align*}
        \mathbf T \oplus \mathcal L_{\mathbf u_1}^{(x_1)} \oplus \ldots \oplus \mathcal L_{\mathbf u_k}^{(x_k)} \cong \mathbf T' \oplus \mathcal L_{\mathbf u_1'}^{(x_1')} \oplus \ldots \oplus \mathcal L_{\mathbf u_{k'}'}^{(x_{k'}')} \,.
    \end{align*}
\end{DEF}

\begin{DEF}{\label{def-major-subtree}}
    Given an unlabeled rooted tree $\mathbf T$, define $\mathbf T_{\iota}$ to be the induced subtree of $\mathbf T$ where the vertex set of $\mathbf T_\iota$ is defined as
    \begin{align*}
        V(\iota) = \big\{ \mathbf v \in V(\mathbf T) : |V(\mathsf{Des}_{\mathbf T}(\mathbf v))| \geq \log^{4}(\iota^{-1}) \big\} \,.
    \end{align*}
    It is straightforward to check that $\mathbf u \hookrightarrow \mathbf v$ and $\mathbf v \in V(\iota)$ imply $\mathbf u \in V(\iota)$. Recalling definitions of chained pairs around \eqref{eq-def-mathsf-Ch-a,b,*}, for $a,b \geq 0$ we define
    \begin{align}
        & \mathfrak{Ch}^{(a,b)}_\mathtt m (\mathbf T) = \{(u,v) \in \mathsf{Ch}^{(a,b)}_\mathtt m (\mathbf T): u,v \in V(\mathbf T_\iota),\mathfrak R(\mathbf T) \not\in \mathfrak p_{\mathbf T}(u,v), 2 \leq\mathsf{Branch}_{\mathbf T}(u), \mathsf{Branch}_{\mathbf T}(v)\leq 3\} \,,  \label{def-chained-signal-pair-in-T_iota}\\
        & \mathfrak{Ch}^{(0,0;*)}_\mathtt m (\mathbf T) = \{(u,v) \in \mathsf{Ch}^{(0,0;*)}_\mathtt m (\mathbf T): u,v \in V(\mathbf T_\iota)\,,\mathfrak R(\mathbf T) \not\in \mathfrak p_{\mathbf T}(u,v)\}\,.\label{def-chained-noise-pair-in-T_iota}
    \end{align}
\end{DEF}

\begin{DEF}{\label{def-tilde-T-K}}
    Define $\widetilde{\mathcal{T}}_{\aleph}$ to be the collection of trees $\mathbf T \in \mathcal R_{\aleph}$ that satisfy the following conditions:
    \begin{enumerate}
        \item[(1)] For every $u \in V(\mathbf T)$, $\mathsf{Deg}^{\mathrm{ess}}_{\mathbf T}(u) \leq \log^2(\iota^{-1})$.
        \item[(2)] The tree $\mathbf T$ contains no arm-path with length at least $\log^2(\iota^{-1})$.
        \item[(3)] $|V(\mathbf T_{\iota})| \geq \frac{ \aleph }{ \log^{9}(\iota^{-1}) }$.
        \item[(4)] For any descendant trees $  \mathbf S,\mathbf S'\hookleftarrow \mathbf T$ such that $|V(\mathbf S)|,|V(\mathbf S')| \geq \log^{2}(\iota^{-1})$, if $\mathbf S \neq \mathbf S'$ and $\mathbf S,\mathbf S'$ share the same parent in $\mathbf T$, then $\mathbf S \not \sim \mathbf S'$. 
        \item[(5)] The root $\mathfrak R(\mathbf T)$ has exactly two children trees $\mathbf T_1,\mathbf T_2$ with $V(\mathbf T_1)= \lfloor \tfrac{\aleph-1}{2} \rfloor, V(\mathbf T_2)= \lceil \tfrac{\aleph-1}{2} \rceil$ respectively such that for all $ \mathbf S \hookleftarrow \mathbf{T}_2$ we have $\mathbf T_1 \not \sim \mathbf S$ and for all $\mathbf S'\hookleftarrow\mathbf T_1$ we have $\mathbf T_2 \not \sim \mathbf S'$.
        \item[(6)] $|\mathfrak{Ch}^{(0,0;*)}_\mathtt m(\mathbf T)| = \mathfrak C_\aleph$ where $\mathfrak C_\aleph$ is a nonnegative integer determined in \eqref{eq-def-mathfrak-C-Au}, and satisfies (recall \eqref{eq-choice-mathfrak-c} for the definition of $\mathfrak c$)
        \begin{align}\label{eq-chained-pairs-control-key}
            \mathfrak C_\aleph-\mathfrak c^{-9} \sum_{0\leq a,b\leq 1} |\mathfrak{Ch}^{(a,b)}_\mathtt m(\mathbf T)| \geq \tfrac{\aleph}{\log^9(\iota^{-1})}\,.
        \end{align}
        \item[(7)] $\operatorname{Aut}(\mathbf T) \in [\mathfrak{Au}_\aleph , 2 \mathfrak{Au}_\aleph]$ where $\mathfrak{Au}_\aleph$ is a positive integer determined in \eqref{eq-def-mathfrak-C-Au}.
    \end{enumerate}
\end{DEF}

The following result provides an estimate for the cardinality of $\widetilde{\mathcal T}_{\aleph}$.
\begin{thm}{\label{thm-num-desired-tree}}
    We have 
    \begin{equation*}
        |\widetilde{\mathcal T}_{\aleph}| \geq \alpha^{-\aleph} \cdot \exp\big\{ -10 e^{-\frac{1}{20}\log^2(\iota^{-1})} \aleph \big\} \,,
    \end{equation*}
    Thus, denoting by $\mathcal T_{\aleph}$ the set of $\mathbf T \in \mathcal V_{\aleph}$ such that there exists an unlabeled rooted tree $(\mathbf T,\mathbf{o}) \in \widetilde{\mathcal T}_{\aleph}$, we have 
    \begin{equation}
        |\mathcal T_{\aleph}| \geq (\alpha+o(1))^{-\aleph} \cdot \exp\big\{ -10 e^{-\frac{1}{20}\log^2(\iota^{-1})} \aleph \big\} \,.
    \end{equation} 
\end{thm}
The proof of Theorem~\ref{thm-num-desired-tree} is postponed to Section~\ref{subsec:proof-tree-enu-thm} of the appendix. Now we discuss the choice of vertex subsets.

\begin{thm}{\label{thm-desired-vertex-sets}}
    For each $\mathbf T \in \widetilde{\mathcal T}_{\aleph}$, we can choose a subset $\mathfrak{Ch'_\mathtt m}(\mathbf T) \subset \mathfrak{Ch}^{(0,0;*)}_\mathtt m(\mathbf T)$ and choose $W_{\mathtt 1}(\mathbf T),\ldots,W_{\mathtt M}(\mathbf T) \subset V(\mathbf T)$ to be pairings of $\mathbf T$ with $\mathtt M = \exp\big( -10e^{-\log^{2}(\iota^{-1})} \aleph \big)\cdot \binom{ \mathfrak C_\aleph/4 }{\iota\aleph}$ such that the following conditions hold:
    \begin{enumerate}
        \item[(1)] $|\mathfrak{Ch'_\mathtt m}(\mathbf T)| = \frac{\mathfrak C_\aleph}{4}$ and $u,v \neq \mathfrak R(\mathbf T)$ for all $(u,v) \in \mathfrak{Ch'_\mathtt m}(\mathbf T)$. In addition, no two pairings in $\mathfrak{Ch'_\mathtt m}(\mathbf T)$ share a common vertex.
        \item[(2)] $W_\mathtt i(\mathbf T) \subset \mathfrak{Ch}'_\mathtt m (\mathbf T)$ for all $\mathtt 1 \leq \mathtt i \leq \mathtt M$.
        \item[(3)] For all $\mathtt 1 \leq \mathtt i \leq \mathtt M$ and all descendant trees $  \mathbf T' \hookleftarrow \mathbf {T}$ with $V(\mathbf T') \geq \log^{4}(\iota^{-1})$, we have
        \begin{align*}
            \big| \mathsf{Vert}(W_{\mathtt i}) \cap V(\mathbf T') \big| \leq \log^{-2}(\iota^{-1}) \cdot |V(\mathbf T')| \,.
        \end{align*}
    \end{enumerate}
\end{thm}
The proof of Theorem~\ref{thm-desired-vertex-sets} is postponed to Section~\ref{subsec:proof-set-enu-thm} of the appendix. We can now state our choice of trees and vertex sets, as incorporated in the following lemma.

\begin{lemma}{\label{lem-enu-decorated-trees}}
    For each $\mathbf T \in \mathcal T_{\aleph}$, define $\mathcal S(\mathbf T)=\{ W_{\mathtt 1},\ldots, W_{\mathtt M} \}$ where $W_{\mathtt 1},\ldots, W_{\mathtt M}$ are the same as in Theorem~\ref{thm-desired-vertex-sets}. Define $\mathcal H=\mathcal H(\iota,\aleph,M;\mathcal T_{\aleph}, \mathcal S(\mathcal T_{\aleph}))$. We then have
    \begin{align*}
        |\mathcal H| = (\alpha+o(1))^{-\aleph} \cdot \tbinom{ \mathfrak C_{\aleph}/4 }{ \iota\aleph } \cdot \exp\big\{ -20 e^{-\frac{1}{20}\log^2(\iota^{-1})} \aleph \big\} \geq\Big( \alpha^{-1}\exp\big( 0.9\iota\log(\iota^{-1}) \big) \Big)^{\aleph} \,.
    \end{align*}
\end{lemma}
\begin{proof}
    The first inequality follows directly from combining Theorems~\ref{thm-num-desired-tree} and \ref{thm-desired-vertex-sets}. As for the second inequality, note that \eqref{eq-chained-pairs-control-key} yields that $\mathfrak C_{\aleph} \geq \tfrac{\aleph}{\log^9(\iota^{-1})}$. Thus,
    \begin{align*}
        \tbinom{ \mathfrak C_{\aleph}/4 }{ \iota\aleph } \geq \tbinom{ \aleph/4\log^9(\iota^{-1}) }{ \iota\aleph } \geq \exp\Big( \iota\aleph \cdot \log\big( \tfrac{ 1 }{ 4\iota\log^9(\iota^{-1}) } \big) \Big) \geq \exp\big( 0.99\iota\log(\iota^{-1})\aleph \big) \,,
    \end{align*}
    leading to the second inequality easily.
\end{proof}

Now we can rigorously state the main result of this paper, which verifies the success of the statistic $f(A,B)$ under the above choices of $\mathcal T_{\aleph}$ and $\mathcal S(\mathcal T_{\aleph})$.
\begin{thm}{\label{THM-validity-of-statistics}}
    Under the conditions \eqref{eq-choice-N-2}--\eqref{eq-choice-aleph}, by choosing $\mathcal H=\mathcal H(\iota,\aleph,\mathtt M;\mathcal T_{\aleph}, \mathcal S(\mathcal T_{\aleph}))$ according to Definition~\ref{def-tilde-T-K} and Theorem~\ref{thm-desired-vertex-sets}, we have 
    \begin{align*}
        \frac{ \mathbb E_{\mathbb P}[ f ]^2 }{ \mathbb E_{\mathbb Q}[ f^2 ] } = \omega(1) \mbox{ and } \frac{ \mathbb E_{\mathbb P}[ f ]^2 }{ \mathbb E_{\mathbb P}[ f^2 ] } = 1+o(1) \,.
    \end{align*}
\end{thm}

\begin{remark}{\label{rmk-algorithm}}
    It is straightforward by Chebyshev's inequality that Theorem~\ref{THM-validity-of-statistics} implies that the testing error satisfies
    \begin{align*}
        \Qb\big( f(A,B)\ge \tau \big) + \Pb\big( f(A,B)\le \tau \big)= o(1) \,,
    \end{align*}
    where the threshold $\tau$ is chosen as $\tau = C\mathbb{E}_{\Pb} [f_{\mathcal{T}}(A,B)]$ for any fixed constant $0<C<1$. 
\end{remark}

\subsection{Discussions}{\label{subsec:discussions}}

Before moving on to the proof of Theorem~\ref{THM-validity-of-statistics}, we feel that it is necessary to explain a bit more about the construction of our statistic $f(A,B)$ and the seemingly daunting choices of $\mathcal T_{\aleph}, \mathcal S(\mathcal T_{\aleph})$. Recall that we are in the supercritical regime $\epsilon^2 \lambda s>1$ where weak community recovery in $A$ and $B$ is possible. Also, assuming that all the community labels in $A$ and $B$ (we denote them by $\sigma_A$ and $\sigma_B$) are known to us, it is easy to check that we can achieve detection below Otter's threshold between $\Pb$ and $\Qb$ using the following statistic: 
\begin{align*}
    g(A,B)= \sum_{ \mathbf H_{c} \in \mathcal H_c } \frac{ s^{\aleph-1} \operatorname{Aut}(\mathbf H_c) }{ n^{\aleph} } \sum_{S_1,S_2 \cong \mathbf H_c} \phi_{S_1,S_2}(A,B) \,,
\end{align*}
where $\mathcal H_c$ is the set of all unlabeled two-colored trees and for all $S_1,S_2 \in \mathsf K_n$, and we say $S_1 \cong \mathbf H_c$ if there exists an isomorphism $\varphi:V(S_1) \to V(\mathbf H_c)$ such that $\sigma_A(i) = \sigma_{\mathbf H_c}(\varphi(i))$ for all $i\in V(S_1)$. Thus, a naive attempt is to first run the weak community recovery algorithm in both $A$ and $B$ which produces an estimator $(\widehat{\sigma}_A, \widehat{\sigma}_B)$ of $(\sigma_A,\sigma_B)$, and then plug this estimator into $g(A,B)$ to obtain a testing variable $\widehat{g}(A,B)$. However, as explained in Section~\ref{subsec:innovations}, it seems of substantial challenge to analyze this naive approach directly because we need to tackle the correlations caused by the community recovery step. 

To overcome this issue, we note that both the task of community recovery (by counting non-backtracking or self-avoiding paths) or correlation detection (by counting trees) can be captured by some low-degree polynomials of $(A,B)$. Thus, instead of running these two algorithms sequentially, we seek for a polynomial which \emph{combines} the low-degree polynomials corresponding to community recovery and correlation detection, and then we analyze this polynomial directly. This explains the intuition behind our definition of decorated trees in Definition~\ref{def-decorated-trees}, where the self-avoiding paths $\mathsf L_i(\mathbf H)$ record the information of the community and the trees $\mathsf T(\mathbf H)$ are used for correlation detection. Although this intuition is helpful, analyzing the statistic obtained by counting all decorated trees remains challenging. Another key to the success of our detection algorithm is to exploit the correlation of subgraph counts in $A$ and $B$ as much as possible, while suppressing the undesirable correlation between different subgraph counts. We next explain why restricting to the special family $\mathcal T_{\aleph}, \mathcal S(\mathcal T_{\aleph})$ is crucial, and outline some basic guidelines for choosing its parameters. To illustrate, let us consider the second moment under $\Pb$ (which reduces to calculating $\mathbb E_{\Pb}[\phi_{S_1,S_2}\phi_{K_1,K_2}]$). For simplicity, assume $\pi = \mathsf{id}$. Our baseline analysis focuses on the case in which $S_1,S_2,K_1,K_2$ are simple (i.e., each edge has multiplicity $1$) and $S_1 \cap S_2$ (respectively, $K_1 \cap K_2$) is a tree containing $\mathsf T(S_1)$ (respectively, $\mathsf T(K_1)$; see Figure~\ref{fig:overall}(a)). In this scenario, $S_1 \cap S_2$ must be $\mathsf T(S_1)$ with some self-avoiding paths attached. Using Corollary~\ref{lem-useful-Aut-bound}, we can bound the enumeration of such $(S_1,S_2;K_1,K_2)$ and the expectation in this case can be calculated precisely. Extending from this simple case (referred to as the baseline) to more general cases involves five key observations for bounding the total variance (although the actual proof does not follow this classification exactly):
\begin{itemize}
    \item If one of the self-avoiding paths (for example, $\mathsf L_i(S_1)$) intersects with the tree (for example, $\mathsf T(S_1)$)  (see Figure~\ref{fig:overall}(b)), there will be additional cycles and thus increase the excess of the graph $G_{\cup}=\widetilde S_1\cup \widetilde S_2 \cup \widetilde K_1 \cup \widetilde K_2$, gaining an extra $1/n$ factor in the second moment bound compared to the contribution of the baseline as (non-rigorously speaking) the enumeration of such graphs depends on the number of vertices of $G_{\cup}$ and the expectation in each case depends on the number of edges of $G_{\cup}$.
    \item If the two trees $\mathsf T(K_1), \mathsf T(K_2)$ (or $\mathsf T(S_1), \mathsf T(S_2)$) do not completely overlap (see Figure~\ref{fig:overall}(c)), there will also be additional cycles and thus increase the excess of the graph $G_{\cup}=S_1\cup S_2 \cup K_1 \cup K_2$, gaining extra factors of $1/n$ in the second moment bound compared to the contribution of the baseline.
    \item If $(S_1,S_2)$ and $(K_1,K_2)$ have overlapping vertices (see Figure~\ref{fig:overall}(d)), this will decrease the enumeration of $(S_1,S_2,K_1,K_2)$, gaining extra factors of $1/n$ in the second moment bound compared to the contribution of the baseline.
    \item If $S_1 \cap S_2$ or $K_1 \cap K_2$ contains a whole length-$\ell$ self-avoiding path (see Figure~\ref{fig:overall}(e)), this will decrease $|E(G_{\cup})|$ by $\ell$ and decrease $\tau(G_{\cup})$ by $1$, gaining extra factors of $n(\epsilon^2 \lambda s)^{-\ell}<1/n$ (thanks to \eqref{eq-choice-aleph}) in the second moment bound compared to the contribution of the baseline (note that the contribution from the pair of paths is $O(1)$ when they completely overlap with each other, compared to the baseline case where the contribution is $\frac{(\epsilon^2\lambda s)^\ell}{n^{\ell}}\cdot n^{\ell-1}$)). 
    \item If $\mathsf T(S_1)=\mathsf T(K_1)$ and $\mathsf T(S_2)=\mathsf T(K_2)$ (see Figure~\ref{fig:overall}(f)), using Item~(2) in Theorem~\ref{thm-desired-vertex-sets} the expectation is smaller, gaining extra factors $\epsilon^{\mathtt m\iota\aleph}$ in the second moment bound compared to the contribution of the baseline.
\end{itemize}
\begin{figure}[ht]
    \centering
    \vspace{0cm}
    \includegraphics[height=7.5cm,width=12.8cm]{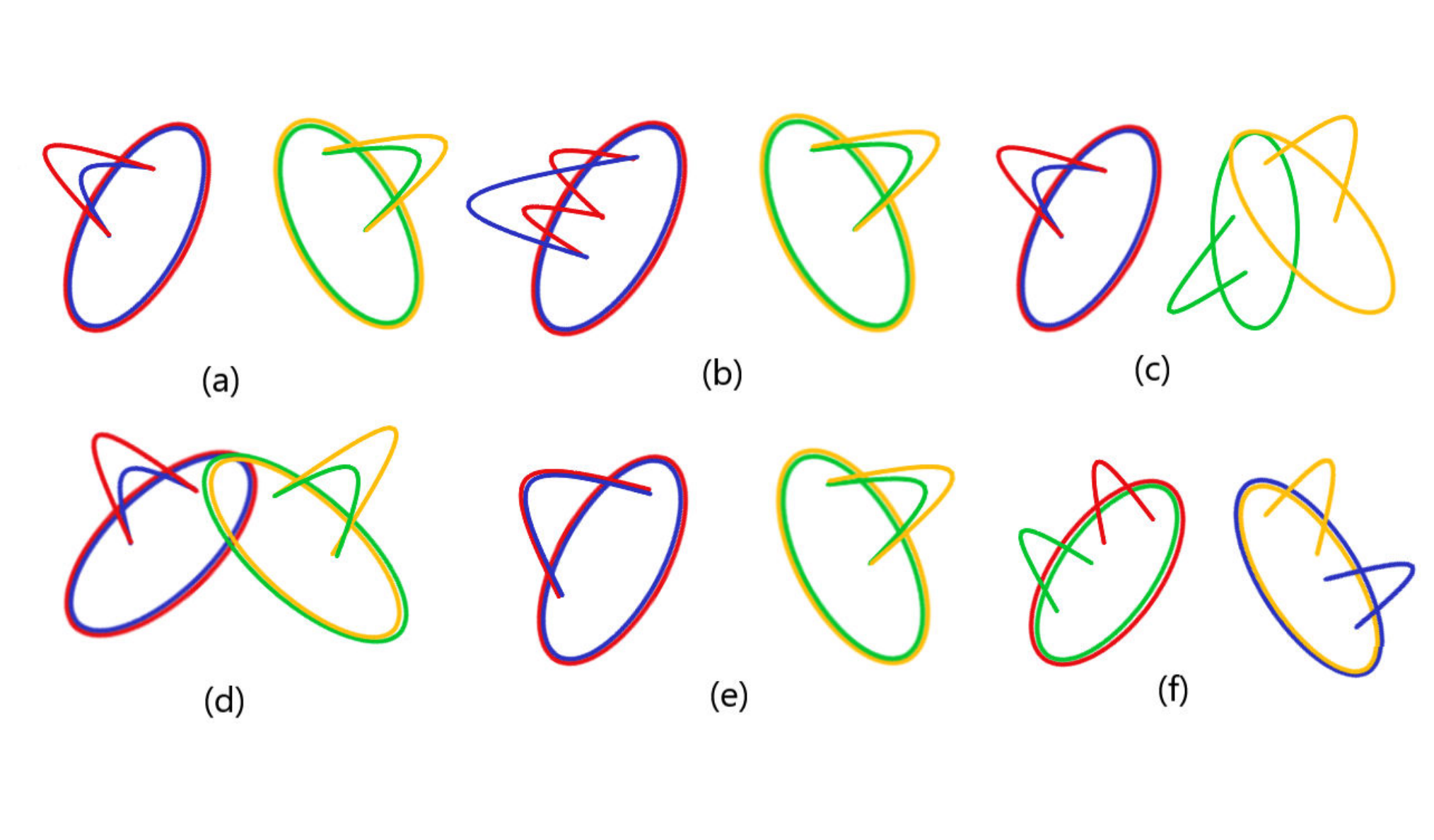}
    \caption{\noindent Examples of overlapping patterns of $S_1,S_2$ and $K_1,K_2$, shown in red/blue and green/orange respectively; the ellipses represent the trees and the curves represent the self-avoiding paths attached to it. (a) $S_1 \cap S_2, T_1 \cap T_2$ are non-intersecting trees containing $\mathsf T(S_1),\mathsf T(K_1)$, respectively (the baseline case); (b) The self-avoiding paths and the trees intersect, creating cycle(s); (c) The two trees $\mathsf T(K_1), \mathsf T(K_2)$ (or $\mathsf T(S_1), \mathsf T(S_2)$) do not completely overlap, creating cycle(s); (d) The four trees $\mathsf T(S_1), \mathsf T(S_2), \mathsf T(K_1), \mathsf T(K_2)$ have non-empty common part; (e) Two self-avoiding paths overlap completely; (f) $\mathsf T(S_1)=\mathsf T(K_1)$ and $\mathsf T(S_2)=\mathsf T(K_2)$.}
    \label{fig:overall}
\end{figure}

\section{Estimation of first moment}{\label{sec:bound-1st-moment}}

The main goal of this section is to prove the following proposition. 
\begin{proposition}{\label{prop-first-moment}}
    We have the following estimate:
    \begin{align}
        s^{2(\aleph-1)} |\mathcal H| \cdot \Big( \tfrac{  (\epsilon^2 \lambda s)^{\ell} (1-\epsilon^{2\mathtt m})^{1/2} }{ n } \Big)^{2\iota\aleph} & \leq \mathbb E_{\mathbb P}[f] \label{eq-Pb-moment-lower-bound} \\
        &\leq s^{2(\aleph-1)} |\mathcal H| \cdot \Big( \tfrac{ (\epsilon^2 \lambda s)^{\ell} (1-\epsilon^{2\mathtt m})^{1/2} }{n} \Big)^{2\iota\aleph} (3e^{2}\Delta^{-1})^{2\iota\aleph} \,.  \label{eq-Pb-moment-upper-bound} 
    \end{align}
\end{proposition}
To prove Proposition~\ref{prop-first-moment}, we first need the following two lemmas:
\begin{lemma}{\label{lem-joint-moment-A-B}}
    For any $r+t \geq 1$, there exist $u_{r,t},v_{r,t}$ such that
    \begin{align}
        \mathbb E_{\Pb_{\sigma,\pi}}\Bigg[ \Bigg( \frac{ A_{i,j} - \tfrac{\lambda s}{n} }{ \big( \tfrac{\lambda s}{n} (1-\tfrac{\lambda s}{n}) \big)^{1/2} } \Bigg)^r \Bigg( \frac{ B_{\pi(i),\pi(j)} - \tfrac{\lambda s}{n} }{ \big( \tfrac{\lambda s}{n} (1-\tfrac{\lambda s}{n}) \big)^{1/2} } \Bigg)^t \ \Bigg] =  u_{r,t} + v_{r,t} \sigma_i \sigma_j \,. \label{eq-joint-moment-A-B}
    \end{align}
    In addition, we have
    \begin{align}
        & u_{0,1} = u_{1,0} = 0\,,\quad v_{0,1} = v_{1,0} = \big(\tfrac{n}{\epsilon^2 \lambda s} \big)^{-1/2} \,; \label{eq-strong-bound-u-v-r,t} \\
        & u_{0,2}=u_{2,0}=1\,, \quad v_{0,2}=v_{2,0}=\epsilon\,, \quad u_{1,1}=s \,, \quad v_{1,1}=\epsilon s \,. \label{eq-strong-second-bound-u-v-r,t} \\
        & 0 \leq u_{r,t}, v_{r,t} \leq \big( \tfrac{n}{\epsilon^2 \lambda s} \big)^{(r+t-2)/2} \mbox{ for } t+r\geq 3\,. \label{eq-bound-u-v-r,t}
    \end{align}
\end{lemma}

\begin{lemma}{\label{lem-expectation-over-chain}}
    For a path $P$ with $V(P)= \{ v_0, \ldots, v_l \} $ and $\operatorname{EndP}(P)=\{ v_0,v_l \}$, we have that for all $a_i,b_i \in \mathbb R$
    \begin{equation}{\label{eq-expectation-over-chain}}
        \mathbb{E}_{\sigma \sim \nu} \Big[ \prod_{i=1}^{l} \big( a_{i} + b_{i} \cdot \sigma_{i-1} \sigma_i \big) \mid \sigma_0, \sigma_l \Big] = \prod_{i=1}^{l} a_i + \sigma_0 \sigma_l \cdot \prod_{i=1}^{l} b_i \,.
    \end{equation}
\end{lemma}
The proofs of Lemmas~\ref{lem-joint-moment-A-B} and \ref{lem-expectation-over-chain} are incorporated in Sections~\ref{subsec:proof-lem-3.3} and \ref{subsec:proof-lem-3.4}, respectively. We now state several useful bounds that control the expectation of $\beta_{S_1}(A)\beta_{S_2}(B)$. 

\begin{lemma}{\label{lem-upper-bound-exp}}
    For connected multigraphs $S_1, S_2$  we have 
    \begin{align}
        \mathbb E_{\Pb_{\pi}}\big[ \beta_{S_1}(A)\beta_{S_2}(B) \big]=0 \mbox{ if } \mathcal L(\pi(S_1)) \cup \mathcal L(S_2) \not\subset V(\pi(S_1))\cap V(S_2) \,.  \label{eq-upper-bound-exp-case-1}
    \end{align} 
    In addition, suppose that $|\mathcal L(\widetilde S_1)|, |\mathcal L(\widetilde S_2)| \leq \aleph$, and for all $(i,j) \in E(S_1) \cup E(S_2)$, denote 
    \begin{align}\label{eq-def-chi}
        \chi(i,j)=(\chi_1(i,j),\chi_2(i,j))=( E(S_1)_{i,j}, E(S_2)_{i,j} ) \,.
    \end{align}
    Then for all $V \subset V(S_1) \cup V(\pi^{-1}(S_2))$ we have that $\mathbb E_{\Pb_{\pi}}\Big[\beta_{S_1}(A)\beta_{S_2}(B) \cdot \prod_{i \in V} \sigma_i \Big]$ is bounded by (recall that $\tau(H)=|E(H)|-|V(H)|$ and recall \eqref{eq-joint-moment-A-B})
    \begin{align}
        2^{ 5\tau(\pi(S_1)\cup S_2)+5|V|+12\aleph } \cdot \prod_{ (i,j) \in E(S_1) \cup E(S_2) } \max\big\{ u_{ \chi_1(i,j),\chi_2(i,j) }, v_{\chi_1(i,j),\chi_2(i,j)} \big\} \,. \label{eq-upper-bound-exp-case-2}
    \end{align}
    In particular, $\mathbb E_{\Pb_{\pi}}\Big[\beta_{S_1}(A)\beta_{S_2}(B) \cdot \prod_{i \in V} \sigma_i \Big]$ is bounded by
    \begin{align}
        2^{ 5\tau(\pi(S_1)\cup S_2)+5|V|+12\aleph } \cdot \big( \tfrac{n}{\epsilon^2 \lambda s} \big)^{ \sum_{(i,j) \in E(S_1) \cup E(S_2)} (\chi_1(i,j)+\chi_2(i,j)-2)/2 } \,.  \label{eq-upper-bound-exp-case-2-plus}
    \end{align}
\end{lemma}

\begin{lemma}{\label{lem-delicate-exp-upper-bound-on-trees}}
    Given any tree $T$ and a subset $\mathtt U \subset V(T)$ such that $|\mathtt U|$ is even. For any $\mathtt V \subset V(T)$ with even number of elements we define
    \begin{align}{\label{def-mathcal-X-of-sets}}
        \mathcal X(\mathtt V) = \Big\{ \big( \mathtt u_1, \mathtt w_1, \ldots, \mathtt u_{|\mathtt V|/2}, \mathtt w_{|\mathtt V|/2} \big): \big\{ \mathtt u_1 , \mathtt w_1 ,\ldots, \mathtt u_{|\mathtt V|/2} , \mathtt w_{|\mathtt V|/2} \big\} = \mathtt V \Big\} \,.
    \end{align}
    Suppose that 
    \begin{align*}
        \min_{(\mathtt u_1,\mathtt w_1,\ldots,\mathtt u_{|\mathtt U|/2}, \mathtt w_{|\mathtt U|/2} ) \in \mathcal X(\mathtt U)} \Big\{ \sum_{\mathtt i=1}^{|\mathtt U|/2} \mathsf{Dist}_T(\mathtt u_\mathtt i , \mathtt w_\mathtt i) \Big\} = \mathtt d\,,
    \end{align*}
    then we have 
    \begin{align}
        &\mathbb E\Big[ \beta_T(A)^2 \cdot \prod_{\mathtt u \in \mathtt U} \sigma_{\mathtt u} \Big] \circeq \epsilon^{ \mathsf{d} } \,, \label{eq-upper-bound-exp-on-trees} \\
        &\mathbb E_{\Pb_{\mathsf{id}}}\Big[ \beta_T(A)\beta_T(B)\cdot \prod_{\mathtt u \in \mathtt U} \sigma_{\mathtt u} \Big] \circeq s^{|E(T)|} \cdot \epsilon^{ \mathsf{d} } \,. \label{eq-upper-bound-exp-on-trees-2}
    \end{align}
\end{lemma}

The proofs of Lemmas~\ref{lem-upper-bound-exp} and \ref{lem-delicate-exp-upper-bound-on-trees} are incorporated in \ref{subsec:proof-lem-3.6} and \ref{subsec:proof-lem-3.7} of the appendix, respectively. We now return to the proof of Proposition~\ref{prop-first-moment}. Since $f$ is invariant under permutations, it suffices to prove the bounds for $\mathbb E_{\Pb_{\mathsf{id}}}[f]$ in place of $\mathbb E_{\Pb}[f]$ in Proposition~\ref{prop-first-moment}. Define $\mathfrak R_{\mathbf H}^*$ to be the subset of $\mathfrak R_{\mathbf H}:=\{ S: S \Vdash \mathbf H \}$ such that each vertex in $V(S)$ appears exactly once in 
\begin{align*}
    V(\mathsf{T}(S)), V(\mathsf{L}_1(S)) \setminus \operatorname{EndP}(\mathsf{L}_1(S)), \ldots, V(\mathsf{L}_{\iota\aleph}(S)) \setminus \operatorname{EndP}(\mathsf{L}_{\iota\aleph}(S)) \,.
\end{align*} 
Note that for $S \in \mathfrak R_{\mathbf H}^*$ all of the edges in $\mathsf T(S)$ and $\mathsf L_i(S)$ for $i=1,\dots,\iota\aleph$ are not multi-edges, and thus $\widetilde S=S$. In addition, define $\mathcal A$ to be a map from a pair of subgraphs $(S_1,S_2) \in \mathfrak R_{\mathbf H} \times \mathfrak R_{\mathbf I}$ to subsets of permutations by
\begin{align}
    \mathcal A(S_1,S_2) =
    \begin{cases}
        \big\{ \pi \in \mathfrak S_n: \pi(S_1) \cap S_2 \mbox{ is a tree containing } \pi(\mathsf T(S_1)) \cup \mathsf T(S_2) \big\} \,, & S_1,S_2 \in \mathfrak R_{\mathbf H}^* \times \mathfrak R_{\mathbf I}^*  \,; \\
        \emptyset \,, & \mbox{otherwise} \,,
    \end{cases}  \label{eq-def-mathcal-A-diamond}
\end{align}
In addition, define (recall $\mathsf P(S)$ in Definition~\ref{def-decorated-trees}) 
\begin{align}
    &\mathtt K_1(S_1,S_2;\pi) = \big\{ i: \pi(\mathsf L_i(S_1)) \cap S_2 = \operatorname{EndP}(\pi(\mathsf L_i(S_1))) \not\in \mathsf P(S_2) \big\}\,, \label{eq-def-K1-s1-S2-pi} \\
    &\mathtt K_2(S_1,S_2;\pi) = \big\{ i: \mathsf L_i(S_2) \cap \pi( S_1) = \operatorname{EndP}(\mathsf L_i(S_2)) \not\in \mathsf P(\pi(S_1)) \big\}\,,  \label{eq-def-K2-s1-S2-pi}\\
    &\mathtt K(S_1,S_2;\pi) = \#\mathtt K_1(S_1,S_2;\pi) + \#\mathtt K_2(S_1,S_2;\pi) \,, \label{eq-def-K-s1-S2-pi}\\
    &\mathtt K_c(S_1,S_2;\pi) = \# \big\{ i: \pi(\mathsf L_i(S_1)) \cap S_2 = \operatorname{EndP}(\pi(\mathsf L_i(S_1))) \in \mathsf P(S_2) \big\} \,, \label{eq-def-Kc-s1-S2-pi}
\end{align}
and denote the principal subset of $\mathcal A(S_1,S_2)$ by 
\begin{align}
    \mathcal A_*(S_1,S_2) = \big\{ \pi \in \mathcal A(S_1,S_2) : \mathtt K(S_1,S_2;\pi) = 0 \big\}\,, \label{eq-def-mathcal-A*-S1-S2}
\end{align}
(see Figures~\ref{fig:1} and \ref{fig:2} for illustrations of possible patterns).
We now introduce the following lemmas.
\begin{lemma}{\label{lem-est-intersection-istree-P-Q}}
    If $S,K$ are simple graphs such that $S \cap K=T$ and $S\cup K=T \sqcup (\sqcup_{i=1}^{m} \mathcal L_i)$ where $\{ \mathcal{L}_i:1\leq i \leq m\}$ are disjoint self-avoiding paths with $V(\mathcal L_i) \cap V(T)=\operatorname{EndP}(\mathcal L_i)$. Then we have the following:
    \begin{align}
        & \mathbb E_{\Pb_{\mathsf{id}}}\big[ \beta_{S}(A) \beta_K(B) \big] \circeq \big( \tfrac{\epsilon^2\lambda s}{n} \big)^{\frac{1}{2} \sum_{i \leq m} |E(\mathcal L_i)| } \mathbb E_{\Pb_{\mathsf{id}}}\Big[ \beta_T(A) \beta_T(B) \cdot \prod_{i\leq m} \prod_{u \in \operatorname{EndP}(\mathcal L_i)} \sigma_u \Big] \,, \label{eq-phi-P-intersec-istree}\\
        & \mathbb E\big[ \beta_{S}(A) \beta_K(A) \big] \circeq \big( \tfrac{\epsilon^2\lambda s}{n} \big)^{ \frac{1}{2}\sum_{i \leq m} |E(\mathcal L_i)| } \mathbb E\Big[ \beta_T(A)^2 \cdot \prod_{i\leq m} \prod_{u \in \operatorname{EndP}(\mathcal L_i)} \sigma_u \Big] \,. \label{eq-phi-Q-intersec-istree}
    \end{align}
\end{lemma}
\begin{lemma}{\label{lem-est-intersection-general-istree-P-Q}}
    Suppose $\mathbf H \,,\mathbf I \in \mathcal H$ and $(S_1, S_2) \in \mathfrak R^*_{\mathbf H} \times \mathfrak R^*_{\mathbf I}$, we have
    \begin{align*}
        \mathbb E_{\Pb}\big[ \beta_{\pi(S_1)}(A)\beta_{S_2}(A) \big], \ \mathbb E_{\Pb_{\pi}}\big[ \beta_{S_1}(A)\beta_{S_2}(B) \big] \geq 0
    \end{align*}
    for all $\pi \in \mathfrak S_n$, and 
    \begin{align*}
        \mathbb E_{\Pb}\big[ \phi_{\pi(S_1)}(A)\phi_{S_2}(A) \big], \ \mathbb E_{\Pb_{\pi}}\big[ \phi_{S_1}(A)\phi_{S_2}(B) \big] \geq 0
    \end{align*}
    for all $\pi\in \mathcal A(S_1,S_2)$. Moreover, for $\pi \in \mathcal A(S_1,S_2)$ there exist $h_n, \tilde h_n = O(\frac{1}{\sqrt{n}})$ such that
    \begin{align}
        \mathbb E_\Qb \big[ \phi_{S_1}(A)\phi_{S_2}(A) \big] \circeq h_n^{\mathtt K(S_1 , S_2 ; \mathsf{id})} (1 -\epsilon^{2\mathtt m})^{\mathtt K_c(S_1 , S_2 ; \mathsf{id})} \mathbb E_\Qb [\beta_{S_1}(A)\beta_{S_2}(A)]\, \label{eq-S,K-intersection-general-A,A}
    \end{align}
    and
    \begin{align}
        \mathbb E_{\Pb_{\pi}}\big[ \phi_{S_1}(A)\phi_{S_2}(B) \big] \circeq \tilde h_n^{\mathtt K(S_1 , S_2 ; \pi)} (1 -\epsilon^{2\mathtt m})^{\mathtt K_c(S_1 , S_2 ; \pi)} \mathbb E_{\Pb_{\pi}}[\beta_{S_1}(A)\beta_{S_2}(B)]\,. \label{eq-S,K-intersection-general-A,B}
    \end{align}
\end{lemma}
The proof of Lemma~\ref{lem-est-intersection-istree-P-Q} is provided in Section~\ref{subsec:proof-lem-est-intersection-istree-P-Q}, and the proof of Lemma~\ref{lem-est-intersection-general-istree-P-Q} is provided in Section~\ref{subsec:proof-lem-3.5}. For any $\{ \mathsf{A}(S_1,S_2) \subset \mathfrak S_n: S_1,S_2 \in \cup_{ \mathbf H \in \mathcal H } \mathfrak R_{\mathbf H} \}$, we define 
\begin{equation}{\label{eq-def-f-*}}
    f_{\mathsf A}= \sum_{ \mathbf H \in \mathcal H } \frac{ s^{\aleph-1} \operatorname{Aut}(\mathsf{T}(\mathbf H)) (\epsilon^2 \lambda s)^{\ell\iota\aleph} }{ n^{\aleph+ \ell\iota\aleph} } \sum_{S_1,S_2 \in \mathfrak R_{\mathbf H}} \phi_{S_1,S_2}(A,B) \cdot \mathbf 1_{\{ \mathsf{id}\in\mathsf A(S_1,S_2) \}} \,.
\end{equation}
We first deal with $\mathbb E_{\Pb_{\mathsf{id}}}[f_{\mathcal A_*}]$, as incorporated in the following lemma. 
\begin{lemma}{\label{lem-first-moment-good-part}}
    For all $S_1,S_2\in\mathfrak R^*_{\mathbf{H}}$, we have (recall \eqref{eq-def-mathcal-A-diamond}) 
    \begin{align}
        s^{2(\aleph-1)} |\mathcal H| \cdot \Big( \tfrac{ (\epsilon^2 \lambda s)^{\ell} (1-\epsilon^{2\mathtt m})^{1/2} }{ n } \Big)^{2\iota\aleph} &\leq \mathbb E_{\Pb_{\mathsf{id}}}\big[ f_{\mathcal A_*} \big]\nonumber \\
        &\leq s^{2(\aleph-1)} |\mathcal H| \cdot \Big( \tfrac{ (\epsilon^2 \lambda s)^{\ell} (1-\epsilon^{2\mathtt m})^{1/2} }{n} \Big)^{2\iota\aleph} \cdot (3e^{2}\Delta^{-1})^{2\iota \aleph} \,.\nonumber
    \end{align}
\end{lemma}
\begin{proof}
    On the one hand, define $\mathcal{A}_\diamond$ to be a map from any pair $(S_1,S_2)$ to subsets of permutations by
    \begin{equation}{\label{eq-def-mathcal-A-diamond-*}}
        \mathcal A_{\diamond}(S_1,S_2) =
        \begin{cases}
            \big\{ \pi \in \mathfrak S_n: \pi(S_1) \cap S_2 = \pi(\mathsf{T}(S_1)) = \mathsf{T}(S_2) \big\} \,, & S_1,S_2 \in \mathfrak R_{\mathbf H}^*  \,; \\
            \emptyset \,, & \mbox{otherwise} \,.
        \end{cases}
    \end{equation}
    It is clear that $\mathcal A_{\diamond} (S_1,S_2)\subset \mathcal A_*(S_1,S_2)$ for all $(S_1,S_2)$. Recall \eqref{eq-def-f-*}. Using Lemma~\ref{lem-est-intersection-general-istree-P-Q}, we have
    \begin{align}
        \mathbb E_{\Pb_{\mathsf{id}}}\big[ f_{\mathcal A_*} \big] \geq \mathbb E_{\Pb_{\mathsf{id}}}\big[ f_{\mathcal A_{\diamond}} \big] \,. \label{eq-lower-bound-first-moment-1}
    \end{align}
    In addition, we have that for $S_1,S_2 \in \mathfrak R^*_{\mathbf H}$, $\mathsf{id} \in \mathcal A_{\diamond}(S_1,S_2)$ if and only if $\mathsf T(S_1)=\mathsf T(S_2)=S_1 \cap S_2$. Thus, we have
    \begin{align*}
        \mathbb E_{\Pb_{\mathsf{id}}}\big[ f_{\mathcal A_{\diamond}} \big] = \sum_{ \mathbf H \in \mathcal H } \frac{ s^{\aleph-1} \operatorname{Aut}(\mathsf{T}(\mathbf H)) (\epsilon^2 \lambda s)^{\ell\iota\aleph} }{ n^{\aleph+ \ell\iota\aleph} } \sum_{ \substack{S_1,S_2 \in \mathfrak R^*_{\mathbf H} \\ S_1 \cap S_2=\mathsf T(S_1)=\mathsf T(S_2)} } \mathbb E_{\Pb_{\mathsf{id}}}\big[ \phi_{S_1,S_2}(A,B) \big] \,.
    \end{align*}
    We now calculate $\mathbb E_{\Pb_{\mathsf{id}}}[\phi_{S_1,S_2}(A,B)]$ for $S_1 \cap S_2=\mathsf T(S_1)=\mathsf T(S_2)$ and $S_1,S_2 \in \mathfrak R_{\mathbf H}^*$. From Theorems~\ref{thm-num-desired-tree} and \ref{thm-desired-vertex-sets}, we see that $\mathsf{Vert}(\mathsf P(S_1)) \subset \mathsf{Fix}(\mathsf T(S_1))$, and thus $S_1 \cap S_2=\mathsf T(S_1)=\mathsf T(S_2)$ and $S_1,S_2 \in \mathfrak R_{\mathbf H}^*$ imply that $\mathsf P(S_1)=\mathsf P(S_2)$, as illustrated in Figure~\ref{fig:1}.  
    \begin{figure}[!ht]
        \centering
        \vspace{0cm}
        \includegraphics[height=4.34cm,width=13.78cm]{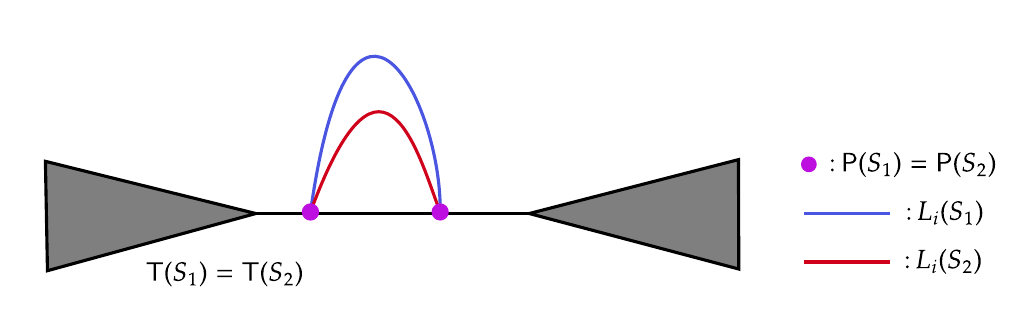}
        \caption{Intersection pattern for $\mathsf{id} \in \mathcal A_{\diamond}(S_1,S_2)$}
        \label{fig:1}
    \end{figure}
    Thus, applying Lemma~\ref{lem-est-intersection-istree-P-Q} and Lemma~\ref{lem-est-intersection-general-istree-P-Q}, we have $\mathbb E_{\Pb_{\mathsf{id}}}[ \phi_{S_1,S_2}(A,B) ]$ equals $(1+o(1))(1-\epsilon^{2 \mathtt m})^{\iota\aleph}$ times
    \begin{align}
        \big( \tfrac{\epsilon^2 \lambda s}{n} \big)^{\ell\iota\aleph} \cdot \mathbb E_{\Pb_{\mathsf{id}}}\Bigg[ \prod_{(i,j) \in E(\mathsf T(S_1))} \tfrac{ A_{i,j}-\frac{\lambda s}{n} }{ \sqrt{ \frac{\lambda s}{n} (1-\frac{\lambda s}{n}) } } \prod_{ (i,j) \in E(\mathsf T(S_2)) } \tfrac{ B_{i,j}-\frac{\lambda s}{n} }{ \sqrt{ \frac{\lambda s}{n} (1-\frac{\lambda s}{n}) } } \Bigg] = \big( \tfrac{\epsilon^2 \lambda s}{n} \big)^{\ell\iota\aleph} s^{\aleph-1} \,.  \label{eq-expectation-pi-in-A-diamond}
    \end{align}
    Thus, we have $\mathbb E_{\Pb_{\mathsf{id}}}\big[ f_{\mathcal A_{\diamond}} \big]$ equals $(1+o(1))(1-\epsilon^{2\mathtt m})^{\iota\aleph}$ times
    \begin{align}
        & \sum_{ \mathbf H \in \mathcal H } \frac{ s^{\aleph-1} \operatorname{Aut}(\mathsf{T}(\mathbf H)) (\epsilon^2 \lambda s)^{\ell\iota\aleph} }{ n^{\aleph+ \ell\iota\aleph} } \sum_{ \substack{S_1,S_2 \in \mathfrak R^*_{\mathbf H} \\ S_1 \cap S_2=\mathsf T(S_1)=\mathsf T(S_2)} } \big( \tfrac{\epsilon^2 \lambda s}{n} \big)^{\ell\iota\aleph} s^{\aleph-1} \nonumber \\
        =\ & \sum_{ \mathbf H \in \mathcal H } \frac{ s^{2(\aleph-1)} \operatorname{Aut}(\mathsf{T}(\mathbf H)) (\epsilon^2 \lambda s)^{2\ell\iota\aleph} }{ n^{\aleph+2\ell\iota\aleph} } \#\big\{ (S_1,S_2): S_1,S_2 \in \mathfrak R^*_{\mathbf H}, S_1 \cap S_2=\mathsf T(S_1)=\mathsf T(S_2) \big\} \nonumber \\
        \circeq\ & \sum_{ \mathbf H \in \mathcal H } \frac{ s^{2(\aleph-1)} \operatorname{Aut}(\mathsf{T}(\mathbf H)) (\epsilon^2 \lambda s)^{2\ell\iota\aleph} }{ n^{\aleph+2\ell\iota\aleph} } \cdot \frac{ n^{\aleph} }{ \operatorname{Aut}(\mathsf{T}(\mathbf H)) } \big( n^{\ell-1} \big)^{2\iota\aleph} \nonumber \\
        =\ & s^{2(\aleph-1)} \big( \tfrac{(\epsilon^2 \lambda s)^{\ell}}{n} \big)^{2\iota\aleph} |\mathcal H|  \,. \label{eq-lower-bound-first-moment-2}
    \end{align}
    Plugging \eqref{eq-lower-bound-first-moment-2} into \eqref{eq-lower-bound-first-moment-1} yields the desired lower bound for $\mathbb E_{\Pb_{\mathsf{id}}}[f_{\mathcal A_*}]$. 
    
    For the upper bound, recall that for all $S_1,S_2 \in \mathfrak R^*_{\mathbf H}$ such that $\mathsf{id}\in\mathcal A_*(S_1,S_2)$, we have $S_1 \cap S_2$ is a tree containing $\mathsf T(S_1) \cup \mathsf T(S_2)$. Thus, denoting $\mathsf P(S_1)=\{ (u_1,u_2), \ldots, (u_{2\iota\aleph-1},u_{2\iota\aleph}) \}$, there must exist paths $\mathcal L_1,\ldots,\mathcal L_{2\iota\aleph}$ with $u_i \in \operatorname{EndP}(\mathcal L_i)$ such that 
    \begin{align*}
        S_1 \cap S_2 = \mathsf{T}(S_1) \oplus \Big( \oplus_{u \in \mathsf{Vert}(\mathsf P (S_1))} \mathcal L_{u} \Big) \,,
    \end{align*}
    as illustrated in Figure~\ref{fig:2}. 
    \begin{figure}[!ht]
        \centering
        \vspace{0cm}
        \includegraphics[height=11.29cm,width=13.84cm]{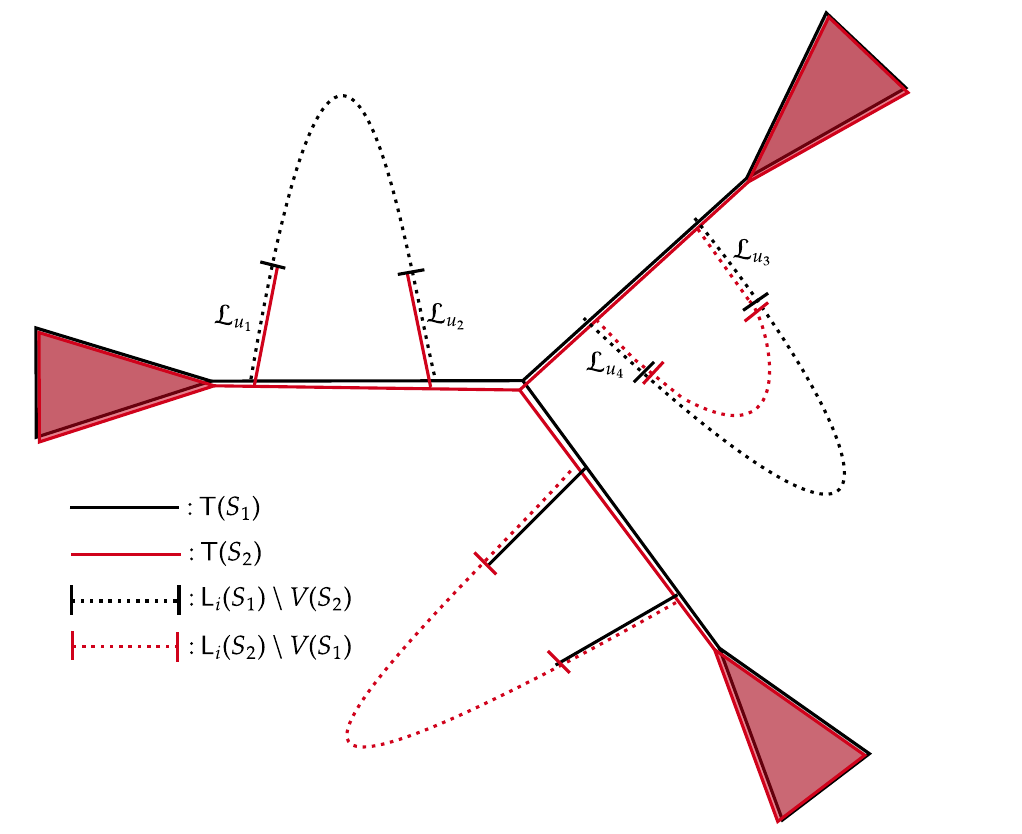}
        \caption{Intersection pattern for $\mathsf{id}\in \mathcal A_*(S_1,S_2)$}
        \label{fig:2}
    \end{figure} 
    Define 
    \begin{align}
        \mathcal O(\{x_u: u \in \mathsf{Vert}(\mathsf P(S_1)) \}) &= \{ S_2 \in \mathfrak R^*_{\mathbf H}:\mbox{ there exist paths } \mathcal L_u \mbox{ with }u \in \operatorname{EndP}(\mathcal L_u)\} \nonumber \\
        \mbox{ such that }|E(\mathcal L_u)|&=x_u\mbox{ and }S_1 \cap S_2 = \mathsf{T}(S_1) \oplus ( \oplus_{u \in \mathsf{Vert}(\mathsf P(S_1))} \mathcal L_{u} )\}\,. \label{eq-def-mathcal-O-set}
    \end{align}
    We first bound the cardinality of $\mathcal O(\{x_u: u \in \mathsf{Vert}(\mathsf P(S_1)) \})$. First note that we must have $\mathsf T(S_2) \subset \mathsf{T}(S_1) \oplus ( \oplus_{u \in \mathsf{Vert}(\mathsf P(S_1))} \mathcal L_{u} )$, and thus the number of possible choices of $\mathsf T(S_2)$ is bounded by
    \begin{align*}
        \#\big\{ T \subset \mathsf{T}(S_1) \oplus ( \oplus_{u \in \mathsf{Vert}(\mathsf P(S_1))} \mathcal L_{u} ): T \cong \mathsf T(\mathbf H) \big\} \leq \frac{ \operatorname{Aut}(\mathsf{T}(\mathbf H) \oplus ( \oplus_{u \in \mathsf{Vert}(\mathsf P(\mathbf H))} \mathcal L_{u} )) }{ \operatorname{Aut}(\mathsf T(\mathbf H)) } \,,
    \end{align*}
    where the inequality follows from Item~(5) in Lemma~\ref{lem-useful-property-trees-and-sets}. In addition, since $S_2 \Vdash \mathbf H$ and $\mathsf{Vert}(\mathsf P(S_2)) \subset \mathsf{Fix}(\mathsf T(S_2))$, we have that the relative position of $\mathsf P(S_2)$ on the tree $\mathsf T(S_2)$ is determined, which implies that the choice of $\mathsf P(S_2)$ is fixed once we have chosen $\mathsf T(S_2)$. Finally, it is straightforward to check that 
    \begin{align}
        & \#\Big( \cup_{1 \leq i \leq \iota\aleph} V(\mathsf L_i(S_2)) \setminus V(S_1) \Big) = (\ell-1)\iota\aleph-\sum x_u \,, \nonumber \\
        & \mathtt K(S_1 ,S_2 ; \mathsf{id}) = 0\,,\quad\mathtt K_c(S_1 ,S_2 ; \mathsf{id}) \geq \iota\aleph - \#\{ u : x_u > 0\} \geq \iota\aleph - \sum x_u \,. \label{eq-P-first-Acal-star-K-estimates}
    \end{align} 
    Thus, given $\mathsf T(S_2)$, the number of choices for $\{ \mathsf L_i(S_2): 1 \leq i \leq \iota \aleph \}$ is bounded by $n^{ (\ell-1)\iota\aleph-\sum x_u }$. As a result, we have
    \begin{align*}
        \#\mathcal O(\{x_u: u \in \mathsf{Vert}(\mathsf P(S_1)) \}) \leq \frac{ \operatorname{Aut}(\mathsf T(\mathbf H) \oplus (\oplus_{u\in\mathsf{Vert}(\mathsf P(\mathbf H))} \mathcal L_{u}))\cdot  n^{ (\ell-1)\iota\aleph-\sum x_u }  }{ \operatorname{Aut}(\mathsf T(\mathbf H)) } \,.
    \end{align*}
    Also, by \eqref{eq-S,K-intersection-general-A,B}, \eqref{eq-phi-P-intersec-istree} and \eqref{eq-upper-bound-exp-on-trees-2} we obtain $\mathbb E_{\Pb_{\mathsf{id}}}\big[ \phi_{S_1,S_2}(A,B) \big]$ is bounded by $[1+o(1)]$ times
    \begin{align}\label{eq-expectation-pi-in--A}
        &s^{\aleph-1+\sum_{u\in\mathsf{Vert}(\mathsf P(S_2))} x_u} \big( \tfrac{\epsilon^2\lambda s}{n} \big)^{\iota\aleph-\sum_{u\in\mathsf{Vert}(\mathsf P(S_2))} x_u} (1-\epsilon^{2\mathtt m})^{\mathtt K_c(S_1 , S_2;\mathsf{id})}\cdot h_n^{\mathtt K(S_1 , S_2 ; \mathsf{id})} \cdot \big( \tfrac{\epsilon^2\lambda s }{n} \big)^{(\ell -1)\iota\aleph}\nonumber \\
        \leq\ & s^{\aleph-1+\sum_{u\in\mathsf{Vert}(\mathsf P(S_2))} x_u} \big( \tfrac{\epsilon^2\lambda s(1-\epsilon^{2\mathtt m})}{n} \big)^{\iota\aleph-\sum_{u\in\mathsf{Vert}(\mathsf P(S_2))} x_u} \cdot \big( \tfrac{\epsilon^2\lambda s }{n} \big)^{(\ell -1)\iota\aleph}\,,
    \end{align}
    where the inequality holds by \eqref{eq-P-first-Acal-star-K-estimates}. Thus (writing $\mathbf{x}\geq 0$ for $x_u\geq 0$ for all $u\in\mathsf{P}(S_2)$),
    \begin{align}\label{eq-first-moment-P-id--A}
        & \mathbb E_{\Pb_{\mathsf{id}}}\big[ f_{\mathcal A_*} \big] = \sum_{ \mathbf H \in \mathcal H } \frac{ s^{\aleph-1} \operatorname{Aut}(\mathsf{T}(\mathbf H)) (\epsilon^2 \lambda s)^{\ell\iota\aleph} (1-\epsilon^{2\mathtt m})^{\iota\aleph}}{ n^{\aleph+ \ell\iota\aleph} } \sum_{ S_1 \in \mathfrak R^*_{\mathbf H} } \sum_{\mathbf{x} \geq 0} \sum_{ S_2 \in \mathcal O(\{x_u\}) } \mathbb E_{\Pb_{\mathsf{id}}}\big[ \phi_{S_1,S_2}(A,B) \big] \nonumber \\
        \leq\ & \sum_{ \mathbf H \in \mathcal H } \operatorname{Aut}(\mathsf{T}(\mathbf H)) \frac{ s^{2(\aleph-1)}  (\epsilon^2 \lambda s)^{2\ell\iota\aleph}(1-\epsilon^{2\mathtt m})^{\iota\aleph} }{ n^{\aleph+ (\ell+1)\iota\aleph} } \sum_{ S_1 \in \mathfrak R^*_{\mathbf H} } \sum_{\mathbf{x} \geq 0} \frac{ \operatorname{Aut}(\mathsf T(\mathbf H) \oplus (\oplus_{u\in\mathsf{Vert}(\mathsf P(\mathbf H))} \mathcal L_{u})) }{ \operatorname{Aut}(\mathsf T(\mathbf H)) (\epsilon^2 \lambda s(1 - \epsilon^{2 \mathtt m}))^{\sum{x_u}} } \nonumber \\
        \leq\ & \sum_{ \mathbf H \in \mathcal H } \operatorname{Aut}(\mathsf{T}(\mathbf H)) \frac{ e^{4\iota\aleph} s^{2(\aleph-1)} (\epsilon^2 \lambda s)^{2\ell\iota\aleph}(1-\epsilon^{2\mathtt m})^{\iota\aleph} }{ n^{\aleph+ (\ell+1)\iota\aleph} } (1-\tfrac{1}{\epsilon^2 \lambda s(1 - \epsilon^{2 \mathtt m})})^{-2\iota\aleph} |\mathfrak R^*_{\mathbf H}| \nonumber \\
        \leq\ & \sum_{ \mathbf H \in \mathcal H } s^{2(\aleph-1)} \Big( \tfrac{ e^2 (\epsilon^2 \lambda s)^{\ell}(1-\epsilon^{2\mathtt m})^{1/2}}{n} \Big)^{2\iota\aleph} (1-\tfrac{1}{\epsilon^2 \lambda s(1 - \epsilon^{2 \mathtt m})})^{-2\iota\aleph} \nonumber \\
        \leq\ & s^{2(\aleph-1)} \cdot (3e^{2}\Delta^{-1})^{2\iota\aleph} \Big( \tfrac{(\epsilon^2 \lambda s)^{\ell}(1-\epsilon^{2\mathtt m})^{1/2}}{n} \Big)^{2\iota\aleph} |\mathcal H|  \,. 
    \end{align}
    where the second inequality follows from Item~(4) in Lemma~\ref{lem-useful-property-trees-and-sets} and \eqref{eq-assumption-s,lambda}, the third inequality follows from  
    \begin{equation}{\label{eq-bound-card-R^*_H}}
        |\mathfrak R^{*}_{\mathbf H}| \leq \frac{ n^{\aleph} }{ \operatorname{Aut}(\mathsf T(\mathbf H)) } \cdot \big(n^{\ell-1} \big)^{\iota\aleph} \,,
    \end{equation}
    and the fourth inequality follows from \eqref{eq-choice-mathtt-m}. This leads to the desired upper bound of $\mathbb E_{\Pb_{\mathsf{id}}}[f_{\mathcal A_*}]$.
\end{proof}

Based on Lemma~\ref{lem-first-moment-good-part}, it remains to deal with the parts $\pi_* \in\mathcal A(S_1,S_2) \setminus \mathcal A_*(S_1,S_2)$ and $\pi_* \not\in \mathcal A(S_1,S_2)$. For notational convenience, for $S_1,S_2\in \cup_{\mathbf H\in\mathcal{H}}\mathfrak R_{\mathbf{H}}$ we denote 
\begin{align}\label{eq-def-tilde-mathcal-A-mathcal-A-c}
    \widetilde{\mathcal{A}}(S_1,S_2)= \mathcal{A}(S_1,S_2) \setminus \mathcal{A}_*(S_1,S_2)\,,\quad \mathcal{A}^c(S_1,S_2)=\mathfrak S_n\setminus \mathcal{A}(S_1,S_2)\,.
\end{align}
And we further define
\begin{align}
    &\mathbb E^{|\cdot|}_{\Pb}\big[ f_{\widetilde{\mathcal A}} \big] = \sum_{ \mathbf H \in \mathcal H } \frac{ s^{\aleph-1} \operatorname{Aut}(\mathsf{T}(\mathbf H)) (\epsilon^2 \lambda s)^{\ell\iota\aleph} }{ n^{\aleph+ \ell\iota\aleph} }  \sum_{ S_1,S_2 \in \mathfrak R_{\mathbf H} } \Big| \mathbb E_{\Pb}\big[ \mathbf{1}_{\{ \pi_* \in \widetilde{\mathcal A}(S_1,S_2) \}} \phi_{S_1,S_2}(A,B) \big]  \Big|\,, \label{eq-def-abs-E-Pb-f-tilde-A} \\
    &\mathbb E^{|\cdot|}_{\Pb}\big[ f_{\mathcal A^c} \big] = \sum_{ \mathbf H \in \mathcal H } \frac{ s^{\aleph-1} \operatorname{Aut}(\mathsf{T}(\mathbf H)) (\epsilon^2 \lambda s)^{\ell\iota\aleph} }{ n^{\aleph+ \ell\iota\aleph} } \sum_{ S_1,S_2 \in \mathfrak R_{\mathbf H} } \Big| \mathbb E_{\Pb}\big[ \mathbf{1}_{\{ \pi_* \in \mathcal A^c(S_1,S_2) \}} \phi_{S_1,S_2}(A,B) \big]  \Big|\,.  \label{eq-def-abs-E-Pb-f-A-c}
\end{align}

\begin{lemma}{\label{lem-first-moment-bad-part}}
    Recall \eqref{eq-def-f-*}. We have 
    \begin{align} 
        & \mathbb E^{|\cdot|}_{\Pb}\big[ f_{\widetilde{\mathcal A}} \big] = o(1) \cdot s^{2(\aleph-1)}  \Big( \tfrac{ (\epsilon^2 \lambda s)^{\ell} (1-\epsilon^\mathtt m )^{1/2} }{ n } \Big)^{2\iota\aleph} |\mathcal H| {\label{eq-first-moment-semi-bad-part}}\,,\\
        & \mathbb E^{|\cdot|}_{\Pb}\big[ f_{\mathcal A^c} \big] = o(1) \cdot s^{2(\aleph-1)}  \Big( \tfrac{ (\epsilon^2 \lambda s)^{\ell} (1-\epsilon^\mathtt m )^{1/2} }{ n } \Big)^{2\iota\aleph} |\mathcal H| \,.{\label{eq-first-moment-bad-part}}
    \end{align}
\end{lemma}

\begin{proof}[Proof of Proposition~\ref{prop-first-moment}]
    It suffices to note that Proposition~\ref{prop-first-moment} follows directly from combining Lemmas~\ref{lem-first-moment-good-part} and \ref{lem-first-moment-bad-part}.
\end{proof}

The rest of this section is devoted to the proof of Lemma~\ref{lem-first-moment-bad-part}. To this end, note that by symmetry we have $\mathbb E_{\Pb}^{|\cdot|}[f_{\widetilde{\mathcal A}}]=\mathbb E^{|\cdot|}_{\Pb_\pi}[f_{\widetilde{\mathcal A}}]$ and $\mathbb E_{\Pb}^{|\cdot|}[f_{\mathcal A^{c}}]=\mathbb E^{|\cdot|}_{\Pb_\pi}[f_{\mathcal A^c}]$ for any $\pi\in\mathfrak S_n$. Thus, it suffices to show \eqref{eq-first-moment-semi-bad-part} and \eqref{eq-first-moment-bad-part} for $\Pb_{\mathsf{id}}$. We prove \eqref{eq-first-moment-semi-bad-part} first. Recall \eqref{eq-def-mathcal-O-set}. By \eqref{eq-S,K-intersection-general-A,B}, \eqref{eq-phi-P-intersec-istree} and \eqref{eq-upper-bound-exp-on-trees-2}, for all $S_1 , S_2 \in \mathfrak R^*_\mathbf H$ such that $\mathsf{id} \in \widetilde{\mathcal A}(S_1,S_2)$ and $S_2 \in \mathcal O(\{x_u: u \in \mathsf{Vert}(\mathsf P(S_1)) \})$ we have $\big|\mathbb E_{\Pb_{\mathsf{id}}}\big[ \phi_{S_1,S_2}(A,B) \big]\big|$ is bounded by $[1+o(1)]$ times 
\begin{align}\label{eq-expectation-pi-in--tilde-A}
    &s^{\aleph-1+\sum_{u\in\mathsf{Vert}(\mathsf P(S_2))} x_u} \big( \tfrac{\epsilon^2\lambda s}{n} \big)^{\iota\aleph-\sum_{u\in\mathsf{Vert}(\mathsf P(S_2))} x_u} (1-\epsilon^{2\mathtt m})^{\mathtt K_c(S_1 , S_2;\mathsf{id})}\cdot |h_n|^{\mathtt K(S_1 , S_2 ; \mathsf{id})} \cdot \big( \tfrac{\epsilon^2\lambda s }{n} \big)^{(\ell -1)\iota\aleph}\nonumber \\
    \leq\ & s^{\aleph-1+\sum_{u\in\mathsf{Vert}(\mathsf P(S_2))} x_u} \big( \tfrac{\epsilon^2\lambda s(1-\epsilon^{2\mathtt m})}{n} \big)^{\iota\aleph-\sum_{u\in\mathsf{Vert}(\mathsf P(S_2))} x_u} \cdot \big(\tfrac{|h_n|}{1-\epsilon^{2\mathtt m}}\big)^{\mathtt K(S_1 , S_2 ; \mathsf{id})} \cdot \big( \tfrac{\epsilon^2\lambda s }{n} \big)^{(\ell -1)\iota\aleph}\nonumber \\
    \leq\ & s^{\aleph-1+\sum_{u\in\mathsf{Vert}(\mathsf P(S_2))} x_u} \big( \tfrac{\epsilon^2\lambda s(1-\epsilon^{2\mathtt m})}{n} \big)^{\iota\aleph-\sum_{u\in\mathsf{Vert}(\mathsf P(S_2))} x_u} \cdot \tfrac{|h_n|}{1-\epsilon^{2\mathtt m}} \cdot \big( \tfrac{\epsilon^2\lambda s }{n} \big)^{(\ell -1)\iota\aleph}\,,
\end{align}
where the first inequality holds by
\begin{align*}
    \mathtt K_c(S_1 , S_2;\mathsf{id}) + \mathtt K(S_1 , S_2;\mathsf{id}) &\geq \#\{ i \in [\iota\aleph] : \mathsf L_i(S_2) \cap \pi( S_1) = \operatorname{EndP}(\mathsf L_i(S_2))\} \\
    &\geq \iota\aleph - \#\{ u \in \mathsf{Vert}(\mathsf P (S_2)):x_u > 0\} \geq \iota\aleph - \sum_{u \in \mathsf{Vert}(\mathsf P (S_2))} x_u\,,
\end{align*}
and the second inequality holds by the fact that $\mathtt K(S_1 , S_2 ; \mathsf{id}) \geq 1$ whenever $\mathsf{id} \in \widetilde{\mathcal A}(S_1 , S_2)$. Therefore, similar to \eqref{eq-first-moment-P-id--A} we have (writing $\mathbf{x}\geq 0$ for $x_u\geq 0$ for all $u\in\mathsf{P}(S_2)$),
\begin{align}\label{eq-first-moment-P-id--tilde-A}
    &\mathbb E^{|\cdot|}_{\Pb_{\mathsf{id}}}\big[ f_{\widetilde{\mathcal A}} \big] = \sum_{ \mathbf H \in \mathcal H } \frac{ s^{\aleph-1} \operatorname{Aut}(\mathsf{T}(\mathbf H)) (\epsilon^2 \lambda s)^{\ell\iota\aleph} (1-\epsilon^{2\mathtt m})^{\iota\aleph}}{ n^{\aleph+ \ell\iota\aleph} } \sum_{ S_1 \in \mathfrak R^*_{\mathbf H} } \sum_{\mathbf{x} \geq 0} \sum_{ \substack{S_2 \in \mathcal O(\{x_u\}) \\ \mathsf{id} \in \widetilde{\mathcal A}(S_1,S_2)}}\big| \mathbb E_{\Pb_{\mathsf{id}}}\big[ \phi_{S_1,S_2}(A,B) \big] \big| \nonumber \\
    \leq\ & 2|h_n|\sum_{ \mathbf H \in \mathcal H } \operatorname{Aut}(\mathsf{T}(\mathbf H)) \frac{ s^{2(\aleph-1)}  (\epsilon^2 \lambda s)^{2\ell\iota\aleph}(1-\epsilon^{2\mathtt m})^{\iota\aleph} }{ n^{\aleph+ (\ell+1)\iota\aleph} } \sum_{ S_1 \in \mathfrak R^*_{\mathbf H} } \sum_{\mathbf{x} \geq 0} \frac{ \operatorname{Aut}(\mathsf T(\mathbf H) \oplus (\oplus_{u\in\mathsf{Vert}(\mathsf P(\mathbf H))} \mathcal L_{u})) }{ \operatorname{Aut}(\mathsf T(\mathbf H)) (\epsilon^2 \lambda s(1 - \epsilon^{2 \mathtt m}))^{\sum{x_u}} } \nonumber \\
    \leq\ & 2|h_n| s^{2(\aleph-1)} e^{18\iota\aleph} \Big( \tfrac{(\epsilon^2 \lambda s)^{\ell}(1-\epsilon^{2\mathtt m})^{1/2}}{n} \Big)^{2\iota\aleph} |\mathcal H| = o(1) \cdot s^{2(\aleph-1)} \Big( \tfrac{(\epsilon^2 \lambda s)^{\ell}(1-\epsilon^{2\mathtt m})^{1/2}}{n} \Big)^{2\iota\aleph} |\mathcal H|\,, 
\end{align}
which yields \eqref{eq-first-moment-semi-bad-part}. Next we prove \eqref{eq-first-moment-bad-part}. Recall Definition~\ref{def-decorated-trees}. For $S \Vdash \mathbf H$ where $\mathbf H$ is a decorated tree and $I \subset [\iota \aleph]$, denote 
\begin{align}\label{eq-def-I-over-S}
    S(I) =\cup_{i \in I} \mathsf L_i(S) \,,\quad S^I = S(I) \cup \mathsf T(S)\,.
\end{align}
By \eqref{eq-def-beta-S}, \eqref{eq-def-psi-S} and \eqref{eq-def-f-S}, we see that
\begin{align}
    \phi_S(X) &= \beta_{\mathsf T(S)}(X) \sum_{I \subset [\iota\aleph]} \Big( -\epsilon^{\mathtt m} \big( \tfrac{\epsilon^2 \lambda s}{n} \big)^{\ell/2} \Big)^{\iota\aleph-|I|} \beta_{S(I)}(X)\nonumber \\
    &= \sum_{I \subset [\iota\aleph]} \Big( -\epsilon^{\mathtt m} \big( \tfrac{\epsilon^2 \lambda s}{n} \big)^{\ell/2} \Big)^{\iota\aleph-|I|} \beta_{S^I}(X)\,. \label{eq-phi-as-linear-combination-of-beta}
\end{align}
For $S_1,S_2 \in \mathfrak R_{\mathbf H}$ and $I,J \subset [\iota\aleph]$, define $G^{I,J}_{\cup}$ to be the (simple) graph obtained by
\begin{align*}
    G^{I,J}_{\cup} = \widetilde{S}^I_1 \cup \widetilde{S}^J_2  \,,
\end{align*}
and in particular we let $G_\cup = G^{[\iota\aleph],[\iota\aleph]}_\cup$. From Item~(3) in Lemma~\ref{lem-useful-property-trees-and-sets}, we know that there must exist $v \in \big(\mathcal L( \mathsf T(S_1 ))\cup \mathcal L( \mathsf T(S_2 ))\big)\setminus \big(V(S_1 \cap S_2)\big)$ if $\mathcal L(S_1) \cup \mathcal L( S_2) \not\subset V(S_1 \cap S_2)$, which implies that $v \in \mathcal L(S_1^I \cup S_2^J)$ for all $I,J \subset [\iota\aleph]$ and therefore $\mathbb E_{\mathbb P_\mathsf{id}}[ \phi_{S_1,S_2}(A,B) ]=0$ from \eqref{eq-phi-as-linear-combination-of-beta} and Lemma~\ref{lem-upper-bound-exp}. In addition, we have $\mathcal L(G_{\cup}) \subset \mathcal L(\widetilde S_1) \cup  \mathcal L(\widetilde S_2)$ and thus $|\mathcal L(G_{\cup})| \leq 2\aleph$. Also we have
\begin{align}
    &\big(\epsilon^\mathtt m (\tfrac{\epsilon^2 \lambda s}{n})^{\ell/2} \big)^{2\iota\aleph - |I| - |J|}\E_{\P_\mathsf{id}} [\beta_{S^I_1}(A)\beta_{S^J_2}(B)] \nonumber \\
    =\ & \mathbb{E}_{\mathbb P_\mathsf{id}} \Bigg[  \prod_{(i,j) \in E(\widetilde{S}^I_1)} \Big( \tfrac{ A_{i,j}-\frac{\lambda s}{n} }{ \sqrt{\lambda s/n} } \Big)^{E(S^I_1)_{i,j}} \prod_{(i,j) \in E(\widetilde S^J_2)} \Big( \tfrac{ B_{i,j}-\frac{\lambda s}{n} }{ \sqrt{\lambda s/n} } \Big)^{E(S^J_2)_{i,j}}  \Bigg] * \big(\epsilon^\mathtt m (\tfrac{\epsilon^2 \lambda s}{n})^{\ell/2} \big)^{2\iota\aleph - |I| - |J|}\nonumber  \\
    =\ & \mathbb{E}_{\sigma\sim\nu} \Bigg\{  \prod_{(i,j) \in E(G_{\cup} \setminus \widetilde S^J_2)} \mathbb{E}_{\mathbb P_{\mathsf{id},\sigma}} \Bigg[ \Big( \tfrac{ A_{i,j}-\frac{\lambda s}{n} }{ \sqrt{\lambda s/n} } \Big)^{E(S^I_1)_{i,j}} \Bigg] \prod_{(i,j) \in E(G_{\cup} \setminus \widetilde S^I_1)} \mathbb{E}_{\mathbb P_{\mathsf{id},\sigma}} \Bigg[ \Big( \tfrac{ B_{i,j}-\frac{\lambda s}{n} }{ \sqrt{\lambda s/n} } \Big)^{E(S^J_2)_{i,j}} \Bigg] \nonumber \\ 
    & \prod_{(i,j) \in E(\widetilde S^I_1 \cap \widetilde S^J_2)} \mathbb{E}_{\mathbb P_{\mathsf{id},\sigma}} \Bigg[ \Big( \tfrac{ A_{i,j}-\frac{\lambda s}{n} }{ \sqrt{\lambda s/n} } \Big)^{E(S^I_1)_{i,j}} \Big( \tfrac{ B_{i,j}-\frac{\lambda s}{n} }{ \sqrt{\lambda s/n} } \Big)^{E(S^J_2)_{i,j}} \Bigg]  \Bigg\} * \big(\epsilon^\mathtt m (\tfrac{\epsilon^2 \lambda s}{n})^{\ell/2} \big)^{2\iota\aleph - |I| - |J|} \,. \label{eq-Pb-f-setminus-*-relax-1}
\end{align}
Using Lemmas~\ref{lem-joint-moment-A-B} and \ref{lem-upper-bound-exp} (with $V=\emptyset$), we see that \eqref{eq-Pb-f-setminus-*-relax-1} is bounded by (below we denote $E(S)_{i,j}=0$ if $(i,j) \not \in E(S)$)
\begin{align*}
    2^{5\tau(G^{I,J}_{\cup})+12\aleph} \prod_{(i,j) \in E(G^{I,J}_{\cup})} \Big( \tfrac{\sqrt{n}}{\sqrt{\epsilon^2\lambda s}} \Big)^{E(S^I_1)_{i,j}+E(S^J_2)_{i,j}-2} \big(\epsilon^\mathtt m (\tfrac{\epsilon^2 \lambda s}{n})^{\ell/2} \big)^{2\iota\aleph - |I| - |J|}\,. 
\end{align*}
Using the fact that $\tau(G^{I,J}_{\cup}) \leq \tau(G_{\cup})$ and
\begin{align*}
    \sum_{(i,j) \in E(G^{I,J}_{\cup})} (E(S^I_1)_{i,j}+E(S^J_2)_{i,j}-2)&= 2(\aleph-1)+\ell(|I|+|J|)-2|E(G^{I,J}_{\cup})| \\
    &\leq 2(\aleph - 1 + \iota\ell\aleph) - 2|E(G_\cup)| - \ell(|I|+|J|) \,,
\end{align*}
we obtain that \eqref{eq-Pb-f-setminus-*-relax-1} is further bounded by 
\begin{align}
    2^{5\tau(G_{\cup})+12\aleph} n^{\aleph-1+\iota\ell\aleph-|E(G_{\cup})|} / (\epsilon^2\lambda s)^{\aleph-1+\iota\ell\aleph-|E(G_{\cup})|} \,.  \label{eq-Pb-f-setminus-*-relax-2}
\end{align}
Combined with \eqref{eq-def-f} and \eqref{eq-phi-as-linear-combination-of-beta}, it yields that 
\begin{align}
    & \mathbb{E}^{|\cdot |}_{\mathbb P_{\mathsf{id}}}[f_{\mathcal A^c}] \leq \sum_{\mathbf H \in \mathcal H} \sum_{\substack{S_1,S_2\Vdash \mathbf H \\ \mathsf{id} \not\in\mathcal A(S_1,S_2)} } 2^{2\iota\aleph} \cdot \frac{ 2^{5\tau(G_{\cup})+12\aleph} s^{\aleph-1} (\epsilon^2 \lambda s)^{\ell\iota\aleph} \operatorname{Aut}(\mathsf T(\mathbf H)) }{ n^{\aleph+\ell\iota\aleph} } \cdot \big( \tfrac{ \epsilon^2\lambda s }{ n } \big)^{ |E(G_{\cup})|-\ell\iota\aleph-\aleph+1 } \nonumber \\
    =\ &n^{-1+o(1)} \sum_{ \substack{ |\mathcal L(G_{\cup})|\leq \aleph \\ |E(G_{\cup})| \leq 2\ell\iota\aleph } } 2^{5\tau(G_{\cup})} \big( \tfrac{ \epsilon^2\lambda s }{ n } \big)^{ |E(G_{\cup})| } \operatorname{ENUM}(G_{\cup}) \,, \label{eq-Pb-f-setminus-*-relax-3} 
\end{align}
where the factor $2^{2\iota\aleph}$ comes from the enumeration of the pairs $(I,J)$ and
\begin{equation}{\label{eq-def-ENUM-G}}
\begin{aligned}
    \operatorname{ENUM}(G_{\cup}) := \#\Big( \cup_{\mathbf H \in \mathcal H} \big\{ & (S_1,S_2) \in \mathfrak R_{\mathbf H} \times \mathfrak R_{\mathbf H}: \widetilde{S}_1 \cup \widetilde{S}_2= G_{\cup}, \\
    & \mathcal L(S_1) \cup \mathcal L(S_2) \subset V(S_1) \cap V(S_2), \mathsf{id} \not\in\mathcal A(S_1,S_2) \big\}  \Big) \,.
\end{aligned}
\end{equation} 
Now it suffices to control \eqref{eq-Pb-f-setminus-*-relax-3}, for which the following lemma is useful.
\begin{lemma}{\label{lem-character-non-principle}}
    For all $S_1,S_2 \in \mathfrak R_{\mathbf H}$ such that $\mathsf{id} \not\in \mathcal A(S_1,S_2)$ and $\mathcal L(S_1) \cup \mathcal L(S_2) \subset V(S_1 \cap S_2)$, we have
    \begin{equation}\label{eq-character-non-principle}
        2\iota\aleph-\tau(G_{\cup}) \leq \tfrac{2\iota\ell\aleph+2\aleph-|E(G_{\cup})|}{\ell/2} \,.
    \end{equation}
\end{lemma}
 
The proof of Lemma~\ref{lem-character-non-principle} is incorporated in Section~\ref{subsec:proof-lem-3.10}. We can now complete the proof of Lemma~\ref{lem-first-moment-bad-part}.

\begin{proof}[Proof of Lemma~\ref{lem-first-moment-bad-part}]
    By Lemma~\ref{lem-character-non-principle}, we have (in what follows we say $\mathbf{G}$ satisfies \eqref{eq-character-non-principle} if \eqref{eq-character-non-principle} holds after replacing $G_\cup$ by $\mathbf{G}$ and we denote by $\mathrm{ENUM}(G)$ as \eqref{eq-def-ENUM-G} after replacing $G_\cup$ with $G$))
    \begin{align}
        \eqref{eq-Pb-f-setminus-*-relax-3} = n^{-1+o(1)} \sum_{ \substack{ |\mathcal L(\mathbf G)| \leq \aleph, |E(\mathbf G)| \leq 2\ell\iota\aleph+2\aleph \\ \mathbf G \text{ satisfying } \eqref{eq-character-non-principle} } } \sum_{ G \cong \mathbf G } 2^{5\tau(G)} \big( \tfrac{ \epsilon^2\lambda s }{ n } \big)^{ |E(G)| } \operatorname{ENUM}(G) \,, \label{eq-Pb-f-setminus-*-relax-4}
    \end{align}
    where the summation is taken over all unlabeled graph $\mathbf G$ such that $\operatorname{ENUM}(G)>0$ for $G \cong \mathbf G$.
    In addition, using Lemma~\ref{lem-enu-decorated-trees-in-given-graph}, for $G$ satisfying \eqref{eq-character-non-principle}, we have 
    \begin{align*}
         \operatorname{ENUM}(G) \leq \binom{|E(G)|}{\aleph}^2\cdot \aleph ^{4\iota\aleph} 2^{2\iota\aleph(\aleph+5(\tau(G)+1))} = n^{o(1)}\,.
    \end{align*}
    Thus, we have 
    \begin{align}
        \eqref{eq-Pb-f-setminus-*-relax-4} &\leq n^{-1+o(1)} \cdot \sum_{ \substack{ |\mathcal L(\mathbf G)|\leq \aleph, |E(\mathbf G)| \leq 2\ell\iota\aleph+2\aleph \\ \mathbf G \text{ satisfying } \eqref{eq-character-non-principle} \\ \operatorname{ENUM}(\mathbf G )> 0 } } 2^{5\tau(\mathbf G)} \big( \tfrac{ \epsilon^2\lambda s }{ n } \big)^{ |E(\mathbf G)| } \cdot \#\big\{ G \subset \mathsf K_n: G \cong \mathbf G \big\} \nonumber \\
        &\leq n^{-1+o(1)} \cdot \sum_{ \substack{ |\mathcal L(\mathbf G)|\leq\aleph, |E(\mathbf G)| \leq 2\ell\iota\aleph+2\aleph \\ \mathbf G \text{ satisfying } \eqref{eq-character-non-principle} \\ \operatorname{ENUM}(\mathbf G )> 0 } }  2^{5\tau(\mathbf G)} \big( \tfrac{ \epsilon^2\lambda s }{ n } \big)^{ |E(\mathbf G)| } n^{|V(\mathbf G)|} 
        \nonumber \\
        &\leq n^{-1+o(1)} (\epsilon^2 \lambda s)^{2\ell\iota\aleph} \sum_{ \substack{ |\mathcal L(\mathbf G)|\leq\aleph, |E(\mathbf G)| \leq 2\ell\iota\aleph+2\aleph \\ \mathbf G \text{ satisfying } \eqref{eq-character-non-principle} \\ \operatorname{ENUM}(\mathbf G )> 0} } 2^{5\tau(\mathbf G)} n^{-\tau(\mathbf G)} 
        (\epsilon^2 \lambda s)^{|E(\mathbf G)|-2\ell\iota\aleph} \nonumber \\
        &\leq n^{-1+o(1)} (\epsilon^2 \lambda s)^{2\ell\iota\aleph} \sum_{\substack{0 \leq x \leq 2\ell\iota\aleph + 2\aleph \\ 2(2\ell\iota\aleph + 2\aleph-x)/\ell \geq 2\iota\aleph-y } } \big( \tfrac{32(\ell\aleph)^3}{n} \big)^{y} (\epsilon^2 \lambda s)^{x-2\ell\iota\aleph}\nonumber \\
        &\leq n^{-1+o(1)} (\epsilon^2 \lambda s)^{2\ell\iota\aleph} \sum_{0 \leq x \leq 2\ell\iota\aleph + 2\aleph } \big( \tfrac{32(\ell\aleph)^3}{n} \big)^{2\iota\aleph} (\tfrac{(\epsilon^2 \lambda s)^{\ell}}{n^{2}})^{(x-2\ell\iota\aleph)/\ell}\nonumber \\
        &\leq n^{-1-2\iota\aleph+o(1)} (\epsilon^2 \lambda s)^{2\ell\iota\aleph} \,, \label{eq-Pb-f-setminus-*-relax-5}
    \end{align}
    where the second inequality follows from $\#\{ G \subset \mathsf K_n: G \cong \mathbf G \}= \frac{n^{|V(\mathbf G)|}}{\operatorname{Aut}(\mathbf G)}$, the fourth inequality holds by Lemma~\ref{lem-enu-union-of-decorated-trees} and the last inequality follows from $(\ell\aleph)^{\aleph}=n^{o(1)}$. Plugging \eqref{eq-Pb-f-setminus-*-relax-5} into \eqref{eq-Pb-f-setminus-*-relax-3}, we obtain
    \begin{align*}
        \eqref{eq-Pb-f-setminus-*-relax-3} \leq n^{-1+o(1)} \cdot \Big( \tfrac{(\epsilon^2 \lambda s)^{\ell}}{n} \Big)^{2\iota\aleph} \overset{\eqref{eq-choice-aleph}}{=} o(1) \cdot s^{2(\aleph-1)} |\mathcal H| \cdot \Big( \tfrac{ (\epsilon^2\lambda s)^{\ell} (1-\epsilon^\mathtt m)^{1/2} }{n} \Big)^{2\iota\aleph} \,,
    \end{align*}
    as desired. 
\end{proof}

\section{Estimation of second moment}{\label{sec:bound-var}}

The main goal of this section is to prove the following proposition. 
\begin{proposition}{\label{prop-second-moment}}
    We have the following estimates:
    \begin{itemize}
        \item[(1)] $\mathbb E_{\mathbb Q}[f^2] = o(1) \cdot \mathbb E_{\Pb}[f]^2$;
        \item[(2)] $\operatorname{Var}_{\mathbb P}[f] = o(1) \cdot \mathbb E_{\Pb}[f]^2$.
    \end{itemize}
\end{proposition}
In particular, combining Proposition~\ref{prop-second-moment} with Proposition~\ref{prop-first-moment} and Lemma~\ref{lem-enu-decorated-trees} (which implies that $\mathbb E_{\mathbb P}[f]\to \infty$) yields Theorem~\ref{THM-validity-of-statistics} (note that Item (1) also implies that $\operatorname{Var}_{\Qb}[f],\mathbb E_{\Qb}[f]^2 \leq \mathbb E_{\Qb}[f^2] = o(1) \cdot \mathbb E_{\Pb}[f]^2$). The remainder of this section is devoted to proving Proposition~\ref{prop-second-moment}.

\subsection{Proof of Item (1)}{\label{subsec:bound-Q-var}}

Recall \eqref{eq-def-f}. By the independence of $A$ and $B$ under $\Qb$, we have
\begin{align}
    \mathbb E_{\Qb}[f^2] &= \sum_{\mathbf H,\mathbf I \in \mathcal H} \frac{ s^{2(\aleph-1)} (\epsilon^2 \lambda s)^{2\ell\iota\aleph} \operatorname{Aut}(\mathsf T(\mathbf H)) \operatorname{Aut}(\mathsf T(\mathbf I)) }{ n^{2(\aleph+\ell\iota\aleph)} } \sum_{ \substack{S_1,S_2\Vdash \mathbf H \\ K_1,K_2 \Vdash \mathbf I} } \mathbb E_{\Qb}[\phi_{S_1,S_2} \phi_{K_1,K_2}] \nonumber \\
    &= \sum_{\mathbf H,\mathbf I \in \mathcal H} \frac{ s^{2(\aleph-1)} (\epsilon^2 \lambda s)^{2\ell\iota\aleph} \operatorname{Aut}(\mathsf T(\mathbf H)) \operatorname{Aut}(\mathsf T(\mathbf I)) }{ n^{2(\aleph+\ell\iota\aleph)} } \Big( \sum_{ S \Vdash \mathbf H,K \Vdash \mathbf I } \big| \mathbb E_{\Qb}[\phi_S(A) \phi_K(A)] \big| \Big)^2 \,. \label{eq-Qb-second-moment-relax-1}
\end{align}
Define $\mathfrak P_{\mathbf H,\mathbf I} = \big\{ (S,K): S \in \mathfrak R_{\mathbf H}, K \in \mathfrak R_{\mathbf I} \big\}$. In addition, we denote
\begin{align}
    \mathfrak P^*_{\mathbf H,\mathbf I} = \big\{ (S,K): S \in \mathfrak R_{\mathbf H}^*, K \in \mathfrak R_{\mathbf I}^*, S \cap K \mbox{ is a tree containing } \mathsf T(S) \cup \mathsf T(K) \big\} \,. \label{eq-def-mathsf-P-*}
\end{align}
Then using the inequality $(a+b)^2\leq 2(a^2+b^2)$, we have that \eqref{eq-Qb-second-moment-relax-1} is bounded by $2$ times 
\begin{align}
    & \sum_{\mathbf H,\mathbf I \in \mathcal H} \frac{ s^{2(\aleph-1)} (\epsilon^2 \lambda s)^{2\ell\iota\aleph} \operatorname{Aut}(\mathsf T(\mathbf H)) \operatorname{Aut}(\mathsf T(\mathbf I)) }{ n^{2(\aleph+\ell\iota\aleph)} } \Big( \sum_{ (S,K) \in \mathfrak P_{\mathbf H,\mathbf I}^* } \big|\mathbb E_{\Qb}[\phi_S(A) \phi_K(A)]\big| \Big)^2  \label{eq-Qb-var-good-part}  \\
    +& \sum_{\mathbf H,\mathbf I \in \mathcal H} \frac{ s^{2(\aleph-1)} (\epsilon^2 \lambda s)^{2\ell\iota\aleph} \operatorname{Aut}(\mathsf T(\mathbf H)) \operatorname{Aut}(\mathsf T(\mathbf I)) }{ n^{2(\aleph+\ell\iota\aleph)} } \Big( \sum_{ (S,K) \not\in \mathfrak P_{\mathbf H,\mathbf I}^* } \big|\mathbb E_{\Qb}[\phi_S(A) \phi_K(A)]\big| \Big)^2 \,.  \label{eq-Qb-var-bad-part} 
\end{align}
In addition, define 
\begin{align}
    & \mathfrak P_{\mathbf H} = \big\{ (S,K) \in \mathfrak P_{\mathbf H,\mathbf H}^*: \mathsf T(S)=\mathsf T(K), \mathsf P(S)=\mathsf P(K) \big\} \,;  \label{eq-def-P-H}  \\
    &\mathfrak Q_{\mathbf H,\mathbf I} = \big\{ (S,K) \in \mathfrak P_{\mathbf H,\mathbf I}^* \setminus \mathfrak P_{\mathbf H}: \exists i\in[\iota\aleph]\mbox{ s.t. } \operatorname{EndP}(\mathsf L_i(K))\not\in \mathsf P(S), E(\mathsf L_i(K) \cap \mathsf T(S)) = \emptyset \big\} \,;  \label{eq-def-Q-H-I} \\
    &\mathfrak U_{\mathbf H,\mathbf I} = \mathfrak P_{\mathbf H,\mathbf I}^* \setminus (\mathfrak P_{\mathbf H} \cup \mathfrak Q_{\mathbf H,\mathbf I} ) \,.  \label{eq-def-V-H-I}
\end{align}
Then by Item (7) of Definition~\ref{def-tilde-T-K} (which enables us to bound $\operatorname{Aut}(\mathsf T(\mathbf H)) \operatorname{Aut}(\mathsf T(\mathbf I))$ by $2\operatorname{Aut}(\mathsf T(\mathbf H))^2$) and Cauchy's inequality, we have that \eqref{eq-Qb-var-good-part} is bounded by $6$ times
\begin{align}
    & \sum_{\mathbf H \in \mathcal H} \frac{ s^{2(\aleph-1)} (\epsilon^2 \lambda s)^{2\ell\iota\aleph} \operatorname{Aut}(\mathsf T(\mathbf H))^2 }{ n^{2(\aleph+\ell\iota\aleph)} } \Big( \sum_{ (S,K) \in \mathfrak P_{\mathbf H} } \big| \mathbb E_{\Qb}[\phi_S(A) \phi_K(A)] \big| \Big)^2 \label{eq-Qb-var-good-part-1} \\
    + & \sum_{\mathbf H,\mathbf I \in \mathcal H} \frac{ s^{2(\aleph-1)} (\epsilon^2 \lambda s)^{2\ell\iota\aleph} \operatorname{Aut}(\mathsf T(\mathbf H))^2 }{ n^{2(\aleph+\ell\iota\aleph)} } \Big( \sum_{ (S,K) \in \mathfrak Q_{\mathbf H,\mathbf I} } \big| \mathbb E_{\Qb}[\phi_S(A) \phi_K(A)] \big| \Big)^2 \label{eq-Qb-var-good-part-2} \\
    + & \sum_{\mathbf H,\mathbf I \in \mathcal H} \frac{ s^{2(\aleph-1)} (\epsilon^2 \lambda s)^{2\ell\iota\aleph}\operatorname{Aut}(\mathsf T(\mathbf H))^2   }{ n^{2(\aleph+\ell\iota\aleph)} } \Big( \sum_{ (S,K) \in \mathfrak U_{\mathbf H,\mathbf I} } \big|\mathbb E_{\Qb}[\phi_S(A) \phi_K(A)]\big| \Big)^2 \label{eq-Qb-var-good-part-3} \,. 
\end{align}
Thus, it suffices to bound \eqref{eq-Qb-var-good-part-1}, \eqref{eq-Qb-var-good-part-2}, \eqref{eq-Qb-var-good-part-3} and \eqref{eq-Qb-var-bad-part} by $o(1) \cdot \mathbb E_{\Pb}[f]^2$ separately. We first control \eqref{eq-Qb-var-good-part-1} via the following lemma:
\begin{lemma}{\label{lem-bound-Qb-var-good-part}}
    We have
    \begin{equation}{\label{eq-bound-Qb-var-good-part}}
        \eqref{eq-Qb-var-good-part-1} \leq s^{2(\aleph-1)} \big( 1-\tfrac{1}{\epsilon^2\lambda s(1-\epsilon^\mathtt m)} \big)^{-4\iota\aleph}
        \Big( \tfrac{(\epsilon^2 \lambda s)^{\ell}}{n} \Big)^{4\iota\aleph} (1-\epsilon^{2\mathtt m})^{2\iota\aleph} |\mathcal H| \,.
    \end{equation}
    In particular, combined with Proposition~\ref{prop-first-moment} and Lemma~\ref{lem-enu-decorated-trees}, it yields that 
    \begin{align}\label{proof-second-moment-item-1-1}
        \eqref{eq-Qb-var-good-part-1} = o(1) \cdot \mathbb E_{\Pb}[f]^2 \,.
    \end{align}
\end{lemma}
\begin{proof}
    Note that for $(S,K) \in \mathfrak P_{\mathbf H}$, there exists a collection of self-avoiding paths $\{ \mathcal L_u:u \in \mathsf{Vert}(\mathsf P(S)) \}$ satisfying $u \in \operatorname{EndP}(\mathcal L_u)$ for all $u \in \mathsf{Vert}(\mathsf P(S))$ such that 
    \begin{align*}
        S \cap K = \mathsf T(S) \oplus ( \oplus_{u \in \mathsf{Vert}(\mathsf P(S))} \mathcal L_u ) \,.
    \end{align*}
    (See Figure~\ref{fig:3} for an illustration.)
    \begin{figure}[!ht]
        \centering
        \vspace{0cm}
        \includegraphics[height=4.33cm,width=13.92cm]{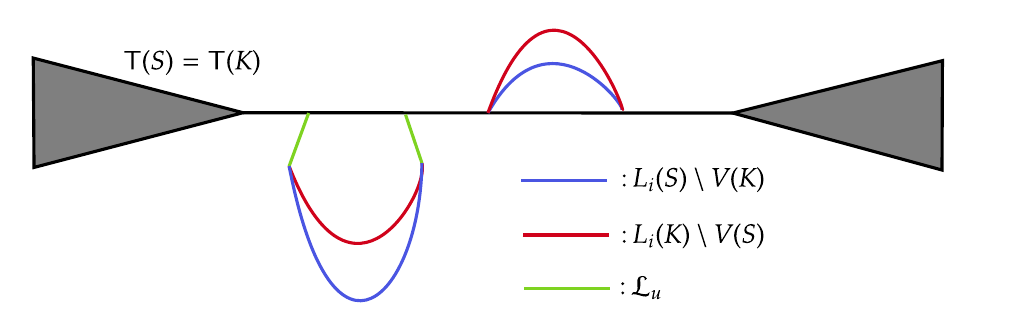}
        \caption{Intersection pattern for $(S,K)\in \mathfrak P_{ \mathbf{H} }$}
        \label{fig:3}
    \end{figure}
    For each $S \in \mathfrak R_{\mathbf H}^*$ and each non-negative sequence $\{x_u\}$, define $\mathfrak P_{\mathbf H}(S,\{ x_u \})$ to be the set of all $K \subset \mathsf K_n$ such that
    \begin{align*}
          (S,K) \in \mathfrak P_{\mathbf H}, K \cap S = \mathsf T(S) \oplus ( \oplus_{u \in \mathsf{Vert}(\mathsf P(S))} \mathcal L_u ) \mbox{ with } |E(\mathcal L_u)|=x_u  \,.
    \end{align*}
    We have
    \begin{align*}
        \#\mathfrak P_{\mathbf H}(S,\{ x_u \}) \circeq n^{ (\ell-1)\iota\aleph - \sum x_u }  \,. 
    \end{align*}
    In addition, by \eqref{eq-phi-Q-intersec-istree} and \eqref{eq-S,K-intersection-general-A,A} and recalling \eqref{eq-def-Kc-s1-S2-pi}, for $K \in \mathfrak P_{\mathbf H}(S,\{ x_u \})$ we have 
    \begin{align*}
        \mathbb E_{\Qb}\big[ \phi_S(A) \phi_K(A) \big] \circeq \big( \tfrac{\epsilon^2 \lambda s}{n} \big)^{ \ell\iota\aleph-\sum x_u } (1-\epsilon^{2\mathtt m})^{\mathtt K_c(S,K;\mathsf{id})} \leq \big( \tfrac{\epsilon^2 \lambda s}{n} \big)^{ \ell\iota\aleph-\sum x_u } (1-\epsilon^{2\mathtt m})^{\iota\aleph-\sum x_u } \,,
    \end{align*} 
    Thus, we have (writing $\mathbf{x}\geq 0$ for $x_u\geq 0$ for all $u\in\mathsf{P}(S)$)
    \begin{align*}
        & \sum_{S,K \in \mathfrak P_{\mathbf H}} \big|\mathbb E_{\Qb}[\phi_S(A)\phi_K(A)]\big| = \sum_{ S \in \mathfrak R_{\mathbf H}^* } \sum_{\mathbf x \geq 0} \sum_{ K \in \mathfrak P_{\mathbf H}(S,\{ x_u \}) } \big|\mathbb E_{\Qb}[\phi_S(A)\phi_K(A)]\big| \\
        \leq\ & \sum_{ S \in \mathfrak R_{\mathbf H}^* } \sum_{\mathbf x \geq 0} n^{ (\ell-1)\iota\aleph - \sum x_u } \cdot \big( \tfrac{\epsilon^2 \lambda s}{n} \big)^{ \ell\iota\aleph-\sum x_u } (1-\epsilon^{2\mathtt m})^{\iota\aleph - \sum x_u}  \\
        \leq\ & \big| \mathfrak R_{\mathbf H}^* \big| \cdot n^{-\iota\aleph} \sum_{\mathbf x \geq 0} \big(\epsilon^2 \lambda s \big)^{\ell\iota\aleph-\sum x_u} (1-\epsilon^{2\mathtt m})^{\iota\aleph - \sum x_u} \\
        \leq\ & \frac{ n^{\aleph} }{ \operatorname{Aut}(\mathsf T(\mathbf H)) } \cdot \big( n^{\ell-1} \big)^{\iota\aleph}\cdot n^{-\iota\aleph} (\epsilon^2\lambda s)^{\ell\iota\aleph} (1-\epsilon^{2\mathtt m})^{\iota\aleph}\Big(
        \sum_{\mathbf x\geq 0}\big(\epsilon^2\lambda s(1-\epsilon^{2\mathtt m})\big)^{-x_u} \Big)  \\
        \leq\ & \frac{n^{\aleph +(\ell-1)\iota\aleph}}{\operatorname{Aut}(\mathsf T(\mathbf H))} n^{-\iota\aleph} \big( 1-\tfrac{1}{ \epsilon^2\lambda s(1-\epsilon^\mathtt m)} \big)^{-2\iota\aleph} (\epsilon^2 \lambda s)^{\ell\iota\aleph} (1-\epsilon^{2\mathtt m})^{\iota\aleph}\,,
    \end{align*}
    where the third inequality follows from \eqref{eq-bound-card-R^*_H}. Plugging this bound into \eqref{eq-Qb-var-good-part-1}, we get that
    \begin{align}
        \eqref{eq-Qb-var-good-part-1} &\leq \sum_{\mathbf H \in \mathcal H} \frac{ s^{2(\aleph-1)} (\epsilon^2 \lambda s)^{2\ell\iota\aleph} \operatorname{Aut}(\mathsf T(\mathbf H))^2 }{ n^{2(\aleph+\ell\iota\aleph)} } \Big( \frac{n^{-\iota\aleph}\big( 1-\tfrac{1}{\epsilon^2\lambda s(1-\epsilon^\mathtt m)} \big)^{-2\iota\aleph} (\epsilon^2 \lambda s)^{\ell\iota\aleph} (1-\epsilon^{2\mathtt m})^{\iota\aleph}}{n^{-\aleph -(\ell - 1)\iota\aleph} \operatorname{Aut}(\mathsf T(\mathbf H))} \Big)^2 \nonumber \\
        &= s^{2(\aleph-1)} \big( 1-\tfrac{1}{\epsilon^2\lambda s(1-\epsilon^\mathtt m)} \big)^{-4\iota\aleph}
        \Big( \tfrac{(\epsilon^2 \lambda s)^{\ell}}{n} \Big)^{4\iota\aleph} (1-\epsilon^{2\mathtt m})^{2\iota\aleph} |\mathcal H| \,, \label{eq-Qb-var-good-part-1*}
    \end{align}
    as desired.
\end{proof}

We then deal with \eqref{eq-Qb-var-good-part-2} and \eqref{eq-Qb-var-good-part-3}. For $(S,K) \in \mathfrak Q_{\mathbf H, \mathbf I} \cup \mathfrak U_{\mathbf H, \mathbf I}$, let 
\begin{align*}
\mathtt V = \mathsf{Vert}(\mathsf P(S)) \,.
\end{align*}
Also, since $(S,K) \in \mathfrak Q_{\mathbf H, \mathbf I} \cup \mathfrak U_{\mathbf H, \mathbf I}$, $S \triangle K$ is a union of self-avoiding paths
$\{ \mathcal K_i : 1 \leq i \leq \mathtt z\}$ with $\mathcal K_i \cap (S \cap K) = \operatorname{EndP}(\mathcal K_i)$. Define 
\begin{align*}
    \mathtt U = \cup_{1 \leq i \leq \mathtt z} \operatorname{EndP}(\mathcal K_i)\,, \quad
    \mathtt W = \triangle_{1 \leq i \leq \mathtt z} \operatorname{EndP}(\mathcal K_i)\,,
\end{align*}
and thus $\mathtt W$ is the set of vertices which appear odd number of times in $\operatorname{EndP}(\mathcal{K}_i)$'s. Note that there exist $\{ x_v \geq 0: v \in \mathtt{V} \}$ and self-avoiding paths $\{\mathcal{L}_v\}$ satisfying $v\in \operatorname{EndP}(\mathcal{L}_v)$ for all $v$ such that $|E(\mathcal L_v)|=x_v$ and
\begin{align*}
    S \cap K = \mathsf T(S) \oplus \big( \oplus_{v \in \mathtt V} \mathcal L_v \big)
\end{align*}
(see Figures~\ref{fig:4} and \ref{fig:5} for an illustration).
\begin{figure}[ht]
    \centering
    \vspace{0cm}
    \includegraphics[height=5.4cm,width=14.1cm]{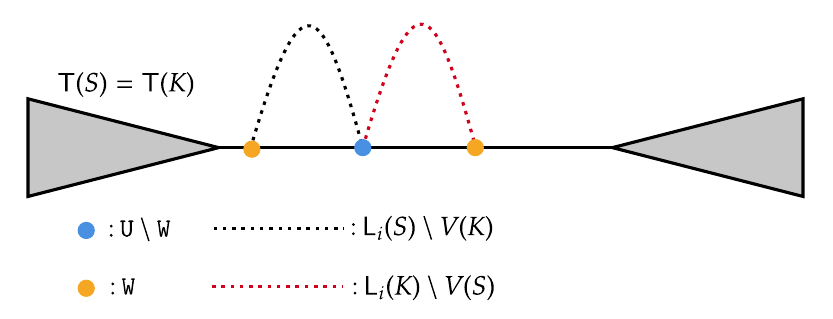}
    \caption{Intersection pattern for $K\in\mathfrak Q_{\mathbf H,\mathbf I}(S,\{x_v\},y)$}
    \label{fig:4}
\end{figure}
\begin{figure}[ht]
    \centering
    \vspace{0cm}
    \includegraphics[height=5.89cm,width=14.89cm]{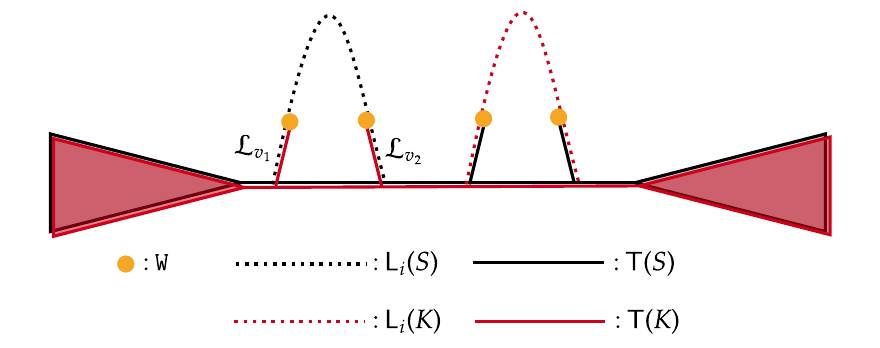}
    \caption{Intersection pattern for $K\in\mathfrak U_{\mathbf H,\mathbf I}(S,\{x_v\},y)$}
    \label{fig:5}
\end{figure}

\noindent In this case we have $\sum_{1 \leq i \leq \mathtt z} |E(\mathcal K_i)| = \ell\iota\aleph - \sum x_v$. Define $\mathfrak Q_{\mathbf H,\mathbf I}(S,\{ x_u\},y)$ to be the set of $K \in \mathfrak R_{\mathbf I}^*$ such that 
\begin{align*}
    (S,K) \in \mathfrak Q_{\mathbf H,\mathbf I}, \quad |E(\mathcal L_u)| = x_u, \quad \#\big\{ i \in [\iota\aleph] : E(\mathsf L_i (K)\cap \mathsf T(S)) \neq \emptyset \big\} = y \,,
\end{align*} 
and similarly define $\mathfrak U_{\mathbf H,\mathbf I}(S,\{ x_u\},y)$ be the set of $K \in \mathfrak R_{\mathbf I}^*$ such that 
\begin{align*}
    (S,K) \in \mathfrak U_{\mathbf H,\mathbf I}, \quad |E(\mathcal L_u)| = x_u, \quad \#\big\{ i \in [\iota\aleph] : E(\mathsf L_i (K)\cap \mathsf T(S)) \neq \emptyset \big\} = y \,.
\end{align*}
We first show the following bounds on the enumerations of $\mathfrak Q_{\mathbf H,\mathbf I}(S,\{ x_u \},y)$ and $\mathfrak U_{\mathbf H,\mathbf I}(S,\{ x_u \},y)$.
\begin{lemma}{\label{lem-enu-mathfrak-Q-mathfrak-U}}
    We have
    \begin{align}\label{eq-enum-frakchoose-pairings-Q-H,I}
        \sum_{\mathbf I \in \mathcal H}\#\big(\mathfrak Q_{\mathbf H,\mathbf I}(S,\{ x_u \},y) \cup \mathfrak U_{\mathbf H,\mathbf I}(S,\{ x_u \},y)\big) \leq n^{(\ell-1)\iota\aleph- \sum x_u+o(1)}\,.
    \end{align}
    In addition, given any $\{ x_u \}$ and $y$ we have (recall the definition of $\mathfrak c$ in \eqref{eq-choice-mathfrak-c})
    \begin{align}
        &\#\big\{ \mathbf I \in \mathcal H: \mathfrak U_{\mathbf H,\mathbf I}(S,\{ x_u \},y) \neq \emptyset \big\} \leq 2^{2\iota\aleph} \tbinom{ 100\mathfrak c^{9}\mathfrak C_\aleph }{y}\,. \label{eq-enum-frakchoose-pairings-U-valid-I-enum} \\
        &\#\mathfrak U_{\mathbf H,\mathbf I}(S,\{ x_u \},y) \leq e^{4\iota\aleph}n^{(\ell-1)\iota\aleph- \sum x_u} \,. \label{eq-enum-frakchoose-pairings-U-valid-I-enum-+}
    \end{align}
\end{lemma}
The proof of Lemma~\ref{lem-enu-mathfrak-Q-mathfrak-U} is incorporated in Section~\ref{subsec:proof-lem-4.3}. Now we can bound \eqref{eq-Qb-var-good-part-2} and \eqref{eq-Qb-var-good-part-3} via the following lemma.
\begin{lemma}{\label{lem-Qb-var-good-part-2,3}}
    We have 
    \begin{align}
        \eqref{eq-Qb-var-good-part-2} &\leq n^{-1+o(1)} \cdot (1-\tfrac{1}{\epsilon^2\lambda s})^{-4\iota\aleph} s^{2(\aleph-1)}|\mathcal H|^2 \big( \tfrac{  (\epsilon^2 \lambda s)^{\ell} }{ n } \big)^{4\iota\aleph} \,, \label{eq-Qb-var-good-part-2*}   \\
        \eqref{eq-Qb-var-good-part-3} &\leq s^{2(\aleph-1)} |\mathcal H| \big( \tfrac{(\epsilon^2 \lambda s)^{\ell} }{ n } \big)^{4\iota\aleph} 2^{-\iota\aleph} \tbinom{\mathfrak C_\aleph / 4}{\iota\aleph} \,. \label{eq-Qb-var-good-part-3*} 
    \end{align}
    In particular, with Proposition~\ref{prop-first-moment} and Lemma~\ref{lem-enu-decorated-trees} we have $\eqref{eq-Qb-var-good-part-2}, \eqref{eq-Qb-var-good-part-3}=o(1) \cdot \mathbb E_{\Pb}[f]^2$.
\end{lemma}
\begin{proof} 
    Note that conditioning on $\sigma_{\mathtt U}=\{ \sigma_u: u \in \mathtt U \}$, we have that
    \begin{align*}
        \Bigg\{ & \prod_{(i,j) \in E(\mathsf T(S))} \Big( \tfrac{ A_{i,j}-\frac{\lambda s}{n} }{ \sqrt{ \frac{\lambda s}{n} (1-\frac{\lambda s}{n}) } } \Big)^2, \prod_{m \leq \iota\aleph} \prod_{(i,j) \in E(\mathsf L_m(S))} \tfrac{ A_{i,j}-\frac{\lambda s}{n} }{ \sqrt{ \frac{\lambda s}{n} (1-\frac{\lambda s}{n}) } } \prod_{(i,j) \in E(\mathsf L_m(K))} \tfrac{ A_{i,j}-\frac{\lambda s}{n} }{ \sqrt{ \frac{\lambda s}{n} (1-\frac{\lambda s}{n}) } } \Bigg\}
    \end{align*}
    are conditionally independent. Also, similar to \eqref{eq-expectation-pi-in-A-diamond}, we have
    \begin{align*}
        & \mathbb E\Bigg\{  \prod_{m \leq \iota\aleph} \prod_{(i,j) \in E(\mathsf L_m(S))} \tfrac{ A_{i,j}-\frac{\lambda s}{n} }{ \sqrt{ \frac{\lambda s}{n} (1-\frac{\lambda s}{n}) } } \prod_{(i,j) \in E(\mathsf L_m(K))} \tfrac{ A_{i,j}-\frac{\lambda s}{n} }{ \sqrt{ \frac{\lambda s}{n} (1-\frac{\lambda s}{n}) } } \mid \sigma_{\mathtt U}  \Bigg\} \\
        =\ & \big( \tfrac{\epsilon^2 \lambda s}{n} \big)^{ \sum_{1 \leq i \leq \mathtt z}|E(\mathcal K_i)| } \prod_{1 \leq i \leq \mathtt z} \prod_{u \in \operatorname{EndP}(\mathcal K_i )} \sigma_u = \big( \tfrac{\epsilon^2 \lambda s}{n} \big)^{ \ell\iota\aleph-\sum x_v } \prod_{u \in \mathtt W} \sigma_u \,.
    \end{align*}
    Thus, for $K \in \mathfrak Q_{\mathbf H,\mathbf I}(S,\{ x_u\},y)\cup \mathfrak U_{\mathbf H,\mathbf I}(S,\{ x_u\},y)$ we have 
    \begin{align}
        & \mathbb E_{\Qb}[ \beta_S(A)\beta_K(A) ] = \mathbb E_{\sigma_{\mathtt U}\sim \nu_{\mathtt U}}\Big\{ \mathbb E_{\Qb}\big[ \beta_S(A)\beta_K(A) \mid \sigma_{\mathtt U} \big] \Big\} \nonumber \\
        =\ & \big( \tfrac{\epsilon^2 \lambda s}{n} \big)^{ \ell\iota\aleph-\sum x_v } \mathbb E_{\sigma_{\mathtt U}\sim \nu_{\mathtt U}} \Bigg\{ \mathbb E\Big[ \prod_{(i,j) \in E(\mathsf T(S))} \Big( \tfrac{ A_{i,j}-\frac{\lambda s}{n} }{ \sqrt{ \frac{\lambda s}{n} (1-\frac{\lambda s}{n}) } } \Big)^2 \mid \sigma_{\mathtt U} \Big] \prod_{u \in \mathtt W} \sigma_u \Bigg\} \nonumber \\
        =\ & \big( \tfrac{\epsilon^2 \lambda s}{n} \big)^{ \ell\iota\aleph-\sum x_v } \mathbb E_{\sigma_{\mathtt U}\sim \nu_{\mathtt U}}\Bigg\{ \prod_{(i,j) \in E(\mathsf T(S))} \Big( \tfrac{ A_{i,j}-\frac{\lambda s}{n} }{ \sqrt{ \frac{\lambda s}{n} (1-\frac{\lambda s}{n}) } } \Big)^2 \prod_{u \in \mathtt W} \sigma_u \Bigg\}  \,. \nonumber
    \end{align}
    Using Lemma~\ref{lem-label-matching-exact}, we see that for any $\{ (\mathtt u_1 ,\mathtt v_1), \ldots, (\mathtt u_\mathtt i,\mathtt v_\mathtt i) \} \in \mathcal X(\mathtt W)$ (recall that we use $\mathcal X(\mathtt W)$ to denote partitions of $\mathtt W$ into disjoint pairs), we have $\sum_{\mathtt j=1}^{\mathtt i} d(\mathtt u_\mathtt j ,\mathtt v_\mathtt j ) \geq y\mathtt m$. Thus we obtain that
    \begin{align}
        \mathbb E_{\Qb}[ \beta_S(A)\beta_K(A) ]  \leq \big( \tfrac{\epsilon^2 \lambda s}{n} \big)^{ \ell\iota\aleph-\sum x_v } \epsilon^{ y\mathtt m}  \,. \label{eq-Qb-var-good-part-case-II-relax-1}
    \end{align}
    Also, note that for $K \in \mathfrak Q_{\mathbf H,\mathbf I}(S,\{x_v\}, y)$ we have $\big|\mathbb E_{\Qb}[ \phi_S(A) \phi_K(A) ]\big| \leq n^{-\tfrac{1}{2}+o(1)} \mathbb E_{\Qb}[ \beta_S(A) \beta_K(A) ]$ by Lemma~\ref{lem-est-intersection-general-istree-P-Q}. Thus, we have that \eqref{eq-Qb-var-good-part-2} is bounded by (writing $\mathbf x\geq 0$ for $x_v \geq 0$ for all $v \in \mathtt V$)
    \begin{align}
        & \sum_{\mathbf H, \mathbf I \in \mathcal H} \frac{ s^{2(\aleph-1)} (\epsilon^2 \lambda s)^{2\ell\iota\aleph}\operatorname{Aut}(\mathsf T(\mathbf H))^2 }{ n^{2(\aleph+\ell\iota\aleph)} } \Big( \sum_{S \in \mathfrak R_{\mathbf H}^*} \sum_{\substack{\mathbf x\geq 0\\0 \leq y \leq \iota\aleph}} \sum_{K \in \mathfrak Q_{\mathbf H,\mathbf I}(S,\{x_v\} , y)} \big|\mathbb E_{\Qb}[ \phi_S(A) \phi_K(A) ]\big| \Big)^2 \nonumber \\
        \leq\ & \sum_{  \mathbf H, \mathbf I \in \mathcal H} \frac{ s^{2(\aleph-1)} (\epsilon^2 \lambda s)^{2\ell\iota\aleph}\operatorname{Aut}(\mathsf T(\mathbf H))^2  }{ n^{2(\aleph+\ell\iota\aleph)} } \Big(\sum_{\substack{\mathbf x\geq 0\\0 \leq y \leq \iota\aleph}} |\mathfrak R_\mathbf H^*| \cdot n^{-\tfrac{1}{2}+o(1)} \sum_{K \in \mathfrak Q_{\mathbf H,\mathbf I} (S,\{x_v\} , y)} \E_{\Qb} [\beta_S(A)\beta_K(A)]\Big)^2  \nonumber \\
        \overset{\eqref{eq-enum-frakchoose-pairings-Q-H,I}}{\leq}\ & \sum_{  \mathbf H, \mathbf I \in \mathcal H} \frac{ s^{2(\aleph-1)} (\epsilon^2 \lambda s)^{2\ell\iota\aleph}\operatorname{Aut}(\mathsf T(\mathbf H))^2  }{ n^{2(\aleph+\ell\iota\aleph)} }\ * \nonumber \\
        & \sum_{\substack{\mathbf x\geq 0\\0 \leq y \leq \iota\aleph}}\Big( \frac{n^{(\ell-1)\iota\aleph+\aleph}}{\operatorname{Aut}(\mathsf T(\mathbf H))} \cdot n^{(\ell-1)\iota\aleph-\sum x_v + o (1)}\cdot n^{-\tfrac{1}{2}+o(1)}\big( \tfrac{\epsilon^2 \lambda s}{n} \big)^{ \ell\iota\aleph-\sum x_v } \Big)^2  \nonumber \\
        \leq\ & (1-\tfrac{1}{\epsilon^2\lambda s})^{-4\iota\aleph} s^{2(\aleph-1)}|\mathcal H|^2 \big( \tfrac{  (\epsilon^2 \lambda s)^{\ell} }{ n } \big)^{4\iota\aleph} n^{-1+o(1)} \,, \nonumber
    \end{align}
    leading to \eqref{eq-Qb-var-good-part-2*}. Now we deal with \eqref{eq-Qb-var-good-part-3}. Using Cauchy-Schwartz inequality, we see that \eqref{eq-Qb-var-good-part-3} is bounded by
    \begin{align}
        & \sum_{\mathbf H, \mathbf I \in \mathcal H} \frac{ s^{2(\aleph-1)} (\epsilon^2 \lambda s)^{2\ell\iota\aleph}\operatorname{Aut}(\mathsf T(\mathbf H))^2 }{ n^{2(\aleph+\ell\iota\aleph)} } \Bigg( \sum_{S \in \mathfrak R_{\mathbf H}^*} \sum_{\substack{\mathbf x\geq 0\\0 \leq y \leq \iota\aleph}} \sum_{K \in \mathfrak U_{\mathbf H,\mathbf I}(S,\{x_v\} , y)} \big|\mathbb E_{\Qb}[ \phi_S(A) \phi_K(A) ]\big| \Bigg)^2 \nonumber \\ 
        \leq\ & \sum_{\mathbf H, \mathbf I \in \mathcal H} \frac{ s^{2(\aleph-1)} (\epsilon^2 \lambda s)^{2\ell\iota\aleph}\operatorname{Aut}(\mathsf T(\mathbf H))^2 }{ n^{2(\aleph+\ell\iota\aleph)} } \Bigg( \sum_{\substack{\mathbf x\geq 0\\0 \leq y \leq \iota\aleph}} (\epsilon^2 \lambda s)^{-\sum x_u} \ * \nonumber  \\
        &  \sum_{\substack{\mathbf x\geq 0\\0 \leq y \leq \iota\aleph}}\Big( (\epsilon^2 \lambda s)^{(\sum x_u)/2} \cdot \sum_{S \in \mathfrak R_{\mathbf H}^*} \sum_{K \in \mathfrak U_{\mathbf H,\mathbf I}(S,\{x_v\} , y)} \big|\mathbb E_{\Qb}[ \phi_S(A) \phi_K(A) ]\big| \Big)^2 \Bigg) \nonumber \\
        \leq\ & (1-\tfrac{1}{\epsilon^2\lambda s})^{-2\iota\aleph}(\iota\aleph+1) \cdot \sum_{\mathbf H, \mathbf I \in \mathcal H} \frac{ s^{2(\aleph-1)} (\epsilon^2 \lambda s)^{2\ell\iota\aleph}\operatorname{Aut}(\mathsf T(\mathbf H))^2 }{ n^{2(\aleph+\ell\iota\aleph)} } \ * \nonumber \\
        & \sum_{\substack{\mathbf x\geq 0\\0 \leq y \leq \iota\aleph}} \Big( \sum_{S \in \mathfrak R_{\mathbf H}^*} \sum_{K \in \mathfrak U_{\mathbf H,\mathbf I}(S,\{x_v\} , y)} (\epsilon^2 \lambda s)^{(\sum x_u)/2} \big|\mathbb E_{\Qb}[ \phi_S(A) \phi_K(A) ]\big| \Big)^2 \nonumber \\
        \leq\ & (1-\tfrac{1}{\sqrt{\epsilon^2\lambda s}})^{-2\iota\aleph} (\iota\aleph+1) \sum_{\mathbf H \in \mathcal H} \sum_{\substack{\mathbf x\geq 0\\0 \leq y \leq \iota\aleph}} \sum_{ \substack{ \mathbf I \in \mathcal H \\ \mathfrak U_{\mathbf H,\mathbf I} \neq \emptyset } } \frac{ s^{2(\aleph-1)} (\epsilon^2 \lambda s)^{2\ell\iota\aleph}\operatorname{Aut}(\mathsf T(\mathbf H))^2 }{ n^{2(\aleph+\ell\iota\aleph)} } \ * \nonumber \\
        & \Big( \sum_{S \in \mathfrak R_{\mathbf H}^*} \sum_{K \in \mathfrak U_{\mathbf H,\mathbf I}(\{x_v\} , y)} (\epsilon^2 \lambda s)^{(\sum x_u)/2} \big|\mathbb E_{\Qb}[ \phi_S(A) \phi_K(A) ]\big| \Big)^2 \,.  \label{eq-Qb-var-good-part-3-relax-1}
    \end{align}
    Using Lemma~\ref{lem-enu-mathfrak-Q-mathfrak-U}, we get that
    \begin{align*}
        \eqref{eq-Qb-var-good-part-3-relax-1} &\overset{\eqref{eq-Qb-var-good-part-case-II-relax-1}}{\leq} \Big( \sum_{S \in \mathfrak R_{\mathbf H}^*} \sum_{K \in \mathfrak U_{\mathbf H,\mathbf I}(\{x_v\},y)} (\epsilon^2 \lambda s)^{(\sum x_u)/2} \epsilon^{y \mathtt m} \big( \tfrac{\epsilon^2 \lambda s}{n} \big)^{ \ell\iota\aleph-\sum x_v } \Big)^2 \\
        &\overset{\eqref{eq-enum-frakchoose-pairings-U-valid-I-enum-+}}{\leq} \Big( \sum_{S \in \mathfrak R_{\mathbf H}^*} e^{4\iota\aleph} 3^{2\iota\aleph} n^{(\ell-1)\iota\aleph-\sum x_v} (\epsilon^2 \lambda s)^{(\sum x_u)/2} \epsilon^{y \mathtt m} \big( \tfrac{\epsilon^2 \lambda s}{n} \big)^{ \ell\iota\aleph-\sum x_v } \Big)^2 \\
        &= \Big( \tfrac{ n^{(\ell-1)\iota\aleph+\aleph} }{ \operatorname{Aut}(\mathsf{T}(\mathbf H)) } \cdot e^{4\iota\aleph} 3^{2\iota\aleph}  (\epsilon^2 \lambda s)^{-(\sum x_u)/2} \epsilon^{y \mathtt m} \big( \tfrac{(\epsilon^2 \lambda s)^\ell}{n} \big)^{ \iota\aleph } \Big)^2 \,.
    \end{align*}
    Thus, we get that \eqref{eq-Qb-var-good-part-3} is bounded by
    \begin{align*}
        & (1-\tfrac{1}{\sqrt{\epsilon^2\lambda s}})^{-2\iota\aleph} (\iota\aleph+1) \sum_{\mathbf H \in \mathcal H} \sum_{\substack{\mathbf x\geq 0\\0 \leq y \leq \iota\aleph}} \sum_{ \substack{ \mathbf I \in \mathcal H \\ \mathfrak U_{\mathbf H,\mathbf I} \neq \emptyset } } \frac{ s^{2(\aleph-1)} (\epsilon^2 \lambda s)^{2\ell\iota\aleph}\operatorname{Aut}(\mathsf T(\mathbf H))^2 }{ n^{2(\aleph+\ell\iota\aleph)} } \ *  \\
        & \Big( \tfrac{ n^{(\ell-1)\iota\aleph+\aleph} }{ \operatorname{Aut}(\mathsf{T}(\mathbf H)) } \cdot e^{4\iota\aleph} (\epsilon^2 \lambda s)^{-(\sum x_u)/2} \epsilon^{y \mathtt m} \big( \tfrac{(\epsilon^2 \lambda s)^\ell}{n} \big)^{ \iota\aleph } \Big)^2  \\
        \leq\ & (1-\tfrac{1}{\sqrt{\epsilon^2\lambda s}})^{-2\iota\aleph} (\iota\aleph+1) \sum_{\mathbf H \in \mathcal H} \sum_{\substack{\mathbf x\geq 0\\0 \leq y \leq \iota\aleph}} \sum_{ \substack{ \mathbf I \in \mathcal H \\ \mathfrak U_{\mathbf H,\mathbf I} \neq \emptyset } } \frac{ s^{2(\aleph-1)} (\epsilon^2 \lambda s)^{4\ell\iota\aleph-\sum x_v} 2^{16\iota\aleph}  \epsilon^{2y\mathtt m} }{ n^{4\iota\aleph} } \\
        \leq\ & (1-\tfrac{1}{\sqrt{\epsilon^2\lambda s}})^{-2\iota\aleph} (\iota\aleph+1) \sum_{\mathbf H \in \mathcal H} \sum_{\substack{\mathbf x\geq 0\\0 \leq y \leq \iota\aleph}}  \frac{ s^{2(\aleph-1)} (\epsilon^2 \lambda s)^{4\ell\iota\aleph-\sum x_v} 2^{18\iota\aleph} \tbinom{ 100\mathfrak c^{9}\mathfrak C_\aleph }{y} \epsilon^{2y\mathtt m} }{ n^{4\iota\aleph} } \\
        \leq\ & (1-\tfrac{1}{\sqrt{\epsilon^2\lambda s}})^{-4\iota\aleph} (\iota\aleph+1) \cdot \frac{ s^{2(\aleph-1)} |\mathcal H| (\epsilon^2 \lambda s)^{4\ell\iota\aleph} 2^{18\iota\aleph} }{ n^{4\iota\aleph} }  \sum_{0 \leq y \leq \iota\aleph} \tbinom{ 100\mathfrak c^{9}\mathfrak C_\aleph }{y} \epsilon^{2y\mathtt m} \\
        \leq\ & s^{2(\aleph-1)} |\mathcal H| \big( \tfrac{(\epsilon^2 \lambda s)^{\ell} }{ n } \big)^{4\iota\aleph} 2^{-\iota\aleph} \tbinom{\mathfrak C_\aleph / 4}{\iota\aleph}  \,,
    \end{align*}
    where the second inequality follows from \eqref{eq-enum-frakchoose-pairings-U-valid-I-enum} and the last inequality holds by the fact that 
    \begin{align*}
        \sum_{0 \leq y \leq \iota\aleph} \tbinom{ 100\mathfrak c^{9}\mathfrak C_\aleph }{y} \epsilon^{2y\mathtt m} \leq (\iota\aleph + 1) \tbinom{ 100\mathfrak c^{9}\mathfrak C_\aleph }{\iota\aleph} \epsilon^{2\mathtt m \iota\aleph} \overset{\eqref{eq-choice-mathtt-m}}{\leq}  \tbinom{ \mathfrak C_\aleph/4 }{\iota\aleph}\epsilon^{\mathtt m \iota\aleph} \,,
    \end{align*}
    and the fact that $\epsilon^{\mathtt m \iota\aleph} \leq(1-\tfrac{1}{\sqrt{\epsilon^2\lambda s}})^{4\iota\aleph} 2^{-29\iota\aleph}$ by \eqref{eq-choice-mathtt-m} and \eqref{eq-choice-iota}. This leads to \eqref{eq-Qb-var-good-part-3*}.
\end{proof}
Now it remains to deal with \eqref{eq-Qb-var-bad-part}. For $S \in \mathfrak R_{\mathbf H}$, $K \in \mathfrak R_{\mathbf I}$ and $I,J \subset [\iota\aleph]$, define $G^{I,J}_{\cup}$ to be the (simple) graph obtained by (recall \eqref{eq-def-I-over-S})
\begin{align*}
    G^{I,J}_{\cup} = \widetilde{S}^I_1 \cup \widetilde{S}^J_2  \,,
\end{align*}
and in particular we define $G_\cup = G^{[\iota\aleph],[\iota\aleph]}_\cup$. From Item~(3) in Lemma~\ref{lem-useful-property-trees-and-sets}, we know that there must exist $v \in \big(\mathcal L( \mathsf T(S))\cup \mathcal L( \mathsf T(K ))\big)\setminus \big(V(S \cap K)\big)$ if $\mathcal L(S) \cup \mathcal L( K) \not\subset V(S \cap K)$, which implies that $v \in \mathcal L(S^I \cup K^J)$ for all $I,J \subset [\iota\aleph]$ and therefore $\mathbb E_{\mathbb Q}[ \phi_{S}(A)\phi_{K}(A) ]=0$ from \eqref{eq-phi-as-linear-combination-of-beta} and Lemma~\ref{lem-upper-bound-exp}. Therefore, we only need to consider the case that $\mathcal L(S) \cup \mathcal L(K) \subset V(S \cap K)$, in which we have $\mathcal L(G_{\cup}) \subset \mathcal L(\widetilde S) \cap \mathcal L(\widetilde K)$ and thus $|\mathcal L(G_{\cup})| \leq \aleph$. Next we give a union upper bound for $\E_\Qb [\phi_S(A)\phi_K(A)]$. Note that
\begin{align}
    &\big(\epsilon^\mathtt m (\tfrac{\epsilon^2 \lambda s}{n})^{\ell/2} \big)^{2\iota\aleph - |I| - |J|}\E_{\P_\mathsf{id}} [\beta_{S^I}(A)\beta_{K^J}(A)] \nonumber \\
    =\ & \mathbb{E}_{\mathbb Q} \Bigg[  \prod_{(i,j) \in E(\widetilde{S}^I)} \Big( \tfrac{ A_{i,j}-\frac{\lambda s}{n} }{ \sqrt{\lambda s/n} } \Big)^{E(S^I)_{i,j}} \prod_{(i,j) \in E(\widetilde K^J)} \Big( \tfrac{ A_{i,j}-\frac{\lambda s}{n} }{ \sqrt{\lambda s/n} } \Big)^{E(K^J)_{i,j}}  \Bigg] * \big(\epsilon^\mathtt m (\tfrac{\epsilon^2 \lambda s}{n})^{\ell/2} \big)^{2\iota\aleph - |I| - |J|}\nonumber  \\
    =\ & \mathbb{E}_{\sigma\sim\nu} \Bigg\{  \prod_{(i,j) \in E(G_{\cup} \setminus \widetilde S^I)} \mathbb{E}_{\mathbb P_{\mathsf{id},\sigma}} \Bigg[ \Big( \tfrac{ A_{i,j}-\frac{\lambda s}{n} }{ \sqrt{\lambda s/n} } \Big)^{E(S^I)_{i,j}} \Bigg] \prod_{(i,j) \in E(G_{\cup} \setminus \widetilde K^J)} \mathbb{E}_{\mathbb P_{\mathsf{id},\sigma}} \Bigg[ \Big( \tfrac{ A_{i,j}-\frac{\lambda s}{n} }{ \sqrt{\lambda s/n} } \Big)^{E(K^J)_{i,j}} \Bigg] \nonumber \\ 
    & \prod_{(i,j) \in E(\widetilde S^I \cap \widetilde K^J)} \mathbb{E}_{\mathbb P_{\mathsf{id},\sigma}} \Bigg[ \Big( \tfrac{ A_{i,j}-\frac{\lambda s}{n} }{ \sqrt{\lambda s/n} } \Big)^{E(S^I)_{i,j}+E(K^J)_{i,j}}  \Bigg]  \Bigg\} * \big(\epsilon^\mathtt m (\tfrac{\epsilon^2 \lambda s}{n})^{\ell/2} \big)^{2\iota\aleph - |I| - |J|} \,. \label{eq-Qb-f1-setminus-*-relax-1}
\end{align}
Using Lemmas~\ref{lem-joint-moment-A-B} and \ref{lem-upper-bound-exp} (with $V=\emptyset$), we see that \eqref{eq-Qb-f1-setminus-*-relax-1} is bounded by (below we denote $E(S)_{i,j}=0$ if $(i,j) \not \in E(S)$)
\begin{align*}
    2^{5\tau(G^{I,J}_{\cup}) +12\aleph} \prod_{(i,j) \in E(G^{I,J}_{\cup})} \Big( \tfrac{\sqrt{n}}{\sqrt{\epsilon^2\lambda s}} \Big)^{E(S^I_1)_{i,j}+E(S^J_2)_{i,j}-2} \big(\epsilon^\mathtt m (\tfrac{\epsilon^2 \lambda s}{n})^{\ell/2} \big)^{2\iota\aleph-|I|-|J|} \,. 
\end{align*}
Using the fact that $\tau(G^{I,J}_{\cup}) \leq \tau(G_{\cup})$ and
\begin{align*}
    \sum_{(i,j) \in E(G^{I,J}_{\cup})} (E(S^I)_{i,j}+E(K^J)_{i,j}-2)&= 2(\aleph-1)+\ell(|I|+|J|)-2|E(G^{I,J}_{\cup})| \\
    &\leq 2(\aleph - 1 + \iota\ell\aleph) - 2|E(G_\cup)| - \ell(|I|+|J|) \,,
\end{align*}
we obtain that \eqref{eq-Qb-f1-setminus-*-relax-1} is further bounded by 
\begin{align}
    2^{5\tau(G_{\cup})+10\aleph} n^{\aleph-1+\iota\ell\aleph-|E(G_{\cup})|} / (\epsilon^2\lambda s)^{\aleph-1+\iota\ell\aleph-|E(G_{\cup})|} \,.  \label{eq-Qb-f1-setminus-*-relax-2}
\end{align}
Therefore, applying \eqref{eq-phi-as-linear-combination-of-beta} we have that
\begin{align}
    \Big| \mathbb E_{\Qb}\big[ \phi_S(A)\phi_K(A) \big] \Big| \leq 2^{5\tau(G_{\cup})+10\aleph+2\iota\aleph} n^{\aleph-1+\iota\ell\aleph-|E(G_{\cup})|} / (\epsilon^2\lambda s)^{\aleph-1+\iota\ell\aleph-|E(G_{\cup})|} \,.  \label{eq-before-lem-4.3}
\end{align}

\begin{lemma}{\label{lem-character-Q-H-I}}
    For all $(S,K) \not\in \mathfrak P^*_{\mathbf H, \mathbf I}$ such that $\widetilde S \cup \widetilde K = G_{\cup}$ and $\mathcal L(S) \cup \mathcal L(K) \subset V(S \cap K)$, we have
    \begin{equation}\label{eq-character-Q-H-I}
        2\iota\aleph-\tau(G_{\cup}) \leq \tfrac{2\ell\iota\aleph+2\aleph-|E(G_{\cup})|}{\ell/2} \,.
    \end{equation}
\end{lemma}
\begin{proof}
    The proof is highly similar to the proof of Lemma~\ref{lem-character-non-principle}, and we omit further details here for simplicity.
\end{proof}

Using Lemma~\ref{lem-character-Q-H-I}, we have
\begin{align}
    & \sum_{ S,K \not\in \mathfrak P_{\mathbf H,\mathbf I}^* } \Big| \Eb_{\Qb}\big[ \phi_S(A)\phi_K(A) \big] \Big|\leq \sum_{ \substack{ |E(G_{\cup})| \leq 2\ell\iota\aleph+2\aleph \\ G_{\cup} \text{ satisfies } \eqref{eq-character-Q-H-I} } } \sum_{S \cup K = G_{\cup}} \Big| \mathbb E_{\Qb}\big[ \phi_S(A)\phi_K(A) \big] \Big| \nonumber \\
    \leq \ & \sum_{ \substack{ |E(G_{\cup})| \leq 2\ell\iota\aleph+2\aleph \\ G_{\cup} \text{ satisfies } \eqref{eq-character-Q-H-I} } } \big( \tfrac{\epsilon^2 \lambda s}{n} \big)^{ -(\aleph-1+\ell\iota\aleph)+|E(G_{\cup})| } \cdot \operatorname{ENUM}'(G_{\cup}) \,, \label{eq-Qb-var-bad-part-relax-1}
\end{align}
where for a subgraph $G$ (e.g., $G=G_\cup$)
\begin{equation}{\label{eq-def-ENUM'-G}}
\begin{aligned}
    \operatorname{ENUM}'(G) = \#\Big( \cup_{\mathbf H \in \mathcal H} \big\{ & (S,K) \in \mathfrak P_{\mathbf H , \mathbf I} \setminus \mathfrak P^*_{\mathbf H , \mathbf I}: \widetilde S \cup \widetilde K= G, \mathcal L(S) \cup \mathcal L(K) \subset V(S \cap K) \big\} \Big) \,.
\end{aligned}
\end{equation}  
Using Lemma~\ref{lem-enu-decorated-trees-in-given-graph}, we have
\begin{align*}
    \operatorname{ENUM}'(G_{\cup}) &\leq \binom{|E(G_{\cup})|}{\aleph} \cdot \Big( |V(G_{\cup})| \cdot \tau(G_{\cup})! \Big)^{2\iota\aleph} \\
    &\leq (2\ell\iota\aleph)^{2\aleph} \cdot \Big(  \tau(G_{\cup})! \Big)^{2\iota\aleph} \leq (2\ell\aleph)^{4\aleph} \leq n^{o(1)} \,.
\end{align*}
Plugging this estimation into \eqref{eq-Qb-var-bad-part-relax-1}, then we have (below we write $G$ satisfies \eqref{eq-character-Q-H-I} if \eqref{eq-character-Q-H-I} holds with $G_\cup$ replaced by $G$)
\begin{align}\label{eq-bound-final-step-ii}
    \eqref{eq-Qb-var-bad-part-relax-1} &\leq n^{o(1)} \cdot \sum_{ \substack{ |E(G)| \leq 2\ell\iota\aleph+2\aleph, |\mathcal L(G)| \leq 2 \aleph \\ G \text{ satisfying } \eqref{eq-character-Q-H-I} \\ \operatorname{ENUM}'(G) > 0 } } \big( \tfrac{\epsilon^2 \lambda s}{n} \big)^{ -(\aleph-1+\ell\iota\aleph)+|E(G)| } \nonumber  \\
    & = n^{o(1)} \cdot \sum_{ \substack{ |E(\mathbf G)| \leq 2\ell\iota\aleph + 2\aleph, |\mathcal L(\mathbf G)| \leq 2 \aleph \\ \mathbf G \text{ satisfying } \eqref{eq-character-Q-H-I} \\ \operatorname{ENUM}'(\mathbf G) > 0 } } \big( \tfrac{\epsilon^2 \lambda s}{n} \big)^{ -(\aleph-1+\ell\iota\aleph)+|E(\mathbf G)| } \cdot \#\{  G \subset \mathsf K_n: G \cong \mathbf G \} \nonumber \\
    &\leq n^{\aleph - 1 +\ell\iota\aleph + o(1)} \cdot \sum_{ \substack{ 0 \leq x \leq 2\aleph + 2\ell\iota\aleph \\ \ell(2\iota\aleph - y) \leq 2(2\ell\iota\aleph+2\aleph-x) } } (\epsilon^2 \lambda s)^{ -(\aleph-1+\ell\iota\aleph)+x }n^{-y} (\ell\aleph)^{ 2y } \,,
\end{align}
where $G$ is summed over $G\cong \mathbf{G}$, $\mathbf{G}$ is summed over $|E(\mathbf{G})|=x,\tau(\mathbf{G})=y$ and then over $x,y$, and the last inequality holds by Lemma~\ref{lem-enu-union-of-decorated-trees}. Thus, we have
\begin{align}
    \eqref{eq-Qb-var-bad-part-relax-1} &\leq n^{\aleph - 1 +\ell\iota\aleph + o(1)} \sum_{0 \leq x \leq 2\aleph + 2\ell\iota\aleph} (\epsilon^2 \lambda s)^{ -(\aleph-1+\ell\iota\aleph)+x }\big(\tfrac{(\ell\aleph)^3}{n}\big)^{-2\iota\aleph+(2(x-2\aleph)/l)} \nonumber \\
    &\circeq n^{\aleph - 1 +\ell\iota\aleph + o(1)}\sum_{0 \leq x \leq 2\aleph + 2\ell\iota\aleph} \big(\tfrac{(\epsilon^2 \lambda s)^\ell}{n^2}\big)^{-\iota\aleph+x/\ell} \nonumber \\
    &\circeq n^{\aleph - 1 +\ell\iota\aleph - 2\iota\aleph + o(1)}(\epsilon^2 \lambda s )^{\ell\iota\aleph}\, .
\end{align}
Plugging this bound into \eqref{eq-Qb-var-bad-part} yields
\begin{align}\label{eq-Qb-var-bad-bound-final-i}
    \eqref{eq-Qb-var-bad-part} &\leq \sum_{\mathbf H,\mathbf I \in \mathcal H} \frac{ s^{2(\aleph-1)} (\epsilon^2 \lambda s)^{2\ell\iota\aleph} \big(\operatorname{Aut}(\mathsf T(\mathbf H)) \big)^2 }{ n^{2+4\iota\aleph+o(1)} }(\epsilon^2 \lambda s )^{2\ell\iota\aleph}\nonumber \\
    &\leq o(1)\cdot s^{2(\aleph - 1)}\big( \tfrac{(\epsilon^2 \lambda s)^\ell(1-\epsilon^{2\mathtt m})^{1/2}}{n}\big)^{4\iota\aleph}\, ,
\end{align}
giving the desired bound on \eqref{eq-Qb-var-bad-part}. In particular, by combining \eqref{eq-Qb-var-bad-bound-final-i} with Proposition~\ref{prop-first-moment} and Lemma~\ref{lem-enu-decorated-trees}, we obtain 
\begin{align}\label{proof-second-moment-item-1-2}
    \eqref{eq-Qb-var-bad-part} = o(1) \cdot \mathbb E_{\Pb}[f]^2 \,.
\end{align}
Finally, combined with Lemmas~\ref{lem-bound-Qb-var-good-part} and \ref{lem-Qb-var-good-part-2,3}, this completes the proof of Item (1) of Proposition~\ref{prop-second-moment}.

\subsection{Proof of Item (2)}{\label{subsec:bound-var-Pb}}

This subsection is devoted to the proof of Item~(2) of Proposition~\ref{prop-second-moment}. We define $\mathfrak R_{\mathbf H,\mathbf I}= \mathfrak R_{\mathbf H} \times \mathfrak R_{\mathbf H} \times \mathfrak R_\mathbf I \times \mathfrak R_\mathbf I$, and introduce 
\begin{align}
    & \mathfrak R^*_{\mathbf H,\mathbf I} = \big\{ (S_1,S_2;K_1,K_2) \in \mathfrak R_{\mathbf H}^* \times \mathfrak R_{\mathbf H}^* \times \mathfrak R_{\mathbf I}^* \times \mathfrak R_{\mathbf I}^*: V(S_1) \cap V(K_1), V(S_2) \cap V(K_2) = \emptyset \big\} \,. \label{eq-def-R^*-H,I} \\
    & \mathfrak R^{**}_{\mathbf H,\mathbf I} = \big\{ (S_1,S_2;K_1,K_2) \in \mathfrak R_{\mathbf H}^* \times \mathfrak R_{\mathbf H}^* \times \mathfrak R_{\mathbf I}^* \times \mathfrak R_{\mathbf I}^*: \mathsf T(S_1)=\mathsf T(K_1), \mathsf T(S_2)=\mathsf T(K_2) \big\} \,. \label{eq-def-R^**-H,I}
\end{align}
Denote
\begin{align*}
    G_{\cup}=\pi(\widetilde S_1) \cup \pi(\widetilde K_1) \cup \widetilde S_2 \cup \widetilde K_2 \mbox{ and } \mathsf T_{\cup} = \pi(\mathsf T(S_1)) \cup \pi(\mathsf T(K_1)) \cup \mathsf T(S_2) \cup \mathsf T(K_2) \,.
\end{align*}
In addition, let $G_{\geq 2} \subset G_{\cup}$ be the (simple) subgraph whose vertex set $V(G_{\geq 2})$ consists of vertices that appear in at least two of the sets $\{ V(S_1),V(S_2),V(K_1),  V(K_2) \}$, and whose edge set $E(G_{\geq 2})$ is defined analogously. We define $\mathcal A_{\bullet}= \mathcal A_{\bullet}(S_1,S_2;K_1,K_2)$ as follows:
\begin{itemize}
    \item If $(S_1,S_2;K_1,K_2) \in \mathfrak R^*_{\mathbf H,\mathbf I}$, then $\mathcal A_{\bullet}$ is the subset of $\mathfrak S_n$ such that $\mathcal L(G_{\cup}) \subset V(G_{\geq 2})$ and $G_{\geq 2}$ is a union of two disjoint trees and contains $\mathsf T_{\cup}$.
    \item If $(S_1,S_2;K_1,K_2) \in \mathfrak R^{**}_{\mathbf H,\mathbf I}$, then $\mathcal A_{\bullet}$ is the subset of $\mathfrak S_n$ such that $\mathcal L(G_{\cup}) \subset V(G_{\geq 2})$ and $\pi(V(S_1 \cup K_1)) \cap V(S_2 \cup K_2) = \emptyset$.
    \item Otherwise, set $\mathcal A_{\bullet}=\emptyset$.
\end{itemize}
Recalling \eqref{eq-def-f} and noting that symmetry ensures $\mathbb E_{\Pb_{\mathsf{id}}}[f^2]= \mathbb E_{\Pb}[f^2], \mathbb E_{\Pb_{\mathsf{id}}}[f]= \mathbb E_{\Pb}[f]$, we may write the variance as 
\begin{align}
    & \operatorname{Var}_{\Pb}[f] = \sum_{ \mathbf H,\mathbf I \in \mathcal H } \frac{ s^{2(\aleph-1)} (\epsilon^2 \lambda s)^{2\ell\iota\aleph} \operatorname{Aut}(\mathsf T(\mathbf H)) \operatorname{Aut}(\mathsf T(\mathbf I)) }{ n^{2(\aleph+\ell\iota\aleph)} } \ \cdot \nonumber \\
    & \sum_{(S_1,S_2;K_1,K_2) \in \mathfrak R_{\mathbf H,\mathbf I}} \Big( \mathbb E_{\Pb_{\mathsf{id}}}[\phi_{S_1,S_2}\phi_{K_1,K_2}] - \mathbb E_{\Pb_{\mathsf{id}}}[\phi_{S_1,S_2}] \mathbb E_{\Pb_{\mathsf{id}}}[\phi_{K_1,K_2}] \Big) \,.  \label{eq-var-Pb-expand}
\end{align}
Thus, (recalling \eqref{eq-def-mathcal-A-diamond}) we may further decompose \eqref{eq-var-Pb-expand} into four parts, where the first part equals 
\begin{align}
    & \sum_{ \mathbf H,\mathbf I \in \mathcal H } \frac{ s^{2(\aleph-1)} (\epsilon^2 \lambda s)^{2\ell\iota\aleph} \operatorname{Aut}(\mathsf T(\mathbf H)) \operatorname{Aut}(\mathsf T(\mathbf I)) }{ n^{2(\aleph+\ell\iota\aleph)} } \ \cdot \sum_{(S_1,S_2;K_1,K_2) \in \mathfrak R^*_{\mathbf H,\mathbf I}}  \nonumber \\
    & \Big( \mathbb E_{\Pb_{\mathsf{id}}}[\phi_{S_1,S_2}\phi_{K_1,K_2} \mathbf 1_{\mathsf{id}\in\mathcal A_{\bullet}(S_1,S_2;K_1,K_2)}] - \mathbb E_{\Pb_{\mathsf{id}}}[\phi_{S_1,S_2} \mathbf 1_{\mathsf{id}\in\mathcal A(S_1,S_2)}] \mathbb E_{\Pb_{\mathsf{id}}}[\phi_{K_1,K_2} \mathbf 1_{\mathsf{id}\in\mathcal A(K_1,K_2)}] \Big) \,, \label{eq-Pb-var-good-part} 
\end{align}
the second part equals
\begin{align}
    & \sum_{ \mathbf H,\mathbf I \in \mathcal H } \frac{ s^{2(\aleph-1)} (\epsilon^2 \lambda s)^{2\ell\iota\aleph} \operatorname{Aut}(\mathsf T(\mathbf H)) \operatorname{Aut}(\mathsf T(\mathbf I)) }{ n^{2(\aleph+\ell\iota\aleph)} } \ \cdot  \nonumber \\
    & \sum_{(S_1,S_2;K_1,K_2) \in \mathfrak R^{**}_{\mathbf H,\mathbf I}} \Big( \mathbb E_{\Pb_{\mathsf{id}}}[\phi_{S_1,S_2}\phi_{K_1,K_2} \mathbf 1_{\mathsf{id}\in\mathcal A_{\bullet}(S_1,S_2;K_1,K_2)}]  \Big) \,, \label{eq-Pb-var-good-part'} 
\end{align}
the third part equals
\begin{align}
    & \sum_{ \mathbf H,\mathbf I \in \mathcal H } \frac{ s^{2(\aleph-1)} (\epsilon^2 \lambda s)^{2\ell\iota\aleph} \operatorname{Aut}(\mathsf T(\mathbf H)) \operatorname{Aut}(\mathsf T(\mathbf I)) }{ n^{2(\aleph+\ell\iota\aleph)} }  \nonumber \\
    & \sum_{ \substack{ (S_1,S_2;K_1,K_2) \in \mathfrak R_{\mathbf H,\mathbf I} } }  \mathbb E_{\Pb}\big[ \phi_{S_1,S_2}\phi_{K_1,K_2} \mathbf 1_{\pi\not\in\mathcal A_{\bullet}(S_1,S_2;K_1,K_2)} \big]  \,, \label{eq-Pb-var-bad-part-2} 
\end{align}
and the fourth part equals
\begin{align}
    \sum_{ \mathbf H,\mathbf I \in \mathcal H }  \sum_{ \substack{ (S_1,S_2;K_1,K_2) \in \mathfrak R_{\mathbf H,\mathbf I} } } \frac{ s^{2(\aleph-1)} (\epsilon^2 \lambda s)^{2\ell\iota\aleph} \operatorname{Aut}(\mathsf T(\mathbf H)) \operatorname{Aut}(\mathsf T(\mathbf I))h(S_1,S_2;K_1,K_2) }{ n^{2(\aleph+\ell\iota\aleph)} } \,. \label{eq-Pb-var-bad-part-1} 
\end{align}
where for $(S_1,S_2;K_1,K_2) \in \mathfrak R^{**}_{\mathbf H,\mathbf I}$ we let $h(S_1,S_2;K_1,K_2)=-\mathbb E_{\Pb_{\mathsf{id}}}[\phi_{S_1,S_2}] \mathbb E_{\Pb_{\mathsf{id}}}[\phi_{K_1,K_2}]$, and otherwise we let
\begin{align*}
    h(S_1,S_2;K_1,K_2) = \mathbb E_{\Pb_{\mathsf{id}}}[\phi_{S_1,S_2}\mathbf 1_{\mathsf{id}\in\mathcal A(S_1,S_2)}] \mathbb E_{\Pb_{\mathsf{id}}}[\phi_{K_1,K_2}\mathbf 1_{\mathsf{id}\in\mathcal A(K_1,K_2)}]-\mathbb E_{\Pb_{\mathsf{id}}}[\phi_{S_1,S_2}] \mathbb E_{\Pb_{\mathsf{id}}}[\phi_{K_1,K_2}] \,. 
\end{align*}
It suffices to show that \eqref{eq-Pb-var-good-part}, \eqref{eq-Pb-var-good-part'}, \eqref{eq-Pb-var-bad-part-2} and \eqref{eq-Pb-var-bad-part-1} are upper-bounded by $o(1) \cdot \mathbb E_{\Pb}[f]^2$. We first consider \eqref{eq-Pb-var-good-part}. Note that when $(S_1,S_2,K_1,K_2) \in \mathfrak R_{\mathbf H,\mathbf I}^*$ and $\mathsf{id} \in \mathcal A_{\bullet}(S_1,S_2;K_1,K_2)$, we have $G_{\geq 2} \subset \widetilde{S}_1 \cup \widetilde K_1$ (since $V(S_2) \cap V(K_2)=\emptyset$) is the union of two disjoint trees. Thus, from $\mathsf T(S_1),\mathsf T(S_2) \subset \mathsf T_{\cup} \subset G_{\geq 2}$ and $V(\mathsf T(S_1))\cap V(\mathsf T(S_2))=\emptyset$ we have that one of the trees in $G_{\geq 2}$ contains $\mathsf T(S_1)$ and the other contains $\mathsf T(S_2)$. Similarly, we can show that one of the trees in $G_{\geq 2}$ contains $\mathsf T(K_1)$ and the other contains $\mathsf T(K_2)$. 
In conclusion, when $(S_1,S_2,K_1,K_2) \in \mathfrak R_{\mathbf H,\mathbf I}^*$ and $\mathsf{id} \in \mathcal A_{\bullet}(S_1,S_2;K_1,K_2)$, one of the two following conditions must hold: (i) $S_1 \cap S_2, K_1 \cap K_2$ are two disjoint trees containing $\mathsf{T}(S_1) \cup \mathsf T(S_2), \mathsf{T}(K_1) \cup \mathsf T(K_2)$ respectively; (ii) $S_1 \cap K_2, K_1 \cap S_2$ are two disjoint trees containing $\mathsf{T}(S_1) \cup \mathsf T(K_2), \mathsf{T}(K_1) \cup \mathsf T(S_2)$ respectively. Thus, if we denote
\begin{equation}{\label{eq-def-Bumpeq}}
    S \Bumpeq K \mbox{ if and only if } S \cap K \mbox{ is a tree containing } \mathsf{T}(S) \cup \mathsf T(K) \,,
\end{equation}
we have 
\begin{align*}
    \mathbf 1_{\mathsf{id}\in\mathcal A_{\bullet}(S_1,S_2;K_1,K_2)} \leq \mathbf 1_{S_1 \Bumpeq S_2} \mathbf 1_{K_1 \Bumpeq K_2} + \mathbf 1_{S_1 \Bumpeq K_2} \mathbf 1_{S_2 \Bumpeq K_1} \,.
\end{align*}
Thus, we have
\begin{align*}
    & \big|\mathbb E_{\Pb_{\mathsf{id}}}[\phi_{S_1,S_2}\phi_{K_1,K_2} \mathbf 1_{\mathsf{id}\in\mathcal A_{\bullet}(S_1,S_2;K_1,K_2)}]\big| = \big|\mathbb E_{\Pb_{\mathsf{id}}}[\phi_{S_1,S_2}\phi_{K_1,K_2} ]\big|\mathbf 1_{\mathsf{id}\in\mathcal A_{\bullet}(S_1,S_2;K_1,K_2)}\\
    \leq\ & \big|\mathbb E_{\Pb_{\mathsf{id}}}[\phi_{S_1,S_2}\phi_{K_1,K_2}]\big| (\mathbf 1_{S_1 \Bumpeq S_2} \mathbf 1_{K_1 \Bumpeq K_2} + \mathbf 1_{S_1 \Bumpeq K_2} \mathbf 1_{K_1 \Bumpeq S_2}) \\
    =\ & \mathbf 1_{S_1 \Bumpeq S_2} \mathbf 1_{K_1 \Bumpeq K_2} \big|\mathbb E_{\Pb_{\mathsf{id}}}[\phi_{S_1,S_2}] \mathbb E_{\Pb_{\mathsf{id}}}[\phi_{K_1,K_2}]\big| + \mathbf 1_{S_1 \Bumpeq K_2} \mathbf 1_{K_1 \Bumpeq S_2} \big|\mathbb E_{\Pb_{\mathsf{id}}}[\phi_{S_1,K_2}] \mathbb E_{\Pb_{\mathsf{id}}}[\phi_{K_1,S_2}]\big| \,.
\end{align*}
Combined with the fact that
\begin{align*}
    \mathbb E_{\Pb_{\mathsf{id}}}[\phi_{S_1,S_2} \mathbf 1_{\mathsf{id}\in\mathcal A(S_1,S_2)}] = \mathbf 1_{S_1 \Bumpeq S_2} \mathbb E_{\Pb_{\mathsf{id}}}[\phi_{S_1,S_2}] \mbox{ and } \mathbb E_{\Pb_{\mathsf{id}}}[\phi_{K_1,K_2} \mathbf 1_{\mathsf{id}\in\mathcal A(K_1,K_2)}] = \mathbf 1_{K_1 \Bumpeq K_2} \mathbb E_{\Pb_{\mathsf{id}}}[\phi_{K_1,K_2}] \,,
\end{align*}
it yields that (note that $S_1 \Bumpeq K_2, S_2 \Bumpeq K_1$ implies $\mathsf T(\mathbf H) \sim \mathsf T(\mathbf I)$)
\begin{align}
    \eqref{eq-Pb-var-good-part} \leq & \sum_{ \substack{ \mathbf H,\mathbf I \in \mathcal H \\ \mathsf T(\mathbf H) \sim \mathsf T(\mathbf I) } } \frac{ s^{2(\aleph-1)} (\epsilon^2 \lambda s)^{2\ell\iota\aleph}  \operatorname{Aut}(\mathsf T(\mathbf H)) \operatorname{Aut}(\mathsf T(\mathbf I)) }{ n^{2(\aleph+\ell\iota\aleph)} }  \sum_{ \substack{ (S_1,S_2;K_1,K_2) \in \mathfrak R^*_{\mathbf H,\mathbf I} \\ S_1 \Bumpeq K_2, S_2 \Bumpeq K_1 } } \big|\mathbb E_{\Pb_{\mathsf{id}}}[\phi_{S_1,K_2}] \mathbb E_{\Pb_{\mathsf{id}}}[\phi_{K_1,S_2}]\big| \,. \label{eq-Pb-var-good-part-1}
\end{align}
Similarly as in \eqref{eq-expectation-pi-in--A} and \eqref{eq-first-moment-P-id--A}, we get that \eqref{eq-Pb-var-good-part} is bounded by
\begin{align}
    & \sum_{ \substack{ \mathbf H,\mathbf I \in \mathcal H \\ \mathsf T(\mathbf H) \sim \mathsf T(\mathbf I) } } \frac{ s^{4(\aleph-1)} (\epsilon^2 \lambda s)^{4\ell\iota\aleph} (1-\epsilon^{2\mathtt m})^{2\iota\aleph}\operatorname{Aut}(\mathsf T(\mathbf H)) \operatorname{Aut}(\mathsf T(\mathbf I)) }{ n^{2(\aleph+(\ell+1)\iota\aleph)} } \cdot |\mathfrak R_{\mathbf H}^*| |\mathfrak R_{\mathbf I}^*| \ \cdot \nonumber \\
    & \Big(\sum_{ \mathbf x \geq 0 } \frac{ \operatorname{Aut}(\mathsf T(\mathbf H) \oplus (\oplus_{u\in\mathsf{Vert}(\mathsf P(\mathbf H))} \mathcal L_{u})) }{ \operatorname{Aut}(\mathsf T(\mathbf I)) (\epsilon^2 \lambda s(1-\epsilon^{2\mathtt m}))^{\sum{x_u}} } \Big) \Big(\sum_{ \mathbf y \geq 0 } \frac{ \operatorname{Aut}(\mathsf T(\mathbf I) \oplus (\oplus_{v\in\mathsf{Vert}(\mathsf P(\mathbf I))} \mathcal L_{v})) }{ \operatorname{Aut}(\mathsf T(\mathbf H)) (\epsilon^2 \lambda s(1-\epsilon^{2\mathtt m}))^{\sum{y_v}} } \Big) \ \nonumber \\
    \leq\ & s^{4(\aleph-1)} e^{8\iota\aleph} (1-\tfrac{1}{\epsilon^2 \lambda s(1-\epsilon^{2\mathtt m})})^{-4\iota\aleph} \big( \tfrac{(\epsilon^2 \lambda s)^{\ell}}{n} \big)^{4\iota\aleph} \cdot \#\{ \mathbf H,\mathbf I \in \mathcal H: \mathsf T(\mathbf H) \sim \mathsf T(\mathbf I)\} \nonumber \\
    =\ & o(1) \cdot \mathbb E_{\Pb}[f]^2 \,,
    \label{eq-Pb-var-good-part-2}
\end{align}
where in the inequality we used Item~(5) in Lemma~\ref{lem-useful-property-trees-and-sets} and in the equality we used Lemma~\ref{lem-remove-sim-relation}, Lemma~\ref{lem-enu-decorated-trees} and Proposition~\ref{prop-first-moment}. Combining \eqref{eq-Pb-var-good-part-1} and \eqref{eq-Pb-var-good-part-2}, we have
\begin{equation}{\label{eq-Pb-var-good-part-1-relax}}
    \eqref{eq-Pb-var-good-part} = o(1) \cdot \mathbb E_{\Pb}[f]^2 \,.
\end{equation}
Now we consider \eqref{eq-Pb-var-good-part'}. For $(S_1,S_2;K_1,K_2)\in \mathfrak R_{\mathbf H,\mathbf I}^{**}$ and $\pi\in\mathcal A_{\bullet}(S_1,S_2;K_1,K_2)$, we have
\begin{align*}
    \mathbb E_{\Pb_{\pi}}\big[ \phi_{S_1,S_2}\phi_{K_1,K_2} \big] &= \mathbb E_{\Pb_{\pi}}\big[ \phi_{S_1}(A)\phi_{K_1}(A) \phi_{S_2}(B)\phi_{K_2}(B) \big]  \\
    &= \mathbb E_{\Pb_{\pi}}\big[ \phi_{S_1}(A)\phi_{K_1}(A) \big] \mathbb E_{\Pb_{\pi}}\big[ \phi_{S_2}(B)\phi_{K_2}(B) \big] \\
    &= \mathbb E_{\Qb}\big[ \phi_{S_1}(A)\phi_{K_1}(A) \big] \mathbb E_\Qb\big[ \phi_{S_2}(B)\phi_{K_2}(B) \big] \,,
\end{align*}  
where the second equality follows from the fact that $\pi(V(S_1 \cup K_1)) \cap V(S_2 \cup K_2) = \emptyset$ for $(S_1,S_2;K_1,K_2)\in \mathfrak R_{\mathbf H,\mathbf I}^{**}$ and $\pi\in\mathcal A_{\bullet}(S_1,S_2;K_1,K_2)$. Thus, we obtain that \eqref{eq-Pb-var-good-part'} is upper-bounded by
\begin{align}
    &\sum_{ \substack{ (S_1,S_2;K_1,K_2) \in \mathfrak R_{\mathbf H,\mathbf I}^{**} } } \frac{ s^{2(\aleph-1)} (\epsilon^2 \lambda s)^{2\ell\iota\aleph} \operatorname{Aut}(\mathsf T(\mathbf H)) \operatorname{Aut}(\mathsf T(\mathbf I)) }{ n^{2(\aleph+\ell\iota\aleph)} } \mathbb E_{\Qb}\big[ \phi_{S_1}(A)\phi_{K_1}(A) \big] \mathbb E_\Qb\big[ \phi_{S_2}(B)\phi_{K_2}(B) \big]   \nonumber \\
    & \leq \mathbb E_{\Qb}[f^2] = o(1) \cdot \mathbb E_{\Pb}[f]^2 \,, \label{eq-Pb-var-good-part'-relax} 
\end{align}
where the equality follows from Item~(1) of Proposition~\ref{prop-second-moment}.  

Now we turn to \eqref{eq-Pb-var-bad-part-1}. By definition of $h(S_1,S_2;K_1,K_2)$ , given $(S_1 , S_2 ;K_1 , K_2)$ we have $h(S_1,S_2;K_1,K_2)$ is either zero or $-\mathbb E_{\Pb_{\mathsf{id}}}[\phi_{S_1,S_2} ] \mathbb E_{\Pb_{\mathsf{id}}}[\phi_{K_1,K_2}]$. Therefore, it suffices to show
\begin{align}
    & \sum_{ \mathbf H,\mathbf I \in \mathcal H } \frac{ s^{2(\aleph-1)} (\epsilon^2 \lambda s)^{2\ell\iota\aleph} \operatorname{Aut}(\mathsf T(\mathbf H)) \operatorname{Aut}(\mathsf T(\mathbf I)) }{ n^{2(\aleph+\ell\iota\aleph)} } \nonumber \\
    & \sum_{ \substack{ (S_1,S_2;K_1,K_2) \in \mathfrak R_{\mathbf H,\mathbf I} } } \Big( \mathbb E_{\Pb_{\mathsf{id}}}[\phi_{S_1,S_2} ]\mathbb E_{\Pb_{\mathsf{id}}}[\phi_{K_1,K_2}] \Big)_{-} = o(1) \cdot \big( \Eb_\Pb[f] \big)^2\,,\label{eq-Pb-var-part-1-reduction}
\end{align}
where $x_- := -\min(x,0)$ is the negative part of $x$. By Lemma~\ref{lem-est-intersection-general-istree-P-Q}, we have $\mathbb E_{\Pb_{\mathsf{id}}}[\phi_{S_1,S_2}]\geq 0$ for $\mathsf{id}\in\mathcal A(S_1,S_2)$ and similarly for $K_1,K_2$. Thus,  
\begin{equation*}
    \big( \mathbb E_{\Pb_{\mathsf{id}}}[\phi_{S_1,S_2} ]\mathbb E_{\Pb_{\mathsf{id}}}[\phi_{K_1,K_2}]\big)_{-} \leq (\mathbf{1}_{\{\mathsf{id} \not\in \mathcal A(S_1,S_2)\}} + \mathbf{1}_{\{\mathsf{id} \not\in \mathcal A(K_1,K_2)\}}) \cdot \big|\mathbb E_{\Pb_{\mathsf{id}}}[\phi_{S_1,S_2} ]\mathbb E_{\Pb_{\mathsf{id}}}[\phi_{K_1,K_2}] \big|\,,
\end{equation*}
which yields that \eqref{eq-Pb-var-part-1-reduction} is bounded by 2 times
\begin{align}
    & \sum_{ \substack{ \mathbf H,\mathbf I \in \mathcal H \\ (S_1,S_2;K_1,K_2) \in \mathfrak R_{\mathbf H,\mathbf I} \\ \mathsf{id} \not\in \mathcal A(S_1,S_2) } } \tfrac{ s^{2(\aleph-1)} (\epsilon^2 \lambda s)^{2\ell\iota\aleph} \operatorname{Aut}(\mathsf T(\mathbf H)) \operatorname{Aut}(\mathsf T(\mathbf I)) }{ n^{2(\aleph+\ell\iota\aleph)} } \cdot \big|\mathbb E_{\Pb_{\mathsf{id}}}[\phi_{S_1,S_2} ]\mathbb E_{\Pb_{\mathsf{id}}}[\phi_{K_1,K_2}]\big| \nonumber \\
    &\leq \sum_{ \substack{ \mathbf H \in \mathcal H \\ (S_1,S_2) \in \mathfrak R_{\mathbf H}\times \mathfrak R_\mathbf H \\ \mathsf{id}\not\in\mathcal A(S_1,S_2) } } \tfrac{ s^{(\aleph-1)} (\epsilon^2 \lambda s)^{\ell\iota\aleph} \operatorname{Aut}(\mathsf T(\mathbf H)) |\mathbb E_{\Pb_{\mathsf{id}}}[\phi_{S_1,S_2} ] | }{ n^{(\aleph+\ell\iota\aleph)} }  \sum_{ \substack{ \mathbf I \in \mathcal H \\ (K_1,K_2) \in \mathfrak R_{\mathbf I} \times \mathfrak R_{\mathbf I} } } \tfrac{ s^{(\aleph-1)} (\epsilon^2 \lambda s)^{\ell\iota\aleph} \operatorname{Aut}(\mathsf T(\mathbf I)) |\mathbb E_{\Pb_{\mathsf{id}}}[\phi_{K_1,K_2}] | }{ n^{(\aleph+\ell\iota\aleph)} }  \nonumber \\
    &\overset{\eqref{eq-def-abs-E-Pb-f-tilde-A}, \eqref{eq-def-abs-E-Pb-f-A-c}}{=} \Big( \mathbb E_{ \Pb_{\mathsf{id}} }^{|\cdot|} \big[ f_{\mathcal A^c} \big]\Big) \Big( \mathbb E_{ \Pb_{\mathsf{id}} }^{|\cdot|} \big[ f_{\mathcal A^c} \big] + \mathbb E_{ \Pb_{\mathsf{id}} }^{|\cdot|} \big[ f_{\widetilde{\mathcal A}} \big] + \mathbb E_{\Pb_{\mathsf{id}}}[f_{\mathcal A_{\star}}] \Big) \nonumber \\
    &\overset{\eqref{eq-first-moment-semi-bad-part}, \eqref{eq-first-moment-bad-part}}{\leq} \Big( o(1) \cdot \mathbb E_{\Pb_{\mathsf{id}}}[f] \Big) \cdot \Big( (1+o(1)) \cdot \mathbb E_{\Pb_{\mathsf{id}}}[f] \Big) = o(1) \cdot \big( \Eb_\Pb[f] \big)^2 \,, \label{eq-Pb-var-bad-part-1-relax}
\end{align}
where in the first equality we again used $\mathbb E_{\Pb_{\mathsf{id}}}[\phi_{S_1,S_2}]\geq 0$ for $\mathsf{id}\in\mathcal A(S_1,S_2)$ and similarly for $K_1,K_2$. Finally, we consider \eqref{eq-Pb-var-bad-part-2}. For $I,J,I',J' \subset [\iota\aleph]$ we denote $G^{I,J;I',J'}_{\cup}=\widetilde S_1^I \cup \widetilde S_2^J \cup \widetilde K_1^{I'} \cup \widetilde K_2^{J'}$. In particular we denote $G_\cup = G^{[\iota\aleph],[\iota\aleph];[\iota\aleph],[\iota\aleph]}_\cup$. Then we have $|\mathcal L(G^{I,J;I',J'}_{\cup})| \leq 4\aleph$ for any $I,J,I',J' \subset [\iota\aleph]$. We establish a lemma which serves as a general estimate first.
\begin{lemma}\label{lem-Pb-var-bad-part-2-expectation-estimate-general}
    Suppose $S_1 ,S_2 , K_1 , K_2$ are multigraphs and $G_{\cup}=\widetilde S_1 \cup \widetilde S_2 \cup \widetilde K_1 \cup \widetilde K_2$. Then we have
    \begin{align}
        \mathbb E_{\Pb_{\mathsf{id}}}[ \beta_{S_1,S_2} \beta_{K_1,K_2} ] \leq \big( \tfrac{n}{\epsilon^2 \lambda s} \big)^{ \tfrac{\mathtt n}{2}-|E(G_{\cup})| } \,,
    \end{align}
    where $\mathtt n =\sum_{(i,j) \in G_\cup} ( E(S_1)_{i,j} + E(S_2)_{i,j} + E(K_1)_{i,j} + E(K_2)_{i,j} )$.
\end{lemma}
\begin{proof}
Recall the definition of $\chi$ in \eqref{eq-def-chi}. By Lemmas~\ref{lem-joint-moment-A-B} and \ref{lem-expectation-over-chain} we have (for simplicity of notations we assume $E(S_1)_{i,j}=0$ for $(i,j) \not \in E(S_1)$ and similarly for $S_2,K_1,K_2$)
\begin{align}
    & \mathbb E_{\Pb_{\mathsf{id}}}[ \beta_{S_1,S_2} \beta_{K_1,K_2} ] = \mathbb E_{\Pb_{\mathsf{id}}} \Bigg\{ \prod_{ (i,j) \in E(G_{\cup}) } \Bigg( \frac{ A_{i,j}-\tfrac{\lambda s}{n} }{ \sqrt{\lambda s/n} } \Bigg)^{\chi_{S_1 \cup K_1}(i,j)} \Bigg( \frac{ B_{i,j}-\tfrac{\lambda s}{n} }{ \sqrt{\lambda s/n} } \Bigg)^{\chi_{S_2 \cup K_2}(i,j)} \Bigg\} \notag \\
    \leq\ & \big( \tfrac{n}{\epsilon^2 \lambda s} \big)^{ \frac{1}{2}( \sum_{(i,j) \in E(G_{\cup})} ( E(S_1)_{i,j} + E(S_2)_{i,j} + E(K_1)_{i,j} + E(K_2)_{i,j}-2 ) } = \big( \tfrac{n}{\epsilon^2 \lambda s} \big)^{ \tfrac{\mathtt n}{2}-|E(G_{\cup})| }\,, \nonumber
\end{align}
which yields our desired result.
\end{proof}

By Lemma~\ref{lem-Pb-var-bad-part-2-expectation-estimate-general} and \eqref{eq-phi-as-linear-combination-of-beta}, we have
\begin{align}
    &\mathbb E_{\Pb_{\mathsf{id}}}[ \phi_{S_1,S_2} \phi_{K_1,K_2} ] = \sum_{I,J,I',J' \subset [\iota\aleph] } \mathbb E_{\Pb_{\mathsf{id}}}\Big[ \beta_{S^I_1,S^J_2} \beta_{K^{I'}_1,K^{J'}_2} \Big] \cdot \big(- \epsilon^\mathtt m (\tfrac{\epsilon^2 \lambda s}{n})^{\ell/2} \big)^{4\iota\aleph - |I|-|J|-|I'|-|J'|} \nonumber \\ 
    \leq\ & \sum_{I,J,I',J' \subset [\iota\aleph] }\big( \tfrac{n}{\epsilon^2 \lambda s} \big)^{ 2(\aleph-1)+\tfrac{\ell}{2}(|I|+|J|+|I'|+|J'|)-|E(G^{I,J;I',J'}_{\cup})| }\big( \epsilon^\mathtt m (\tfrac{\epsilon^2 \lambda s}{n})^{\ell/2} \big)^{4\iota\aleph - |I|-|J|-|I'|-|J'|} \nonumber \\
    \leq \ &2^{4\iota\aleph} \big( \tfrac{n}{\epsilon^2 \lambda s} \big)^{ 2(\aleph-1+\ell\iota\aleph)-|E(G_{\cup})| }
    \label{eq-Pb-var-bad-part-2-expectation} \,,
\end{align}
where in the first inequality we use the fact that 
\begin{align*}
    \sum_{(i,j) \in E(G^{I,J;I',J'}_{\cup})} ( E(S_1)_{i,j} + E(S_2)_{i,j} + E(K_1)_{i,j} + E(K_2)_{i,j} ) = 4(\aleph-1)+\ell(|I|+|J|+|I'|+|J'|)
\end{align*}
for $I,J,I',J' \subset [\iota\aleph]$, and in the second inequality we use the fact that
\begin{align*}
    |E(G_\cup)| \leq |E(G^{I,J;I',J'}_\cup)| + \ell (4\iota\aleph - |I| - |J| - |I'| - |J'|)\,.
\end{align*}
Thus, we have
\begin{align}
    \eqref{eq-Pb-var-bad-part-2} &\leq \sum_{ \mathbf H,\mathbf I \in \mathcal H } \sum_{ \substack{ (S_1,S_2;K_1,K_2) \\ \mathsf{id} \not \in \mathcal A_{\bullet}(S_1,S_2;K_1,K_2) } } \frac{ s^{2\aleph} (\epsilon^2 \lambda s)^{2\aleph} \operatorname{Aut}(\mathsf T(\mathbf H)) \operatorname{Aut}(\mathsf T(\mathbf I)) }{ n^{2} } \cdot \Big( \frac{\epsilon^2 \lambda s}{n} \Big)^{|E(G_{\cup})|} \nonumber \\
    &= \sum_{ \mathbf H,\mathbf I \in \mathcal H } \sum_{ \substack{ (S_1,S_2;K_1,K_2) \\ \mathsf{id} \not \in \mathcal A_{\bullet}(S_1,S_2;K_1,K_2) } } \frac{ 1 }{ n^{2-o(1)} } \cdot \Big( \frac{\epsilon^2 \lambda s}{n} \Big)^{|E(G_{\cup})|} \,. \label{eq-Pb-var-bad-part-2-relax-I}
\end{align}
Recall the definition of $G_{\geq 2}$ in the paragraph below \eqref{eq-def-R^**-H,I}. The following claim is crucial for our proof.
\begin{lemma}{\label{lem-characterize-setminus-R-H-I}}
    We have 
    \begin{equation}{\label{eq-characterize-setminus-R-H-I}}
        4\iota\aleph-1-\tau(G_{\cup}) \leq \tfrac{4\ell\iota\aleph+8\aleph - |E(G_{\cup})|}{\ell/2}
    \end{equation}
    for $\mathsf{id}\not\in \mathcal A_{\bullet}(S_1,S_2;K_1,K_2)$ and $\mathcal L(G_{\cup}) \subset V(G_{\geq 2})$.
\end{lemma}

\begin{remark}{\label{cor-of-lem-4.4}}
    Lemma~\ref{lem-characterize-setminus-R-H-I} remains unchanged if the self-avoiding paths $\{ \mathsf L_j(S_i), \mathsf L_j(K_i): 1 \leq j \leq \iota\aleph, i=1,2 \}$ are replaced by non-backtracking paths. The proof is similar and we omit further details.
\end{remark}

The proof of Lemma~\ref{lem-characterize-setminus-R-H-I} is incorporated in Section~\ref{subsec:proof-lem-4.6}. Based on Lemma~\ref{lem-characterize-setminus-R-H-I}, we now claim our bound on \eqref{eq-Pb-var-bad-part-2-relax-I}. Note that we can write \eqref{eq-Pb-var-bad-part-2-relax-I} as
\begin{align}\label{eq-Pb-var-bad-part-2-relax-I-rewrite}
    & n^{-2+o(1)} \sum_{ \substack{ |E(G)| \leq 4 \ell\iota\aleph, |\mathcal L(G)|\leq 4\aleph \\ G \text{ satisfying } \eqref{eq-characterize-setminus-R-H-I} \\ \operatorname{ENUM}''(G) > 0 } } \big( \tfrac{\epsilon^2 \lambda s}{n} \big)^{|E(G)|} \cdot \operatorname{ENUM}''(G) \,,
\end{align}
where $G$ is summed over $G\subset \mathsf K_n$ with $|E(G)|\leq 4\ell \iota\aleph, |\mathcal{L}(G)|\leq 4\aleph$ and the graphs satisfying inequality \eqref{eq-characterize-setminus-R-H-I} with $G_\cup$ replaced by $G$, and
\begin{align*}
    \operatorname{ENUM}''(G) = \#\big\{ (S_1,S_2;T_1,T_2) : S_1 \cup S_2 \cup T_1 \cup T_2=G, \mathcal L(G) \subset V(G_{\geq 2}(S_1,S_2,T_1,T_2)) \big\}\,.
\end{align*}
By Lemma~\ref{lem-enu-decorated-trees-in-given-graph}, we have
\begin{align*}
    \operatorname{ENUM}''(G) \leq \binom{4\ell\aleph +4\aleph}{\aleph}^4 \Big( (\ell\aleph)^{2}2^{5\tau(G)+5} \Big)^{4\iota\aleph} \,.
\end{align*}
Therefore, we have (below $\mathbf{G}$ is summed over $|E(\mathbf G)|=x,\tau(\mathbf{G})=y$ and then also over $x, y$)
\begin{align}
    \eqref{eq-Pb-var-bad-part-2-relax-I-rewrite} &\leq n^{-2+o(1)} \sum_{ \substack{ |E(G)| \leq 4 \ell\iota\aleph \\ G \text{ satisfying } \eqref{eq-characterize-setminus-R-H-I} \\ \operatorname{ENUM}''(G) > 0 } } \big( \tfrac{\epsilon^2 \lambda s}{n} \big)^{|E(G)|} \binom{4\ell\aleph +4\aleph}{\aleph}^4 \Big( (\ell\aleph)^{2}2^{5\tau(G)+5} \Big)^{4\iota\aleph} \nonumber  \\
    &\leq n^{-2+o(1)} \sum_{ \substack{ |E(\mathbf G)| \leq 4\ell\iota\aleph \\ \mathbf G \text{ satisfying } \eqref{eq-characterize-setminus-R-H-I} \\ \operatorname{ENUM}''(G) > 0 } } \big( \tfrac{\epsilon^2 \lambda s}{n} \big)^{|E(\mathbf G)|}  (\ell\aleph)^{\tau(\mathbf G)} \#\{ G \subset \mathsf K_n :G \cong \mathbf G \} \nonumber  \\
    &\leq n^{-2+o(1)} \sum_{ \substack{ 0 \leq x \leq 4\ell\iota\aleph \\ \ell(4\iota\aleph-1-y)\leq 2(4\ell\iota\aleph + 8\aleph - x) } } \big( \tfrac{\epsilon^2 \lambda s}{n} \big)^{x}  (\ell\aleph)^{y} n^{x-y} (\ell\aleph)^{2y} \nonumber \\
    &\leq n^{-2+o(1)} \sum_{ \substack{ 0 \leq x \leq 4\ell\iota\aleph \\ 4\iota\aleph-1-y\leq 2(4\ell\iota\aleph + 8\aleph - x)/\ell } } (\epsilon^2 \lambda s)^{x}  \big( \tfrac{(\ell\aleph)^3}{n}\big)^{y}  \nonumber \\
    &= n^{-2+o(1)}(\epsilon^2 \lambda s)^{4\ell\iota\aleph} \sum_{ \substack{ 0 \leq x \leq 4\ell\iota\aleph } }  \big( \tfrac{(\ell\aleph)^3}{n}\big)^{4\iota\aleph-1} \big( \tfrac{(\epsilon^2 \lambda s)^{\ell}}{n^2}\big)^{(x - 4\ell\iota\aleph )/\ell} \nonumber \\
    &= \frac{ (\epsilon^2 \lambda s)^{4\ell\iota\aleph} }{n^{2-o(1)}} \cdot n^{-(4\iota\aleph-1)+o(1)} \overset{\text{Proposition~\ref{prop-first-moment}}}{=} o(1) \cdot \mathbb E_{\Pb}[f]^2 \,, \label{eq-Pb-var-bad-part-2-relax-II}
\end{align}
where the second inequality holds by $\binom{4\ell\aleph +4\aleph}{\aleph}^4 (\ell\aleph)^{8\iota\aleph}2^{20\iota\aleph}= n^{o(1)}$ and $2^{20\iota\aleph} \leq \ell\aleph $, and the third inequality holds by Lemmas~\ref{lem-enu-union-of-decorated-trees} and \ref{lem-characterize-setminus-R-H-I}. Combining \eqref{eq-Pb-var-bad-part-2-relax-I} and \eqref{eq-Pb-var-bad-part-2-relax-II}, we see that
\begin{equation}
    \eqref{eq-Pb-var-bad-part-2} = o(1) \cdot \mathbb E_{\Pb}[f]^2 \,.
\end{equation}
Combined with \eqref{eq-Pb-var-bad-part-1-relax}, this completes the proof of Item (2) in Proposition~\ref{prop-second-moment}.

\section{Approximating the statistic}{\label{sec:approx-via-color-coding}}

In this section, we present a polynomial-time algorithm to approximately compute the statistic $f=f(A,B)$ defined in \eqref{eq-def-f}. The key starting point is to replace self-avoiding paths with non-backtracking paths, simplifying the computation while retaining essential structural properties. Our approach further integrates the ideas of color-coding and the matrix power method to achieve this efficiency.
\begin{DEF}\label{def-decorated-trees-NBP}
    For each $\mathbf H \in \mathcal H$, we say a multigraph $S \vdash \mathbf H$ \underline{(note its difference to $S \Vdash \mathbf H$ defined} \underline{in Definition~\ref{def-decorated-trees})}, if (counting edge multiplicities) $S$ can be decomposed into a tree $\mathsf T(S)$ and $\iota\aleph$ \underline{non-backtracking paths} $\mathsf M_1(S),\ldots,\mathsf M_{\iota\aleph}(S)$ such that the following conditions hold:
    \begin{enumerate}
        \item[(1)] $V(\mathsf T (S)) \subset [n]$;
        \item[(2)] $\operatorname{EndP}(\mathsf M_i(S))=\{ u_i,v_i \}$ where $u_i,v_i \in V(\mathsf T(S))$;
        \item[(3)] Denoting $\mathsf P(S)=\{ (u_1,v_1),\ldots,(u_{\iota\aleph},v_{\iota\aleph}) \}$, there exists a graph isomorphism $\varphi: S \to \mathbf H$ such that $\varphi$ maps $\mathsf T(S)$ to $\mathsf T(\mathbf H)$,maps $\mathsf P(S)$ to $\mathsf P(\mathbf H)$ and maps $\mathsf M_k (S)$ to $\mathsf M_k (\mathbf H)$ for $1 \leq k \leq \iota\aleph$.
    \end{enumerate}
    Given $S_1,S_2 \vdash \mathbf H \in \mathcal H$, we similarly write $\phi_{S_1,S_2}(A,B)=\phi_{S_1}(A) \phi_{S_2}(B)$, where for $S \vdash \mathbf H$ (recall \eqref{eq-def-beta-S})
    \begin{align*}
        \phi_S(X) = \beta_{\mathsf T(S)}(X) \cdot  \prod_{1 \leq i \leq \iota\aleph}\Big( \beta_{\mathsf M_i(S)}(X) - \epsilon^{\mathtt m} \big(\tfrac{\epsilon^2 \lambda s}{n}\big)^{\ell/2} \Big)
    \end{align*}
    and 
    \begin{equation*}
        \beta_S(X) = \prod_{ (i,j) \in E(S) } \frac{ X_{i,j}-\tfrac{\lambda s}{n} }{ \sqrt{ \tfrac{\lambda s}{n} (1-\tfrac{\lambda s}{n}) } }  \,.
    \end{equation*}
    Finally, we define
    \begin{equation}{\label{eq-def-tilde-f}}
        \widetilde{f}=\widetilde{f}(A,B)= \sum_{ \mathbf H \in \mathcal H } \frac{ s^{\aleph-1} \operatorname{Aut}(\mathsf{T}(\mathbf H)) (\epsilon^2 \lambda s)^{\ell\iota\aleph} }{ n^{\aleph+ \ell\iota\aleph} } \sum_{S_1,S_2 \vdash \mathbf H} \phi_{S_1,S_2}(A,B) \,.
    \end{equation}
\end{DEF}

\begin{lemma}{\label{lem-replace-by-non-backtracking-path}}
    We have that
    \begin{equation}{\label{eq-diff-tilde-f-minus-f}}
        \frac{ \widetilde{f}(A,B) - f(A,B) }{ \mathbb E_{\Pb}[f(A,B)] } \longrightarrow 0
    \end{equation}
    in probability under both $\Pb$ and $\Qb$.
\end{lemma}

The proof of Lemma~\ref{lem-replace-by-non-backtracking-path} is provided in Section~\ref{subsec:proof-lem-5.2}. We now show that $\widetilde{f}(A,B)$ can be approximated efficiently. To this end, we proceed as follows: Given a standardized adjacency matrix $M$ of a graph on $[n]$, we generate a random coloring $\mu:[n] \to [\aleph]$ that assigns a color to each vertex of $M$ from the color set $[\aleph]$ independently and uniformly at random. Given any $V \subset [n]$, let $\chi_{\mu}(V)$ be the indicator for the event that $\mu(V)$ is colorful, i.e., $\mu(x)\neq\mu(y)$ for any distinct $x,y \in V$. In particular, if $|V|=\aleph$, then $\chi_{\mu}(V)=1$ with probability
\begin{equation}{\label{eq-def-r}}
    r= \frac{ \aleph! }{ \aleph^{\aleph} } \,.
\end{equation}
For any graph $\mathbf H$ with $\aleph$ vertices, we define
\begin{equation}{\label{eq-def-X-mathbf-H}}
    X_{\mathbf H}(M,\mu) = \sum_{S \vdash \mathbf H} \chi_{\mu}(V(\mathsf T(S))) \prod_{(i,j)\in E(S)} M_{i,j} \,. 
\end{equation}
This construction ensures that $\frac{1}{r}X_{\mathbf H}(M,\mu)$ is an unbiased estimator of 
\begin{align*}
    \sum_{S \vdash \mathbf H} \prod_{(i,j)\in E(S)} M_{i,j} \,.
\end{align*}
For $t \geq 1$, we generate $2t$ random colorings $\{ \mu_i : 1 \leq i \leq t \}$ and $\{ \nu_j:1 \leq j \leq t \}$ that are independent copies of $\mu$. Then, we define (denote $\overline{A}_{i,j}=A_{i,j} - \tfrac{\lambda s}{n}$ and $\overline{B}_{i,j}=B_{i,j} - \tfrac{\lambda s}{n}$ respectively)
\begin{equation}{\label{eq-def-overline-f}}
    \overline{f}(A,B) = \frac{1}{r^2} \sum_{ \mathbf H \in \mathcal H } \frac{ s^{\aleph-1} \operatorname{Aut}(\mathsf{T}(\mathbf H)) (\epsilon^2 \lambda s)^{\ell\iota\aleph} }{ n^{\aleph+ \ell\iota\aleph} } \Big( \frac{1}{t} \sum_{i=1}^{t} X_{\mathbf H}(\overline{A},\mu_i) \Big) \Big( \frac{1}{t} \sum_{j=1}^{t} X_{\mathbf H}(\overline{B},\nu_i) \Big) \,.
\end{equation}
Similarly as \cite[Proposition~3]{MWXY21+}, we can show that when $t=\lceil \frac{1}{r} \rceil$ 
\begin{equation}{\label{eq-diff-overline-f-minus-tilde-f}}
    \frac{ \overline{f}(A,B) - \widetilde f(A,B) }{ \mathbb E_{\Pb}[f(A,B)] } \overset{L^2}{\longrightarrow} 0
\end{equation}
under both $\Pb$ and $\Qb$. Combining \eqref{eq-diff-tilde-f-minus-f} and \eqref{eq-diff-overline-f-minus-tilde-f}, we obtain the following: 
\begin{align*}
    \Qb( \overline{f}(A,B) \ge \tau ) + \Pb( \overline{f}(A,B) \le \tau )= o(1) \,,
\end{align*}
where the threshold $\tau$ is chosen as $\tau = C\mathbb{E}_{\Pb} [f_{\mathcal{T}}(A,B)]$ for any fixed constant $0<C<1$. It remains to show that $\overline{f}(A,B)$ can be computed efficiently, as in the next lemma. 

\begin{lemma}{\label{lem-color-coding}}
     There exists an algorithm with running time $O(n^C)$ (see Algorithm~\ref{alg:dynamic-programming}) to compute $X_{\mathbf H}(M,\mu)$ given any $\mathbf H \in \mathcal H$, a weighted graph $M$ on $[n]$, and a coloring $\mu:[n] \to [\aleph]$.    
\end{lemma}

Using Lemma~\ref{lem-color-coding}, we can calculate $\overline{f}_{A,B}$ using Algorithm~\ref{alg:cal-overline-f} below.

\begin{breakablealgorithm}{\label{alg:cal-overline-f}}
    \caption{Computation of $\overline{f}_{A,B}$}
    \begin{algorithmic}[1]
    \STATE {\bf Input:} Adjacency matrices $A$ and $B$, correlation parameter $s$, divergence parameter $\epsilon$, edge-density parameter $\lambda$.
    \STATE Choose $\aleph,\ell$ and $\iota$ according to \eqref{eq-choice-aleph} and \eqref{eq-choice-iota}. Let $t=\tfrac{\aleph^{\aleph}}{\aleph!}$.
    \STATE Apply the constant-time free tree generation algorithm in \cite{Dinneen15, WROM86} to list all non-isomorphic unlabeled trees with $\aleph$ edges and return $\mathcal T'$ as the resulting set.
    \STATE For each $\mathbf T \in \mathcal T'$, add $\mathbf{T}\in\mathcal{T}$ iff there exists $\mathbf o \in V(\mathbf T)$ such that $(\mathbf T,\mathbf o)$ satisfies Definition~\ref{def-tilde-T-K}.
    \STATE For each $\mathbf T \in \mathcal T$, compute $\operatorname{Aut}(\mathbf H)$ using the algorithm in \cite{CB81}.
    \STATE For each $\mathbf T \in \mathcal T$, enumerate all $V \subset V(\mathbf T)$ and select $\mathcal S(\mathbf T)$ according to Theorem~\ref{thm-desired-vertex-sets}.
    \STATE Construct $\mathcal H=\mathcal H(\mathcal T,\mathcal S(\mathcal T))$ according to Definition~\ref{def-decorated-trees}. 
    \STATE Generate i.i.d.\ random colorings $\{ \mu_i : 1 \leq i \leq t \}$ and $\{ \nu_j:1 \leq j \leq t \}$ that map $[n]$ to $[\aleph]$ uniformly at random. 
    \FOR{each $1 \leq i \leq t$}
    \STATE For each $\mathbf H \in \mathcal H$, compute $X_{\mathbf H}(\overline{A},\mu_i)$ and $X_{\mathbf H}(\overline{B},\nu_i)$ via Algorithm~\ref{alg:dynamic-programming}. 
    \ENDFOR
    \STATE Compute $\overline{f}(A,B)$ according to \eqref{eq-def-overline-f}.
    \STATE {\bf Output:} $\overline{f}(A,B)$.
    \end{algorithmic}
\end{breakablealgorithm}

It remains to prove Lemma~\ref{lem-color-coding}. We now describe how to calculate $X_{\mathbf H}(M,\mu)$ via dynamical programming. For $l\geq 1$, denote by $\operatorname{NB}_l(x,y)$ the collection of non-backtracking paths from $x$ to $y$ with $l$ edges. Denote $\mathcal H'$ as the set of all $\mathbf H' =(\mathbf T', \mathsf P'(\mathbf T'))$ such that there exists $(\mathbf T, \mathsf P(\mathbf T)) \in \mathcal H$ with $\mathbf T' \subset \mathbf T$ and $\mathsf P'(\mathbf T')$ is the restriction of $\mathsf P(\mathbf T)$ in $\mathbf T'$. 
\begin{breakablealgorithm}{\label{alg:dynamic-programming}}
    \caption{Computation of $X_{\mathbf H}(M,\mu)$}
    \begin{algorithmic}[1]
    \STATE \textbf{Input}: A weighted host graph $M$ on $[n]$ with its vertices colored by $\mu$, and an element $\mathbf H\in \mathcal H'$. 
    \STATE For each $x,y \in [n]$ and $1 \leq d \leq \aleph$, use \cite[Section~3.2]{MNS18} to calculate 
    \begin{align*}
        \mathcal L(x,y) = \sum_{ \gamma \in \mathrm{NB}_{\ell}(x,y)  } \Big( \prod_{e \in \gamma }M_e \Big) \,.
    \end{align*}
    \STATE Choose $\mathbf o = \mathfrak R(\mathsf T(\mathbf H))$. If $\mathbf H$ is a single point we simply let 
    \begin{align*}
        Y(x,\mathbf H,\{ c \},\mu) = \mathbf 1_{\{ \mu_x = c \}} \mbox{ for each } x \in [n] \,.
    \end{align*}
    \STATE For every $x \in [n]$, $2 \leq y \leq \aleph$ and every subset $C \subset [\aleph]$ of colors with $|C| = y$, compute $Y(x,\mathbf A_{y} , C,\mu)$ (Recall the definition of $\mathbf A_y$ in \eqref{eq-AN}) recursively by (note that the case that $y=1$ is already calculated)
    \begin{align}
        Y(x,\mathbf A_{y} , C,\mu) = \sum_{z \in [n]}\sum_{c \in C} Y(z, \mathbf A_{y-1} ,C\setminus \{c\} ,\mu)M_{x,z}\,.
    \end{align}
    \IF{$\mathbf o \not\in \mathsf{Vert}(\mathsf P(\mathbf H))$ }
    \STATE There is an edge $\mathbf e =(\mathbf o , \mathbf o')$ in $\mathbf H$. By Theorem~\ref{thm-desired-vertex-sets}, \eqref{def-chained-noise-pair-in-T_iota} and \eqref{eq-def-mathsf-Ch-a,b,*}, for each $(\mathbf v_1 , \mathbf v_2) \in \mathsf{P}(\mathbf H)$ we have $\mathbf o \not\in \mathfrak p_{\mathbf H}(\mathbf v_1 , \mathbf v_2)$, which yields that by removing $\mathbf e$ in $\mathbf H$, $\mathbf H$ can be partitioned into $2$ rooted subgraphs $\{ \mathbf H_1 , \mathbf H_2\}$ such that $\mathfrak R(\mathbf H_1) = \mathbf o$, $\mathfrak R(\mathbf H_2) = \mathbf o'$, and there does not exist $\mathbf u \in V(\mathbf H_1) , \mathbf v\in V(\mathbf H_2)$ such that $(\mathbf u, \mathbf v) \in \mathsf{P}(\mathbf H)$. Moreover, for any $(\mathbf u', \mathbf v') \in \mathsf{P}(\mathbf H_1)$ and $w \in V(\mathfrak p_{\mathbf H_1}(\mathbf u , \mathbf v))$ we have $\mathsf{Branch}_{\mathbf H_1}(w) = 1$ and $\mathsf{Deg}_{\mathbf H_1}(w) \geq 2$, which gives that $w \neq \mathfrak R(\mathbf H_1)$. Therefore, it holds for $\mathbf H_1$ (and similarly for $\mathbf H_2$) that either $\mathfrak R(\mathbf H_1) \in \mathsf P(\mathbf H_1)$, or for every $(\mathbf v_1' , \mathbf v_2') \in \mathsf{P}(\mathbf H_1)$ we have $\mathfrak R(\mathbf H_1) \not\in \mathfrak p_{\mathbf H_1}(\mathbf v_1' , \mathbf v_2')$.
    \STATE For every $x \in [n]$ and every subset $C \subset [\aleph]$ of colors with $|C|=|V(\mathsf T(\mathbf H))|$, compute recursively for each $x \in [n]$ 
    \begin{align}
       Y(x,\mathbf H,C,\mu) = \sum_{ y\in [n] , y \neq x} \sum_{C = C_1 \sqcup C_2}  Y(x,\mathbf H_1,C_1,\mu)Y(y,\mathbf H_2,C_2,\mu) M_{x,y} \,. \label{eq-recursive-Y-case-1}
    \end{align}
    \ELSE
    \STATE There exists $\mathbf{o}'$ such that $(\mathbf o,\mathbf o') \in \mathsf P(\mathbf H)$. In this case, by Theorem~\ref{thm-desired-vertex-sets} we can suppose that $\mathfrak p_{\mathbf H}(\mathbf o,\mathbf o') = (\mathbf w_0 ,\mathbf w_1 , \ldots ,\mathbf w_\mathtt m) \in \mathfrak{Ch}^{(0,0;*)}_\mathtt m(\mathsf{T}(\mathbf H))$ with $\mathbf w_0 = \mathbf o ,\mathbf w_\mathtt m = \mathbf o'$. By removing the edges of $\mathfrak p_{\mathbf H}(\mathbf o, \mathbf o')$ in $\mathbf H$, $\mathbf H$ can be partitioned into $2$ rooted subgraphs $\{ \mathbf H_1 , \mathbf H_2\}$ such that $\mathfrak R(\mathbf H_1) = \mathbf o$, $\mathfrak R(\mathbf H_2) = \mathbf o'$, and similar to the analysis in Step~6 given any $i \in \{1,2\}$ we have either $\mathfrak R(\mathbf H_i) \in \mathsf P(\mathbf H_i)$, or for every $(\mathbf v' , \mathbf v'') \in \mathsf{P}(\mathbf H_i)$ we have $\mathfrak R(\mathbf H_i) \not\in \mathfrak p_{\mathbf H_i}(\mathbf v' , \mathbf v'')$.
    \STATE Denote $\mathbf H_0 = (\mathbf w_1 , \ldots , \mathbf w_{\mathtt m-1})$. For every $x \in [n]$ and every subset $C \subset [\aleph]$ of colors with $|C|=|V(\mathsf T(\mathbf H))|$, compute recursively for each $x \in [n]$ 
    \begin{align}
       Y(x,\mathbf H,C,\mu) = \sum_{ \substack{\{y_0 , \ldots , y_\mathtt m\}\subset [n] \\ y_0 = x , y_\mathtt m = y}} &\sum_{C = \{ c_1 , \ldots , c_{\mathtt m-1} \} \sqcup C_1 \sqcup C_2} \Big( Y(x,\mathbf H_1,C_1,\mu)Y(y,\mathbf H_2,C_2,\mu) \cdot \nonumber \\
       &\mathcal L(x,y) \prod_{i=1}^{\mathtt m} (M_{y_{i-1},y_{i}} (\mathbf 1_{\{ \mu (y_i) = c_i \,, i\neq \mathtt m\}}+\mathbf{1}_{i=\mathtt m})) \Big)\,. \label{eq-recursive-Y-case-2}
    \end{align}
    \ENDIF
    \STATE \textbf{Output}: $\tfrac{1}{\overline{\mathbf{Aut}}(\mathbf H_\mathbf o)}\sum_{x \in[n]} Y(x,\mathbf H_{\mathbf o},C,\mu)$, where $\overline{\mathbf{Aut}}(\mathbf H_\mathbf o)$ denotes the number of bijections $\pi : V(\mathbf H_\mathbf o) \longrightarrow V(\mathbf H_\mathbf o)$ such that $\pi (\mathsf P(\mathbf H_\mathbf o)) = \mathsf P(\mathbf H_\mathbf o)$ and $\pi$ is an automorphism of $\mathsf T(\mathbf H_\mathbf o)$.
    \end{algorithmic}
\end{breakablealgorithm}

\begin{proof}[Proof of Lemma~\ref{lem-color-coding}]
    Calculating $\mathcal L(x,y)$ for all $x,y$ takes time $O(n^2)\cdot O(n^{3+o(1)})=O(n^{5+o(1)})$.
    The cardinality of $\mathcal H'$ is bounded by $10^{\aleph} \cdot \tbinom{\aleph}{2\iota\aleph} = n^{o(1)}$, the total number of all subsets $C \subset [\aleph]$ is bounded by $2^{\aleph}=n^{o(1)}$, and the total number of all $(C_1,C_2)$ is bounded by $2^{\aleph}=n^{o(1)}$. Also, computing $\overline{\mathbf{Aut}}(\mathbf H_\mathbf o)$ for each $\mathbf H_\mathbf o$ takes time $O(\aleph^{2\aleph})$. Thus, according to \eqref{eq-recursive-Y-case-1} and \eqref{eq-recursive-Y-case-2}, the total time complexity of computing $Y(x,\mathbf H,C,\mu)$ for all $x \in [n]$ and all $C \subset [\aleph]$ is bounded by
    \begin{align*}
        O(n^{o(1)}) \cdot O(n^{4\mathtt m}) = O(n^{4\mathtt m+o(1)}) \,.
    \end{align*}
    Thus, the total time complexity of Algorithm~\ref{alg:dynamic-programming} is bounded by
    \begin{align*}
        O(n^{5+o(1)}) + O(n) \cdot O(2^{\aleph}) \cdot O(n^{4\mathtt m+o(1) }) = O\big( n^{4\mathtt m+1+o(1)} \big) \,.
    \end{align*}
    The proof that Algorithm~\ref{alg:dynamic-programming} outputs $X_{\mathbf{H}}(M,\mu)$ is almost identical to the proof of \cite[Lemma~2]{MWXY21+}; the only difference is that we need to invoke the value of $\mathcal L(x,y)$ when we compute $Y(x,\mathbf H_{\mathbf o},C,\mu)$ recursively. However, this can be done efficiently as we have already stored the value of all $\mathcal L(x,y)$ at Step~2, and we omit further details here for simplicity. 
\end{proof}

\appendix

\section{Preliminary results on graphs}{\label{sec:prelim-graphs}}

\begin{lemma}{\label{lem-remove-sim-relation}}
    For any $N \geq \log(\iota^{-1})$ and $\mathbf T \in \cup_{m \leq N} \mathcal R_m$, we have 
    \begin{align*}
        \#\big\{ \mathbf T' \in \cup_{m \leq N} \mathcal R_m: \mathbf T' \sim \mathbf T \big\} \leq e^{ 6N \log\log(\iota^{-1}) /\log(\iota^{-1})} \,.
    \end{align*}
\end{lemma}
\begin{proof}
    Note that for any fixed $\mathbf T \in \cup_{m \leq N} \mathcal R_N$, in order to have $\mathbf T' \sim \mathbf T$, there must exist 
    \begin{align*}
        & k \leq \log^{-2}(\iota^{-1}) \cdot |V(\mathbf T)|, k'\leq \log^{-2}(\iota^{-1}) \cdot |V(\mathbf T')| \,, \\
        & \{\mathbf u_1,\ldots,\mathbf u_k\} \subset V(\mathbf T), \{\mathbf u_1',\ldots,\mathbf u_{k'}'\}\subset V(\mathbf T') \,, \\
        & 3\leq \mathsf{Deg}_{\mathbf T}(\mathbf u_1), \ldots ,\mathsf{Deg}_{\mathbf T}(\mathbf u_k), \mathsf{Deg}_{\mathbf T'}(\mathbf u'_1), \ldots ,  \mathsf{Deg}_{\mathbf T'}(\mathbf u'_{k'}) \leq 4 \,, \\
        & 0 \leq x_1,\ldots,x_k, x_1',\ldots,x_{k'}' \leq \log^{2}(\iota^{-1})
    \end{align*}
    such that 
    \begin{align*}
        \mathbf T \oplus \mathcal L_{\mathbf u_1}^{(x_1)} \oplus \ldots \oplus \mathcal L_{\mathbf u_k}^{(x_k)} \cong \mathbf T' \oplus \mathcal L_{\mathbf u_1'}^{(x_1')} \oplus \ldots \oplus \mathcal L_{\mathbf u_{k'}'}^{(x_{k'}')} \,.
    \end{align*}
    Observe that the total number of ways to enumerate $\{ \mathbf u_k, \mathbf u'_{k'} \}$ is bounded by
    \begin{align*}
        \binom{N}{N/\log(\iota^{-1})} \binom{N+\log^2(\iota^{-1})N}{N/\log(\iota^{-1})} \leq e^{ 3N\log\log(\iota^{-1})/ \log(\iota^{-1})}\,.
    \end{align*}
    Additionally, the number of possible choices for $\{ x_k,x'_{k'} \}$ is bounded by $\big( \log^2(\iota^{-1}) \big)^{4\iota N}$. Given $\{ \mathbf u_k, \mathbf u'_{k'} \}$ and $\{ x_k,x'_{k'} \}$, the graph $\mathbf T \oplus \mathcal L_{\mathbf u_1}^{(x_1)} \oplus \ldots \oplus \mathcal L_{\mathbf u_k}^{(x_k)}$ is also determined. Since the degree of $\mathbf u_i'$ (for $1 \leq i \leq k'$) in $\mathbf T \oplus \mathcal L_{\mathbf u_1}^{(x_1)} \oplus \ldots \oplus \mathcal L_{\mathbf u_k}^{(x_k)}$ is bounded by $4+1=5$, the enumeration of $\mathcal L_{\mathbf u_1'}^{(x_1')} , \ldots , \mathcal L_{\mathbf u_{k'}'}^{(x_{k'}')}$ is bounded by $5^{3N/\log(\iota^{-1})}$. Combining these bounds, we have
    \begin{align*}
        \#\big\{ \mathbf T' \in \cup_{m \leq N} \mathcal R_m: \mathbf T \sim \mathbf T' \big\} &\leq e^{ 3N\log\log(\iota^{-1})/\log(\iota^{-1})} \cdot \big( \log(\iota^{-1}) \big)^{8\iota N} \cdot 5^{3N/\log(\iota^{-1})} \\
        &\leq e^{ 6N\log\log(\iota^{-1}) /\log(\iota^{-1}) } \,.\qedhere
    \end{align*}
\end{proof}

Recall Definition~\ref{def-tilde-T-K} and Theorem~\ref{thm-desired-vertex-sets}. We now describe some useful properties for our choice of trees and vertex sets. 

\begin{lemma}{\label{lem-useful-property-trees-and-sets}}
    Let $\mathbf H \in \mathcal H$ and let $S \Vdash \mathbf H$. We have the following properties: 
    \begin{enumerate}
        \item[(1)] For all $\{ x_{\mathbf u} \geq 0:\mathbf u \in \mathsf{Vert}(\mathsf P(\mathbf H)) \}$ and all injection $\varphi: \mathsf T(\mathbf H) \to \mathsf T(\mathbf H) \oplus ( \oplus_{\mathbf u \in \mathsf{Vert}(\mathsf P(\mathbf H))} \mathcal L_{\mathbf u}^{(x_{\mathbf u})})$, we have that $\varphi$ maps $\mathfrak R(\mathsf T(\mathbf H))$ to $\mathfrak R(\mathsf T(\mathbf H))$.
        \item[(2)] For all $\{ 0 \leq x_{\mathbf u} \leq \log^2(\iota^{-1}):\mathbf u \in \mathsf{Vert}(\mathsf P(\mathbf H)) \}$, we have (recall $\mathsf{Fix}$ in the notation section)
        \begin{align*}
        \mathsf{Vert}(\mathsf P(\mathbf H)) \subset V(\mathsf T(\mathbf H)_\iota) \subset \mathsf{Fix}(\mathsf T(\mathbf H) \oplus ( \oplus_{\mathbf u \in \mathsf{Vert}(\mathsf P(\mathbf H))} \mathcal L_{\mathbf u}^{(x_{\mathbf u})} )). 
        \end{align*}
        \item[(3)] Recall in Section~\ref{subsec:notations} we denote $\widetilde{S}$ as the simple graph corresponding to $S$. We have $\mathcal L(\widetilde S) = \mathcal L(S) \subset \mathcal L(\mathsf T(S))$.
        \item[(4)] For all $\{ x_{\mathbf u} \geq 0:\mathbf u \in \mathsf{Vert}(\mathsf P(\mathbf H)) \}$, we have
        \begin{align*}
            \frac{ \operatorname{Aut}( \mathsf T(\mathbf H) \oplus ( \oplus_{\mathbf u \in \mathsf{Vert}(\mathsf P(\mathbf H))} \mathcal L_{\mathbf u}^{(x_{\mathbf u})} ) ) }{ \operatorname{Aut}(\mathsf T(\mathbf H)) } \leq e^{4\iota \aleph} \,.
        \end{align*}
        \item[(5)] The number of injections from $\mathsf T(\mathbf H)$ to $\mathsf T(\mathbf H) \oplus (\oplus_{\mathbf u \in \mathsf{Vert}(\mathsf P(\mathbf H))} \mathcal L_{\mathbf u}^{(x_{\mathbf u})})$ equals
        \begin{align*}
            \frac{ \operatorname{Aut}( \mathsf T(\mathbf H) \oplus ( \oplus_{\mathbf u \in \mathsf{Vert}(\mathsf P(\mathbf H))} \mathcal L_{\mathbf u}^{(x_{\mathbf u})} ) ) }{ \operatorname{Aut}(\mathsf T(\mathbf H)) } \,.
        \end{align*}
    \end{enumerate}
\end{lemma}
\begin{proof}
    Denote $\mathbf T=\mathsf T(\mathbf H)$ and $\mathbf T^{\oplus} = \mathbf T \oplus ( \oplus_{\mathbf u \in \mathsf{Vert}(\mathsf P(\mathbf H))} \mathcal L_{\mathbf u}^{(x_{\mathbf u})} )$ for simplicity in the following proof. We first show Item~(1). Since $\mathbf T$ does not have any arm-path with length at least $\log^2(\iota^{-1})$, then $\sigma$ must map $\mathbf T$ to $\mathbf T \oplus ( \oplus_{\mathbf v \in \mathsf P(\mathbf H)} \mathcal L_{\mathbf v}^{ (x_{\mathbf v} \wedge\log^2(\iota^{-1})) } )$. Thus, without loss of generality we may assume that $x_{\mathbf v}\leq \log^2(\iota^{-1})$ for all $\mathbf v$. Recall the notation $\mathbf{T}_1$ and $\mathbf{T}_2$ in Item~(5) of Definition~\ref{def-tilde-T-K}. Note that in this case if $\sigma$ maps $\mathfrak R(\mathbf T)$ to another vertex, then either $\mathbf T_1$ is mapped into a descendant tree of $\mathbf T_2 \oplus (\oplus_{\mathbf v \in \mathsf P (\mathbf H) \cap V(\mathbf T_2)} \mathcal L_{\mathbf v}^{(x_{\mathbf v})})$, or $\mathbf T_2$ is mapped into a descendant tree of $\mathbf T_1 \oplus (\oplus_{\mathbf v \in \mathsf P (\mathbf H) \cap V(\mathbf T_1)} \mathcal L_{\mathbf v}^{(x_{\mathbf v})})$. This implies that $\mathbf T_1$ has $\sim$ relation with a subtree of $\mathbf T_2$ or $\mathbf T_2$ has $\sim$ relation with a subtree of $\mathbf T_1$, contradicting to Item~(5) of Definition~\ref{def-tilde-T-K} and thus leading to Item~(1). 
    
    We now show Item~(2). Note that $\mathsf{Vert}(\mathsf P(\mathbf H)) \subset V( \mathbf T_\iota )$ holds by Item~(1) of Theorem~\ref{thm-desired-vertex-sets}. It remains to prove $V(\mathbf T_\iota) \subset \mathsf{Fix}(\mathbf T^{\oplus})$. Suppose that there exists a vertex $v \in V(\mathbf T_\iota) \setminus \mathsf{Fix}(\mathbf T^{\oplus} )$. Then there is an automorphism $\pi$ of $\mathbf T^{\oplus}$ such that $\pi(v) \neq v$. Now denote by $\mathcal L^*_{v}$ the unique self-avoiding path in $\mathbf T$ connecting $v$ and $\mathfrak R(\mathbf T)$, and similarly define $\mathcal L^*_{\pi (v)}$. Then $\mathcal L^*_{v} , \mathcal L^*_{\pi(v)} \subset \mathbf T$. By Item~(1), $\pi (\mathfrak R (\mathbf T)) = \mathfrak R (\mathbf T)$ and therefore we have $\pi(\mathcal L^*_{v}) = \mathcal L^*_{\pi (v)}$ and $\mathcal L^*_{v} \cap \mathcal L^*_{\pi (v)} \neq \emptyset$. Denote $v^* = \operatorname{arg} \max_{u \in \mathcal L^*_{v} \cap \mathcal L^*_{\pi (v)}} \mathsf{Dep}_{\mathbf T^{\oplus}}(u)$. By $\mathsf{Dep}_{\mathbf T^{\oplus}}(v^*) = \mathsf{Dep}_{\mathbf T^{\oplus}}(\pi(v^*))$ and $\mathsf{Dep}_{\mathbf T^{\oplus}}(v) = \mathsf{Dep}_{\mathbf T^{\oplus}}(\pi(v))$, we have $\pi (v^*) = v^*$ and therefore $v^* \neq v$. Denote by $v^{**}$ the neighbor of $v^*$ in $\mathcal L^*_{v}$ such that $\mathsf{Dep}_{\mathbf T^{\oplus}}(v^{**})>\mathsf{Dep}_{\mathbf T^{\oplus}}(v^*)$. Then $\mathsf{Des}_{\mathbf T^{\oplus}}(v^{**})$ and $\mathsf{Des}_{\mathbf T^{\oplus}}(\pi(v^{**}))$ are two distinct connected components of $(\mathbf T^{\oplus})_{ \setminus \{ v^* \}} = (\mathbf T^{\oplus})_{ \setminus \{ \pi(v^*) \}}$ (recall the definition of $H_{\setminus A}$ for any unlabeled graph $H$ and any vertex set $A \subset V(H)$ in Section~\ref{subsec:notations}, where we have defined vertex induced subgraphs). Therefore $\pi (\mathsf{Des}_{\mathbf T^{\oplus}}(v^{**})) = \mathsf{Des}_{\mathbf T^{\oplus}}(\pi(v^{**}))$, and we have $|V(\mathsf{Des}_{\mathbf T}(\pi(v^{**})))| = |V(\mathsf{Des}_{\mathbf T}(v^{**}))| \geq |V(\mathsf{Des}_{\mathbf T}(v))| \geq \log^2(\iota^{-1})$. Since the same parent $v^*$ is shared by $\mathsf{Des}_{\mathbf T}(v^{**}) \neq \mathsf{Des}_{\mathbf T}(\pi(v^{**}))$, by Item~(4) of Definition~\ref{def-tilde-T-K} we have $\mathsf{Des}_{\mathbf T}(v^{**}) \not \sim \mathsf{Des}_{\mathbf T}(\pi(v^{**}))$. Meanwhile, by $\mathsf{Des}_{\mathbf T^{\oplus}}(v^{**}) \cong \mathsf{Des}_{\mathbf T^{\oplus}}(\pi(v^{**}))$ we have (denote $V_{\mathsf P}(v^{**}) = \mathsf{Vert}(\mathsf P(\mathbf H)) \cap V(\mathsf{Des}_{\mathbf T}(v^{**}))$ and similarly for $V_{\mathsf P}(\pi(v^{**}))$)
    \begin{align*}
        \mathsf{Des}_{\mathbf T}(v^{**}) \oplus ( \oplus_{\mathbf u \in V_{\mathsf P}(v^{**})} \mathcal L_{\mathbf u}^{(x_{\mathbf u})} ) \cong \mathsf{Des}_{\mathbf T}(\pi(v^{**}))  \oplus ( \oplus_{\mathbf v \in V_{\mathsf P}(\pi(v^{**}))} \mathcal L_{\mathbf v}^{(x_{\mathbf v})} ) \,.
    \end{align*}
    Combining this with Item~(2) of Theorem~\ref{thm-desired-vertex-sets}, we have $\mathsf{Des}_{\mathbf T}(v^{**}) \sim \mathsf{Des}_{\mathbf T}(\pi(v^{**}))$, which contradicts to Item~(4) of Definition~\ref{def-tilde-T-K} . This leads to Item~(2).
    
    We now show Item~(3). Note that for all $u \in \mathcal L(\widetilde S)$, since each $\mathsf L_i(S)$ is self-avoiding, we see that all vertices in $V(\mathsf L_i(S)) \setminus \operatorname{EndP}(\mathsf L_i(S))$ have degree at least $2$ in $\widetilde{S}$ and thus $u \not\in V(\mathsf L_i(S)) \setminus \operatorname{EndP}(\mathsf L_i(S))$. Similarly we have $u \not\in V(\mathsf T(S)) \setminus \mathcal L(\mathsf T(S))$. This shows that $\mathcal L(\widetilde{S}) \subset \mathcal L(\mathsf T(S))$. We now show that $\mathcal L(\widetilde S)=\mathcal L(S)$. Clearly we have $\mathcal L(S) \subset \mathcal L(\widetilde S)$. In addition, for all $u \in \mathcal L(\widetilde S) \subset \mathcal L(\mathsf T(S))$, denote $(u,v)$ the edge in $\mathsf T(S)$. For all $1 \leq i \leq \iota\aleph$, we must have $u \not \in V(\mathsf L_i(S)) \setminus \operatorname{EndP}(\mathsf L_i(S))$ since otherwise the degree of $u$ in $\widetilde{S}$ is at least $2$. This shows that $(u,w)$ has multiplicity $1$ in the multigraph $S$ and thus $u \in \mathcal L(S)$. This concludes $\mathcal L(\widetilde S)=\mathcal L(S)$.
 
    For Item~(4), denote $\mathbf T=\mathsf T(\mathbf H)$ and $\mathsf{Vert}(\mathsf P(\mathbf H))=\{ \mathbf w_1, \ldots, \mathbf w_{2\iota\aleph} \}$. If there exists some $x_i>\log^2(\iota^{-1})$, we may assume that $x_i > \log^2(\iota^{-1})$ if and only if $i\leq \Lambda$ for some $\Lambda \geq 1$. By Item~(2) of this lemma, we see that for all $\pi \in \operatorname{Aut}( \mathbf T \oplus (\oplus_{i \leq 2\iota\aleph} \mathcal L_{\mathbf w_i}^{(x_i)} ))$, $\pi$ must map $\{ \mathbf w_1,\ldots,\mathbf w_{\Lambda} \}$ to itself. Thus,
    \begin{align*}
        \operatorname{Aut}( \mathbf T \oplus (\oplus_{i \leq 2\iota\aleph} \mathcal L_{\mathbf w_i}^{(x_i)}) ) = \operatorname{Aut}( \mathbf T \oplus (\oplus_{i>\Lambda} \mathcal L_{\mathbf w_i}^{(x_i)}) ) \,.
    \end{align*}
    Therefore, without loss of generality, we may assume that $x_1,\ldots,x_{2\iota\aleph} \leq \log^2(\iota^{-1})$. Denote $\mathbf T(j) = \mathbf T \oplus (\oplus_{i < j} \mathcal L_{\mathbf w_i}^{(x_i)})$. Then we have 
    \begin{align*}
        \frac{ \operatorname{Aut}( \mathbf T \oplus (\oplus_{i \leq 2\iota\aleph} \mathcal L_{\mathbf w_i}^{(x_i)}) ) }{ \operatorname{Aut}(\mathbf T) } = \prod_{j=1}^{2\iota\aleph} \frac{ \operatorname{Aut}(\mathbf T(j) \oplus \mathcal L_{\mathbf w_j}^{(x_j)}) }{ \operatorname{Aut}(\mathbf T(j)) } \,.
    \end{align*}
    By Item~(2), we have $\mathbf w_j \in \mathsf{Fix}(\mathbf T(j) \oplus \mathcal L_{\mathbf w_j}^{(x_j)} )$. Using the fact that $\mathsf{Deg}_{\mathbf T}(\mathbf w_j) \leq \mathsf{Branch}_{\mathbf T}(\mathbf w_j) + 1 \leq 4$ (recall \eqref{def-chained-noise-pair-in-T_iota} and \eqref{def-chained-signal-pair-in-T_iota}), we see that
    \begin{align*}
        \frac{ \operatorname{Aut}(\mathbf T(j) \oplus \mathcal L_{\mathbf w_j}^{(x_j)}) }{ \operatorname{Aut}(\mathbf T(j)) } \leq 4 \,,
    \end{align*}
    which gives
    \begin{equation*}
        \frac{ \operatorname{Aut}( \mathbf T \oplus (\oplus_{i \leq 2\iota\aleph} \mathcal L_{\mathbf w_i}^{(x_i)}) ) }{ \operatorname{Aut}(\mathbf T) } \leq e^{4\iota\aleph} \,,
    \end{equation*}
    leading to the desired upper bound. 
    
    For Item~(5), denote $\mathbf T=\mathsf T(\mathbf H)$. Note that if there exists $x_{\mathbf u}>\log^2(\iota^{-1})$, since $\mathbf T$ does not have an arm-path with length at least $\log^2(\iota^{-1})$, then the number of injections from $\mathbf T$ to $\mathbf T \oplus (\oplus_{\mathbf v} \mathcal L_{\mathbf v}^{(x_{\mathbf v})})$ equals the number of injections from $\mathbf T$ to $\mathbf T \oplus (\oplus_{\mathbf v \neq \mathbf u} \mathcal L_{\mathbf v}^{(x_{\mathbf v})} )$. Thus, (by applying the preceding observation inductively) we may assume that $x_{\mathbf u} \leq \log^2(\iota^{-1})$ without loss of generality. Note that using Item~(4) in Definition~\ref{def-tilde-T-K}, the injections from $\mathbf T$ to $\mathbf T \oplus (\oplus_{\mathbf u} \mathcal L_{\mathbf u})$ must map $\mathbf T_{\iota}$ to itself. Note that we can write
    \begin{align*}
        \mathbf T = \mathbf T_{\iota} \oplus \big( \oplus_{\mathbf w \in V(\mathbf T_{\iota})} \mathbf T_{\mathbf w} \big) \,,
    \end{align*}
    where $\mathbf T_{\mathbf w}=\mathsf{Des}_{\mathbf T}(\mathbf w)$ are descendant trees rooted at $\mathbf w$ with $|V(\mathbf T_{\mathbf w})| \geq (\log(\iota^{-1}))^2$. Denoting $\mathbf T^{\oplus}_{\mathbf w}=\mathbf T_{\mathbf w}$ if $\mathbf w \not\in \mathsf{Vert}(\mathsf P(\mathbf T))$ and $\mathbf T^{\oplus}_{\mathbf w}=\mathbf T_{\mathbf w} \oplus \mathcal L_{\mathbf w}^{(x_{\mathbf w})}$ if $\mathbf w \in \mathsf{Vert}(\mathsf P(\mathbf T))$, it can be easily checked that the number of injections from $\mathbf T_{\mathbf w}$ to $\mathbf T^{\oplus}_{\mathbf w}$ equals to $\frac{\operatorname{Aut}(\mathbf T^{\oplus}_{\mathbf w})}{\operatorname{Aut}(\mathbf T_{\mathbf w})}$. Thus, the number of injections from $\mathbf T$ to $\mathbf T^{\oplus}$ equals to
    \begin{align*}
        \prod_{\mathbf w \in V(\mathbf T_{\iota})} \frac{\operatorname{Aut}(\mathbf T^{\oplus}_{\mathbf w})}{\operatorname{Aut}(\mathbf T_{\mathbf w})} = \frac{\operatorname{Aut}(\mathbf T^{\oplus})}{\operatorname{Aut}(\mathbf T)} \,,
    \end{align*}
    which completes the proof. 
\end{proof}

\begin{cor}{\label{lem-useful-Aut-bound}}
    For all $\mathbf H \in \mathcal H$ and constant $0 \leq \kappa\leq \tfrac{1}{\epsilon^2 \lambda s}$ (note that $\epsilon^2 \lambda s >1+\Delta$ from \eqref{eq-assumption-s,lambda}) we have
    \begin{align}
        \sum_{x_1,\ldots,x_{2\iota\aleph} \geq 0} \kappa^{x_1+\ldots+x_{2\iota\aleph}} \cdot \frac{ \operatorname{Aut}( \mathsf T(\mathbf H) \oplus ( \oplus_{\mathbf u \in \mathsf{Vert}(\mathsf P(\mathbf H))} \mathcal L_{\mathbf u}^{(x_{\mathbf u})} ) )^2 }{ \operatorname{Aut}(\mathsf T(\mathbf H))^2 } \leq \exp\big( (8+4\Delta)\iota\aleph \big) \,. \label{eq-useful-Aut-bound}
    \end{align}
\end{cor}
\begin{proof}
    Using Item~(4) of Lemma~\ref{lem-useful-property-trees-and-sets}, the left-hand side of \eqref{eq-useful-Aut-bound} is bounded by 
    \begin{align*}
        e^{8\iota\aleph} \sum_{x_1,\ldots,x_{2\iota\aleph} \geq 0} \kappa^{x_1+\ldots+x_{2\iota\aleph}} \overset{\eqref{eq-choice-iota}}{\leq} \exp\big( (8+4\Delta)\iota\aleph \big) \,,
    \end{align*}
    as desired (we use $(1-\Delta)^{-1}<e^{2\Delta}$ for $\Delta<0.1$ in the last inequality).
\end{proof}

The following lemmas provide several properties of the subgraphs of a given graph $S$.
\begin{lemma}[\cite{CDGL24+}, Lemma~A.3]{\label{lem-decomposition-H-Subset-S}}
    For $H \subset S$, we can decompose $E(S)\setminus E(H)$ into $\mathtt m$ cycles ${C}_{\mathtt 1}, \ldots, {C}_{\mathtt m}$ and $\mathtt t$ paths ${P}_{\mathtt 1}, \ldots, {P}_{\mathtt t}$ for some $\mathtt m, \mathtt t\geq 0$ such that the following hold:
    \begin{enumerate}
        \item[(1)] ${C}_{\mathtt 1}, \ldots, {C}_{\mathtt m}$ are vertex-disjoint (i.e., $V(C_{\mathtt i}) \cap V(C_{\mathtt j})= \emptyset$ for all $\mathtt i \neq \mathtt j$) and $V(C_{\mathtt i}) \cap V(H)=\emptyset$ for all $1\leq\mathtt i\leq \mathtt m$.
        \item[(2)] $\operatorname{EndP}({P}_{\mathtt j}) \subset V(H) \cup (\cup_{\mathtt i=1}^{\mathtt m} V(C_{\mathtt i})) \cup (\cup_{\mathtt k=1}^{\mathtt j-1} V(P_{\mathtt k})) \cup \mathcal L(S)$ for $1 \leq \mathtt j \leq \mathtt t$.
        \item[(3)] $\big( V(P_{\mathtt j}) \setminus \operatorname{EndP}(P_{\mathtt j}) \big) \cap \big( V(H) \cup (\cup_{\mathtt i=1}^{\mathtt m} V(C_{\mathtt i})) \cup (\cup_{\mathtt k=1}^{\mathtt j-1} V(P_{\mathtt k}) ) \cup \mathcal L (S) \big) = \emptyset$ for $\mathtt 1 \leq \mathtt j \leq \mathtt t$.
        \item[(4)] $\mathtt t = |\mathcal L(S) \setminus V(H)|+\tau(S)-\tau(H)$.
    \end{enumerate}
\end{lemma}

\begin{lemma}[\cite{CDGL24+}, Corollary~A.4]{\label{lem-revised-decomposition-H-Subset-S}}
    For $H \subset S$, we can decompose $E(S)\setminus E(H)$ into $\mathtt m$ cycles ${C}_{\mathtt 1}, \ldots, {C}_{\mathtt m}$ and $\mathtt t$ paths ${P}_{\mathtt 1}, \ldots, {P}_{\mathtt t}$ for some $\mathtt m, \mathtt t\geq 0$ such that the following hold: 
    \begin{enumerate}
        \item[(1)] ${C}_{\mathtt 1}, \ldots, {C}_{\mathtt m}$ are independent cycles in $S$.
        \item[(2)] $V(P_{\mathtt j}) \cap \big( V(H) \cup (\cup_{\mathtt i=1}^{\mathtt m} V(C_{\mathtt i})) \cup (\cup_{\mathtt k \neq \mathtt j} V(P_{\mathtt k}) ) \cup \mathcal L (S) \big) = \operatorname{EndP}(P_{\mathtt j})$ for $1 \leq \mathtt j \leq \mathtt t$.
        \item[(3)] $\mathtt t \leq 5(|\mathcal L (S) \setminus V(H) | + \tau(S)-\tau(H))$.
    \end{enumerate}
\end{lemma}

\begin{lemma}{\label{lem-decomposition-H-plus-path}}
    Suppose that $H \subset \mathsf K_n$ and $L_1,\ldots,L_k$ are self-avoiding paths with $\operatorname{EndP}(L_i) \subset V(H)$ and 
    $|E(L_i)|\leq M$. Denoting $S=H \cup L_1 \cup \ldots \cup L_k$, we can decompose $E(S)\setminus E(H)$ into $\mathtt t$ self-avoiding paths ${P}_{\mathtt 1}, \ldots, {P}_{\mathtt t}$ for some $\mathtt t\geq 0$ such that the following hold:
    \begin{enumerate}
        \item[(1)] $V(P_\mathtt j )\cap\big( V(H) \cup (\cup_{\mathtt k=1}^{\mathtt j-1} V(P_{\mathtt k})) \cup \mathcal L(S) \big) = \operatorname{EndP}(P_{\mathtt j})$ for $1 \leq \mathtt j \leq \mathtt t$.
        \item[(2)] $|E(P_{\mathtt i})|\leq M$.
        \item[(3)] $\mathtt t =\tau(S)-\tau(H)$.
    \end{enumerate}
\end{lemma}
\begin{proof}
    We begin with the case $k=1$, i.e., when $S=H\cup L_1$. Let $L_1 = (v_0,\ldots,v_{M_1})$ with $M_1 \leq M$. It is clear that $V(H) \cap V(L_1)$ splits $L_1$ into a collection of new self-avoiding paths that we will denote by $\{ P_\mathtt j :1 \leq \mathtt j \leq \mathtt t\}$. Thus we have $\cup_{1\leq \mathtt j \leq \mathtt t} E(P_\mathtt j) = E(S) \setminus E(H)$, and also Items~(1) and (2) are satisfied. Furthermore, each $\operatorname{EndP}(P_\mathtt j)$ is a subset of $V(H)$ by our construction, which yields $V(S)\setminus V(H) = \sqcup_{1 \leq \mathtt j \leq \mathtt t} \big( V(P_\mathtt j) \setminus \operatorname{EndP}(P_\mathtt j)\big)$ combined with Item~(1). Therefore we have 
    \begin{align*}
        |E(S)| &= |E(H)| + \sum_{\mathtt j = 1}^{\mathtt t} |E(P_{\mathtt j})| \,, \\
        |V(S)| &= |V(H)| + \sum_{\mathtt j =1}^{\mathtt t} | V(P_\mathtt j) \setminus \operatorname{EndP}(P_\mathtt j)| = |V(H)| - \mathtt t + \sum_{\mathtt j =1}^{\mathtt t} |E(P_\mathtt j) | \,,
    \end{align*}
    which yields Item~(3). This completes our proof for $k=1$. 
    
    For the general case, it suffices to denote $H_i = H \cup L_1 \cup \ldots \cup L_i$ (let $H_0=H$ for convenience) and apply the case $k=1$ to $H =H_i$ and $S=H_{i+1}$, and then a simple induction concludes the proof.
\end{proof}

The following lemma provides an upper bound on the number of certain subgraphs contained in a given connected graph $J$.

\begin{lemma}{\label{lem-enu-decorated-trees-in-given-graph}}
    Given $J \subset \mathsf K_n$, we have
    \begin{align}
        \#\big\{ S \subset J: S \Vdash \mathbf H \mbox{ for some } \mathbf H\in \mathcal H \big\} \leq \binom{|E(J)|}{\aleph} \cdot \aleph^{2\iota\aleph} 2^{\iota\aleph(\aleph+5(\tau(J)+1))}  \,. \label{eq-enu-decorated-trees-in-given-graph}
    \end{align}
\end{lemma}
\begin{proof}
    Note that the number of possible choices of $\mathsf T(S)$ is bounded by $\binom{|E(J)|}{\aleph}$, and given $\mathsf T(S)$ the number of possible choices of $\mathsf P(S)$ is bounded by $\aleph^{2\iota\aleph}$. Given $\mathsf T(S)$ and $\mathsf P(S)$, each $\mathsf L_{i}(S)$ can be determined as follows: using Lemma~\ref{lem-revised-decomposition-H-Subset-S}, $J \setminus \mathsf T(S)$ can be decomposed into self-avoiding paths $P_{\mathtt 1},\ldots,P_{\mathtt t}$ with $\mathtt t \leq 5(\tau(J)+1)$ such that Items (1)--(4) of Lemma~\ref{lem-revised-decomposition-H-Subset-S} hold. Since $\mathsf L_i(S)$ is also a self-avoiding path, there exist $E_0 \subset E(\mathsf T(S))$ and $I \subset [\mathtt t]$ such that
    \begin{align*}
        E(\mathsf L_i(S)) = E_0 \cup \Big( \cup_{\mathtt i \in I} E(P_{\mathtt i}) \Big) \,.
    \end{align*}
    Thus, for each $1 \leq i \leq \iota\aleph$ the number of possible choices of $\mathsf L_i(S)$ is bounded by $2^{\aleph+5(\tau(J)+1)}$, and the desired result follows from the multiplication principle.
\end{proof}

The final lemma provides an upper bound on the number of graphs $\mathbf{G}$ that have a specific value of $\tau(\mathbf{G})$ and are isomorphic to a union of decorated trees.
\begin{lemma}{\label{lem-enu-union-of-decorated-trees}}
Recall Definition~\ref{def-decorated-trees-NBP}. For $r \leq 4$, we have 
\begin{align}\label{eq-lem-union-dt}
    \#\Bigg\{ 
    \begin{split}
        & \mbox{unlabeled graph } \mathbf G: \exists S_i \subset \mathsf K_n , S_i \vdash \mathbf H_i \mbox{ for some } \mathbf H_i \in \mathcal H,  \\
        & 1\leq i \leq r \mbox{ such that } \cup_{1\leq i\leq r} \widetilde{S}_i \cong \mathbf{G} , \tau (\mathbf G ) = j 
    \end{split}
    \Bigg\} \leq n^{o(1)} (\ell\aleph)^{5j}\,.
\end{align}
\end{lemma}
\begin{proof}
    First, the enumeration of isomorphic classes for $\cup_{1\leq i \leq r} \mathsf T(S_i)$ is bounded by $\tbinom{4\aleph}{2}^{4\aleph-4} = n^{o(1)}$. Then, given $\cup_{1\leq i \leq r} \mathsf T(S_i)$, we can decompose $\big(\cup_{1\leq i \leq r} \widetilde{S}_i \big) \setminus \big(\cup_{1 \leq i \leq r}\mathsf T(S_i)\big)$ into $\mathtt t = 5\tau \big(\cup_{1\leq i \leq r} \widetilde{S}_i \big) - 5\tau \big(\cup_{1\leq i \leq r} \mathsf T( S_i) \big) \leq 5(j+4)$ self-avoiding paths (denoted by $\{ P_{\mathtt i}\}_{1 \leq \mathtt i \leq \mathtt t}$) by Lemma~\ref{lem-revised-decomposition-H-Subset-S} such that for $1 \leq \mathtt i \leq \mathtt t$ we have $|E(P_{\mathtt i})| \leq \ell$ and
    \begin{align*}
        V(P_\mathtt i)\cap\Big( V\big(\cup_{1\leq i \leq r} \mathsf T(S_i)\big) \cup \big(\cup_{\mathtt k=1}^{\mathtt i-1} V(P_{\mathtt k})\big) \cup \mathcal L\big(\cup_{1\leq i \leq r} S_i\big) \Big) = \operatorname{EndP}(P_{\mathtt i})\,.
    \end{align*}
    Therefore, we only need to bound the enumeration of $\{ P_{\mathtt i}\}_{1 \leq \mathtt i \leq \mathtt t}$ up to isomorphism. Since it suffices to enumerate the endpoints and length of $P_\mathtt i$ for $\mathtt i = 1,\ldots ,\mathtt t$ in order to enumerate $\{ P_{\mathtt i}\}_{1 \leq \mathtt i \leq \mathtt t}$ up to isomorphism, the enumeration of $\{ P_{\mathtt i}\}_{1 \leq \mathtt i \leq \mathtt t}$ is bounded by $((\aleph +\iota\aleph\ell)^2 \cdot \ell )^{5(j+4)} = n^{o(1)}(\ell\aleph)^{5j}$ up to isomorphism. Combining the enumerations above yields our desired result. 
\end{proof}

\section{Proofs of Theorems~\ref{thm-num-desired-tree} and \ref{thm-desired-vertex-sets}}

In this section, we complete the postponed proofs of Theorems~\ref{thm-num-desired-tree} and \ref{thm-desired-vertex-sets}.

\subsection{Proof of Theorem~\ref{thm-num-desired-tree}}{\label{subsec:proof-tree-enu-thm}}

This subsection is dedicated to the proof of Theorem~\ref{thm-num-desired-tree}. We denote $\mathtt L=\log^2(\iota^{-1})$ in this subsection. In addition, denote $\eta_N$ to be the uniform measure over $\mathcal R_{N}$. Recall that for $N \geq 1$ we denote $\mathcal R_N$ as the set of all unlabeled rooted trees with size $N$. Let $ \mathcal R^{(1,2,4)}_N $ be the collection of trees in $\mathcal R_N$ that satisfy Items~(1), (2) and (4) in Definition~\ref{def-tilde-T-K}. Let $\eta^*_N$ be a uniform measure on the set $\mathcal R_N^{(1,2,4)}$. We first show that it suffices to prove
\begin{align}
    \big| \mathcal R^{(1,2,4)}_{\lfloor (\aleph-1)/2 \rfloor} \big| \geq (\alpha+o(1))^{-\aleph/2} \cdot \exp\big\{ -5 e^{-\frac{1}{10}\log^2(\iota^{-1})} \aleph/2 \big\} \,, \label{eq-thm-2.7-relax} 
\end{align}
and 
\begin{align}
    \Eb_{\mathbf T \sim \eta^*_{\lfloor (\aleph-1)/2 \rfloor}} \Big[ \big| \mathfrak{Ch}^{(0,0;*)}_\mathtt m(\mathbf T) \big| - \mathfrak c^{-9} \sum_{0 \leq a,b \leq 1} \big| \mathfrak{Ch}^{(a,b)}_\mathtt m(\mathbf T) \big| \Big] \geq \frac{(\tfrac{\aleph-1}{2}-\mathtt L^2-\mathtt m-2)_+}{\mathtt L^4 \mathfrak c^{9+\mathtt m}} \,.\label{eq-thm-2.7-key-expectation} 
\end{align}
Indeed, if \eqref{eq-thm-2.7-relax} and \eqref{eq-thm-2.7-key-expectation} hold, we claim first that we can select a subset $\mathcal R^{\text{pre}}_{\lfloor (\aleph-1)/2 \rfloor}$ such that $|\mathcal R^{\text{pre}}_{\lfloor (\aleph-1)/2 \rfloor}| \geq \frac{1}{\log^9 (\iota^{-1})} |\mathcal R^{(1,2,4)}_{\lfloor (\aleph-1)/2 \rfloor}|$ and for all $\mathbf T \in \mathcal R^{\text{pre}}_{\lfloor (\aleph-1)/2 \rfloor}$
\begin{align}\label{eq-1,2,4-partly-satisfy-3,6}
    \tfrac{1}{\mathtt m} \big| E(\mathbf T_\iota) \big| \geq \big| \mathfrak{Ch}^{(0,0;*)}_\mathtt m(\mathbf T) \big| - \mathfrak c^{-9} \sum_{0\leq a,b\leq 1} \big| \mathfrak{Ch}^{(a,b)}_\mathtt m(\mathbf T) \big| \geq \tfrac{10\aleph}{\log^9(\iota^{-1})} \,.
\end{align}
We now prove \eqref{eq-1,2,4-partly-satisfy-3,6}. Denote  
\begin{equation}\label{eq-def-mathbf-F-T}
    \mathbf F(\mathbf T)=\big|\mathfrak{Ch}^{(0,0;*)}_\mathtt m(\mathbf T) \big| - \mathfrak c^{-9} \sum_{0 \leq a,b \leq 1} \big|\mathfrak{Ch}^{(a,b)}_\mathtt m(\mathbf T)\big| \,,
\end{equation}
for any unlabeled rooted tree $\mathbf T$. By \eqref{eq-thm-2.7-key-expectation} we have $ \Eb_{\mathbf{T}\sim\eta^*_{\lfloor (\aleph-1)/2 \rfloor}} \big[\mathbf F(\mathbf T) \big] \geq \tfrac{100\aleph}{\log^{9}(\iota^{-1})}$. Therefore, we have
\begin{align*}
    \tfrac{100\aleph}{\log^9 (\iota^{-1})} \leq \aleph\cdot \eta^*_{\lfloor (\aleph-1)/2 \rfloor} \Big( \mathbf F(\mathbf T) \geq \tfrac{10\aleph}{\log^9(\iota^{-1})} \Big) + \tfrac{10\aleph}{\log^9 (\iota^{-1})} \cdot \Big( 1-\eta^*_{\lfloor (\aleph-1)/2 \rfloor} \Big(\mathbf F(\mathbf T) \geq \tfrac{10\aleph}{\log^9(\iota^{-1})} \Big) \Big)\,,
\end{align*}
which yields 
\begin{align*}
    \eta^*_{\lfloor (\aleph-1)/2 \rfloor} \Big( \mathbf F(\mathbf T) \geq \tfrac{10\aleph}{\log^9(\iota^{-1})} \Big) \geq \tfrac{2}{\log^9 (\iota^{-1})}\,.
\end{align*}
 Next, recalling the definition of $\mathfrak{Ch}_\mathtt m^{(0,0,*)}$ in \eqref{def-chained-noise-pair-in-T_iota}, we immediately have $\tfrac{1}{\mathtt m}|E(\mathbf T_\iota)|\geq |\mathfrak{Ch}_\mathtt m^{(0,0,*)}(\mathbf T)|$, which shows that any unlabeled rooted tree $\mathbf T$ satisfies the first inequality of \eqref{eq-1,2,4-partly-satisfy-3,6}, leading to the existence of $\mathcal R^{\text{pre}}_{\lfloor (\aleph-1)/2 \rfloor}$.

Next, we consider all unlabeled rooted trees $\mathbf T$ whose root $\mathfrak R(\mathbf T)$ has exactly two children subtrees $\mathbf T \hookrightarrow \mathbf T_1,\mathbf T_2$ with $\mathbf T_1 \in \mathcal R^{\text{pre}}_{\lfloor (\aleph-1)/2 \rfloor}$ and $\mathbf T_2 \in \mathcal R^{\text{pre}}_{\lceil (\aleph-1)/2 \rceil}$ (which can be constructed in a similar manner), and denote the set of such trees satisfying Item (5) in Definition~\ref{def-tilde-T-K} as $\widetilde{\mathcal{T}}^{\text{pre}}_\aleph$. Then it follows immediately from \eqref{eq-1,2,4-partly-satisfy-3,6} that $|E(\mathbf T_\iota )| \geq \tfrac{\aleph}{\log^9(\iota^{-1})}$ for $\mathbf T \in \widetilde{\mathcal{T}}^{\text{pre}}_\aleph$, which yields $\mathbf T$ satisfies Item (3) in Definition~\ref{def-tilde-T-K}. Also, for all $\mathbf T \in \widetilde{\mathcal T}_{\aleph}^{\text{pre}}$ we have
\begin{align}
\mathbf F(\mathbf T) &\geq 2\cdot \big( \tfrac{\aleph}{\log^9(\iota^{-1})} - \mathfrak c^{-9}\big)\geq \tfrac{\aleph}{\log^9(\iota^{-1})}\label{eq-item-6-is-satisfied}
\end{align}
by definitions of $\mathfrak{Ch}^*_\mathtt m$ and $\mathfrak{Ch}_\mathtt m$, which yields that $\mathbf T$ satisfies Item (6) in Definition~\ref{def-tilde-T-K}. Using Lemma~\ref{lem-remove-sim-relation} we see that
\begin{align*}
    \big| \widetilde{\mathcal T}_{\aleph}^{\text{pre}} \big| &\geq \big| \mathcal R^{\text{pre}}_{\lfloor (\aleph-1)/2 \rfloor} \big|^2 - 2 \big| \mathcal R^{\text{pre}}_{\lfloor (\aleph-1)/2 \rfloor} \big| \cdot \aleph e^{\frac{6\log\log(\iota^{-1})}{\log(\iota^{-1})}\aleph} \\
    &= (\alpha+o(1))^{-\aleph} \cdot \exp\big\{ -10 e^{-\frac{1}{10}\log^2(\iota^{-1})} \aleph \big\} \,.
\end{align*}
Last, for $1\leq i\leq \aleph$ and $1 \leq j \leq \aleph^2$ we denote
\begin{align}
    \widetilde{\mathcal T}_{\aleph;i,j}^{\text{pre}} = \Big\{ \mathbf T \in  \widetilde{\mathcal T}_{\aleph}^{\text{pre}}: |\mathfrak{Ch}^{(0,0;*)}_\mathtt m(\mathbf T)| = i,  \operatorname{Aut}(\mathbf T) \in [2^{j-1},2^j] \Big\}\,,
\end{align}
then we have $\bigcup_{1\leq i \leq \aleph; 1\leq j \leq \aleph^2} \widetilde{\mathcal T}_{\aleph;i,j}^{\text{pre}} = \widetilde{\mathcal T}_{\aleph}^{\text{pre}}$ by $2^{\aleph^2 -1} \geq \aleph! \geq \operatorname{Aut}(\mathbf T)$ for any tree $\mathbf T$ with $\aleph$ vertices. Therefore, there exists a choice of $(i,j)=(i_*, j_*)$ such that 
\begin{equation}
    \big| \widetilde{\mathcal T}_{\aleph;i_*,j_*}^{\text{pre}} \big| \geq \tfrac{1}{\aleph^3} \big| \widetilde{\mathcal T}_{\aleph}^{\text{pre}} \big| =(\alpha+o(1))^{-\aleph} \cdot \exp\big\{ -10 e^{-\frac{1}{10}\log^2(\iota^{-1})} \aleph \big\} \,.
\end{equation}
By definition of $\widetilde{\mathcal T}_{\aleph;i
_*,j_*}^{\text{pre}}$ and \eqref{eq-item-6-is-satisfied} we have $\widetilde{\mathcal T}_{\aleph} = \widetilde{\mathcal T}_{\aleph;i_*,j_*}^{\text{pre}}$ satisfies Items (1)--(7) of Definition~\ref{def-tilde-T-K} if we select 
\begin{align}{\label{eq-def-mathfrak-C-Au}}
    \mathfrak C_\aleph = i_*\,,\quad\mathfrak{Au}_\aleph = 2^{j_*-1}\,,
\end{align}
in Definition~\ref{def-tilde-T-K}, which implies the conclusion of Theorem~\ref{thm-num-desired-tree}. The remainder of this subsection is devoted to the proof of \eqref{eq-thm-2.7-relax} and \eqref{eq-thm-2.7-key-expectation}.

We prove \eqref{eq-thm-2.7-relax} first. To simplify the proof, we work on a strengthened version of \eqref{eq-thm-2.7-relax} instead of \eqref{eq-thm-2.7-relax} itself. Define a strengthened version of Item (1) of Definition~\ref{def-tilde-T-K} as follows:
\begin{itemize}
    \item[(1')] For every $u \in V(\mathbf T)$, $\mathsf{Branch}_{\mathbf T}(u) \leq \log^2(\iota^{-1})$.
\end{itemize}
We refer to this condition by ``Item (1') of Definition~\ref{def-tilde-T-K}" for convenience. Also, let $\mathcal R_N^{(1',2,4)}$ be the collection of trees in $\mathcal R_N$ satisfying Items (1'), (2), (4) of Definition~\ref{def-tilde-T-K}. Thus to show \eqref{eq-thm-2.7-relax}, it suffices to prove
\begin{equation}{\label{eq-thm-2.7-relax-strengthened}}
    \big| \mathcal R^{(1',2,4)}_{\lfloor (\aleph-1)/2 \rfloor} \big| \geq (\alpha+o(1))^{-\aleph/2} \cdot \exp\big\{ -5 e^{-\frac{1}{10}\mathtt L} \aleph/2 \big\} \,.
\end{equation} 
Next we prove \eqref{eq-thm-2.7-relax-strengthened} by the following several lemmas.
\begin{lemma}{\label{lem-enu-trees-condition-1,2}}
    Denoting $\mathcal R^{(1',2)}_N$ the set of $\mathbf T \subset \mathcal R_N$ satisfying Items (1') and (2) in Definition~\ref{def-tilde-T-K}, we have
    \begin{align*}
        \big| \mathcal R^{(1',2)}_N \big| \geq \big| \mathcal R_N \big| \cdot \exp\big( -e^{-\frac{1}{8}\mathtt L} N \big) \,.
    \end{align*}
\end{lemma}
\begin{proof}
    Denote $\mathcal R^{(1')}_N$ as the set of $ \mathbf T \subset \mathcal R_N $ such that each vertex in $V(\mathbf T)$ has at most $\mathtt L$ children (i.e, satisfying Item~(1') in Definition~\ref{def-tilde-T-K}). It has been shown in \cite[Theorem~1]{GS94} that
    \begin{align*}
        \big| \mathcal R^{(1')}_N \big| \geq \big| \mathcal R_N \big| \cdot \exp\big( - e^{ -\frac{1}{4}\mathtt L } N \big) \,.
    \end{align*}
    Thus, it suffices to show that
    \begin{align}
         \big| \mathcal R^{(1',2)}_N \big| \geq \big| \mathcal R^{(1')}_N \big| \cdot \exp\big( - e^{ -\frac{1}{5}\mathtt L } N \big) \,. \label{eq-compare-num-trees-with-condition-2-and-1,2}
    \end{align}
    Denote $\Gamma_N = \big| \mathcal R^{(1')}_N \big|$ and $\widetilde{\Gamma}_N = \big| \mathcal R^{(1',2)}_N \big|$. Recall $\Gamma_N=|\mathcal{R}_N|$. It was known in \cite[Page~585]{Otter48} that
    \begin{align}\label{eq-induction-formula-Gamma_N}
        \Gamma_{N+1} = \sum_{ \substack{ \mu_1 + \ldots + \mu_N \leq \mathtt L \\ \mu_1+2\mu_2+\ldots+N\mu_N=N } } \prod_{i=1}^{N} \binom{ \Gamma_i+\mu_i-1 }{ \mu_i } \,,
    \end{align}
    where for each $i$, $\mu_i$ counts the number of children trees with $i$ vertices and the combinatorial number counts the ways to choose $\mu_i$ trees with $i$ vertices (with repetition). In addition, by definition of $\mathcal R^{(1')}_N$ and $\mathcal R^{(1',2)}_N$ it can be checked that $\widetilde{\Gamma}_k=\Gamma_k$ for $k \leq \mathtt L$. We next derive an induction formula for $\widetilde{\Gamma}_{N+1}$, assuming that 
    \begin{align}
        \widetilde{\Gamma}_k \geq \Gamma_k \cdot \exp\big( - e^{ -\frac{1}{5} \mathtt L } k \big) \mbox{ for } k \leq N \,. \label{eq-induction-hypothesis-trees}
    \end{align}
    Clearly, suppose that $\mathbf T \in \mathcal{R}^{(1',2)}_{N+1}$ has children trees $\mathbf T_1,\ldots,\mathbf T_m \in \mathcal{R}^{(1',2)}_{N}$ such that $\#\{j:|V(\mathbf T_j)|=i\} = \mu_i$, we see that $\mathbf T \in \mathcal R_{N+1}^{(1',2)}$ implies that each $\mathbf T_i$ is not an arm-path with length $\mathtt L-1$. Thus, we see that
    \begin{align*}
        \widetilde{\Gamma}_{N+1} &= \sum_{ \substack{ \mu_1+\ldots+\mu_N \leq \mathtt L \\ \mu_1+2\mu_2+\ldots+N\mu_N=N } } \Bigg( \binom{ \widetilde \Gamma_{\mathtt L-1}+\mu_{\mathtt L-1}-2 }{ \mu_{\mathtt L-1} } \cdot \prod_{i \neq \mathtt L-1} \binom{ \widetilde{\Gamma}_i+\mu_i-1 }{ \mu_i } \Bigg) \\
        &= \sum_{ \substack{ \mu_1+\ldots+\mu_N \leq \mathtt L \\ \mu_1+2\mu_2+\ldots+N\mu_N=N } } \Big( 1-\frac{\mu_{\mathtt L-1}}{\widetilde{\Gamma}_{\mathtt L-1}+\mu_{\mathtt L-1}-1} \Big) \cdot \prod_{i=1}^{N} \binom{ \widetilde{\Gamma}_i+\mu_i-1 }{ \mu_i } \\
        &\geq \Big( 1-\mathtt L e^{-\mathtt L} \Big) \cdot \sum_{ \substack{ \mu_1+\ldots+\mu_N \leq \mathtt L \\ \mu_1+2\mu_2+\ldots+N\mu_N=N } }  \prod_{i=1}^{N} \binom{ \widetilde{\Gamma}_i+\mu_i-1 }{ \mu_i } \,,
    \end{align*}
    where the inequality follows from $\mu_{\mathtt L-1} \leq \mathtt L$ (since the degrees are bounded by $\mathtt L$) and \eqref{eq-choice-iota} (recalling \eqref{eq-choice-N-2}, this implies that $\mathtt L \geq M$ and thus $\widetilde \Gamma_{\mathtt L-1}=\Gamma_{\mathtt L-1} \geq \alpha^{-\mathtt L}/100 \mathtt L^2$ from \eqref{eq-choice-N-2}). Recalling \eqref{eq-induction-hypothesis-trees}, we have
    \begin{align*}
        \binom{ \widetilde{\Gamma}_k+\mu_k-1 }{ \mu_k } \geq \binom{ \Gamma_k+\mu_k-1 }{ \mu_k } \cdot \exp\big( - e^{ -\frac{1}{5} \mathtt L } k\mu_k \big) \,.
    \end{align*}
    Thus, we have
    \begin{align*}
        \widetilde{\Gamma}_{N+1} &\geq \Big( 1-\mathtt L e^{-\mathtt L} \Big) \cdot \sum_{ \substack{ \mu_1+\ldots+\mu_N \leq \mathtt L \\ \mu_1+2\mu_2+\ldots+N\mu_N=N } } \prod_{i=1}^{N} \exp\big( - e^{ -\frac{1}{5} \mathtt L } i\mu_i \big) \binom{ \Gamma_i+\mu_i-1 }{ \mu_i } \\
        &= \Big( 1-\mathtt L e^{-\mathtt L} \Big) \cdot \exp\big( -N e^{ -\frac{1}{5} \mathtt L } \big) \cdot \sum_{ \substack{ \mu_1+\ldots+\mu_N \leq \mathtt L \\ \mu_1+2\mu_2+\ldots+N\mu_N=N } } \prod_{i=1}^{N} \binom{ \Gamma_i+\mu_i-1 }{ \mu_i } \,.
    \end{align*}
    Combined with \eqref{eq-induction-formula-Gamma_N}, this completes the induction proof for \eqref{eq-induction-hypothesis-trees} and thus leads to \eqref{eq-compare-num-trees-with-condition-2-and-1,2}.
\end{proof}


\begin{lemma}{\label{lem-eum-trees-condition-1,2,4}}
    Denote $\mathcal R^{(1',2,4)}_N$ the set of $\mathbf T \subset \mathcal R_N$ satisfying Items~(1'), (2), (4) in Definition~\ref{def-tilde-T-K}. Then 
    \begin{align*}
        \big| \mathcal R^{(1',2,4)}_N \big| \geq \big| \mathcal R^{(1',2)}_N \big| \cdot \exp\big( -e^{-\frac{1}{10}\mathtt L} N \big) \,.
    \end{align*}
\end{lemma}
\begin{proof}
    Recall that we define $\eta_N$ to be the uniform measure on $\mathcal R^{(1',2)}_N$. And we further denote $\eta_N'$ as the uniform measure on $\mathcal R^{(1',2,4)}_N$. Note that the lemma is equivalent to 
    \begin{align}
        \eta_N'(\mathbf T)/\eta_N(\mathbf T) \leq \exp\big( e^{-\frac{1}{10} \mathtt L} N \big) \mbox{ for all } \mathbf T \in \mathcal R_N^{(1',2,4)} \,. \label{eq-lemma-B.3-equivalent-form}
    \end{align}
    We prove \eqref{eq-lemma-B.3-equivalent-form} by induction. Clearly \eqref{eq-lemma-B.3-equivalent-form} holds for $N \leq \mathtt L$. Now suppose that \eqref{eq-lemma-B.3-equivalent-form} holds for all $m \leq N$. For each 
    \begin{align*}
        a+b \leq \mathtt L \,, \gamma_1+\ldots+ \gamma_{a+b}=N \,, \gamma_i \leq \mathtt L \mbox{ if and only if } i \leq a \,,
    \end{align*}
    denote $\mathcal W_{a,b;\gamma_1,\ldots,\gamma_{a+b}}$ the set of $\mathbf T \in \mathcal R_{N+1}^{(1',2)}$ such that $\mathbf T$ has $ a+b $ children trees $\mathbf T_1,\ldots,\mathbf T_{a+b}$ with $|V(\mathbf T_i)|=\gamma_i$ for each $1 \leq i \leq a+b$; we will ignore the orders among those $\mathbf T_i$'s with the same cardinality. In addition, denote $\widetilde{\mathcal W}_{\gamma_1,\ldots,\gamma_a} = \mathcal W_{a,0;\gamma_1,\ldots,\gamma_a}$. Denote $\eta_{a,b;\gamma_1,\ldots, \gamma_{a+b}}$ the uniform measure on $\mathcal W_{a,b;\gamma_1,\ldots,\gamma_{a+b}}$ and $\eta'_{a,b;\gamma_1,\ldots, \gamma_{a+b}}$ the uniform measure on $\mathcal R_{N+1}^{(1',2,4)} \cap \mathcal W_{a,b;\gamma_1,\ldots,\gamma_{a+b}}$; denote $\widetilde{\eta}_{\gamma_1,\ldots,\gamma_a} = \eta_{a,0;\gamma_1,\ldots,\gamma_a}$. For any fixed $\{ \gamma_1,\ldots,\gamma_{a+b} \}$, define 
    \begin{align*}
        \varkappa = \varkappa( \gamma_1,\ldots,\gamma_{a+b} ) = \prod_{ m>\log^2(\iota^{-1}) } \varkappa_m ! \,, \mbox{ where } \varkappa_m =\#\{ a+1 \leq i \leq a+b: \gamma_i=m \} \,.
    \end{align*}
    Denote $\mathbf T_1 \oplus \ldots \oplus \mathbf T_a$ as the unlabeled rooted tree with children trees given by $\mathbf T_1,\ldots,\mathbf T_a$. From the definition of $\eta'$ and $\mathcal{R}_{N+1}^{(1',2,4)}$, we then have (for two measures $\mu$ and $\nu$, we define $\mu\otimes\nu$ to be their product measure)
    \begin{align}
        & \eta'_{a,b;\gamma_1,\ldots,\gamma_{a+b}}(\mathbf T) \nonumber \\ 
        =\ & \eta_{ a,b;\gamma_1,\ldots,\gamma_{a+b} } \big( \mathbf T \mid \mathbf T_{a+i} \in \mathcal R^{(1',2,4)}_{\gamma_{a+i}} \mbox{ for } 1 \leq i \leq b, \mathbf T_{a+i} \not\sim \mathbf T_{a+j}\mbox{ for all } 1\leq i\neq j\leq b \big) \nonumber \\
        =\ & \varkappa \cdot \widetilde{\eta}_{\gamma_1,\ldots,\gamma_a} \otimes \eta'_{\gamma_{a+1}} \otimes \ldots \otimes \eta'_{\gamma_{a+b}}( \mathbf T_1 \oplus \ldots \oplus \mathbf T_a ; \mathbf T_{a+1}, \ldots, \mathbf T_{a+b} \mid \mathcal G ) \,,\label{eq-mu-star-decomposition-calculation}
    \end{align}
    where 
    \begin{align*}
        \mathcal G=\big\{ \mathbf T_{a+i} \not\sim \mathbf T_{a+j}\text{ for all } 1\leq i\neq j\leq b \big\} \,.
    \end{align*}
    (Note that $\mathbf T_{a+i} \not\sim \mathbf T_{a+j}$ implies $\mathbf T_{a+i} \neq \mathbf T_{a+j}$, and the term $\varkappa$ comes from the enumeration that we can permute $\mathbf T_{a+i}$'s with the same cardinality.)
    Note that we have 
    \begin{align}
        \big| \mathcal W_{a,b;\gamma_1,\ldots,\gamma_{a+b}} \big| &= \big| \widetilde{\mathcal W}_{\gamma_1,\ldots,\gamma_a} \big| \cdot \prod_{m>\mathtt L} \tbinom{ |\mathcal R^{(1',2)}_m| + \varkappa_m-1 }{ \varkappa_m } \\
        &\leq \prod_{ m>\mathtt L} \frac{1}{\varkappa_m!} \cdot \big| \widetilde{\mathcal W}_{\gamma_1,\ldots,\gamma_a} \big| \prod_{a+1 \leq i \leq a+b} \Big( \big| \mathcal R^{(1',2)}_{\gamma_i} \big| + b \Big) \,. \label{eq-card-mathcal-W-upper-bound}
    \end{align}
    By Lemma~\ref{lem-remove-sim-relation} and a simple union bound, we have
    \begin{align*}
        & \widetilde{\eta}_{\gamma_1,\ldots,\gamma_a} \otimes \eta'_{\gamma_{a+1}} \otimes \ldots \otimes \eta'_{\gamma_{a+b}}( \exists i\neq j\,,\mathbf T_{a+i} \sim \mathbf T_{a+j} ) \\
        \leq\ & \sum_{1 \leq i<j \leq b} \eta'_{\gamma_{a+i}} \otimes \eta'_{\gamma_{a+j}}( \mathbf T_{a+i} \sim \mathbf T_{a+j} ) \\
        \leq\ & \sum_{1 \leq i<j \leq b} \exp\big( e^{-\tfrac{1}{10}\mathtt L} (\gamma_{a+i}+\gamma_{a+j}) \big) \cdot \eta_{\gamma_{a+i}} \otimes \eta_{\gamma_{a+j}}( \mathbf T_{a+i} \sim \mathbf T_{a+j} ) \\
        \leq\ & \sum_{1 \leq i<j \leq b} \exp\big( e^{-\tfrac{1}{10}\mathtt L} (\gamma_{a+i}+\gamma_{a+j}) \big) \cdot \alpha^{\gamma_{a+j}}\gamma^{1.5}_{a+j} \#\{ \mathbf T \in \mathcal R_{\gamma_{a+j}} :\mathbf T \sim \mathbf T_{a+i}\} \\
        \leq\ & b^2 \cdot \exp\big( (2e^{-\tfrac{1}{10}\mathtt L} +\tfrac{7\log \log (\iota^{-1})}{\log (\iota^{-1})} +\log \alpha )\cdot \mathtt L \big) \leq \exp\big( -\tfrac{1}{8}\mathtt L \big) \,,
    \end{align*}
    where the second inequality follows from the induction hypothesis of \eqref{eq-lemma-B.3-equivalent-form}, the third inequality follows by \eqref{eq-choice-N-2}, the fourth inequality follows from Lemma~\ref{lem-remove-sim-relation} and $\gamma_{a+i} \geq \mathtt L$ and the last inequality is from $b \leq \mathtt L$. Thus, we have 
    \begin{align*}
        & \widetilde{\eta}_{\gamma_1,\ldots,\gamma_a} \otimes \eta'_{\gamma_{a+1}} \otimes \ldots \otimes \eta'_{\gamma_{a+b}}( \mathbf T_1 \oplus \ldots \oplus \mathbf T_a ; \mathbf T_{a+1}, \ldots, \mathbf T_{a+b} \mid \mathcal G ) \\
        \leq\ & \Big( 1-e^{-\tfrac{1}{8}\mathtt L} \Big)^{-1} \cdot \widetilde{\eta}_{\gamma_1,\ldots,\gamma_a} \otimes \eta'_{\gamma_{a+1}} \otimes \ldots \otimes \eta'_{\gamma_{a+b}}( \mathbf T_1 \oplus \ldots \oplus \mathbf T_a ; \mathbf T_{a+1}, \ldots, \mathbf T_{a+b} ) \\
        \leq\ & \Big( 1-e^{-\tfrac{1}{8}\mathtt L} \Big)^{-1} \cdot \exp\big( e^{-\tfrac{1}{10}\mathtt L} (\gamma_{a+1}+\ldots+\gamma_{a+b}) \big) \cdot \\
        & \widetilde{\eta}_{\gamma_1,\ldots,\gamma_a} \otimes \eta_{\gamma_{a+1}} \otimes \ldots \otimes \eta_{\gamma_{a+b}}( \mathbf T_1 \oplus \ldots \oplus \mathbf T_a ; \mathbf T_{a+1}, \ldots, \mathbf T_{a+b} ) \,,
    \end{align*}
    where the second inequality follows from the induction hypothesis of \eqref{eq-lemma-B.3-equivalent-form}. Plugging this estimation into \eqref{eq-mu-star-decomposition-calculation}, we have that for any $\mathbf T \in \mathcal R^{(1',2,4)}_{N+1} \cap \mathcal W_{a,b;\gamma_1,\ldots,\gamma_{a+b}}$, 
    \begin{align}
        \frac{ \eta'_{a,b;\gamma_1,\ldots,\gamma_{a+b}}(\mathbf T) }{ \eta_{a,b;\gamma_1,\ldots,\gamma_{a+b}}(\mathbf T) } \leq\ & \big( 1-e^{-\tfrac{1}{8}\mathtt L} \big)^{-1} \cdot \exp\big( e^{-\tfrac{1}{10}\mathtt L} (\gamma_{a+1}+\ldots+\gamma_{a+b}) \big) \cdot \nonumber \\
        & \frac{ \varkappa\cdot \widetilde{\eta}_{\gamma_1,\ldots,\gamma_a} \otimes \eta_{\gamma_{a+1}} \otimes \ldots \otimes \eta_{\gamma_{a+b}}( \mathbf T_1 \oplus \ldots \oplus \mathbf T_a ; \mathbf T_{a+1}, \ldots, \mathbf T_{a+b} ) }{ \eta_{a,b;\gamma_1,\ldots,\gamma_{a+b}}(\mathbf T) } \nonumber \\
        \leq\ & \exp\big( e^{-\tfrac{1}{10}\mathtt L} (\gamma_{a+1}+\ldots+\gamma_{a+b}) + e^{-\tfrac{1}{8}\mathtt L} + e^{-\tfrac{1}{2}\mathtt L} \big) \nonumber \\
        \leq\ & \exp\big( e^{-\tfrac{1}{10}\mathtt L} (N+1) \big)  \,, \label{eq-inequality-to-be-sums-up}
    \end{align}
    where in the second inequality we used the fact that
    \begin{align}
        & \frac{ \widetilde{\eta}_{\gamma_1,\ldots,\gamma_a} \otimes \eta_{\gamma_{a+1}} \otimes \ldots \otimes \eta_{\gamma_{a+b}}( \mathbf T_1 \oplus \ldots \oplus \mathbf T_a ; \mathbf T_{a+1}, \ldots, \mathbf T_{a+b} ) }{ \eta_{a,b;\gamma_1,\ldots,\gamma_{a+b}}(\mathbf T) } \nonumber \\
        =\ & \frac{ |\mathcal W_{a,b;\gamma_1,\ldots,\gamma_{a+b}}| }{ |\widetilde{\mathcal W}_{\gamma_1,\ldots,\gamma_a}| |\mathcal R^{(1',2)}_{\gamma_{a+1}}| \ldots |\mathcal R^{(1',2)}_{\gamma_{a+b}}| } \nonumber \\
        =\ & \frac{ |\mathcal W_{a,b;\gamma_1,\ldots,\gamma_{a+b}}| }{ |\widetilde{\mathcal W}_{\gamma_1,\ldots,\gamma_a}| (|\mathcal R^{(1',2)}_{\gamma_{a+1}}|+b)\cdot \ldots \cdot (|\mathcal R^{(1',2)}_{\gamma_{a+b}}|+b) }\cdot\frac{(|\mathcal R^{(1',2)}_{\gamma_{a+1}}|+b)\cdot \ldots \cdot (|\mathcal R^{(1',2)}_{\gamma_{a+b}}|+b)}{|\mathcal R^{(1',2)}_{\gamma_{a+1}}| \ldots |\mathcal R^{(1',2)}_{\gamma_{a+b}}|} \nonumber \\
        \leq \ & \frac{ |\mathcal W_{a,b;\gamma_1,\ldots,\gamma_{a+b}}| }{ |\widetilde{\mathcal W}_{\gamma_1,\ldots,\gamma_a}| (|\mathcal R^{(1',2)}_{\gamma_{a+1}}|+b) \ldots (|\mathcal R^{(1',2)}_{\gamma_{a+b}}|+b) } \cdot \Big( \tfrac{|\mathcal R^{(1',2)}_{\mathtt L}|+\mathtt L}{|\mathcal R^{(1',2)}_{\mathtt L}|}\Big)^{\mathtt L} \nonumber \\
        \overset{\eqref{eq-card-mathcal-W-upper-bound}}{\leq}& \varkappa^{-1} \cdot \Big( \tfrac{|\mathcal R^{(1',2)}_{\mathtt L}|+\mathtt L}{|\mathcal R^{(1',2)}_{\mathtt L}|}\Big)^{\mathtt L}
        \overset{\text{Lemma~\ref{lem-enu-trees-condition-1,2}}}{\leq}  \varkappa^{-1} \cdot \exp (e^{-\tfrac{1}{2} \mathtt L})\,. \nonumber
    \end{align}
    where in the first inequality we used $\frac{|\mathcal R^{(1',2)}_{\gamma_{a+j}}|+b}{|\mathcal R^{(1',2)}_{\gamma_{a+j}}|}\leq \frac{|\mathcal{R}_\mathtt L^{(1',2)}|+b}{|\mathcal{R}_\mathtt L^{(1',2)}|} \leq \frac{|\mathcal{R}_\mathtt L^{(1',2)}|+\mathtt L}{|\mathcal{R}_\mathtt L^{(1',2)}|}$ for each $j\leq b$ and $b\leq\mathtt L$. Summing \eqref{eq-inequality-to-be-sums-up} over $\{ (a,b); \gamma_1,\ldots,\gamma_{a+b} \}$ yields \eqref{eq-lemma-B.3-equivalent-form}. 
\end{proof}

Then the proof of \eqref{eq-thm-2.7-relax} follows by combining Lemmas~\ref{lem-enu-trees-condition-1,2} and \ref{lem-eum-trees-condition-1,2,4}. Next, we deal with \eqref{eq-thm-2.7-key-expectation} using the methods of mapping and induction on layers. We first define some frequently used notations. Recall that for $N \geq 1$ we defined
$
\eta^{*}_N = \operatorname{Uniform}(\mathcal R_{N}^{(1,2,4)})
$
to be a uniform measure. Define
\begin{align}\label{possible-states-for-1,2,4-trees}
    \mathcal D_N = \Big\{ &\big(\{\nu_i\}_{i \geq 1}, \{\vartheta_i\}_{i \geq 1}\big) : \nu_i,\vartheta_i \in \mathbb N, \nu_i \geq \vartheta_i \text{ for }i \geq 1\,; \nonumber \\
    &\sum_{i\geq 1} i\nu_i = N - 1, \sum_{i \geq 1} (\nu_i - \vartheta_i ) \leq \mathtt L \Big\}\,.
\end{align}
Once $\big(\{\nu_i\}_{i \geq 1}, \{\vartheta_i\}_{i \geq 1}\big) \in \mathcal D_N$, we also denote $\big((\nu_1 , \ldots , \nu_N), (\vartheta_1 , \ldots , \vartheta_N)\big) \in \mathcal D_N$ since $\nu_i = 0$ must hold for $i > N$. For $\varrho = \big((\nu_1 , \ldots , \nu_N), (\vartheta_1 , \ldots , \vartheta_N)\big) \in \mathcal D_N$, we define (recall \eqref{def-sub-on-root-size-i} and \eqref{def-armsub-on-root-size-i})
\begin{align}\label{eq-def-type-varrho}
    \mathsf{Type}_\varrho = \Big\{ \mathbf T \in \mathcal R^{(1,2,4)}_N: |\mathsf{Subs}_i(\mathbf T)| = \nu_i \,, |\mathsf{ASubs}_i(\mathbf T)| = \vartheta_i\mbox{ for } 1 \leq i \leq N \Big\} 
\end{align} 
to further decompose $\mathcal R_N^{(1,2,4)}$, and we define
\begin{align}\label{def-Enum-p-etastar-rho}
    \mathsf{Enum}(\varrho) = |\mathsf{Type}_{\varrho}| \,, \quad p(\varrho)=\frac{|\mathsf{Type}_\varrho|}{|\mathcal R^{(1,2,4)}_N|} \,,\quad \eta^*_{N,\varrho} = \operatorname{Uniform}(\mathsf{Type}_\varrho)\,.
\end{align}
Given $\varrho =\big((\nu_1,\ldots,\nu_N), (\vartheta_1,\ldots,\vartheta_N) \big) \in \mathcal D_N$ and $\mathbf T \in \mathsf{Type}_\varrho$, we denote $\varrho(\mathbf T) = \varrho$, and we randomly label the trees in $\mathsf{Subs}(\mathbf T)$ by $\mathbf T_{i,1},\ldots,\mathbf T_{i,\nu_i}$ for $1 \leq i \leq N$ so that $(\mathbf T_{i,\vartheta_i+1}, \ldots, \mathbf T_{i,\nu_i})$ is a uniform random ordering of $\mathsf{Subs}_i(\mathbf T) \setminus \mathsf{ASubs}_i(\mathbf T)$ independently for $1 \leq i \leq N$. Also, given $\varrho \in \mathcal D_N$ define $\eta^*_{N,\varrho;i}$ to be the distribution of $\mathbf T_{i,\vartheta_i+1}$ given $\mathbf T \sim \eta^*_{N,\varrho}\,.$ To treat arm-trees separately from other unlabeled rooted trees, for $3\leq i<N$ and $\varrho \in \mathcal D_N$ we define (recall the definition of $\mathbf{A}_i$ in \eqref{eq-AN})
\begin{align}\label{eq-def-eta-starstar-i}
    \eta^{**}_N(\cdot) = \eta^*_N(\cdot\mid\cdot \neq \mathbf A_N) \,,\quad \eta^{**}_{N,\varrho;i}(\cdot) = \eta^*_{N,\varrho;i}(\cdot\mid\cdot \neq \mathbf A_i) \,.
\end{align}
Also, for $\mathbf T \in \mathsf{Type}_\varrho$ where $\varrho = \big( (\vartheta_1, \ldots , \vartheta_N)\,, (\nu_1 , \ldots ,\nu_N )\big) \in \mathcal D_N$, define $\vartheta(\mathbf T) = \vartheta (\varrho) = (\vartheta_1, \ldots , \vartheta_N)$, and for $\vec{\vartheta} = (\vartheta_1, \ldots , \vartheta_N)$ we define $\Sigma^*(\vec{\vartheta}) = \sum_{i \geq 1}i\vartheta_i$.

Then, we formulate a generating function which is useful when dealing with the arm-paths of the root. Defining
\begin{equation}\label{eq-def-mathsf-p-j}
    \prod_{i \geq 1} (1 -x^i)^{-1}=\sum_{j \geq 0} \mathsf a_j x^j\,,
\end{equation}
then we can see that for any $1 \leq i \leq N$ 
\begin{equation}\label{eq-comb-meaning-mathsf-p-i}
    \# \Big\{ \vec{x} =(x_1 , \ldots , x_N) : x_j \geq 0 \mbox{ for all } 1\leq j \leq N , \Sigma^*(\vec{x}) =i \Big\} = \mathsf a_i\,.
\end{equation}
Also, by \cite[Equation~(1.41)]{HR17}, there is a universal constant $\mathfrak c'$ such that
\begin{equation}\label{eq-def-mathfrak-c-prime}
    \mathsf a_i < \mathfrak c' \cdot 1.01^i\,.
\end{equation}
By definition of $\Sigma^*$, we have
\begin{equation}\label{eq-sigma-star-N-1}
    \#\Big\{ \mathbf T \in \mathcal R^{(1,2,4)}_N : \Sigma^*(\vartheta(\mathbf T)) =N-1 \Big\} = \mathsf a_{N-1}\,.
\end{equation}
Next, we formulate some general mappings that help in enumerating specific types of trees. For $\varrho =\big((\nu_1 , \ldots , \nu_N), (\vartheta_1 , \ldots , \vartheta_N)\big) \in \mathcal D_N$, define $\bar{\varrho} = \big((\nu_1 - \vartheta_1 , \ldots , \nu_N - \vartheta_N), (0,\ldots,0)\big) \in \mathcal D_N$ and we construct the following mapping:
\begin{align}\label{def-HCut}
    \mathsf{HCut}:&\mathbf T\in\mathsf{Type}_\varrho \longrightarrow \mbox{the unique unlabeled rooted tree } \mathbf T' \in \mathsf{Type}_{\bar{\varrho}}\mbox{ such that}\nonumber\\
    &\mathfrak R(\mathbf T') = \mathfrak R(\mathbf T)\,,\mathsf{Subs}(\mathbf T') = \{ \mathbf T_{i,j}:1 \leq i \leq N , \vartheta_i < j \leq \nu_i \}\,,
\end{align}
which ``cuts down" the arm trees among the subtrees of the root (see Figure~\ref{fig:HCut} for an illustration). 
\begin{figure}[ht]
    \centering
    \vspace{0cm}
    \includegraphics[height=6cm,width=12cm]{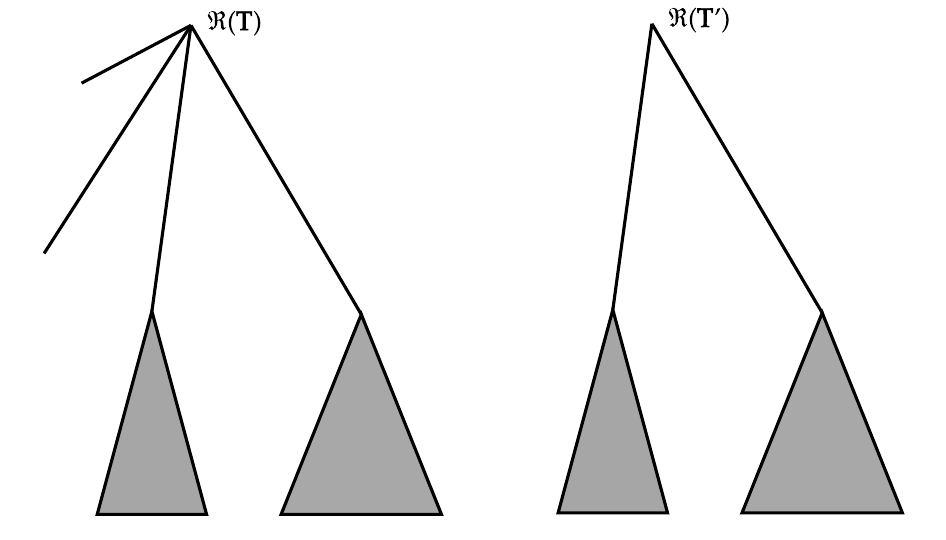}
    \caption{Example of $\mathbf T'=\mathsf{HCut}(\mathbf T)$}
    \label{fig:HCut}
\end{figure}
Also, for $A>0$ we define 
\begin{align}
    &\Sigma_{\leq A,\varrho}=1+\sum_{i \leq A} i\nu_i\,, \quad \Sigma_{>A,\varrho} =1+\sum_{i>A} i\nu_i \,, \\ \label{eq-def-Sigmas}
    & \varrho_{\leq A} = \big((\nu_i)_{\{ i \leq A\}}, (\vartheta_i)_{\{ i \leq A\}}\big)\,,\\
    & \varrho_{>A} = \big( (\nu_i)_{\{ i>A \}}, (\vartheta_i)_{\{ i>A \}}\big)\,,
\end{align}
and for $\mathbf T \in \mathcal R^{(1,2,4)}_N$ and $A>0$ we define the splitting mapping by
\begin{align}
    &\mathsf{Split}_A(\mathbf T)= (\mathbf T_{\leq A},\mathbf T_{>A}) \mbox{ where } \mathbf T_{\leq A} \in \mathcal R^{(1,2,4)}_{\Sigma_{\leq A,\varrho}}\,, \mathbf T_{>A} \in \mathcal R^{(1,2,4)}_{\Sigma_{>A,\varrho}} \mbox{ with } \nonumber \\
    &\mathsf{Subs}(\mathbf T_{\leq A})= \mathsf{Subs}(\mathbf T) \cap \big( \cup_{i \leq A} \mathcal R^{(1,2,4)}_{i}\big)\,, \mathsf{Subs}(\mathbf T_{>A}) = \mathsf{Subs}(\mathbf T) \cap \big( \cup_{i>A} \mathcal R^{(1,2,4)}_{i} \big)\,.\label{eq-def-split-A-T}
\end{align}
Also, for $a,b > 0$ we define the concatenation mapping from two unlabeled rooted trees $\mathbf T_1 \in \mathcal R^{(1,2,4)}_a \,, \mathbf T_2 \in \mathcal R^{(1,2,4)}_b$ to one unlabeled rooted tree by
\begin{align}\label{eq-def-concat}
    \mathsf{Concat}(\mathbf T_1,\mathbf T_2) = \mathbf T \in \mathcal R^{(1,2,4)}_{a+b-1} \mbox{ with }\mathsf{Subs}(\mathbf T) = \mathsf{Subs}(\mathbf T_1) \cup \mathsf{Subs}(\mathbf T_2)\,.
\end{align}
(see Figure~\ref{fig:Concat} for an illustration.)
\begin{figure}[ht]
    \centering
    \vspace{0cm}
    \includegraphics[height=4.72cm,width=13.6cm]{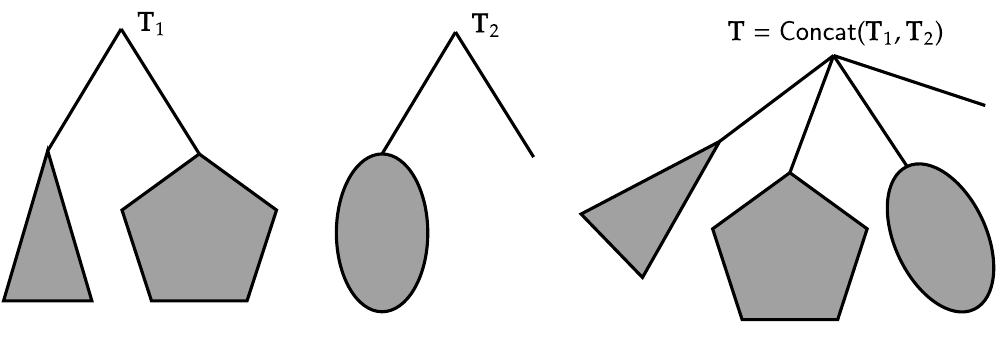}
    \caption{Example of $\mathbf T=\mathsf{Concat}(\mathbf T_1,\mathbf T_2)$}
    \label{fig:Concat}
\end{figure}

Also, for $i \geq 1$ and any unlabeled rooted tree $\mathbf T \in \mathcal R_N$ satisfying $\mathfrak m(\mathbf T) \geq i$ (recall \eqref{eq-def-mathfrak-m-T}) we construct a mapping
\begin{align}\label{def-VCut}
    \mathsf{VCut}_i:& \mathbf T \longrightarrow \mbox{the unique unlabeled rooted tree } \mathbf T' \in \mathcal R_{N-i}\mbox{ such that}\nonumber\\
    &\mathbf T' = \mathsf{Des}_{\mathbf T} (u)\mbox{ for the unique }u\in V(\mathbf T) \mbox{ satisfying }\mathsf{Dep}_{\mathbf T}(u) = i\,,
\end{align}
which ``cuts down'' the path of length $i$ starting from $\mathfrak R(\mathbf T)$ in $\mathbf T$ (see Figure~\ref{fig:VCut} for an illustration). 
\begin{figure}[ht]
    \centering
    \vspace{0cm}
    \includegraphics[height=6.7cm,width=14.9cm]{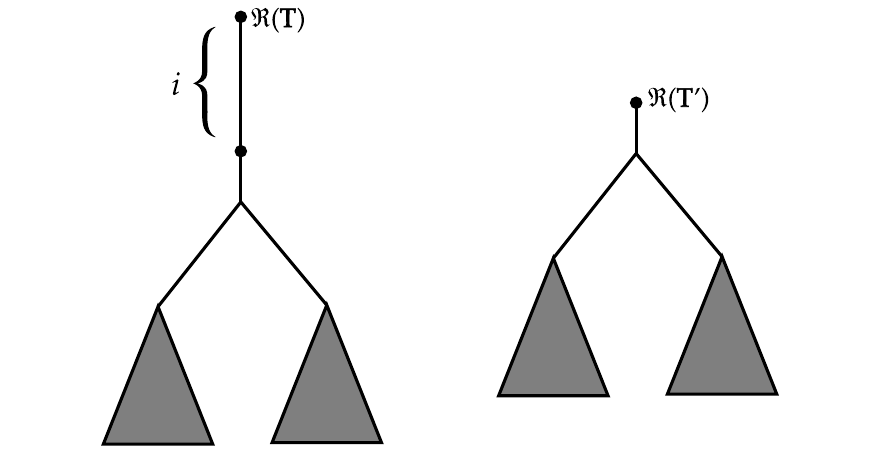}
    \caption{Example of $\mathbf T'=\mathsf{VCut}_i(\mathbf T)$}
    \label{fig:VCut}
\end{figure}

We then prove several lemmas to further understand the ``refined" unlabeled rooted tree set $\mathcal R^{(1,2,4)}_N$ and the (conditional) uniform distributions on $\mathcal R^{(1,2,4)}_N$.
\begin{lemma}\label{lem-prelim-tree-ratio}
Define
\begin{equation}\label{eq-def-mathcal-R-1,2}
    \mathcal R^{(1,2)}_{N} = \Big\{ \mathbf T \in \mathcal R_N:\mathbf T \mbox{ satisfies Items (1), (2) in Definition }\ref{def-tilde-T-K} \Big\}\,,
\end{equation}
and
\begin{equation}\label{eq-def-mathcal-R-1,2,4}
    \mathcal R^{(1,2,4)}_{N} = \Big\{ \mathbf T \in \mathcal R_N:\mathbf T \mbox{ satisfies Items (1),(2),(4) in Definition }\ref{def-tilde-T-K} \Big\}\,,
\end{equation}
Then we have $|\mathcal R^{(1,2)}_{N+1}| \geq 2 |\mathcal R^{(1,2)}_{N}| + \mathbf{1}_{\{N \geq 8\}}$ and $|\mathcal R^{(1,2,4)}_{N+1}| \geq 2 |\mathcal R^{(1,2,4)}_{N}| + \mathbf{1}_{\{N \geq 8\}}$.
\end{lemma}
\begin{proof}
    We only prove the first inequality (i.e. $|\mathcal R^{(1,2)}_{N+1}| \geq 2 |\mathcal R^{(1,2)}_{N}|+\mathbf{1}_{\{N\geq 8\}}$) since the second inequality (i.e. $|\mathcal R^{(1,2,4)}_{N+1}| \geq 2 |\mathcal R^{(1,2,4)}_{N}|+\mathbf{1}_{\{N\geq 8\}}$) follows similarly to the proof of the first inequality. Given a tree $\mathbf T \in \mathcal R^{(1,2)}_N$, we map $\mathbf T$ to two trees in $\mathcal R_{N+1}$, namely
    \begin{align}\label{eq-def-T-+}
        \mathbf T_{+} = \mbox{ the unique unlabeled rooted tree in }\mathcal R^{(1,2)}_{N+1} \mbox{ such that } \mathsf{VCut}_1(\mathbf T_+) = \mathbf T\,,
    \end{align}
    and
    \begin{align}\label{eq-def-T-*}
        \mathbf T_{*} = \mbox{ the unique unlabeled rooted tree in }\mathcal R^{(1,2)}_{N+1}\mbox{ such that }\mathsf{HCut}(\mathbf T_*) =\mathbf T\,.
    \end{align}
    By definition of $\mathbf T_+$ and $\mathbf T_*$, $\mathbf T_+$ is obtained from $\mathbf T$ by adding a parent on $\mathfrak R(\mathbf T)$ and $\mathbf T_*$ is obtained from $\mathbf T$ by adding a child on $\mathfrak R(T)$, leading to the fact that we always have $d_{\mathbf T_+}(\mathfrak R(\mathbf T_+)) = 1$ and $d_{\mathbf T_*}(\mathfrak R(\mathbf T_*)) \geq 2$. Therefore, $\{ \mathbf T_+:\mathbf T \in \mathcal R^{(1,2)}_N\}$ and $\{ \mathbf T_*:\mathbf T \in \mathcal R^{(1,2)}_N\}$ are disjoint. Also, for distinct $\mathbf T, \widetilde{\mathbf T} \in R^{(1,2)}_N$ we have $\mathbf T_+ \neq \widetilde{\mathbf T}_+$ and $\mathbf T_* \neq \widetilde{\mathbf T}_*$ by definition of $\mathbf T_+$ and $\mathbf T_*$, which yields $\#\{ \mathbf T_+:\mathbf T \in \mathcal R^{(1,2)}_N\} = \#\{ \mathbf T_* :\mathbf T \in \mathcal R^{(1,2)}_N\} = |\mathcal R_{N+1}^{(1,2)}|$. Therefore we have
    \begin{align*}
        |\mathcal R_{N+1}^{(1,2)}| \geq \#\{ \mathbf T_+ :\mathbf T \in \mathcal R^{(1,2)}_N\} + \#\{ \mathbf T_*:\mathbf T \in \mathcal R^{(1,2)}_N\} = 2|\mathcal R_N^{(1,2)}|\,.
    \end{align*}
    Now, if we can construct an unlabeled rooted tree $\mathbf T_{\text{special}}$ in $\mathcal R^{(1,2)}_{N+1}$ that cannot be obtained by either \eqref{eq-def-T-+} or \eqref{eq-def-T-*}, then the first inequality become strict when $N \geq 8$, which yields our desired result. 
    \begin{figure}[ht]
        \centering
        \vspace{0cm}
        \includegraphics[height=5.28cm,width=7.84cm]{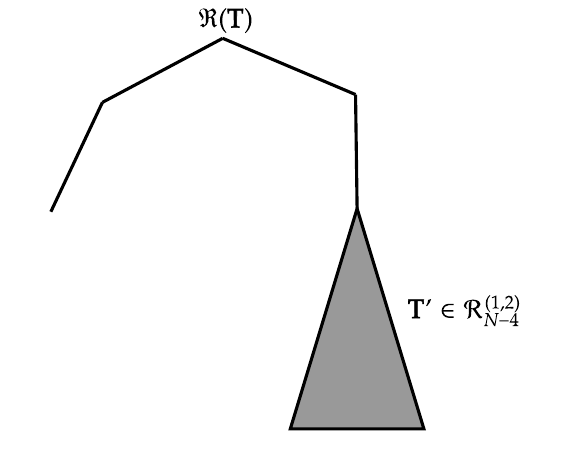}
        \caption{Illustration of $\mathbf T=\mathbf T_{\text{special}}$}
        \label{fig:Special-T}
    \end{figure}
    In order for the construction, we set $\mathbf T_{\text{special}}=\mathbf T$ with $\mathsf{Subs}(\mathbf T) = \{ \mathbf A_2 , \mathbf T'_{+} \}$ for an arbitrary $\mathbf T' \in \mathcal R^{(1,2)}_{N-4}$ where $\mathbf T'_{+}$ is defined in \eqref{eq-def-T-+} (see Figure~\ref{fig:Special-T} for an illustration). Note that since the degree of $\mathfrak R(\mathbf T)$ is 2, $\mathbf T$ cannot be obtained by \eqref{eq-def-T-+}. In addition, since $\mathbf T$ does not have any length-1 arm path at the root, $\mathbf T$ cannot be obtained by \eqref{eq-def-T-*}.  
\end{proof}

\begin{lemma}\label{lem-tree-simple-horizontal-split-enum}
For all $N \geq 1$, define (recall \eqref{eq-def-frakAc})
    \begin{equation}\label{eq-def-mathcal-N-1,2,4}
       \mathcal N^{(1,2,4)}_{N} = \Big\{ \mathbf T \in \mathcal R_N:\mathbf T \mbox{ satisfies Items }(1),(2),(4)\mbox{ in Definition }\ref{def-tilde-T-K}\,,\mathfrak k(\mathfrak{R}(\mathbf T)) = 0 \Big\}\,.
    \end{equation}
    Then for $0 \leq i \leq N-4$ we have
    \begin{align*}
         \#\Big\{ \mathbf T \in \mathcal R^{(1,2,4)}_N: \Sigma^*(\vartheta(\mathbf T)) =i \Big\}= \mathsf a_i \cdot \big| \mathcal N^{(1,2,4) }_{N-i} \big|\,.
    \end{align*}
\end{lemma}
\begin{proof}
For $0 \leq i \leq N-4$ we have
    \begin{align}
        & \#\Big\{ \mathbf T \in \mathcal R^{(1,2,4)}_N:\Sigma^*(\vartheta(\mathbf T)) =i \Big\} \nonumber \\
        =\ & \# \{ \mathbf T \in \mathcal R^{(1,2,4)}_N : \mathsf{HCut}(\mathbf T) \in \mathcal N^{(1,2,4) }_{N-i}\,, \mathbf T{\setminus \mathsf{HCut}(\mathbf T)} \in \mathcal R^{(1,2,4)}_{i+1}\,, \Sigma^*(\vartheta(\mathbf T{\setminus \mathsf{HCut}(\mathbf T)})) = i  \}\nonumber \\
        =\ & \# \{ (\mathbf T^{(1)},\mathbf T^{(2)}):\mathbf T^{(1)} \in \mathcal N^{(1,2,4) }_{N-i}\,, \mathbf T^{(2)} \in \mathcal R^{(1,2,4)}_{i+1}\,, \Sigma^*(\vartheta(\mathbf T^{(2)})) = i  \}=\mathsf a_i \cdot \big| \mathcal N^{(1,2,4) }_{N-i} \big| \,,\label{eq-sigma-star-N-2}
    \end{align}
    where the last equality holds by applying \eqref{eq-sigma-star-N-1}, completing the proof.
\end{proof}

\begin{lemma}\label{lem-tree-N-N+1-ratio-lowers}
    Recall \eqref{eq-def-eta-starstar-i}. For all $\varrho = \big( (\nu_1 , \ldots, \nu_N), (\vartheta_1 , \ldots , \vartheta_N)\big) \in \mathcal D_N$ such that $\vartheta_i < \nu_i$ we have 
    \begin{align*}
        \mathsf{TV}(\eta^{**}_{N,\varrho;i} , \eta^{**}_i ) \leq \mathtt L \cdot 0.7^i \cdot \mathbf{1}_{\{i \geq \mathtt L\}} \,,
    \end{align*}
    and the same bound also holds for $\mathsf{TV}(\eta^{*}_{N,\varrho;i} , \eta^{*}_i )$.
\end{lemma}
\begin{proof}
    The proofs of the bounds for $\mathsf{TV}(\eta^{**}_{N,\varrho;i} , \eta^{**}_i )$ and $\mathsf{TV}(\eta^{*}_{N,\varrho;i} , \eta^{*}_i )$ are similar. For simplicity, we only present the proof of the upper bound for $\mathsf{TV}(\eta^{**}_{N,\varrho;i} , \eta^{**}_i )$. When $i<\mathtt L$, the distribution of $\mathbf T_{i,\vartheta_i+1}$ is uniform in $\mathcal R^{(1,2,4)}_{i} \setminus \{ \mathbf A_i\}$ given $\{ \mathbf T_{i',\ldots}\}_{i' \neq i}$. which yields that in this case $\eta^{**}_{N,\varrho}(i)=\eta^{**}_i$. Suppose that $i \geq \mathtt L$, we prove by considering the conditional distribution of $\mathbf T_{i,j}(\vartheta_i + 2 \leq j \leq \nu_i)$ and $\mathbf T_{i,j}(i' > i, \vartheta_{i'} < j \leq \nu_{i'})$ given the unlabeled rooted trees $\mathbf T_{i',j}(i'<i, \vartheta_{i'}<j \leq\nu_{i'} )$. For any $\mathtt i>0$, suppose we wish to choose an unlabeled rooted tree in $\mathcal{R}_i$ such that it is similar to $\mathtt L$ given unlabeled rooted trees with at most $i$ vertices. Then, by Lemma~\ref{lem-remove-sim-relation} the number of such choices is at most
    \begin{align}\label{eq-def-H-iota-mathtt-i}
        \mathtt L\exp\big(\tfrac{6\log\log(\iota^{-1})\mathtt i}{\log(\iota^{-1})} \big) \mathbf{1}_{\{\mathtt i\geq \mathtt L \}} := H(\iota,\mathtt i)\,.
    \end{align}
    Therefore, for any $\mathsf{tree}_{i',j} \in \mathcal R^{(1,2,4)}_{i'} \setminus \{ \mathbf A_{i'}\}(i' \leq i, \vartheta_{i'} < j \leq \nu_{i'})$ , $\mathsf{tree}^* \in \mathcal R^{(1,2,4)}_{i} \setminus \{ \mathbf A_{i}\}$ and $1 \leq a \leq \nu_i - \vartheta_i$, we have 
    \begin{align}
        & \eta^{*}_{N,\varrho}\Big( \mathbf T_{i,\vartheta_i+a} = \mathsf{tree}^* \mid \mathbf T_{i',j} = \mathsf{tree}_{i',j} \mbox{ for all } i'<i, \vartheta_{i'} < j \leq \nu_{i'}; \mathbf T_{i,j} = \mathsf{tree}_{i,j} \mbox{ for all }\vartheta_i < j < \vartheta_i + a \Big) \nonumber \\
        =\ & \tfrac{1}{\nu_i-\vartheta_i - a + 1} \cdot \eta^{*}_{N,\varrho } \Big( \mathsf{tree}^* \in \{ \mathbf T_{i,j} : \vartheta_i + a \leq j \leq \nu_i \} \mid \mathbf T_{i',j} = \mathsf{tree}_{i',j} \mbox{ for all } i' < i, \vartheta_{i'} < j \leq \nu_{i'}; \nonumber \\
        \ & \mathbf T_{i,j} = \mathsf{tree}_{i,j} \mbox{ for all }\vartheta_i < j < \vartheta_i + a  \Big) \nonumber \\
        \geq\ & \tfrac{1}{\nu_i - \vartheta_i - a + 1} \cdot \frac{\tbinom{|\mathcal R^{(1,2,4)}_{i}| + \nu_{i} - \vartheta_i - 3 - H(\iota , i)}{\nu_{i} - \vartheta_i - a} \cdot \prod_{i' > i} \tbinom{|\mathcal R^{(1,2,4)}_{i'}| + \nu_{i'} - \vartheta_{i'} - 2 - H(\iota , i')}{\nu_{i'} - \vartheta_{i'}}}{ \tbinom{|\mathcal R^{(1,2,4)}_{i}| + \nu_{i} - \vartheta_{i} - a }{\nu_{i} - \vartheta_{i} - a + 1} \cdot \prod_{i' > i} \tbinom{|\mathcal R^{(1,2,4)}_{i'}| + \nu_{i'} - \vartheta_{i'} - 2 }{\nu_{i'} - \vartheta_{i'}}} \nonumber \\
        \geq\ & \frac{1}{|\mathcal R^{(1,2,4)}_{i}|} \cdot \prod_{i' \geq i} \Big( 1 - \tfrac{H(\iota,i')}{|\mathcal R^{(1,2,4)}_{i'}|} \Big)^{\nu_{i'} - \vartheta_{i'}}\,,\label{eq-temp-in-proving-TV-dist-low}
    \end{align}
    where the first inequality holds since $\mathbf T \in \mathcal R^{(1,2,4)}_N$ and $i \geq \mathtt L$ imply that $\mathbf T_{i,\vartheta_i + a}, \ldots, \mathbf T_{i,\nu_i}$ are not similar to each other and therefore distinct, and the second inequality comes from applying \eqref{eq-def-H-iota-mathtt-i}.
    To further estimate \eqref{eq-temp-in-proving-TV-dist-low}, observe that for $i \geq \mathtt L$ we have
    \begin{equation}\label{eq-lower-bound-mathcal-R-1,2,4}
        |\mathcal R^{(1,2,4)}_i| \geq 2^{i-2}
    \end{equation}
    by Lemma~\ref{lem-prelim-tree-ratio}, which shows that for any $i_0 \geq \mathtt L$, and nonnegative integers $k_{i_0} , k_{i_0 + 1}, \ldots$ such that $\sum_{j \geq i_0} k_i \leq \mathtt L$ we have 
    \begin{align}
        \prod_{j \geq i_0} \Big(1-\tfrac{H(\iota, j)}{|\mathcal R^{(1,2,4)}_{j}|}\Big)^{k_j} \geq \exp\big(-2\sum_{j \geq i_0}(0.65^j \cdot k_j)\big) \geq \exp (-0.68^{i_0})\,.\label{eq-estimate-loss-of-condition-4-general}
    \end{align}
    This further implies
    \begin{equation}
        \eqref{eq-temp-in-proving-TV-dist-low} \geq \frac{1}{|\mathcal R^{(1,2,4)}_{i}|} \cdot (1 - 0.68^{i}) \overset{\eqref{eq-lower-bound-mathcal-R-1,2,4}}{\geq} \frac{1}{|\mathcal R^{(1,2,4)}_{i}|-1} \cdot (1 - 0.7^{i}) \,.
    \end{equation}
    Therefore, we have
    \begin{align*} 
        \mathsf{TV}(\eta^{**}_{N,\varrho;i} \,, \eta^{**}_i ) \leq \sum_{\mathbf T \in \mathcal R^{(1,2,4)}_{i}} \frac{1}{| \mathcal R^{(1,2,4)}_{i}|} (1 - (1 - 0.7^i )) \cdot \mathbf{1}_{\{i \geq \mathtt L\}} \leq 0.7^i \cdot \mathbf{1}_{\{i \geq \mathtt L\}}\,,
    \end{align*}
    which yields our desired result.
\end{proof}

We next investigate the upper bound of $\frac{|\mathcal R^{(1,2,4)}_{N+1}|}{|\mathcal R^{(1,2,4)}_{N}|}$ using a more complicated mapping argument. To this end, we begin with a preliminary lemma controlling the number of leaf isomorphism classes in a random unlabeled tree.
\begin{lemma}\label{lem-high-probability-NL-large}
    For any $\mathbf T \in \mathcal R_N$, denote
    \begin{align*}
        \mathsf{N}(\mathbf T) = \Big\{ u \in V(\mathbf T) : \mathcal L'(\mathbf T)\cap \mathsf{Nei}_{\mathbf T}(u) \neq \emptyset \Big\}\,,
    \end{align*}
    (Recall the definition of $\mathcal L'(\mathbf T) := \mathcal L(\mathbf T) \setminus \{ \mathfrak R(\mathbf T) \}$ from Section~\ref{subsec:notations}) and for all $u \in \mathsf N(T)$ we arbitrary choose $\mathsf{leaf}_u \in \mathcal L'(\mathbf T) \cap \mathsf{Nei}_{\mathbf T}(u)$.
    In addition, denote $\mathsf{NL}(\mathbf T)$ as the isomorphism class of $\mathsf N(\mathbf T)$ under $\operatorname{Aut}(\mathbf T)$, i.e., $\mathsf{NL}(\mathbf T)$ is the largest subset of $\mathsf{N}(\mathbf T)$ such that $\{ \mathbf T|_{ V(\mathbf T) \setminus \{ \mathsf{leaf}_u \} }: u \in \mathsf{NL}(\mathbf T) \}$ is a collection of trees which are pairwisely non-isomorphic. Then 
    \begin{equation}\label{eq-high-probability-NL-large}
        \eta_{N+1}^*\big( |\mathsf{NL}(T)| \leq 0.0001(N+1) \big) \leq \exp(-0.1N) \,.
    \end{equation}
\end{lemma}
\begin{proof}
    Note that by Lemma~\ref{lem-eum-trees-condition-1,2,4}, it suffices to show that
    \begin{align}
        \#\big\{ \mathbf T \in \mathcal R_N: |\mathsf{NL}(\mathbf T)| \leq 0.0001(N+1) \big\} \leq e^{-0.11N} \cdot |\mathcal R_N| \,. \label{eq-goal-lem-B.5-relax}
    \end{align}
    Consider the following two-step procedure (denoted by $\mathtt{Compress}$) on a tree $\mathbf T_{\text{initial}} \in \mathcal R^{(1,2,4)}_{N+1}$.

    \textbf{Step 1: Shrinking long paths and labeling edges.} We search an arbitrary vertex set $\{ \mathbf u,\mathbf v \} \subset V(\mathbf T) $ (if exists) with maximum $\mathsf{Dist}_{\mathbf T}(\mathbf u,\mathbf v)$ such that $\mathfrak R(\mathbf T) \not\in V(\mathfrak p_{\mathbf T}(\mathbf u,\mathbf v))\setminus \{\mathbf u,\mathbf v\}$, $\mathsf{Deg}_{\mathbf T}(\mathbf w) = 2$ for all $\mathbf w \in V(\mathfrak p_{\mathbf T}(\mathbf u,\mathbf v)) \cup \{\mathbf u,\mathbf v\}$ and all edges in $\mathfrak p_{\mathbf T}(\mathbf u,\mathbf v)$ are not labeled. Then we renew $\mathbf T$ by replacing $\mathfrak p_{\mathbf T}(\mathbf u,\mathbf v)$ with a single edge $(\mathbf u,\mathbf v)$ with label $|E(\mathfrak p_{\mathbf T}(\mathbf u,\mathbf v))|$. We finish Step 1 if such a vertex set $\{\mathbf u,\mathbf v\}$ cannot be found. Denote by $\mathbf T=\mathbf T_{\text{step1}}$ the tree we obtain after Step 1.

    \textbf{Step 2: Deleting repeated subtrees and labeling vertices.} Suppose $\mathsf{Depth}(\mathbf T) = \mathsf D$, and we list $V(\mathbf T)$ as $\mathtt{List} = \{ \mathbf v_1 , \mathbf v_2 , \ldots , \mathbf v_{\mathtt x}\}$ such that $\mathsf{Dep}_{\mathbf T}(\mathbf v_i)$ is increasing in $i$. Next, we label $\mathfrak R(\mathbf T)$ by 1 and repeat the following process until all vertices in $\mathtt{List}$ are labeled: 
    \begin{itemize}
        \item[(2-1)] Select the first item in $\mathtt{List}$ without a label, denoted by $\mathbf u$;
        \item[(2-2)] Find all $\mathbf w \in \mathtt{List}$ which shares the same parent as $\mathbf{u}$ satisfying $\mathsf{Des}_{\mathbf T}(\mathbf w) = \mathsf{Des}_{\mathbf T}(\mathbf u)$, and denote them by $\mathbf w_1, \ldots,\mathbf w_\mathtt y$ ($\mathtt y=0$ is possible); 
        \item[(2-3)] Label $\mathbf u$ by $\mathtt y + 1$;
        \item[(2-4)] Update $\mathtt{List}$ by setting it to be $\mathtt{List} \setminus \cup_{1 \leq i \leq \mathtt y} V(\mathsf{Des}_{\mathbf T}(\mathbf w_i))$;
        \item[(2-5)] Update $\mathbf T$ by setting it to be $\mathbf T|_{\mathtt{List}}$.
    \end{itemize}

    Denote by $\mathbf T_{\text{final}} = \mathtt{Compress}(\mathbf T_{\text{initial}})$ the tree we obtain after Step 2. Then $\mathbf T_{\text{final}}$ is a labeled tree, where a label $\mathtt l(\mathbf e)$ on an edge $\mathbf e$ represents a self-avoiding path with $\mathtt l(\mathbf e)$ edges in $\mathbf T_{\text{initial}}$, and a label $\mathtt l(\mathbf v)$ on a vertex $\mathbf v$ represents $\mathtt l(\mathbf v)$ isomorphic subtrees sharing the parent of $\mathbf v$ in $\mathbf T_{\text{initial}}$. Therefore, our procedure $\mathtt{Compress}$ is an injection from $\mathcal R^{(1,2,4)}_N$ to the set of unlabeled rooted trees with positive integer labels on edges and vertices with no more than $N+1$ vertices. Also, by definition of our procedure we have $V(\mathbf T_{\text{final}}) \subset V(\mathbf T_{\text{initial}})$; by definition of our Step 2, $\{\mathbf T|_{V(\mathbf T)\setminus \{\mathbf v\}}:\mathbf v\in\mathcal L'(\mathbf T_{\text{final}})\}$ are distinct; by definition of our Step 1, for all $\mathbf v \in V(\mathbf T_{\text{final}}) \setminus \{ \mathfrak R(\mathbf T_{\text{initial}} )\}$ we have $\mathsf{Deg}_{\mathbf T_{\text{initial}}}(\mathbf v) \neq 2$. 

    Next, denote $\mathcal R^{\text{small}}_{N+1} = \{ \mathbf T \in \mathcal R^{(1,2,4)}_{N+1}: |\mathsf{NL}(\mathbf T)| \leq 0.0001(N+1)\}$ and 
    \begin{align}
        \widetilde{\mathcal R^{\text{small}}_{N+1}} = \Big\{ \widetilde{\mathbf T} = (\mathbf T; \{ \mathtt l(\mathbf e) \}_{\mathbf e \in E(\mathbf T)}, \{ \mathtt l(\mathbf v) \}_{\mathbf v \in V(\mathbf T)} ) : \widetilde{\mathbf T} = \mathtt{Compress}(\mathbf T')\mbox{ for some } \mathbf T'\in \mathcal R^{\text{small}}_{N+1} \Big\} \,.  \label{eq-def-tilde-R-small}
    \end{align}
    Suppose $\widetilde{\mathbf T} = (\mathbf T; \{ \mathtt l(\mathbf e) \}_{\mathbf e \in E(\mathbf T)}, \{ \mathtt l(\mathbf v) \}_{\mathbf v \in V(\mathbf T)} ) = \mathtt{Compress}(\mathbf T') \in \widetilde{\mathcal R^{\text{small}}_{N+1}}$. By definition of our procedure we have $\mathcal L'(\mathbf T) \subset \mathcal L'(\mathbf T')$ and 
    \begin{align*}
        \Big\{ \mathbf T'|_{V(\mathbf T') \setminus \{ \mathbf v \}} :\mathbf v \in \mathcal L'(\mathbf T)) \Big\} \mbox{ are distinct} \,,
    \end{align*}
    which yields that $|\mathcal L'(\mathbf T)| \leq |\mathsf{NL}(\mathbf T')|$. Thus, $\mathcal L'(\mathbf T)\leq 0.0001N$ if we assume $|\mathsf{NL}(\mathbf T')|\leq 0.0001N$. Next, by definition of $\{ \mathtt l(\mathbf e) \}_{\mathbf e \in E(\mathbf T)}$ we have 
    \begin{equation}\label{eq-sum-of-label-compressed-edges}
        \sum_{\mathbf e \in E(\mathbf T)}  \mathtt l(\mathbf e) \leq |E(\mathbf T')| = N \,.
    \end{equation}
    Now we control $|V(\mathbf T)|$. For all $\mathbf v \in V(\mathbf T)$ define
    \begin{align}\label{eq-def-DistL}
        \mathsf{DistL}_{\mathbf T}(\mathbf v) = \mathsf{Depth}( \mathsf{Des}_{\mathbf T}(\mathbf v) ) \,.
    \end{align}
    Since $\mathbf T$ is a tree, we have 
    \begin{align}
        &\#\big\{ \mathbf v \in V(\mathbf T): \mathsf{DistL}_{\mathbf T}(\mathbf v) = i \big\} \leq |\mathcal L'(\mathbf T)| \leq 0.0001N \,; \label{eq-control-near-leaf-vert} \\
        &\#\big\{ \mathbf v \in V(\mathbf T): \mathsf{Deg}_{\mathbf T}(\mathbf v) \geq 3 \big\} \leq |\mathcal L'(\mathbf T)| \leq 0.0001N \,.  \label{eq-control-deg>2-vert}
    \end{align}
    In addition, note that if $\mathsf{Deg}_{\mathbf T}(\mathbf v)=2$, then according to Step~1 we should have $\mathsf{Deg}_{\mathbf T'}(\mathbf v)\geq 3$, and thus $\mathtt l(\mathbf v) \geq 2$. Thus, since at least
    \begin{align*}
        \sum_{\mathbf v \in V(\mathbf T)} \mathsf{DistL}_{\mathbf T}(\mathbf v)(\mathtt l(\mathbf v)-1)
    \end{align*}
    vertices are deleted in Step~2 of our procedure $\mathtt{Compress}$, we have
    \begin{align*}
        N \geq \sum_{\mathbf v \in V(\mathbf T)} \mathsf{DistL}_{\mathbf T}(\mathbf v) (\mathtt l(\mathbf v)-1) \geq 100 \#\big\{ \mathbf v: \mathsf{Deg}_{\mathbf T}(\mathbf v)=2, \mathsf{DistL}_{\mathbf T}(\mathbf v) \geq 101 \big\} \,,  
    \end{align*}
    which leads to 
    \begin{align}
        \#\big\{ \mathbf v: \mathsf{Deg}_{\mathbf T}(\mathbf v)=2, (\mathsf{DistL}_{\mathbf T}(\mathbf v)) \geq 101 \big\} \leq 0.01N \,. \label{eq-control-far-leaf-vert}
    \end{align}
    Thus, we have
    \begin{align}
        |V(\mathbf T)| \leq\ & 1+|\mathcal L'(\mathbf T)|+ \#\big\{ \mathbf v \in V(\mathbf T): \mathsf{Deg}_{\mathbf T}(\mathbf v) \geq 3 \big\} + \#\big\{ \mathbf v \in V(\mathbf T): \mathsf{DistL}_{\mathbf T}(\mathbf v) \leq 100 \big\} \nonumber \\
        &+ \#\big\{ \mathbf v: \mathsf{Deg}_{\mathbf T}(\mathbf v)=2, (\mathsf{DistL}_{\mathbf T}(\mathbf v)) \geq 101 \big\} \nonumber \\
        \overset{ \eqref{eq-control-deg>2-vert}, \eqref{eq-control-near-leaf-vert}, \eqref{eq-control-far-leaf-vert} }{\leq} \ & 1+0.0001N + 0.0001N + 100 \cdot 0.0001N + 0.01N \leq 0.03N \,.  \label{eq-compressed-num-of-vert-upper-bound}
    \end{align}
    Finally, since at least $\sum_{\mathbf v \in V(\mathbf T)} \mathsf{DistL}_{\mathbf T}(\mathbf v)(\mathtt l(\mathbf v)-1)$ vertices is deleted in Step 2 of our procedure $\mathtt{Compress}$ we have 
    \begin{align}
        \sum_{\mathbf v \in V(\mathbf T)} \mathtt l(\mathbf v) &\leq |V(\mathbf T)| + |\mathcal L'(\mathbf T)| + \sum_{\mathbf v\in V(\mathbf T) , \mathsf{DistL}_{\mathbf T}(\mathbf v)>0}(\mathtt l(\mathbf v)-1)\nonumber \\
        &\leq 0.03N + 0.0001N + \sum_{\mathbf v\in V(\mathbf T) , \mathsf{DistL}_{\mathbf T}(\mathbf v)>0} \mathsf{DistL}_{\mathbf T}(\mathbf v)(\mathtt l(\mathbf v)-1) \nonumber \\
        &\leq 0.04N + |V(\mathbf T')|+1 \leq 1.05N\,.\label{eq-sum-of-label-compressed-vertices}
    \end{align}
    Combining \eqref{eq-compressed-num-of-vert-upper-bound}, \eqref{eq-sum-of-label-compressed-edges} and \eqref{eq-sum-of-label-compressed-vertices} we have
    \begin{align*}
        |\mathcal R^{\text{small}}_{N+1}| &= |\widetilde {\mathcal R^{\text{small}}_{N+1}}| \leq \sum_{1 \leq i \leq 0.03N} |\mathcal R_i| \cdot \#\{ (x_1 , \ldots ,x_k):x_i \in \mathbb Z_+,k\leq 0.0001N, \sum_{1\leq i \leq k}x_i \leq N\}  \\
        &\cdot \#\{ (x_1 , \ldots ,x_k):x_i \in \mathbb Z_+ ,k\leq 0.03N, \sum_{1\leq i \leq k}x_i \leq 1.05N\}  \\
        &\leq 2|\mathcal R_{\lfloor 0.03N\rfloor}| \cdot 100\tbinom{N}{0.0001N} \tbinom{1.05N}{0.03N} \leq \exp(-0.1N)|\mathcal R_{N+1}|  \leq e^{-0.11N} |\mathcal R_{N}| \,, 
    \end{align*}
    which yields \eqref{eq-goal-lem-B.5-relax}.
\end{proof}

Now we can prove the following lemma.

\begin{lemma}\label{lem-tree-ratio-1,2,4-upper}
    Recall \eqref{eq-choice-mathfrak-c} and \eqref{eq-def-mathcal-N-1,2,4}. We have $|\mathcal R^{(1,2,4)}_{N+1}| \leq \tfrac{\sqrt{\mathfrak c}}{2}|\mathcal R^{(1,2,4)}_{N}| \leq \mathfrak c |\mathcal N^{(1,2,4)}_N|$ for $N \geq 4$.
\end{lemma}
\begin{proof}
    Note that if the left-hand side inequality holds, as $\mathbf T\longrightarrow \mathbf T_+$ (recall \eqref{eq-def-T-+}) is an injection from $\mathcal R^{(1,2,4)}_{N-1} \setminus \{ \mathbf A_{N-1}\}$ to $\mathcal N^{(1,2,4)}_{N}$, which yields that 
    \begin{align*}
        |\mathcal N_{N}^{(1,2,4)}| \geq |\mathcal R_{N-1}^{(1,2,4)}|-1 \geq \tfrac{1}{2\sqrt{\mathfrak c}} |\mathcal R_N^{(1,2,4)}| \,.
    \end{align*}
    Thus, it suffices to show the left-hand side inequality $|\mathcal R^{(1,2,4)}_{N+1}| \leq \tfrac{\sqrt{\mathfrak c}}{2}|\mathcal R^{(1,2,4)}_{N}|$.
    Intuitively, if we remove a uniform random leaf of a tree in $\mathcal R^{(1,2,4)}_{N+1}$ we will typically obtain a tree in $\mathcal R^{(1,2,4)}_{N}$. To make this heuristic rigorous, for $\mathbf T\in \mathcal R^{(1,2,4)}_{N+1}$ and $\mathbf v \in \mathsf N(\mathbf T)$, define $\mathsf{Cut}(\mathbf v,\mathbf T)$ the tree in $\mathcal R_N$ obtained by deleting a leaf in $\mathcal L'(\mathbf T) \cap \mathsf{Nei}_{\mathbf T}(\mathbf v)$. For $N \geq 2$ and $\varrho \in \mathcal D_N$, define
    \begin{equation}\label{eq-def-mathfrak-E-N,varrho}
        \mathfrak{E}_{N,\varrho}= \Big\{ (\mathbf v,\mathbf T):\mathbf T\in \mathcal R^{(1,2,4)}_N \cap \mathsf{Type}_\varrho \,, \mathbf v\in \mathsf{N}(\mathbf T)\,, \mathsf{Cut}(\mathbf v,\mathbf T) \not\in \mathcal R^{(1,2,4)}_{N-1} \Big\} \,,
    \end{equation}
    and
    \begin{equation}\label{eq-def-kappa-N,varrho}
        \mathtt{p}_{N,\varrho} = \frac{|\mathfrak{E}_{N,\varrho} |}{(N-1)|\mathcal R^{(1,2,4)}_N \cap \mathsf{Type}_\varrho|} \,.
    \end{equation}
    We claim that
    \begin{equation}\label{eq-uniformly-cut-leaf-typically-remain-1,2,4}
        \mathtt{p}_N := \max_{\varrho \in \mathcal D_N}\mathtt{p}_{N,\varrho} \leq \log_{\frac{4}{3}}\big( N \wedge \mathtt L 3^{\mathtt L} \big) \cdot \exp(-0.1\mathtt L) + \sum_{j \geq 1, (4/3)^{j+1} 3^{\mathtt L}\leq  N} \exp(-0.1 (\mathtt L \vee (4/3)^j))\,.
    \end{equation}
    To show \eqref{eq-uniformly-cut-leaf-typically-remain-1,2,4}, we first note that \eqref{eq-uniformly-cut-leaf-typically-remain-1,2,4} holds for $N \leq \mathtt L$. Next, given any $N \geq \mathtt L + 1$ and $\varrho =\big((\nu_1 , \ldots , \nu_N), (\vartheta_1, \ldots , \vartheta_N)\big) \in \mathcal D_N$ such that \eqref{eq-uniformly-cut-leaf-typically-remain-1,2,4} holds for any integer less than $N$, we divide our analysis into the following three cases.

    \textbf{Case 1: $\sum_{i\geq 1}(\nu_{i} - \vartheta_i ) \leq 1$.} If $\sum_{i\geq 1}(\nu_{i} - \vartheta_i ) = 0$, then removing any leaf in $\mathcal L'(\mathbf T)$ for any $\mathbf T \in \mathcal R^{(1,2,4)}_N \cap \mathsf{Type}_\varrho$ yields a tree in $\mathcal R^{(1,2,4)}_{N-1}$, which gives $\mathtt p_{N,\varrho} = 0$. Else, we have $\sum_{i\geq 1}(\nu_{i} - \vartheta_i ) = 1$, and therefore there is a unique $j_0$ such that $\nu_{j_0} - \vartheta_{j_0} = 1$. For a tree $\mathbf T$ in this subcase, all but one tree in $\mathsf{Subs}(\mathbf T)$ are arm-trees. Therefore we have
    \begin{align}\label{eq-varkappa-bound-prelim-case}
        \mathtt p_{N,\varrho} \leq \frac{|\mathfrak{E}_{N,\varrho}|}{(N-1)|\mathcal R_{j_0}^{(1,2,4)}|} = \frac{\sum_{\varrho' \in \mathcal D_{j_0}} |\mathfrak{E}_{j_0,\varrho'}|}{(N-1) \sum_{\varrho' \in \mathcal D_{j_0}} |\mathcal R_{j_0}^{(1,2,4)} \cap\mathsf{Type}_{\varrho'}|} \leq \max_{\varrho' \in \mathcal D_{j_0}} \mathtt p_{j_0,\varrho'} = \mathtt p_{j_0}\,.
    \end{align}
    Therefore, in Case 1 we always have
    \begin{align}\label{eq-varkappa-bound-simplest}
    \mathtt p_{N,\varrho}  \leq \max_{j_0 < N} \mathtt p_{j_0}\,.
    \end{align}

    \textbf{Case 2: $\sum_{i\geq 1}(\nu_{i} - \vartheta_i ) \geq 2$, and there exists $A \geq 1$ such that $\sum_{A+1\leq i \leq 3A} (\nu_i - \vartheta_i ) = 0$, $\sum_{i > 3A} (\nu_i - \vartheta_i ) > 0$ and $\sum_{i \leq A} (\nu_i - \vartheta_i ) > 0$.} Denote $N_1 =\Sigma_{\varrho , \leq A}$ and $N_2 = \Sigma_{\varrho , >A}$ for simplicity. Note that in this case it must hold that $\nu_i=\vartheta_i$ for all $A+1 \leq i \leq 3A$. Recall \eqref{eq-def-Sigmas} -- \eqref{eq-def-split-A-T} and recall \eqref{eq-def-mathfrak-E-N,varrho}. Define
    \begin{align}
        \mathfrak E^{\leq A}_{N, \varrho} = \Big\{ (\mathbf v,\mathsf{Concat}(\mathbf T_1,\mathbf T_2)): (\mathbf v,\mathbf T_1) \in \mathfrak E_{N_1,\varrho_{\leq A}} \,, \mathbf T_2 \in R^{(1,2,4)}_{N_2} \cap \mathsf{Type}_{\varrho_{>A}} \Big\}\,,\label{eq-def-mathfrak-E-leq-A}\\
        \mathfrak E^{>A}_{N,\varrho} = \Big\{ (\mathbf v,\mathsf{Concat}(\mathbf T_1,\mathbf T_2)):(\mathbf v,\mathbf T_2) \in \mathfrak E_{N_2,\varrho_{>A}}\,, \mathbf T_1 \in R^{(1,2,4)}_{N_1} \cap \mathsf{Type}_{\varrho_{\leq A}}  \Big\}\,. \label{eq-def-mathfrak-E->A}
    \end{align}
    We shall show
    \begin{align}
        \mathfrak E_{N, \varrho} \subset \mathfrak E^{\leq A}_{N, \varrho} \cup \mathfrak E^{>A}_{N, \varrho}\,.\label{eq-split-cut-equals-cut-split}
    \end{align}
    To show \eqref{eq-split-cut-equals-cut-split}, note that two similar trees with $A_1,A_2$ vertices always satisfy $\tfrac{1}{2} \leq\tfrac{A_1}{A_2} \leq 2$ by Definition~\ref{def-similar-relation}. Therefore, suppose $(\mathbf v,\mathbf T) \in\mathfrak E_{N, \varrho}$, if $\mathbf v \in\mathbf T_{\leq A}$ and $\mathsf{Cut}(\mathbf v,\mathbf T_{\leq A}) \in \mathcal R^{(1,2,4)}_{N_1-1,\,\varrho_{\leq A}}$, then by assumption $(\mathbf v,\mathbf T) \in\mathfrak E_{N, \varrho}$ we know that $\mathsf{Concat}(\mathsf{Cut}(\mathbf v,\mathbf T_{\leq A}),\mathbf T_{>A}) = \mathsf{Cut}(\mathbf v,\mathbf T) \not\in \mathcal R^{(1,2,4)}_{N-1}$. This implies that there exist $\mathbf T_1 \in \mathsf{Subs}(\mathsf{Cut}(\mathbf v,\mathbf T_{\leq A}))$ and $\mathbf T_2 \in \mathsf{Subs}(\mathbf T_{>A})$ such that $\mathbf v \in V(\mathbf T_1)$ and $\mathbf T_1 \sim \mathbf T_2$. Since cutting the unique leaf of an arm tree always generates an arm tree, which cannot be similar to any unlabeled rooted tree, we know that $\mathbf{T}_1,\mathbf{T}_2$ are not arm trees. Thus, we get
    \begin{align*}
        \nu_{|V(\mathbf T_1)|+1}-\vartheta_{|V(\mathbf T_1)|+1} &> 0\,, \nu_{|V(\mathbf T_2)|}-\vartheta_{|V(\mathbf T_2)|} >0\,; \\
        |V(\mathbf T_1)|+1 &\leq A\,, |V(\mathbf T_2)| > 3A\,.
    \end{align*}
    This immediately leads to the contradiction $\tfrac{1}{2} \leq\tfrac{A}{3A} \leq 2$. Therefore, if $v \in T_{\leq A}$ we always have $\mathsf{Cut}(v,T_{\leq A}) \not\in \mathcal R^{(1,2,4)}_{N_1-1,\,\varrho_{\leq A}}$, which yields $(v,T) \in \mathfrak E^{\leq A}_{N, \varrho}$. In a similar manner, we can show that if $v \in T_{> A}$ then we always have $(v,T) \in \mathfrak E^{> A}_{N, \varrho}$. Therefore, \eqref{eq-split-cut-equals-cut-split} holds. By induction hypothesis this implies
    \begin{align}
        |\mathfrak E_{N, \varrho}| &\leq |\mathfrak E^{\leq A}_{N, \varrho}| + |\mathfrak E^{> A}_{N, \varrho}| = |\mathfrak E_{N_1,\, \varrho_{\leq A}}|\cdot |\mathcal R^{(1,2,4)}_{N_2} \cap \mathsf{Type}_{\varrho_{>A}}| + |\mathcal R^{(1,2,4)}_{N_1} \cap \mathsf{Type}_{\varrho_{\leq A}}| \cdot |\mathfrak E_{N_2,\, \varrho_{> A}}|\nonumber \\
        &\leq \big((N_1 - 1)\mathtt{p}_{N_1} + (N_2 - 1)\mathtt{p}_{N_2}\big)\cdot|\mathcal R^{(1,2,4)}_{N_1} \cap \mathsf{Type}_{\varrho_{\leq A}}|\cdot|\mathcal R^{(1,2,4)}_{N_2} \cap \mathsf{Type}_{\varrho_{>A}}|\nonumber\\
        &\leq (N-1)|\mathcal R^{(1,2,4)}_{N} \cap \mathsf{Type}_{\varrho}|(\mathtt{p}_{N_1} \vee \mathtt{p}_{N_2})\,. \label{eq-varkappa-bound-case-1}
    \end{align}
    In order to see the last inequality, we see that no two trees $\mathbf T_1,\mathbf T_2$ with $(\mathbf T_1,\mathbf T_2) \in (\mathcal R^{(1,2,4)}_{N_1} \cap \mathsf{Type}_{\varrho_{\leq A}}) \times (\mathcal R^{(1,2,4)}_{N_2} \cap \mathsf{Type}_{\varrho_{>A}})$ can be similar, and thus
    \begin{align*}
        \mathcal R^{(1,2,4)}_{N} \cap \mathsf{Type}_{\varrho}=\Big\{ \mathsf{Concat}(\mathbf T_1,\mathbf T_2):(\mathbf T_1, \mathbf T_2) \in\mathcal R^{(1,2,4)}_{N_1} \cap \mathsf{Type}_{\varrho_{\leq A}}\times\mathcal R^{(1,2,4)}_{N_2} \cap \mathsf{Type}_{\varrho_{>A}} \Big\}\,,
    \end{align*}
    which implies that
    \begin{equation*}
        |\mathcal R^{(1,2,4)}_{N} \cap \mathsf{Type}_{\varrho}|=|\mathcal R^{(1,2,4)}_{N_1} \cap \mathsf{Type}_{\varrho_{\leq A}}| \cdot|\mathcal R^{(1,2,4)}_{N_2} \cap \mathsf{Type}_{\varrho_{>A}}| \,.
    \end{equation*}

    \textbf{Case 3: $\sum_{i\geq 1}(\nu_{i} - \vartheta_i ) \geq 2$, but there does not exist $A \geq 1$ such that $\sum_{A+1\leq i \leq 3A} (\nu_i - \vartheta_i ) = 0$, $\sum_{i > 3A} (\nu_i - \vartheta_i ) > 0$ and $\sum_{i \leq A} (\nu_i - \vartheta_i ) > 0$.} Recall the definition of $\mathsf{Subs}^*(\mathbf T)$ in \eqref{def-sub-except-arm-on-rooted-trees}. Note that in this case we always have $|\mathsf{Subs}^*(\mathbf T)| \geq 2$ for any $\mathbf T \in \mathcal R^{(1,2,4)}_N \cap \mathsf{Type}_\varrho$. Therefore, we have 
    \begin{align*}
        (N-1) \geq \sum_{\mathbf T_s \in \mathsf{Subs}^*(\mathbf T)}|V(\mathbf T_s)| \geq (1+\tfrac{1}{3})\max\big\{ |V(\mathbf T_s)|:\mathbf T_s \in \mathsf{Subs}^*(\mathbf T) \big\}\,,
    \end{align*}
    and therefore $\max\{|V(\mathbf T_s)|:\mathbf T_s \in \mathsf{Subs}^*(\mathbf T) \} \leq \tfrac{3}{4}(N-1)$. Also, for any $\mathbf T \in \mathcal R^{(1,2,4)}_N \cap \mathsf{Type}_\varrho$ we have $\tfrac{\max\{|V(\mathbf T_s)|:\mathbf T_s \in \mathsf{Subs}^*(\mathbf T)\}}{\min\{|V(\mathbf T_s)|: \mathbf T_s \in \mathsf{Subs}^*(\mathbf T)\}} \leq 3^{\mathtt L}$ by the fact that $|\mathsf{Subs}^*(\mathbf T)| \leq \mathtt L$ from Item (1) of Definition~\ref{def-tilde-T-K}. Denote $\mathtt{Max}_{\varrho} = \max\{|V(\mathbf T_s)|: \mathbf T_s \in \mathsf{Subs}^*(\mathbf T)\}$ (which is determined only by $\varrho$). Next, note that cutting a leaf on an arm subtree yields another arm subtree with shorter arm length, and any arm tree is not similar to any unlabeled rooted tree. Therefore, for $\mathbf T\in \mathcal R^{(1,2,4)}_N$, $\mathbf T' \in \mathsf{ASubs}(\mathbf T)$ and $v \in \mathsf{N}(\mathbf T')$ we always have $\mathsf{Cut}(\mathbf v,\mathbf T) \in \mathcal R^{(1,2,4)}_{N-1}$. Thus,
    \begin{align}
        & |\mathfrak{E}_{N,\varrho}| = \sum_{\substack{ \mathbf T\in \mathcal R^{(1,2,4)}_N \cap \mathsf{Type}_\varrho\\ \mathbf T' \in \mathsf{Subs}^*(\mathbf T)}} \#\Big\{ \mathbf v:\mathbf v\in \mathsf{N}(\mathbf T')\,, \mathsf{Cut}(\mathbf v,\mathbf T) \not\in \mathcal R^{(1,2,4)}_{N-1} \Big\} \nonumber\\
        \leq\ & \#\Bigg( \bigcup_{\substack{ \mathbf T\in \mathcal R^{(1,2,4)}_N \cap \mathsf{Type}_\varrho\\ \mathbf T' \in \mathsf{Subs}^*(\mathbf T)}} \Big\{ (\mathbf v,\mathbf T):\mathbf v\in \mathsf{N}(\mathbf T')\,,\mathsf{Cut}(\mathbf v,\mathbf T') \in \mathcal R^{(1,2,4)}_{|V(\mathbf T')|-1}\,, \mathsf{Cut}(\mathbf v,\mathbf T) \not\in \mathcal R^{(1,2,4)}_{N-1} \Big\} \Bigg) \label{eq-cut-decomposition-part-1} \\
        &+ \sum_{\substack{ \mathbf T\in \mathcal R^{(1,2,4)}_N \cap \mathsf{Type}_\varrho\\ \mathbf T' \in \mathsf{Subs}^*(\mathbf T)}} \#\Big\{ \mathbf v:\mathbf v\in \mathsf{N}(\mathbf T')\,,\mathsf{Cut}(\mathbf v,\mathbf T') \not\in \mathcal R^{(1,2,4)}_{|V(\mathbf T')|-1} \Big\}\label{eq-cut-decomposition-part-2}\,,
    \end{align}
    We first estimate \eqref{eq-cut-decomposition-part-2}, which is bounded by (recall \eqref{eq-def-eta-starstar-i})
    \begin{align}
        & |\mathcal R^{(1,2,4)}_N \cap \mathsf{Type}_\varrho|\cdot \sum_{\substack{1\leq i \leq N}} (\nu_i - \vartheta_i) \E_{\mathbf T'\sim\eta^{**}_{N,\varrho;i}}\Big[ \#\big\{ \mathbf v:\mathbf v\in \mathsf{N}(\mathbf T')\,,\mathsf{Cut}(\mathbf v,\mathbf T') \not\in \mathcal R^{(1,2,4)}_{|V(\mathbf T')|-1} \big\} \Big]   \nonumber \\
        \leq\ & |\mathcal R^{(1,2,4)}_N \cap \mathsf{Type}_\varrho|\cdot \sum_{\substack{[3^{-\mathtt L}\mathtt{Max}_{\varrho}]\vee \mathtt L\leq i \leq \mathtt{Max}_{\varrho}}} (\nu_i - \vartheta_i) \E_{\mathbf T'\sim\eta^{*}_{i}}\Big[ \#\big\{ \mathbf v:\mathbf v\in \mathsf{N}(\mathbf T')\,,\mathsf{Cut}(\mathbf v,\mathbf T') \not\in \mathcal R^{(1,2,4)}_{|V(\mathbf T')|-1} \big\} \Big] \nonumber \\
        &+ |\mathcal R^{(1,2,4)}_N \cap \mathsf{Type}_\varrho|\cdot \sum_{\substack{[3^{-\mathtt L }\mathtt{Max}_{\varrho}]\vee \mathtt L\leq i \leq \mathtt{Max}_{\varrho}}} (\nu_i - \vartheta_i)(i-1)\mathsf{TV}(\eta^*_i , \eta^{**}_{N,\varrho;i}) \nonumber \\  
        \leq\ & |\mathcal R^{(1,2,4)}_N \cap \mathsf{Type}_\varrho|\cdot \sum_{\substack{[3^{-\mathtt L }\mathtt{Max}_{\varrho}]\vee \mathtt L\leq i \leq \mathtt{Max}_{\varrho}}} (\nu_i - \vartheta_i)(i-1) (\mathsf{TV}(\eta^*_i , \eta^{**}_{N,\varrho;i}) + \mathtt{p}_i) \nonumber \\ 
        \leq\ & (N-1)|\mathcal R^{(1,2,4)}_N \cap \mathsf{Type}_\varrho|\cdot \Big(0.7^{[3^{-\mathtt L }\mathtt{Max}_{\varrho}]\vee \mathtt L} + \max_{[3^{-\mathtt L }\mathtt{Max}_{\varrho}]\vee \mathtt L\leq i \leq \mathtt{Max}_{\varrho}} \mathtt{p}_i \Big)\,,\label{eq-cut-part-2-analysis}
    \end{align}
    where the last inequality holds by Lemma~\ref{lem-tree-N-N+1-ratio-lowers} and the fact that $\eta^*_i =\eta^{**}_i $ for $i \geq \mathtt L$. We then estimate \eqref{eq-cut-decomposition-part-1}, and for convenience we denote \eqref{eq-cut-decomposition-part-1} as $\# \mathfrak{E}^{\square}_{N,\varrho}$. For $i\geq \mathtt L$ such that $\nu_i-\vartheta_i> 0$ we denote
    \begin{align*}
        \varrho_i' &= \big((\nu_j - \mathbf{1}_{\{j = i\}}+\mathbf{1}_{\{j = i-1\}})_{1\leq j \leq N-1}\,, (\vartheta_j)_{1\leq j \leq N-1}\big):=\big((\nu'_{i;j})_{1\leq j \leq N-1}\,, (\vartheta_{j})_{1\leq j \leq N-1}\big) \,; \\
        \mathfrak{T}^{\square}_{N,i} &=\{ \mathbf T \in \mathsf{Type}_{\varrho'_i} \setminus \mathcal R^{(1,2,4)}_{N-1}:\mathbf T' \in \mathcal R^{(1,2,4)}_{|V(\mathbf T')|} \mbox{ for all } \mathbf T' \in \mathsf{Subs}(\mathbf T)\} \,.
    \end{align*}
    Then each tree in $\mathfrak{E}^{\square}_{N,\varrho}$ can be obtained by adding a leaf to a subtree with $i-1$ vertices of a tree in $\mathfrak{T}^{\square}_{N,i}$ for some $i \geq \mathtt L$ with $\nu_i - \vartheta_i> 0$. By the fact that an arm path cannot be similar to any tree, \eqref{eq-cut-decomposition-part-1} can be bounded by
    \begin{align}
        \sum_{\substack{1\leq i \leq N}}&\Big((\nu_{i-1} - \vartheta_{i-1} + 1)(i-1)\#\mathfrak{T}^{\square}_{N,i}\Big)\,. \label{eq-cut-part-1-equivalent}
    \end{align}
    Consider any $i \geq \mathtt L$ such that $\nu_i-\vartheta_i > 0$. Note that there are at least (recall the definition of $H(\iota,j)$ \eqref{eq-def-H-iota-mathtt-i})
    \begin{equation*}
        \prod_{1 \leq j < \mathtt L} \binom{|\mathcal R^{(1,2,4)}_j| + \nu_{i;j}' - \vartheta_j - 2 }{\nu_{i;j}' - \vartheta_j} \cdot \prod_{j \geq\mathtt L} \binom{|\mathcal R^{(1,2,4)}_j| + \nu_{i;j}' - \vartheta_j - 1 - H(\iota,j)}{\nu_{i;j}' - \vartheta_j}
    \end{equation*}
    ways of choosing $T \in \mathsf{Type}_{\varrho'_i} \cap \mathcal R^{(1,2,4)}_{N-1}$. Moreover, there are at most 
    \begin{equation}
        \Big( \prod_{1 \leq j < \mathtt L} \tbinom{|\mathcal R^{(1,2,4)}_j| + \nu_{i;j}' - \vartheta_j - 2}{\nu_{i;j}' - \vartheta_j}   \Big)
        \Big( \prod_{j \geq \mathtt L} \tbinom{|\mathcal R^{(1,2,4)}_j| + \nu_{i;j}' - \vartheta_j - 1}{\nu_{i;j}' - \vartheta_j} \Big)
    \end{equation}
    to enumerate $\mathsf{Type}_{\varrho_i'}$. Therefore, we have
    \begin{align}
        \#\mathfrak{T}^{\square}_{N,i} &\leq \Big( \prod_{1 \leq j < \mathtt L} \tbinom{|\mathcal R^{(1,2,4)}_j| + \nu_{i;j}' - \vartheta_j - 2}{\nu_{i;j}' - \vartheta_j} \Big)\Big( \prod_{j \geq \mathtt L} \tbinom{|\mathcal R^{(1,2,4)}_j| + \nu_{i;j}' - \vartheta_j - 1}{\nu_{i;j}' - \vartheta_j} \Big) - |\mathsf{Type}_{\varrho'_i} \cap \mathcal R^{(1,2,4)}_{N-1}|\nonumber \\
        &\leq |\mathsf{Type}_{\varrho'_i} \cap \mathcal R^{(1,2,4)}_{N-1}|\cdot\Big( \prod_{j \geq \mathtt L}(1-\tfrac{H(\iota ,j)}{|\mathcal R^{(1,2,4)}_j|})^{-\nu_{i;j}' + \vartheta_j} -1\Big)\nonumber \\
        &\leq |\mathsf{Type}_{\varrho'_i} \cap \mathcal R^{(1,2,4)}_{N-1}| \cdot (\exp(0.68^{[3^{-\mathtt L }\mathtt{Max}_{\varrho}]\vee \mathtt L}) - 1)\,,\label{eq-cut-part-1-analysis-1}
    \end{align}
    where the third inequality holds by \eqref{eq-estimate-loss-of-condition-4-general}. Also, note that
    \begin{align}
        \frac{|\mathsf{Type}_{\varrho'_i} \cap \mathcal R^{(1,2,4)}_{N-1}|}{|\mathsf{Type}_{\varrho} \cap \mathcal R^{(1,2,4)}_{N}|} &\leq \frac{\prod_{1 \leq j < \mathtt L} \tbinom{|\mathcal R^{(1,2,4)}_j| + \nu_{i;j}' - \vartheta_j - 2}{\nu_{i;j}' - \vartheta_j} \cdot \prod_{j \geq \mathtt L} \tbinom{|\mathcal R^{(1,2,4)}_j| + \nu_{i;j}' - \vartheta_j - 1}{\nu_{i;j}' - \vartheta_j} }{\prod_{1 \leq j < \mathtt L} \tbinom{|\mathcal R^{(1,2,4)}_j| + \nu_{j} - \vartheta_j- 2}{\nu_{j} - \vartheta_j} \cdot \prod_{j \geq \mathtt L} \tbinom{|\mathcal R^{(1,2,4)}_j| + \nu_{j} - \vartheta_j - H(\iota,j)- 1}{\nu_{j} - \vartheta_j} }\nonumber \\
        &\leq \tfrac{(\nu_i - \vartheta_i )}{(\nu_{i-1} - \vartheta_{i-1}+1)} \cdot \tfrac{|\mathcal R^{(1,2,4)}_{i-1}| + \nu_{i-1} - \vartheta_{i-1} }{|\mathcal R^{(1,2,4)}_i| + \nu_{i} - \vartheta_i - 1} \cdot \prod_{j \geq \mathtt L} (1-\tfrac{H(\iota , j)}{|\mathcal R^{(1,2,4)}_j| })^{\vartheta_j - \nu_j}\nonumber \\
        &\leq \tfrac{(\nu_i - \vartheta_i )}{(\nu_{i-1} - \vartheta_{i-1}+1)} \cdot\tfrac{|\mathcal R^{(1,2,4)}_{i-1}| + \mathtt L}{|\mathcal R^{(1,2,4)}_i| - 1} \cdot \exp(0.68^{[3^{-\mathtt L }\mathtt{Max}_{\varrho}]\vee \mathtt L})\leq \tfrac{(\nu_i - \vartheta_i )}{(\nu_{i-1} - \vartheta_{i-1}+1)} \,,\label{eq-cut-part-1-analysis-2}
    \end{align}
    where in the second to last inequality we used $|\mathcal R^{(1,2,4)}_{i-1}|\leq |\mathcal R^{(1,2,4)}_{i}|$ which follows from Lemma~\ref{lem-prelim-tree-ratio}. Combining \eqref{eq-cut-part-1-analysis-1} and \eqref{eq-cut-part-1-analysis-2} we obtain that
    \begin{align}
        \eqref{eq-cut-part-1-equivalent} &\leq (\exp(0.68^{[3^{-\mathtt L }\mathtt{Max}_{\varrho}]\vee \mathtt L}) - 1)\cdot\sum_{i\geq \mathtt L , \nu_i - \vartheta_i > 0}\big(|\mathsf{Type}_{\varrho} \cap \mathcal R^{(1,2,4)}_{N}| \cdot (\nu_i - \vartheta_i)(i-1)\big) \nonumber \\
        &\leq 0.7^{[3^{-\mathtt L}\mathtt{Max}_{\varrho}]\vee\mathtt L}\cdot (N-1) |\mathsf{Type}_{\varrho} \cap \mathcal R^{(1,2,4)}_{N}|\,.\label{eq-cut-part-1-final}
    \end{align}
    Therefore, combining \eqref{eq-cut-part-1-final} and \eqref{eq-cut-part-2-analysis}, in Case 3 we have
    \begin{align}
        \mathtt{p}_{N, \varrho } &\leq 0.72^{[3^{-\mathtt L} \mathtt{Max}_{\varrho}]\vee \mathtt L} + \max_{ i \leq \mathtt{Max}_{\varrho}} \mathtt{p}_i \nonumber  \\
        &\leq \exp(-0.2(\mathtt L \vee \tfrac{\mathtt{Max}_{\varrho}}{ 3^{\mathtt L}})) + \log_{\frac{4}{3}}\big( \mathtt{Max}_{\varrho} \wedge \mathtt L 3^{\mathtt L} \big) e^{-0.1\mathtt L} + \sum_{j \geq 1, (4/3)^{j+1} 3^{\mathtt L}\leq  \mathtt{Max}_{\varrho}} \exp(-0.1 (\mathtt L \vee (4/3)^j))\nonumber\\
        &\leq  \log_{\frac{4}{3}}\big( N \wedge \mathtt L 3^{\mathtt L} \big) e^{-0.1\mathtt L} + \sum_{j \geq 1, (4/3)^{j+1} 3^{\mathtt L}\leq  N} \exp(-0.1 (\mathtt L \vee (4/3)^j)) \,, \label{eq-varkappa-bound-case-2}
    \end{align}
    where in the last inequality we use the fact that
    \begin{align*}
        \varphi(t) + e^{ -0.2(\mathtt L \vee \tfrac{t}{ 3^{\mathtt L}}) } \leq \varphi( \tfrac{4t}{3} ) \mbox{ where } \varphi(t) = \log_{\frac{4}{3}}\big( t \wedge \mathtt L 3^{\mathtt L} \big) e^{-0.1\mathtt L} + \sum_{ \substack{ j \geq 1, (4/3)^{j+1} 3^{\mathtt L}\leq t } } e^{-0.1 (\mathtt L \vee (4/3)^j)} \,.
    \end{align*}
    In conclusion, \eqref{eq-uniformly-cut-leaf-typically-remain-1,2,4} always holds combining \eqref{eq-varkappa-bound-simplest}, \eqref{eq-varkappa-bound-case-1} and \eqref{eq-varkappa-bound-case-2}. Thus, combining \eqref{eq-high-probability-NL-large} and \eqref{eq-uniformly-cut-leaf-typically-remain-1,2,4} we have
    \begin{align*}
        N|\mathcal R^{(1,2,4)}_N| &\geq \#\Big\{(\mathbf v,\mathbf T):\mathbf T \in \mathcal R^{(1,2,4)}_N \,, \mathbf v\in V(\mathbf T)\,, \mathbf T\oplus \mathcal L^{(1)}_{\mathbf v} \in \mathcal R^{(1,2,4)}_{N+1} \Big\}\\
        &= \#\Big\{(\mathbf v,\mathbf T):\mathbf T \in \mathcal R^{(1,2,4)}_{N+1} \,, \mathbf v\in \mathsf{N}(\mathbf T)\,, \mathsf{Cut}(\mathbf v,\mathbf T)\in \mathcal R^{(1,2,4)}_{N} \Big\}\\
        &\geq \sum_{\mathbf T \in \mathcal R^{(1,2,4)}_{N+1}: |\mathsf{N}(\mathbf T)| \geq 0.0001N} \Big( |\mathsf{N}(T)| - \#\{ \mathbf v\in \mathsf{N}(\mathbf T):\, \mathsf{Cut}(\mathbf v,\mathbf T)\not\in \mathcal R^{(1,2,4)}_{N}\} \Big) \\
        &\geq \sum_{\mathbf T \in \mathcal R^{(1,2,4)}_{N+1}:|\mathsf{N}(\mathbf T)| \geq 0.0001N} \Big( 0.0001N - \#\big\{\mathbf v\in \mathsf{N}(\mathbf T):\, \mathsf{Cut}(\mathbf v,\mathbf T)\not\in \mathcal R^{(1,2,4)}_{N} \big\} \Big) \\
        &\geq 0.0001N \cdot \eta^*_{N+1}\big[ |\mathsf{N}(\mathbf T)|\geq 0.0001N \big] \cdot |\mathcal R^{(1,2,4)}_{N+1}| - \mathtt{p}_N \cdot N | \mathcal R^{(1,2,4)}_{N+1}|\\
        &\overset{\eqref{eq-high-probability-NL-large},\eqref{eq-uniformly-cut-leaf-typically-remain-1,2,4}}{\geq} 0.0001N\cdot(1-\exp(-0.05N) - 10\mathtt L^2 0.8^{\mathtt L})|\mathcal R^{(1,2,4)}_{N+1}| \geq 2 \cdot 10^{-5}N|\mathcal R^{(1,2,4)}_{N+1}|\,,
    \end{align*}
    which yields our desired result that $|\mathcal R^{(1,2,4)}_{N+1}| \leq \tfrac{\sqrt{\mathfrak c}}{2}|\mathcal R^{(1,2,4)}_N|$ according to \eqref{eq-choice-mathfrak-c} and therefore completes the proof.
\end{proof}

Recall \eqref{eq-def-mathbf-F-T}. Before finally proving \eqref{eq-thm-2.7-key-expectation}, we introduce a lemma that analyzes the increment of $\mathbf F(\mathbf T)$ when we perform induction argument on layers of $\mathbf{T}$ (where we always denote $\varrho(\mathbf T)$ by $( (\vartheta_1,\ldots,\vartheta_N), (\nu_1,\ldots,\nu_N))$). Recall \eqref{eq-choice-mathtt-m}, \eqref{eq-def-mathfrak-m-T}, \eqref{eq-def-mathfrak-v-T}, \eqref{def-chained-noise-pair-in-T_iota} and \eqref{def-chained-signal-pair-in-T_iota}.
\begin{lemma}\label{lem-induction-estimates-on-rooted-trees-function}
    For $N\geq\mathtt L^2+\mathtt m+3$, define (the $\{\vartheta_x\},\{\nu_x\}$ below depend on $\mathbf{T}$)
    \begin{align}
        & \mathbf F_{1;i,j}(\mathbf T) = \mathbf{1}_{\{\mathfrak m(\mathbf T_{i,j})=\mathtt m-1,\mathsf{Branch}_{\mathbf T_{i,j}}(\mathfrak v(\mathbf T_{i,j}))=2,\mathfrak k(\mathfrak v(\mathbf T_{i,j})) = 0, \sum \nu_x = 2 , \sum \vartheta_x = 0, i \geq \mathtt L^2+\mathtt m-1 \}} \,,  \label{eq-def-mathbf-F-1} \\
        &\mathbf F_{2;i,j}(\mathbf T) = \sum_{0 \leq a,b \leq 1} \mathbf{1}_{\{\mathfrak m(\mathbf T_{i,j}) = \mathtt m-1, \mathfrak k(\mathfrak v(\mathbf T_{i,j})) = a, \sum \vartheta_x = b, \sum_x \nu_x \leq 3, i \geq \mathtt L^2+\mathtt m-1 \}} \,. \label{eq-def-mathbf-F-2}
    \end{align}
    Then we have the following estimates:
    \begin{itemize}
        \item[(1)] $\Eb_{\eta^*_N}\big[ \sum_{1 \leq i \leq N, 1 \leq j \leq \nu_i } \mathbf F_{1;i,j}(\mathbf T) \big] \geq \mathfrak c^{-8-\mathtt m}$;
        \item[(2)] $\Eb_{\eta^*_N}\big[ \sum_{1 \leq i \leq N, 1 \leq j \leq \nu_i } \mathbf F_{2;i,j}(\mathbf T) \big] \leq \mathfrak c^{8} \cdot \Eb_{\eta^*_N}\big[ \sum_{1 \leq i \leq N, 1 \leq j \leq \nu_i } \mathbf F_{1;i,j}(\mathbf T) \big]$.
    \end{itemize}
\end{lemma}

\begin{proof}
As for Item~(1), it suffices to prove (recall \eqref{def-Enum-p-etastar-rho})
\begin{equation}\label{eq-lem-induction-estimates-(1)-expansion}
    \sum_{\substack{\varrho \in \mathcal D_N \\ \sum \nu_x = 2 \\ \sum \vartheta_x = 0}} p(\varrho) \sum_{\substack{\mathtt L^2 + \mathtt m - 1 \leq i \leq N\\1 \leq j \leq \nu_i}} \eta^*_{N,\varrho,i} \Big(\mathfrak m(\mathbf T_{i,j}) = \mathtt m -1 , \mathsf{Branch}_{\mathbf T_{i,j}}(\mathfrak v(\mathbf T_{i,j})) = 2, \mathfrak k(\mathfrak v(\mathbf T_{i,j})) = 0 \Big) \geq \mathfrak c^{-8-\mathtt m} \,.
\end{equation}
Define $\varrho^\diamond = \big( (0,\ldots,0), (\nu^\diamond_1 , \ldots , \nu^\diamond_N) \big)$ where $\nu^\diamond_i = \mathbf{1}_{\{ i=3\}} + \mathbf{1}_{\{ i=N-4\}}$. Noting that $\mathcal N^{(1,2,4)}_3$ is non-empty and using Lemma~\ref{lem-tree-ratio-1,2,4-upper}, we have
\begin{equation}\label{eq-lem-induction-estimates-(1)-p-diamond-varrho}
    p(\varrho^\diamond) \geq \frac{|\mathcal N^{(1,2,4)}_{N-4}|}{|\mathcal R^{(1,2,4)}_N|} \geq \mathfrak c^{-4}\,,
\end{equation}
Note that under $\rho^\diamond$, $\eta^*_{N,\varrho^\diamond;N-4}$ and $\eta^{**}_{N-4}$ are the same measure. Therefore, we have
\begin{align}\label{eq-lem-induction-estimates-(1)-mu-N-diamond-varrho}
    & \eta^*_{N,\varrho^\diamond;N-4} \Big(\mathfrak m(\mathbf T_{N-4,1}) = \mathtt m -1 , \mathsf{Branch}_{\mathbf T_{N-4,1}}(\mathfrak v(\mathbf T_{N-4,1})) = 2, \mathfrak k(\mathfrak v(\mathbf T_{N-4,1})) = 0 \Big) \nonumber \\
    =\ & \eta^{**}_{N-4} \Big(\mathfrak m(\mathbf T) = \mathtt m -1 , \mathsf{Branch}_{\mathbf T}(\mathfrak v(\mathbf T)) = 2, \mathfrak k(\mathfrak v(\mathbf{T}))=0 \Big) \, \nonumber \\
    =\ & \frac{\#\{\mathbf T \in \mathcal R^{(1,2,4)}_{N-4}: \mathfrak m(\mathbf T)=\mathtt m-1 , |\mathsf{Subs}(\mathsf{VCut}_{\mathtt m -1}(\mathbf T)) | = 2,\mathfrak k(\mathfrak v(\mathbf{T}))=0 \}}{|\mathcal R^{(1,2,4)}_{N-4}| - 1} \nonumber \\
    =\ & \frac{\#\{\mathbf T \in \mathcal R^{(1,2,4)}_{N-3-\mathtt m} : |\mathsf{Subs}(\mathbf T) | = 2 ,\mathfrak k(\mathfrak R(\mathbf{T}))=0\}}{|\mathcal R^{(1,2,4)}_{N-4}| - 1}
    \geq \mathfrak{c}^{-4} \cdot \frac{|\mathcal R^{(1,2,4)}_{N-3-\mathtt m}|}{|\mathcal R^{(1,2,4)}_{N-4}| - 1} \overset{\text{Lemma~\ref{lem-tree-ratio-1,2,4-upper}}}{\geq} \mathfrak c^{-4-\mathtt m}\,,
\end{align}
where the second to last inequality follows similarly from \eqref{eq-lem-induction-estimates-(1)-p-diamond-varrho}.
Combining \eqref{eq-lem-induction-estimates-(1)-p-diamond-varrho} and \eqref{eq-lem-induction-estimates-(1)-mu-N-diamond-varrho} yields \eqref{eq-lem-induction-estimates-(1)-expansion}. The proof of Item~(1) is complete.

As for Item~(2), it suffices to prove that for all $0 \leq a,b \leq 1$ we have
\begin{align}
    & \sum_{\substack{\varrho \in \mathcal D_N \\ \sum \vartheta_x = b \\ \sum \nu_x \leq 3}}  p(\varrho) \cdot \sum_{\substack{ \mathtt L^2+\mathtt m-1 \leq i \leq N\\ 1 \leq j \leq \nu_i }} \eta^*_{N,\varrho;i} \Big(\mathfrak m(\mathbf T_{i,j}) = \mathtt m -1 , \mathfrak k(\mathfrak v(\mathbf T_{i,j})) = a \Big) \label{eq-lem-induction-estimates-(2)-expansion} \\ 
    \leq\ & \mathfrak c^4\sum_{\substack{\varrho \in \mathcal D_N\\ \sum \nu_x = 2 \\ \sum \vartheta_x = 0}} p(\varrho) \sum_{\substack{\mathtt L^2 + \mathtt m - 1 \leq i \leq N\\1 \leq j \leq \nu_i}} \eta^*_{N,\varrho;i} \Big(\mathfrak m(\mathbf T_{i,j}) = \mathtt m -1 , \mathsf{Branch}_{\mathbf T_{i,j}}(\mathfrak v(\mathbf T_{i,j})) = 2, \mathfrak k(\mathfrak v(\mathbf T_{i,j})) = 0 \Big) \,. \label{eq-lem-induction-estimates-(2)-expansion-goal}
\end{align}
From Lemma~\ref{lem-tree-N-N+1-ratio-lowers}, we see that 
\begin{align*}
    \eqref{eq-lem-induction-estimates-(2)-expansion} &\leq \sum_{\substack{\varrho \in \mathcal D_N \\ \sum \vartheta_x = b \\ \sum \nu_x \leq 3}}  p(\varrho) \cdot \sum_{\substack{ \mathtt L^2+\mathtt m-1 \leq i \leq N\\ 1 \leq j \leq \nu_i }} \Big( \eta^*_{i} \big(\mathfrak m(\mathbf T_{i,j}) = \mathtt m -1 , \mathfrak k(\mathfrak v(\mathbf T_{i,j})) = a \big) + \operatorname{TV}(\eta^*_{i},\eta^*_{N,\varrho;i}) \Big) \\
    &\overset{\text{Lemma}~\ref{lem-tree-N-N+1-ratio-lowers}}{\leq} \sum_{\substack{\varrho \in \mathcal D_N \\ \sum \vartheta_x = b \\ \sum \nu_x \leq 3}}  p(\varrho) \cdot \sum_{\substack{ \mathtt L^2+\mathtt m-1 \leq i \leq N\\ 1 \leq j \leq \nu_i }} \Big( \eta^*_{i} \big(\mathfrak m(\mathbf T_{i,j}) = \mathtt m -1 , \mathfrak k(\mathfrak v(\mathbf T_{i,j})) = a \big) + 0.7^{i} \Big) \\
    &\leq 3\cdot 0.7^{\mathtt L+\mathtt m-1} + \sum_{\substack{\varrho \in \mathcal D_N \\ \sum \vartheta_x = b \\ \sum \nu_x \leq 3}}  p(\varrho) \cdot \sum_{\substack{ \mathtt L+\mathtt m-1 \leq i \leq N\\ 1 \leq j \leq \nu_i }} \eta^*_{i} \big(\mathfrak m(\mathbf T') = \mathtt m -1 , \mathfrak k(\mathfrak v(\mathbf T')) = a \big) \,.
\end{align*}
Similarly, we have
\begin{align}
    & \tfrac{1}{2}\eqref{eq-lem-induction-estimates-(2)-expansion-goal} \nonumber\\
    \geq & \tfrac{1}{2}\sum_{\substack{\varrho \in \mathcal D_N\\ \sum \nu_x = 2 \\ \sum \vartheta_x = 0}} p(\varrho) \sum_{\substack{\mathtt L^2 + \mathtt m - 1 \leq i \leq N\\1 \leq j \leq \nu_i}} \eta^*_{i} \Big(\mathfrak m(\mathbf T') = \mathtt m -1 , \mathsf{Branch}_{\mathbf T'}(\mathfrak v(\mathbf T')) = 2, \mathfrak k(\mathfrak v(\mathbf T')) = 0 \Big)-0.7^{\mathtt L+\mathtt m-1}\,.\label{eq-lem-induction-estimates-(2)-intermediate-(1)}
\end{align}
By Item~(1), we have that the first term in \eqref{eq-lem-induction-estimates-(2)-intermediate-(1)} is lower-bounded by $\frac{1}{2}\mathfrak c^{-8-\mathtt m}$. This implies
\begin{align}
    &\eqref{eq-lem-induction-estimates-(2)-expansion-goal} \geq \frac{1}{2}\eqref{eq-lem-induction-estimates-(2)-expansion-goal}+\frac{1}{2}\mathfrak c^{-8-\mathtt m} - 0.7^{\mathtt L+\mathtt m-1}\,.
\end{align}
Applying the lower bound in \eqref{eq-lem-induction-estimates-(2)-intermediate-(1)} again, we obtain
\begin{equation}
\begin{split}
    &\eqref{eq-lem-induction-estimates-(2)-expansion-goal}\geq  \text{ first term in }\eqref{eq-lem-induction-estimates-(2)-intermediate-(1)}+  \frac{1}{2}\mathfrak c^{-8-\mathtt m} -2\cdot 0.7^{\mathtt L+\mathtt m-1}\\
   & \geq \text{ first term in }\eqref{eq-lem-induction-estimates-(2)-intermediate-(1)} + 0.7^{\mathtt L+\mathtt m-1}\,.
\end{split}\notag
\end{equation}
Thus, it suffices to prove that given any $\mathtt L^2 + \mathtt m - 1 \leq i \leq N$ we have
\begin{align*}
    & \sum_{\substack{\varrho \in \mathcal D_N \\ \sum \vartheta_x = b \\ \sum \nu_x \leq 3,\nu_i>0}} p(\varrho) \nu_i\eta^*_{i} \big( \mathfrak m(\mathbf T') = \mathtt m -1 , \mathfrak k(\mathfrak v(\mathbf T')) = a \big) \\
    \leq\ & \mathfrak c^{7} \cdot \sum_{\substack{\varrho \in \mathcal D_N\\\sum \vartheta_x = 0\\  \sum \nu_x = 2 ,\nu_i > 0}} p(\varrho) \nu_i\eta^*_{i} \Big(\mathfrak m(\mathbf T') = \mathtt m -1 , \mathsf{Branch}_{\mathbf T'}(\mathfrak v(\mathbf T')) = 2, \mathfrak k(\mathfrak v(\mathbf T')) = 0 \Big) 
\end{align*}
We will prove this inequality by showing that for each $i$ with $\mathtt L^2+\mathtt m-1 \leq i \leq N$, we have
\begin{align}{\label{eq-lem-induction-estimates-(2)-expansion-part-I}}
    \sum_{\substack{\varrho \in \mathcal D_N:\\ \sum \vartheta_x = b, \\\sum \nu_x \leq 3,\nu_i>0}} p(\varrho) \leq 3\mathfrak c^2 \cdot \sum_{\substack{\varrho \in \mathcal D_N: \\\sum \vartheta_x = 0,\\\sum \nu_x = 2,\nu_i>0}} p(\varrho)
\end{align}
and 
\begin{align}{\label{eq-lem-induction-estimates-(2)-expansion-part-II}}
    & \eta^*_{i} \big( \mathfrak m(\mathbf T') = \mathtt m -1 , \mathfrak k(\mathfrak v(\mathbf T')) = a \big) \nonumber \\
    \leq\ & 4\mathfrak c^4 \eta^*_{i} \big(\mathfrak m(\mathbf T') = \mathtt m -1 , \mathsf{Branch}_{\mathbf T'}(\mathfrak v(\mathbf T')) = 2, \mathfrak k(\mathfrak v(\mathbf T')) = 0 \big)\,.
\end{align}
 We first show \eqref{eq-lem-induction-estimates-(2)-expansion-part-I}. Fix $\mathtt L^2+\mathtt m-1 \leq i \leq N$. Given any integer sequence $(\vartheta^*_1 , \ldots , \vartheta^*_N)$, by the definition of $\mathsf{HCut}$ in \eqref{def-HCut}, we have that $\mathsf{HCut}$ is a bijection from $\{\mathbf T \in \mathcal R^{(1,2,4)}_N : \vartheta_j = \vartheta^*_j \mbox{ for all }1 \leq j \leq N\}$ to $\mathcal N^{(1,2,4)}_{N - \sum x\vartheta^*_x}$. Also, for any tree $\mathbf T' \in \mathcal N^{(1,2,4)}_{N - \sum x\vartheta^*_x}$ such that $1 \leq |\mathsf{Subs}_i(\mathbf T')| \leq 3$, we can map $\mathbf T'$ to a tree pair $(\mathbf T'_* , \mathbf T'_{**}) \in \mathcal R_i^{(1,2,4)} \times \mathcal N^{(1,2,4)}_{N-i-\sum x\vartheta^*_x}$ in at most 3 ways by
\begin{align*}
\mathbf T'_* \in \mathsf{Subs}_i(\mathbf T')\,, \quad \mathbf T'_{**} = \mathbf T'|_{\setminus \mathbf T'_{*}} \,,
\end{align*}
where $\mathbf T'_{**}$ represents the tree obtained by deleting $\mathbf{T}_*'$ in $\mathbf{T}'$. Therefore, we have (recall the definition of $\mathsf a_j$ in \eqref{eq-def-mathsf-p-j})
\begin{align*}
    \sum_{\substack{\varrho \in \mathcal D_N: \sum \vartheta_x = b \\ \sum \nu_x \leq 3,\nu_i>0}} p(\varrho) &\leq 3 \sum_{\{\vartheta^*_x\} : \sum \vartheta^*_x = b} \frac{ |\mathcal R^{(1,2,4)}_i||\mathcal N^{(1,2,4)}_{N-i-\sum x\vartheta^*_x}|}{|\mathcal R^{(1,2,4)}_N|} \leq \sum_{\{\vartheta^*_x \}: \sum \vartheta^*_x = b} \frac{ 3|\mathcal R^{(1,2,4)}_i| |\mathcal R^{(1,2,4)}_{N-i-\sum x\vartheta^*_x}|}{|\mathcal R^{(1,2,4)}_N|} \nonumber \\
    &\leq \frac{ 3|\mathcal R^{(1,2,4)}_i| |\mathcal R^{(1,2,4)}_{N-i}|}{|\mathcal R^{(1,2,4)}_N|} \sum_{\{\vartheta^*_x\} : \sum \vartheta^*_x = b} 2^{-\sum x\vartheta^*_x} \\
    &\leq \frac{ 3|\mathcal R^{(1,2,4)}_i| |\mathcal R^{(1,2,4)}_{N-i}|}{|\mathcal R^{(1,2,4)}_N|} \sum_{\{\vartheta^*_x\} : \sum x\vartheta^*_x \geq b} 2^{-\sum x\vartheta^*_x} \leq \frac{ 3\mathfrak{c}|\mathcal R^{(1,2,4)}_i| |\mathcal N^{(1,2,4)}_{N-i}|}{|\mathcal R^{(1,2,4)}_N|} \sum_{i \geq b} \mathsf a_i 2^{-i} \\
    & \leq 3\mathfrak{c}^2 \frac{ |\mathcal R^{(1,2,4)}_i| |\mathcal N^{(1,2,4)}_{N-i}|}{|\mathcal R^{(1,2,4)}_N|} = 3\mathfrak{c}^2\sum_{\substack{\varrho \in \mathcal D_N: \sum \vartheta_x = 0, \sum \nu_x = 2, \nu_i>0}} p(\varrho) \,,
\end{align*}
where the third inequality holds by Lemma~\ref{lem-prelim-tree-ratio}. We have shown \eqref{eq-lem-induction-estimates-(2)-expansion-part-I}. Next we show \eqref{eq-lem-induction-estimates-(2)-expansion-part-II}. Since $i \geq \mathtt L ^2 + \mathtt m - 1$, we see that $\eta^*_{i}$ and $\eta^{**}_i$ are the same measures. Therefore, we have
\begin{align}
    & \eta^*_{i} \big(\mathfrak m(\mathbf T') = \mathtt m -1 , \mathfrak k(\mathfrak v(\mathbf T')) = a \big)= \eta^{**}_i\big(\mathfrak m(\mathbf T) = \mathtt m -1 , \mathfrak k(\mathfrak v(\mathbf T)) = a \big)\nonumber \\
    \leq\ & \frac{\#\{\mathbf T' \in \mathcal R^{(1,2,4)}_i:\mathfrak m(\mathbf T') = \mathtt m -1 , \mathfrak k(\mathsf{VCut_{\mathtt m-1}}(\mathbf T')) = a\} }{\big(|\mathcal R^{(1,2,4)}_i|-1\big)}\nonumber \\
    \leq\ & \frac{2\#\{\mathbf T' \in \mathcal R^{(1,2,4)}_{i - \mathtt m + 1}:\mathfrak k(\mathbf T') = a \}}{|\mathcal R^{(1,2,4)}_i|} \leq \frac{4|\mathcal R^{(1,2,4)}_{i-\mathtt m+1}| }{|\mathcal R^{(1,2,4)}_i|} \leq \frac{4\mathfrak c^4 |\mathcal N^{(1,2,4)}_{i-\mathtt m-3}| }{|\mathcal R^{(1,2,4)}_i|} \nonumber \\
    \leq\ & 4\mathfrak c^4 \eta^*_{i-\mathtt m + 1} \big(\mathsf{Branch}_{\mathbf T'}(\mathfrak R(\mathbf T')) = 2, \mathfrak k(\mathfrak R(\mathbf T')) = 0 \big) \nonumber\\
    =\ & 4\mathfrak c^4 \eta^*_{i} \big(\mathfrak m(\mathbf T') = \mathtt m -1 , \mathsf{Branch}_{\mathbf T'}(\mathfrak v(\mathbf T')) = 2, \mathfrak k(\mathfrak v(\mathbf T')) = 0 \big) \,, \label{eq-lem-induction-estimates-(2)-max-mu}
\end{align}
where the last inequality follows the same lines as the derivation of the first inequality of \eqref{eq-lem-induction-estimates-(1)-p-diamond-varrho}. The proof of \eqref{eq-lem-induction-estimates-(2)-expansion-part-II} is complete.
\end{proof}

We now turn to the proof of \eqref{eq-thm-2.7-key-expectation}. First, recall \eqref{def-chained-noise-pair-in-T_iota} and \eqref{def-chained-signal-pair-in-T_iota}. By definition of $\mathbf F$ in \eqref{eq-def-mathbf-F-T}, it requires at least $\mathtt L^2 + \mathtt m + 3$ vertices for a tree $\mathbf T$ to make $\mathbf F(\mathbf T) \neq 0$, so we will prove \eqref{eq-thm-2.7-key-expectation} by induction on $N$. Suppose $N \geq \mathtt L^2+\mathtt m+3$ and \eqref{eq-thm-2.7-key-expectation} holds for any positive integer smaller than $N$. We have (we always denote $\varrho(\mathbf T) = \big( (\vartheta_1 , \ldots , \vartheta_N); (\nu_1 , \ldots , \nu_N) \big)$ below, and note that $\vartheta_i = 0$ for $i \geq \mathtt L$)
\begin{align}
    \Eb_{\eta^*_N}[\mathbf F(\mathbf T)] \geq\ & \Eb_{\eta^*_N}\Bigg[\sum_{\substack{\mathtt L^2+\mathtt m-1 \leq i \leq N \\ 1 \leq j \leq \nu_i} }\Big( \mathbf F(\mathbf T_{i,j}) + \mathbf{1}_{\{ \mathfrak m(\mathbf T_{i,j}) = \mathtt m -1 , \mathsf{Branch}_{\mathbf T_{i,j}}(\mathfrak v(\mathbf T_{i,j})) = 2 , \sum \nu_x = 2 , \sum \vartheta_x = 0\} } \nonumber \\
    & - \mathfrak c^{-9} \sum_{0 \leq a,b \leq 1} \mathbf{1}_{\{ \mathfrak m(\mathbf T_{i,j}) = \mathtt m -1 , \mathfrak k( \mathfrak v (\mathbf T_{i,j})) = a , \sum \vartheta_x = b\} } \Big) \Bigg] \nonumber \\ 
    =\ & \Eb_{\eta^*_N}\Big[\sum_{\substack{1 \leq i \leq N \\ 1 \leq j \leq \nu_i} }\big(\mathbf F(\mathbf T_{i,j}) + \mathbf F_{1;i,j}(\mathbf T) - \mathfrak c^{-9}\mathbf F_{2;i,j} (\mathbf T)\big) \Big] \nonumber \\
    \geq\ & \mathfrak c^{-9-\mathtt m} + \sum_{\varrho \in \mathcal D_N} p(\varrho)\cdot \big(\sum_{\substack{\mathtt L \leq i \leq N \\ 1 \leq j \leq \nu_i}} \E_{\eta^*_{N,\varrho;i}}[\mathbf F(\mathbf T)] \big)\,,\label{eq-proof-of-expectation-mathbf-F-1}
\end{align}
where the last inequality holds by Items (1) and (2) of Lemma~\ref{lem-induction-estimates-on-rooted-trees-function}. For any $\mathtt L \leq i \leq N$ and $\varrho \in \mathcal D_N$, if $i \leq \mathtt L^2 + \mathtt m - 2$ we have $\E_{\eta^*_{N,\varrho;i}}[\mathbf F(\mathbf T)] = 0$; else by Lemma~\ref{lem-tree-N-N+1-ratio-lowers}, the induction hypothesis and the fact that $|\mathbf F(\mathbf T)| \leq |V(\mathbf T)|$ we have 
\begin{align*}
    \E_{\eta^*_{N,\varrho;i}}[\mathbf F(\mathbf T)] &\geq \E_{\eta^*_{i}}[\mathbf F(\mathbf T)] -i \cdot 0.7^{i} \geq \tfrac{1}{\mathtt L^4\mathfrak c^{9+\mathtt m}}(i-\mathtt L^2 - \mathtt m -2)_+ - i \cdot 0.7^i \\
    &\geq \tfrac{1}{\mathtt L^4\mathfrak c^{9+\mathtt m}}(i-\mathtt L^2 - \mathtt m -2.5)_+ \,.
\end{align*}
Therefore, we always have $\E_{\eta^*_{N,\varrho;i}}[\mathbf F(\mathbf T)] \geq \tfrac{1}{\mathtt L^4\mathfrak c^{9+\mathtt m}} (i-\mathtt L^2 - \mathtt m - 2.5)_+$. Plugging this into \eqref{eq-proof-of-expectation-mathbf-F-1} yields
\begin{align}
    \eqref{eq-proof-of-expectation-mathbf-F-1} &\geq \mathfrak c^{-9-\mathtt m} + \sum_{\varrho \in \mathcal D_N} p(\varrho)\cdot \big(\sum_{\substack{\mathtt L \leq i \leq N \\ 1 \leq j \leq \nu_i}} \tfrac{1}{\mathtt L^4\mathfrak c^{9+\mathtt m}}(i-\mathtt L^2 - \mathtt m - 2.5)_+ \big) \nonumber \\
    &= \mathfrak c^{-9-\mathtt m} + \sum_{\varrho \in \mathcal D_N} p(\varrho)\cdot \big(\sum_{\substack{1\leq i \leq N \\ 1 \leq j \leq \nu_i}} \tfrac{1}{\mathtt L^4\mathfrak c^{9+\mathtt m}}(i-\mathtt L^2 - \mathtt m - 2.5)_+ \big) \nonumber \\
    &\geq \Big(\mathfrak c^{-9-\mathtt m} + \sum_{\varrho \in \mathcal D_N} p(\varrho)\cdot \big(\sum_{\substack{1\leq i \leq N \\ 1 \leq j \leq \nu_i}} \tfrac{1}{\mathtt L^4\mathfrak c^{9+\mathtt m}}(i-\mathtt L^2 - \mathtt m - 2.5) \big) \Big)_+ \nonumber \\
    &= \Big(\mathfrak c^{-9-\mathtt m} + \frac{N-1}{\mathtt L^4\mathfrak c^{9+\mathtt m}} - \frac{\Eb_{\eta^*_N}[|\mathsf{Subs}(\mathbf T)|](\mathtt L^2 + \mathtt m + 2.5)}{\mathtt L^4 \mathfrak c^{9+\mathtt m}} \Big)_+\,.\label{eq-proof-of-expectation-mathbf-F-2}
\end{align}
Next, we bound \eqref{eq-proof-of-expectation-mathbf-F-2} by showing $\Eb_{\eta^*_N}[|\mathsf{Subs}(\mathbf T)|] \leq 2\mathtt L$ for $N \geq \mathtt L^2 + \mathtt m + 3$. By Item (1) of Definition~\ref{def-tilde-T-K}, it suffices to show $\Eb_{\eta^*_N}[\mathfrak k(\mathbf T)] \leq \mathtt L$. By \eqref{eq-sigma-star-N-1} and \eqref{eq-sigma-star-N-2} we have
\begin{align}
    \Eb_{\eta^*_N}[\mathfrak k(\mathbf T)] &\leq \Eb_{\eta^*_N}[\Sigma^*(\vartheta(\mathbf T))] =\sum_{i \geq 1}i\eta^*_N (\Sigma^*(\vartheta(\mathbf T)) = i) \nonumber \\
    &= \frac{(N-1)\mathsf a_{N-1}}{|\mathcal R^{(1,2,4)}_N|} + \sum_{1 \leq i \leq N-4}\frac{i\mathsf a_i |\mathcal N^{(1,2,4)}_{N-i}|}{|\mathcal R^{(1,2,4)}_N|} \nonumber \\
    &\leq 2\mathfrak c'\big((N-1)(\tfrac{1.01}{2})^N + (1-\tfrac{1.01}{2})^{-2} \big) \leq \mathtt L\,, \label{eq-bounded-expectation-of-mathfrak-k-T-1,2,4}
\end{align}
where the second equality holds by Lemmas~\ref{lem-prelim-tree-ratio} and \ref{lem-tree-simple-horizontal-split-enum}. Plugging \eqref{eq-bounded-expectation-of-mathfrak-k-T-1,2,4} into \eqref{eq-proof-of-expectation-mathbf-F-2} yields
\begin{align}
    \eqref{eq-proof-of-expectation-mathbf-F-2} &\geq \Big(\mathfrak c^{-9-\mathtt m} + \frac{N-1}{\mathtt L^4 \mathfrak c^{9+\mathtt m}} - \frac{\Eb_{\eta^*_N}[\mathfrak k(\mathbf T) + \mathtt L](\mathtt L^2 + \mathtt m + 2.5)}{\mathtt L^4 \mathfrak c^{9+\mathtt m}} \Big)_+ \nonumber \\
    &\geq \Big(\mathfrak c^{-9-\mathtt m} + \frac{N-1}{\mathtt L^4 \mathfrak c^{9+\mathtt m}} - \frac{2\mathtt L(\mathtt L^2 + \mathtt m + 2.5)}{\mathtt L^4\mathfrak c^{9+\mathtt m}}\Big)_+ \nonumber \\
    &\geq \frac{1}{\mathtt L^4 \mathfrak c^{9+\mathtt m}} (N - \mathtt L^2 -\mathtt m - 2)_+\,,
\end{align}
which yields \eqref{eq-thm-2.7-key-expectation} as desired.

\subsection{Proof of Theorem~\ref{thm-desired-vertex-sets}}{\label{subsec:proof-set-enu-thm}}

This subsection is dedicated to the proof of Theorem~\ref{thm-desired-vertex-sets}. Denote $\mathfrak C'_\aleph = \tfrac{\mathfrak C_\aleph}{4}$ for simplicity of notation. To begin our proof, fix $\mathbf T \in \widetilde{\mathcal T}_{\aleph}$. Note that by definition of $\mathfrak{Ch}^*_\mathtt m(\mathbf T)$, for each $v \in V(\mathbf T_\iota)$ there exist at most 3 pairings in $\mathfrak{Ch}^*_\mathtt m(\mathbf T)$ containing $v$ as one of the endpoints. As a result, there exists a subset $\mathfrak{Ch}'_\mathtt m(\mathbf T)$ of $\mathfrak{Ch}^*_\mathtt m(\mathbf T)$, such that $|\mathfrak{Ch}'_\mathtt m(\mathbf T)| = \tfrac{1}{4}|\mathfrak{Ch}^*_\mathtt m(\mathbf T)|=\mathfrak C'_\aleph$ and no two pairings in $\mathfrak{Ch}'_\mathtt m(\mathbf T)$ share a common vertex. Therefore, it suffices to show that with probability at least $\exp\big( -10e^{-\log^{2}(\iota^{-1})} \aleph \big)$, a uniform random subset of $\mathfrak{Ch}'_\mathtt m(\mathbf T)$ with $\iota\aleph$ elements will satisfy Item (3) of Theorem~\ref{thm-desired-vertex-sets}, which will be done in the following lemmas. (We use the probability argument to show the existence of such pairings; in our algorithm we just use brutal-force search which takes at most $O(2^{\aleph^2}) = n^{o(1)}$ time.)

We first record a lemma on binomial distribution that will help us choose the appropriate subset of $\mathfrak{Ch}^*_\mathtt m (\mathbf T)$.

\begin{lemma}{\label{lem-binomial-distribution-fundamental-estimate}}
    Denote $X(\kappa) \sim \mathrm{Binom}(\mathfrak C'_\aleph, \tfrac{(1+\kappa)\iota \aleph}{\mathfrak C'_\aleph})$. There exists $\kappa \in (0,1)$ such that 
    \begin{equation*}
        \Pb (X(\kappa)=\iota\aleph)=\exp\big( -4e^{-\log^{2}(\iota^{-1})} \aleph \big) \mbox{ and } \Pb( X(\kappa) < \iota\aleph) \leq \exp\big( -3e^{-\log^{2}(\iota^{-1})} \aleph \big) \,.
    \end{equation*}
\end{lemma}
\begin{proof}
    By definition of binomial distribution, for $X (\kappa) \sim \mathrm{Binom}(\mathfrak C'_\aleph, \tfrac{(1+\kappa)\iota \aleph}{\mathfrak C'_\aleph})$ and $0 \leq i \leq \mathfrak C'_\aleph-1$ it can be shown that
    \begin{align*}
        \frac{\Pb(X(\kappa)=i+1)}{\Pb(X(\kappa)=i)} = \frac{\mathfrak C'_\aleph-i}{i+1}\cdot \frac{\tfrac{(1+\kappa)\iota \aleph}{\mathfrak C'_\aleph}}{1-\tfrac{(1+\kappa)\iota \aleph}{\mathfrak C'_\aleph}}\,.
    \end{align*}
    Thus, $\frac{\Pb(X(\kappa)=i+1)}{\Pb(X(\kappa)=i)} \geq 1$ is equivalent to
    \begin{align}\label{eq-consequent-binom-values}
        (\mathfrak C'_\aleph-i)\cdot \tfrac{(1+\kappa)\iota \aleph}{\mathfrak C'_\aleph} \geq (i+1) \cdot (1-\tfrac{(1+\kappa)\iota \aleph}{\mathfrak C'_\aleph}) \Leftrightarrow i \leq (1+\kappa)(1+\tfrac{1}{\mathfrak C'_\aleph})\iota\aleph - 1 \,.
    \end{align}
    Therefore, if we let $\kappa$ equals 
    \begin{align}\label{eq-def-kappa-1-binom}
        \kappa_1 := \frac{1+\tfrac{1}{\iota\aleph}}{1+\tfrac{1}{\mathfrak C'_\aleph}} - 1 \in (0,1) \mbox{ where we have }(1+\kappa_1)(1+\tfrac{1}{\mathfrak C'_\aleph})\iota\aleph - 1 = \iota\aleph\,,
    \end{align}
    then we have 
    \begin{equation}
        \Pb(|X(\kappa_1)| = \iota\aleph) = \max_{0 \leq i \leq \mathfrak C'_\aleph} \Pb(|X(\kappa_1)| = i) \geq \tfrac{1}{\mathfrak C'_\aleph + 1} > \tfrac{1}{\aleph}\,,\label{eq-X-kappa1}
    \end{equation}
    where the first inequality follows from the fact that $\sum_{0 \leq i \leq \mathfrak C'_\aleph} \P(|X(\kappa_1)| = i) = 1$. Moreover, if we select $\kappa = 1$ (note that $\mathbb E(X(1))=2\iota \aleph$)
    \begin{align}\label{eq-X-1}
        \P(X(1)=\iota\aleph) \leq \P(X(1)\leq (1-\tfrac{1}{2})\cdot2\iota\aleph) \leq \exp(-\tfrac{1}{2}\iota\aleph \cdot (\tfrac{1}{2})^2) = \exp(-\tfrac{1}{8}e^{-\log(\iota^{-1})}\aleph)\,
    \end{align}
    by Chernoff's Inequality. Therefore, combining \eqref{eq-X-kappa1} and \eqref{eq-X-1}, we know that there must exist $\kappa_0 \in (\kappa_1 , 1)$ such that
    \begin{align}\label{eq-selection-kappa-binom}
        \P(X(\kappa_0)=\iota\aleph) = \exp(-4e^{-\log^2(\iota^{-1})}\aleph)\,.
    \end{align}
    We conclude the proof by showing that $\kappa_0$ satisfies the second condition of the lemma. By \eqref{eq-def-kappa-1-binom} and \eqref{eq-consequent-binom-values} we have
    \begin{align*}
        \P(X(\kappa_0) \leq \iota\aleph) &= \sum_{0 \leq i \leq \iota\aleph}\P(X(\kappa_0)=i) \leq (\iota\aleph+1)\max_{0 \leq i \leq \iota\aleph} \P(X(\kappa_0)=i) \\
        &= (\iota\aleph+1)\P(X(\kappa_0)=\iota\aleph) \overset{\eqref{eq-selection-kappa-binom}}{\leq} \exp\big( -3e^{-\log^{2}(\iota^{-1})} \aleph\big)\,,
    \end{align*}
    which yields that $\kappa = \kappa_0$ satisfies our claim.
\end{proof}

Then we turn to the following lemma which will finish our proof.

\begin{lemma}{\label{lem-random-choice-W-suffice}}
    For each $\mathbf T \in \widetilde{\mathcal T}_{\aleph}$, we sample a random subset $W \subset \mathfrak{Ch}'_\mathtt m(\mathbf T)$ by independently including each pair $(u,v) \in \mathfrak{Ch}'_\mathtt m(\mathbf T)$ with probability $\frac{(1+\kappa)\iota \aleph}{\mathfrak C'_\aleph}$, where $\kappa \in (0,1)$ is chosen as the $\kappa$ in Lemma~\ref{lem-binomial-distribution-fundamental-estimate}. Then 
    \begin{align*}
        & \mathbb P\Big( |W|=\iota\aleph; |W \cap V(\mathbf T')| \leq \big[ \log^{-2}(\iota^{-1}) \cdot |V(\mathbf T')| \big] \mbox{ for all } |V(\mathbf T')| \geq \log^{4}(\iota^{-1}) \mbox{ with } \mathbf T \hookrightarrow \mathbf T' \Big) \\
        \geq\ & \exp\big( -10e^{-\log^{2}(\iota^{-1})} \aleph \big) \,. 
    \end{align*}
\end{lemma}
\begin{proof}
    We first show that for each event $\mathcal A$ measurable and decreasing in $W$, 
    \begin{align*}
        \exp\big( 4e^{-\log^{2}(\iota^{-1})} \aleph \big) \cdot \mathbb P(|W|=\iota\aleph, \mathcal A) \geq \mathbb P(|W|\geq \iota\aleph,\mathcal A) \,.
    \end{align*}
    In fact, we have 
    \begin{align}
         \mathbb P(|W|\geq \iota\aleph,\mathcal A) &= \sum_{k \geq \iota\aleph} \mathbb P(|W|=k,\mathcal A) = \sum_{k \geq \iota\aleph} \mathbb P(|W|=k) \mathbb P(\mathcal A \mid |W|=k) \nonumber\\
         &\leq \sum_{k \geq \iota\aleph} \mathbb P(|W|=k) \mathbb P(\mathcal A \mid |W|=\iota\aleph) = \mathbb P(|W|\geq \iota\aleph) \mathbb P(\mathcal A \mid |W|=\iota\aleph) \nonumber\\
         &\overset{Lemma~\ref{lem-binomial-distribution-fundamental-estimate}}{\leq} \exp\big( 4e^{-\log^{2}(\iota^{-1})} \aleph \big) \mathbb P(|W|=\iota\aleph) \mathbb P(\mathcal A \mid |W|=\iota\aleph) \nonumber\\
         &= \exp\big( 4e^{-\log^{2}(\iota^{-1})} \aleph \big) \mathbb P(|W|=\iota\aleph,\mathcal A) \,.\label{eq-choose-subset-by-FKG}
    \end{align}
    Here the first inequality follows from the fact that conditioned on $|W|=k$, $W$ is uniformly distributed over subsets of $W_0$ with cardinality $k$ and there is a natural coupling over the uniform $k$-subset of $W_0$ (denoted as $W_k$) and the uniform $(k+1)$-subset of $W_0$ (denoted as $W_{k+1}$) such that $W_k \subset W_{k+1}$. For each $u \in V(\mathbf T_{\iota})$, recall that we use $\mathsf{Des}_{\mathbf T_{\iota}}(u)$ to denote the descendant tree of $\mathbf T_{\iota}$ rooted at $u$. Thus, recalling Definition~\ref{def-major-subtree}, by \eqref{eq-choose-subset-by-FKG} it suffices to show that 
    \begin{align}
        & \mathbb P\Big( |W|\geq\iota\aleph ; |W \cap V(\mathbf T')| \leq \big[ \log^{-2}(\iota^{-1}) \cdot |V(\mathbf T')| \big] \vee \log^2(\iota^{-1}) \,, \nonumber\\ &\qquad\qquad\qquad \mbox{for all } \mathbf{T}'\hookleftarrow\mathbf{T} \mbox{ with } |V(\mathbf T')| \geq \log^2(\iota^{-1})\Big) \nonumber \\
        \geq\ & \exp\big( -6e^{-\log^{2}(\iota^{-1})} \aleph \big) \,.  \label{eq-goal-lem-B.5}
    \end{align}
    Since the left-hand side of \eqref{eq-goal-lem-B.5} is lower-bounded by
    \begin{align}
        & \mathbb P\Big( |W \cap V(\mathbf T')| \leq \big[ \log^{-2}(\iota^{-1}) \cdot |V(\mathbf T')| \big] \vee \log^2(\iota^{-1}) \mbox{ for all }\mathbf{T}'\hookleftarrow\mathbf{T}\mbox{ with } |V(\mathbf T')| \geq \log^2(\iota^{-1}) \Big) \label{eq-lem-B.5-final-goal} \\
        &-\P\Big( |W|<\iota\aleph \Big)\,, \label{eq-lem-B.5-ancillary}
    \end{align}
    and $\eqref{eq-lem-B.5-ancillary} \leq \exp\big( -3e^{-\log^{2}(\iota^{-1})} \aleph \big)$ by Lemma~\ref{lem-binomial-distribution-fundamental-estimate}, it suffices to show that $\eqref{eq-lem-B.5-final-goal} \geq \exp\big( -2e^{-\log^{2}(\iota^{-1})} \aleph \big)$. For each $u \in V(\mathbf T_{\iota})$ such that $|V( \mathsf{Des}_{\mathbf T_{\iota}}(u) )| \geq \log^2(\iota^{-1})$, define $E_u$ to be the event 
    \begin{align*}
        \Big\{ |W \cap V(\mathsf{Des}_{\mathbf T_{\iota}}(u) )| \leq \big[ \log^{-2}(\iota^{-1}) \cdot |V(\mathsf{Des}_{\mathbf T_{\iota}}(u) )| \big] \vee \log^2(\iota^{-1}) \Big\} \,.
    \end{align*}
   Since the events $\{ E_u : u \in V(\mathbf T_{\iota}) \}$ are decreasing, by the FKG-inequality we have that \eqref{eq-lem-B.5-final-goal} is bounded by 
    \begin{align*}
        & \mathbb P\Big( \cap_{u \in V(\mathbf T_{\iota})} E_u \Big) \\
        \geq\ & \prod_{ \substack{u \in V(\mathbf T_{\iota}) \\ |V( \mathsf{Des}_{\mathbf T_{\iota}}(u) )| \geq \log^2(\iota^{-1})} } \mathbb P\Big( |W \cap V(\mathsf{Des}_{\mathbf{T}_\iota}(u))| \leq \tfrac{|V(\mathsf{Des}_{\mathbf{T}_\iota}(u))|}{\log^2(\iota^{-1})} \vee \log^2(\iota^{-1}) \Big) \\
        \geq\ & \prod_{ \substack{u \in V(\mathbf T_{\iota}) \\ |V( \mathsf{Des}_{\mathbf T_{\iota}}(u) )| \geq \log(\iota^{-1})} } \mathbb P\Big( \mathrm{Binom}\big( |V(\mathsf{Des}_{\mathbf T_{\iota}}(u))|, \tfrac{(1+\kappa)\iota \aleph}{\mathfrak C'_\aleph} \big) \leq \tfrac{|V(\mathsf{Des}_{\mathbf{T}_\iota}(u))|}{\log^2(\iota^{-1})} \vee \log^2(\iota^{-1}) \Big) \\
        \geq\ & \Big( 1-\exp\big( -\log^2(\iota^{-1}) \big) \Big)^{\aleph} \geq \exp \big( -1.5e^{-\log^2(\iota^{-1})} \aleph \big) \,.  \qedhere
    \end{align*}
\end{proof}

\section{Supplementary Proofs}{\label{sec:supp-proofs}}

\subsection{Proof of Lemma~\ref{lem-joint-moment-A-B}}{\label{subsec:proof-lem-3.3}}
The case that $r+t=1$ and $r+t=2$ can be verified via direct calculations, as 
\begin{align*}
    & \mathbb E_{\Pb_{\sigma,\pi}}\Bigg[ \frac{ A_{i,j} - \tfrac{\lambda s}{n} }{ \big( \tfrac{\lambda s}{n} (1-\tfrac{\lambda s}{n}) \big)^{1/2} } \Bigg] = \mathbb E_{\Pb_{\sigma,\pi}}\Bigg[ \frac{ B_{i,j} - \tfrac{\lambda s}{n} }{ \big( \tfrac{\lambda s}{n} (1-\tfrac{\lambda s}{n}) \big)^{1/2} } \Bigg] \circeq \frac{ \sigma_i \sigma_j \sqrt{\epsilon^2 \lambda s} }{ \sqrt{n} } \,, \\
    & \mathbb E_{\Pb_{\sigma,\pi}} \Bigg[ \Bigg( \frac{ A_{i,j} - \tfrac{\lambda s}{n} }{ \big( \tfrac{\lambda s}{n} (1-\tfrac{\lambda s}{n}) \big)^{1/2} } \Bigg) \Bigg( \frac{ B_{\pi(i),\pi(j)} - \tfrac{\lambda s}{n} }{ \big( \tfrac{\lambda s}{n} (1-\tfrac{\lambda s}{n}) \big)^{1/2} } \Bigg) \Bigg] \circeq s(1+\epsilon\sigma_i\sigma_j) \,, \\
    & \mathbb E_{\Pb_{\sigma,\pi}} \Bigg[ \Bigg( \frac{ A_{i,j} - \tfrac{\lambda s}{n} }{ \big( \tfrac{\lambda s}{n} (1-\tfrac{\lambda s}{n}) \big)^{1/2} } \Bigg)^2 \Bigg] = \mathbb E_{\Pb_{\sigma,\pi}} \Bigg[ \Bigg( \frac{ B_{i,j} - \tfrac{\lambda s}{n} }{ \big( \tfrac{\lambda s}{n} (1-\tfrac{\lambda s}{n}) \big)^{1/2} } \Bigg)^2 \Bigg] \circeq (1+\epsilon\sigma_i\sigma_j) \,.
\end{align*}
For the case where $r \geq 3$ or $r+t \geq 3$, it suffices to note that in this case 
\begin{align}
    & \mathbb E_{\Pb_{\sigma,\pi}}\Bigg[ \Bigg( \frac{ A_{i,j} - \tfrac{\lambda s}{n} }{ \big( \tfrac{\lambda s}{n} (1-\tfrac{\lambda s}{n}) \big)^{1/2} } \Bigg)^r \ \Bigg] = \mathbb E_{\Pb_{\sigma,\pi}}\Bigg[ \Bigg( \frac{ B_{i,j} - \tfrac{\lambda s}{n} }{ \big( \tfrac{\lambda s}{n} (1-\tfrac{\lambda s}{n}) \big)^{1/2} } \Bigg)^r \ \Bigg] \nonumber \\
    =\ & (1+O(\tfrac{1}{n}))\big( \tfrac{n}{\lambda s} \big)^{\frac{r}{2}} \cdot \mathbb P_{\sigma,\pi}\Big( A_{i,j}=1 \Big) \leq (1+\epsilon\sigma_i \sigma_j) \cdot (n/\epsilon^2\lambda s)^{(r-2)/2} \,, \label{eq-same-order-u,v-r,t-1}  \\
    & \mathbb E_{\Pb_{\sigma,\pi}}\Bigg[ \Bigg( \frac{ A_{i,j} - \tfrac{\lambda s}{n} }{ \big( \tfrac{\lambda s}{n} (1-\tfrac{\lambda s}{n}) \big)^{1/2} } \Bigg)^r \Bigg( \frac{ B_{\pi(i),\pi(j)} - \tfrac{\lambda s}{n} }{ \big( \tfrac{\lambda s}{n} (1-\tfrac{\lambda s}{n}) \big)^{1/2} } \Bigg)^t \ \Bigg] \nonumber \\
    =\ & (1+O(\tfrac{1}{n}))\big( \tfrac{n}{\lambda s} \big)^{\frac{r+t}{2}} \mathbb P_{\sigma,\pi}\Big( A_{i,j}=B_{\pi(i),\pi(j)}=1 \Big) \leq (1+\epsilon\sigma_i \sigma_j) \cdot (n/\epsilon^2\lambda s)^{(r+t-2)/2} \,, \label{eq-same-order-u,v-r,t-2} 
\end{align}
thereby completing the proof.

\subsection{Proof of Lemma~\ref{lem-expectation-over-chain}}{\label{subsec:proof-lem-3.4}}

Recall that $\nu$ is the uniform distribution on $\{-1,+1\}^n$. By independence, we see that 
\begin{align*}
    \mathbb E_{\sigma \sim \nu} \Big[ \prod_{i \in I} (\sigma_{i-1} \sigma_i) \mid \sigma_0, \sigma_l \Big] = 0 \mbox{ if } I \subsetneq [l] \,.
\end{align*}
Thus,
\begin{align*}
    \mathbb{E}_{\sigma \sim \nu} \Big[ \prod_{i=1}^{l} \big( a_i + b_i \sigma_{i-1} \sigma_i \big) \mid \sigma_0, \sigma_l \Big] &= \prod_{i=1}^{l} a_i + \mathbb{E}_{\sigma \sim \nu} \Big[ \prod_{i=1}^{l} (\sigma_{i-1}\sigma_i) \mid \sigma_0, \sigma_l \Big] \cdot \prod_{i=1}^{l} b_i \\
    &= \prod_{i=1}^{l} a_i + \sigma_0 \sigma_l \cdot \prod_{i=1}^{l} b_i \,, 
\end{align*}    
completing the proof of Lemma~\ref{lem-expectation-over-chain}.

\subsection{Proof of Lemma~\ref{lem-upper-bound-exp}}{\label{subsec:proof-lem-3.6}}

The first result \eqref{eq-upper-bound-exp-case-1} is obvious. For \eqref{eq-upper-bound-exp-case-2}, we may assume that $\pi=\mathsf{id}$ without loss of generality. Denote $G=\widetilde S_1 \cup \widetilde S_2$. Then  
\begin{align*}
    \mathcal L(G) \subset \mathbb L:= \mathcal L(\widetilde{S}_1) \cup \mathcal L(\widetilde{S}_2) \,.
\end{align*}
Denote $G_0$ as the graph with $V(G_0)=V \cup \mathbb L$ and $E(G_0)=\emptyset$. We then have $\mathcal L(G) \subset V(G_0)$. Applying Lemma~\ref{lem-revised-decomposition-H-Subset-S} with $G_0 \subset G$, we see that $E(G)$ can be written as $E(G) = P_{\mathtt 1} \cup \ldots \cup P_{\mathtt t} \cup C_{\mathtt 1} \cup \ldots \cup C_{\mathtt m}$ satisfying Items~(1)--(4) in Lemma~\ref{lem-revised-decomposition-H-Subset-S}. In particular, we have 
\begin{align*}
    \mathtt t \leq 5\Big( \tau(G)-\tau(G_0) \Big) \leq 5 \tau(G) + 5|V| + 5|\mathbb L| \leq 5\tau(G)+5|V|+10\aleph \,.
\end{align*}
In addition, since $S_1,S_2$ is connected, we see that $\mathtt m\leq 2$ (since $G$ has at most 2 connected components and each $C_{\mathtt i}$ is a connected component of $G$). Also, conditioned on $\{\sigma_u : u\in V(G_0) \}$, we have that 
\begin{align*}
    \Bigg\{ & \prod_{(i,j) \in P_{\mathtt k}} \big( u_{\chi_1(i,j),\chi_2(i,j)} + v_{\chi_1(i,j),\chi_2(i,j)} \sigma_i \sigma_j \big) : \mathtt 1 \leq \mathtt k \leq \mathtt t, \\
    & \prod_{(i,j) \in C_{\mathtt k}} \big( u_{\chi_1(i,j),\chi_2(i,j)} + v_{\chi_1(i,j),\chi_2(i,j)} \sigma_i \sigma_j \big) : \mathtt 1 \leq \mathtt k \leq \mathtt m \Bigg\}
\end{align*}
is a collection of conditionally independent variables. In addition, from Lemma~\ref{lem-expectation-over-chain} we have 
\begin{align*}
    & \mathbb E_{ \sigma\sim\nu } \Big[ \prod_{(i,j) \in {P}_{\mathtt j}} \big( u_{\chi_1(i,j),\chi_2(i,j)} + v_{\chi_1(i,j),\chi_2(i,j)} \sigma_i \sigma_j \big) \mid \{\sigma_u : u\in U\} \Big] \\
    \leq\ & \prod_{(i,j) \in {P}_{\mathtt j}} u_{\chi_1(i,j),\chi_2(i,j)} + \prod_{w \in \operatorname{EndP}({P}_{\mathtt j})} \sigma_w \cdot \prod_{(i,j) \in {P}_{\mathtt j}} v_{\chi_1(i,j),\chi_2(i,j)} \\
    \leq\ & 2 \prod_{(i,j) \in {P}_{\mathtt j}} \max\big\{ u_{\chi_1(i,j),\chi_2(i,j)}, v_{\chi_1(i,j),\chi_2(i,j)} \big\} \,.
\end{align*}
Similarly, we have
\begin{align*}
    & \mathbb E_{ \sigma\sim\nu } \Big[ \prod_{(i,j) \in C_{\mathtt j}} \big( u_{\chi_1(i,j),\chi_2(i,j)} + v_{\chi_1(i,j),\chi_2(i,j)} \sigma_i \sigma_j \big) \mid \{\sigma_u : u\in U\} \Big] \\
    \leq\ & 2 \prod_{(i,j) \in C_{\mathtt j}} \max\big\{ u_{\chi_1(i,j),\chi_2(i,j)}, v_{\chi_1(i,j),\chi_2(i,j)} \big\} \,.
\end{align*}
Combining the above analysis with Lemma~\ref{lem-joint-moment-A-B} implies that 
\begin{align*}
    \mathbb E_{\Pb_{\mathsf{id}}}[ \phi_{S_1,S_2} \mid \{\sigma_u : u\in U\} ] \leq 2^{\mathtt t+\mathtt m} \prod_{(i,j) \in G} \max\big\{ u_{\chi_1(i,j),\chi_2(i,j)}, v_{\chi_1(i,j),\chi_2(i,j)} \big\} \,,
\end{align*}
which yields the desired result as $\mathtt t \leq 5\tau(G)+5|V|+10\aleph$ and $\mathtt m \leq 2$.

\subsection{Proof of Lemma~\ref{lem-delicate-exp-upper-bound-on-trees}}{\label{subsec:proof-lem-3.7}}

It suffices to deal with \eqref{eq-upper-bound-exp-on-trees} and Equation \eqref{eq-upper-bound-exp-on-trees-2} can be derived in a similar manner. Note that given $\mathtt U \subset V(T)$, we can write $T$ as 
\begin{align*}
    T = \mathbb{T}_{\mathtt U} \oplus \big( \oplus_{\mathtt u \in \mathtt U} T_{\mathtt u} \big) \,,
\end{align*}
where $\mathbb{T}_{\mathtt U} \subset T$ is the minimum sub-tree of $T$ containing $\mathtt U$, and $T_{\mathtt u}$ is the descendant tree of $T$ rooted at $\mathtt u$. Thus, we have
\begin{align}
    \mathbb E\Big[ \beta_T(A)^2 \cdot \prod_{\mathtt u \in \mathtt U} \sigma_{\mathtt u} \Big] &= \mathbb E\Big[ \beta_{\mathbb{T}_{\mathtt U}}(A)^2 \cdot \prod_{\mathtt u \in \mathtt U} \beta_{T_{\mathtt u}}(A)^2 \cdot \prod_{\mathtt u \in \mathtt U} \sigma_{\mathtt u} \Big] \nonumber \\
    &= \mathbb E\Bigg\{ \prod_{\mathtt u \in \mathtt U} \sigma_{\mathtt u} \cdot \mathbb E\Big[ \beta_{\mathbb{T}_{\mathtt U}}(A)^2 \cdot \prod_{\mathtt u \in \mathtt U} \beta_{T_{\mathtt u}}(A)^2 \mid \{ \sigma_{\mathtt u}: \mathtt u \in \mathtt U \} \Big] \Bigg\} \,. \label{eq-exp-on-tree-relax-1}
\end{align}
Note that conditioned on $\{ \sigma_{\mathtt u}: \mathtt u \in \mathtt U \}$, we have that $\{ \beta_{\mathbb{T}_{\mathtt U}}(A), \beta_{T_{\mathtt u}}(A): \mathtt u \in \mathtt U \}$ is a collection of conditionally independent variables, with 
\begin{align*}
    \mathbb E\Big[ \beta_{T_{\mathtt u}}(A)^2 \mid \{ \sigma_{\mathtt u}: \mathtt u \in \mathtt U \} \Big]=1 \,.
\end{align*}
Thus, we have
\begin{align}
    \eqref{eq-exp-on-tree-relax-1} = \mathbb E\Big[ \beta_{\mathbb{T}_{\mathtt U}}(A)^2 \cdot \prod_{\mathtt u \in \mathtt U} \sigma_{\mathtt u} \Big] \,. \label{eq-exp-on-tree-relax-2}
\end{align}
Observe that $\mathcal L(\mathbb{T}_{\mathtt U}) \subset \mathtt U \subset V(\mathbb{T}_{\mathtt U})$. By applying Lemma~\ref{lem-revised-decomposition-H-Subset-S} to $\mathbb{T}_{\mathtt U}$ and $\mathtt U$ we can decompose $\mathbb{T}_{\mathtt U}$ into $\mathtt t$ self-avoiding paths $P_{\mathtt 1},\ldots,P_{\mathtt t}$ satisfying Items (1)--(4) in Lemma~\ref{lem-revised-decomposition-H-Subset-S}. In particular, we have $\mathtt t \leq 5|\mathcal L(\mathbb{T}_{\mathtt U})|+|\mathtt U| \leq 6|\mathtt U|$. Also, we have
\begin{enumerate}
    \item[(I)] $V(P_{\mathtt i}) \cap (\cup_{\mathtt j \neq \mathtt i} V(P_{\mathtt j})) = \operatorname{EndP}(V(P_{\mathtt i}))$. 
    \item[(II)] $\mathtt U \subset \mathtt W:= \cup_{\mathtt i \leq \mathtt t} \operatorname{EndP}(V(P_{\mathtt i}))$.
\end{enumerate}
By conditioning on $\sigma_{\mathtt W}=\{ \sigma_{\mathtt u}: \mathtt u \in \mathtt W \}$, we obtain
\begin{align}
    \eqref{eq-exp-on-tree-relax-2} &= \mathbb E_{\sigma_{\mathtt W} \sim \nu_{\mathtt W}} \Bigg\{ \prod_{\mathtt u \in \mathtt U} \sigma_{\mathtt u} \mathbb E\Big[ \prod_{\mathtt i \leq \mathtt t} \beta_{P_{\mathtt i}}(A)^2 \mid \sigma_{\mathtt W} \Big] \Bigg\} \nonumber \\ 
    &\circeq \mathbb E_{\sigma_{\mathtt W}\sim\nu_{\mathtt W}} \Bigg\{ \prod_{\mathtt u \in \mathtt U} \sigma_{\mathtt u} \prod_{\mathtt i \leq \mathtt t} \Big( 1+\epsilon^{|E(P_{\mathtt i})|} \prod_{u,v=\operatorname{EndP}(P_{\mathtt i})} \sigma_u \sigma_v \Big) \Bigg\} \nonumber \\
    &= \sum_{ \Lambda \subset [\mathtt t] } \prod_{\mathtt i \in \Lambda} \epsilon^{ |E(P_{\mathtt i})| } \mathbb E_{\sigma_{\mathtt W} \sim \nu_{\mathtt W}} \Bigg\{ \prod_{\mathtt u \in \mathtt U \cup \cup_{\mathtt i \in \Lambda} \operatorname{EndP}(P_{\mathtt i}) } \sigma_u \Bigg\} = \sum_{\Lambda \in \mathcal M} \prod_{\mathtt i \in \Lambda} \epsilon^{ |E(P_{\mathtt i})| } \,, \label{eq-exp-on-tree-relax-3}
\end{align}
where $\mathcal M$ is the collection of $\Lambda \subset [\mathtt t]$ such that each vertex appears an even number of times in $\mathtt U \cup \cup_{\mathtt i \in \Lambda} \operatorname{EndP}(P_{\mathtt i})$. Next, recall the definition of $\mathcal X(\mathtt U)$ from \eqref{def-mathcal-X-of-sets}. We claim the following two facts: 
\begin{enumerate}
    \item[(i)] Suppose $( \mathtt x_1 , \mathtt y_1 , \ldots , \mathtt x_{|\mathtt U|/2} , \mathtt y_{|\mathtt U|/2} ) \in \mathcal X(\mathtt U)$ that minimizes $\sum_{i=1}^{|\mathtt U|/2} \mathsf{Dist}_T(\mathtt x_i , \mathtt y_i)$. Then for all $1\leq i \leq |\mathtt U|/2$ we have $\mathfrak p_T(\mathtt x_i , \mathtt y_i) = \cup_{j \in \Lambda_i} P_j$ for some pairwise disjoint subsets $\Lambda_i \subset [\mathtt t']$;
    \item[(ii)] $\Lambda^* = \cup_{1 \leq i \leq |\mathtt U|/2} \Lambda_i$ is the unique element in $\mathcal M$.
\end{enumerate}
Once (i) and (ii) are shown, we have
\begin{align}\label{eq-exp-on-tree-final-reduction}
    \eqref{eq-exp-on-tree-relax-3} = \prod_{i \in \Lambda^*} \epsilon^{|E(P_i)|} = \epsilon^{\mathtt d}\,,
\end{align}
which yields \eqref{eq-upper-bound-exp-on-trees}. To show (i), note that given two distinct vertices $\mathtt x,\mathtt y \in \mathtt U$, both $\mathfrak p_{\mathbb{T}_{\mathtt U}}(\mathtt x,\mathtt y)$ and $\mathfrak p_T(\mathtt x,\mathtt y) \subset \mathbb{T}_{\mathtt U}$ are self-avoiding paths included in $T$ with endpoints $\mathtt x , \mathtt y$. Since such a path is unique, we get $\mathfrak p_T(\mathtt x,\mathtt y) = \mathfrak p_{\mathbb{T}_{\mathtt U}}(\mathtt x,\mathtt y) \subset \mathbb{T}_{\mathtt U}$. Therefore, given any two distinct vertices $\mathtt x,\mathtt y \in \mathtt U$, we have $\mathfrak p_T(\mathtt x, \mathtt y) = \cup_{j \in \Lambda_{\mathtt x , \mathtt y}} P_j$ for some $\Lambda_{\mathtt x , \mathtt y} \subset [\mathtt t']$. It remains to prove $E\big( \mathfrak p_T(\mathtt x_i,\mathtt y_i)\cap \mathfrak p_T(\mathtt x_j,\mathtt y_j)\big) = \emptyset$ for $1 \leq i < j \leq |\mathtt U|/2$. Otherwise, there exists $1 \leq i < j \leq |\mathtt U|/2$ such that $E\big( \mathfrak p_T(\mathtt x_i,\mathtt y_i)\cap \mathfrak p_T(\mathtt x_j,\mathtt y_j)\big) \neq \emptyset$, and in this case we have
\begin{align*}
    \mathsf{Dist}_T(\mathtt x_i , \mathtt y_i) &+ \mathsf{Dist}_T(\mathtt x_j , \mathtt y_j) = |E\big(\mathfrak p_T(\mathtt x_i,\mathtt y_i)| + |E\big(\mathfrak p_T(\mathtt x_j,\mathtt y_j)|\\
    &> \big|E\big(\mathfrak p_T(\mathtt x_i,\mathtt y_i)\triangle \mathfrak p_T(\mathtt x_j,\mathtt y_j)\big)\big| = \sum_{z\in \Lambda_i \triangle \Lambda_j}|E(P_z)| \\
    &= \big(\mathsf{Dist}_T(\mathtt x_i , \mathtt y_j) + \mathsf{Dist}_T(\mathtt x_j , \mathtt y_i)\big) \wedge \big(\mathsf{Dist}_T(\mathtt x_i , \mathtt x_j) + \mathsf{Dist}_T(\mathtt y_i , \mathtt y_j)\big) \,,
\end{align*}
which contradicts the fact that $\sum_{i=1}^{|\mathtt U|/2} \mathsf{Dist}_T(\mathtt x_i , \mathtt y_i)$ is minimized among $( \mathtt x_1 , \mathtt y_1 , \ldots , \mathtt x_{|\mathtt U|/2} , \mathtt y_{|\mathtt U|/2} ) \in \mathcal X(\mathtt U)$.

To show (ii), it suffices to show $|\mathcal M| \leq 1$ since we have already shown $\Lambda^* \in \mathcal M$ by (i). Otherwise, there exists $\Lambda' \neq \Lambda^*$ such that $\Lambda' \in \mathcal M$. Then the degree of each vertex is even in $\cup_{j \in \Lambda^* \triangle \Lambda'} P_j \neq \emptyset$, which yields that $\cup_{j\in \Lambda^* \triangle \Lambda'} P_j$ is a union of cycles. Since a tree contains no cycle, this is a contradiction. Therefore, (i) and (ii) hold, leading to our desired result \eqref{eq-exp-on-tree-final-reduction}.

\subsection{Proof of Lemma~\ref{lem-est-intersection-istree-P-Q}}{\label{subsec:proof-lem-est-intersection-istree-P-Q}}

In this section, we prove Lemma~\ref{lem-est-intersection-istree-P-Q}. We will only prove \eqref{eq-phi-P-intersec-istree} and the proof of \eqref{eq-phi-Q-intersec-istree} follows similarly. Denote $\Lambda=\{ 1\leq p \leq m: \mathcal L_p \subset S \}$. In addition, define 
\begin{equation}
    \mathtt U = \cup_{1\leq p \leq m} \operatorname{EndP}(\mathcal L_p) \,.
\end{equation}
Conditioned on $\{ \sigma_v:v\in \mathtt U\}$, we have that
\begin{align*}
    \Bigg\{ & \prod_{ (i,j) \in E(T) } \tfrac{ A_{i,j}-\frac{\lambda s}{n} }{ \sqrt{ \frac{\lambda s}{n} (1-\frac{\lambda s}{n}) } } \prod_{ (i,j) \in E(T) } \tfrac{ B_{i,j}-\frac{\lambda s}{n} }{ \sqrt{ \frac{\lambda s}{n} (1-\frac{\lambda s}{n}) } } \,, \\
    & \prod_{(i,j) \in E(\mathcal L_p)} \tfrac{ A_{i,j}-\frac{\lambda s}{n} }{ \sqrt{ \frac{\lambda s}{n} (1-\frac{\lambda s}{n}) } } : p \in \Lambda \,, \quad  
    \prod_{(i,j) \in E(\mathcal L_q)} \tfrac{ B_{i,j}-\frac{\lambda s}{n} }{ \sqrt{ \frac{\lambda s}{n} (1-\frac{\lambda s}{n}) } } : q \in [m] \setminus \Lambda  \Bigg\}
\end{align*}
is a collection of conditionally independent variables. In addition, using Lemma~\ref{lem-expectation-over-chain} we have for each $p \in \Lambda$ 
\begin{align*}
    &\mathbb E \Bigg[ \prod_{(i,j) \in E(\mathcal L_p)} \tfrac{ A_{i,j}-\frac{\lambda s}{n} }{ \sqrt{ \frac{\lambda s}{n} (1-\frac{\lambda s}{n}) } } \mid \big\{ \sigma_u : u \in \mathtt U \big\} \Bigg] \\
    =\ &\mathbb E \Bigg[ \prod_{(i,j) \in E(\mathcal L_p)} \tfrac{ A_{i,j}-\frac{\lambda s}{n} }{ \sqrt{ \frac{\lambda s}{n} (1-\frac{\lambda s}{n}) } } \mid \big\{ \sigma_u : u \in \operatorname{EndP}(\mathcal{L}_p) \big\} \Bigg] \circeq \big( \tfrac{\epsilon^2 \lambda s}{n} \big)^{|E(\mathcal L_p)|/2} \prod_{u \in \operatorname{EndP}(\mathcal{L}_p)} \sigma_u \,,
\end{align*}
where the last asymptotic equality is from the fact that $(1-\tfrac{\lambda s}{n})^{O(\log n)}=1+o(1)$. The similar result holds for $\mathcal L_q$ for $q \in [m] \setminus \Lambda$ with $A$ replaced by $B$. Thus, we have (for a vertex subset $\mathtt V$, we define $\nu_{\mathtt V}$ to be the uniform distribution of the labelings on $\mathtt V$) 
\begin{align}
    & \mathbb E_{\Pb_{\mathsf{id}}}\Big[ \beta_{S}(A) \beta_K(B) \Big] = \mathbb E_{ \sigma_{\mathtt U} \sim \nu_{\mathtt U} } \mathbb E_{\Pb_{\mathsf{id}}}\Big[ \beta_{S}(A) \beta_K(B) \mid \{ \sigma_u:u \in \mathtt U \} \Big] \nonumber \\
    =\ & \mathbb E_{ \sigma_{\mathtt U} \sim \nu_{\mathtt U} } \Bigg\{ \mathbb E_{\Pb_{\mathsf{id}}}\Big[ \prod_{(i,j) \in E(T)} \Big( \tfrac{ A_{i,j}-\frac{\lambda s}{n} }{ \sqrt{ \frac{\lambda s}{n} (1-\frac{\lambda s}{n}) } } \tfrac{ B_{i,j}-\frac{\lambda s}{n} }{ \sqrt{ \frac{\lambda s}{n} (1-\frac{\lambda s}{n}) } } \Big) \nonumber \\
    & \cdot \prod_{p \in \Lambda} \prod_{ (i,j) \in E(\mathcal L_p) } \tfrac{ A_{i,j}-\frac{\lambda s}{n} }{ \sqrt{ \frac{\lambda s}{n} (1-\frac{\lambda s}{n}) } } \prod_{q \in [m] \setminus \Lambda} \prod_{ (i,j) \in E(\mathcal L_q) } \tfrac{ B_{i,j}-\frac{\lambda s}{n} }{ \sqrt{ \frac{\lambda s}{n} (1-\frac{\lambda s}{n}) } } \mid \sigma_{\mathtt U} \Big] \Bigg\} \nonumber \\
    \circeq\ & \big( \tfrac{\epsilon^2 \lambda s}{n} \big)^{\sum_{p \leq m} |E(\mathcal L_p)|/2} \cdot \mathbb E_{\Pb_{\mathsf{id}}} \Big[ \beta_T(A) \beta_T(B) \prod_{1 \leq j \leq m} \prod_{ u \in \operatorname{EndP}(\mathcal L_j) } \sigma_u \Big] \,, \nonumber
\end{align}
completing the proof.

\subsection{Proof of Lemma~\ref{lem-est-intersection-general-istree-P-Q}}{\label{subsec:proof-lem-3.5}}

We first prove that $\mathbb E_{\Pb_{\pi}}[\beta_{S_1}(A)\beta_{S_2}(B)] \geq 0$ for all $\pi \in \mathfrak S_n$. Without loss of generality, assume $\pi = \mathsf{id}$. Note that 
\begin{align*}
    \mathbb E_{\Pb_{\mathsf{id}}}[\beta_{S_1,S_2}] = \mathbb E_{\Pb_{\mathsf{id}}}\Bigg[ \prod_{ (i,j) \in E(S_1) } \Big( \tfrac{ A_{i,j} - \frac{\lambda s}{n} }{ \sqrt{\frac{\lambda s}{n}(1-\frac{\lambda s}{n})} } \Big)^{ E( S_1)_{i,j} } \prod_{ (i,j) \in E(S_2) } \Big( \tfrac{ B_{i,j} - \frac{\lambda s}{n} }{ \sqrt{\frac{\lambda s}{n}(1-\frac{\lambda s}{n})} } \Big)^{ E(S_2)_{i,j} } \Bigg] \,.
\end{align*}
For all $e=(i,j) \in E(S_1) \cup E(S_2)$, recall that $\chi(e)=(\chi_1(e),\chi_2(e))=( E(S_1)_{i,j}, E(S_2)_{i,j} )$. We have that $\mathbb E_{\Pb_{\mathsf{id}}}[\beta_{S_1,S_2}]$ equals to
\begin{align}\label{eq-proof-lem-positive-expansion-formula-1}
    & \mathbb E_{\sigma\sim\nu} \Bigg[ \prod_{(i,j)\in E(S_1) \cup E(S_2)} \big( u_{\chi_1(i,j),\chi_2(i,j)} + v_{\chi_1(i,j),\chi_2(i,j)} \sigma_i \sigma_j \big) \Bigg] \nonumber \\
    =\ & \sum_{ E \subset E(S_1) \cup E(S_2) } \mathbb E_{\sigma\sim\nu}\Bigg[  \prod_{(i,j) \in E} \big( v_{\chi_1(i,j),\chi_2(i,j)} \sigma_i \sigma_j \big) \prod_{(i,j) \in E(S_1) \cup E(S_2) \setminus E} \big( u_{\chi_1(i,j),\chi_2(i,j)} \big) \Bigg] \,,
\end{align}
which is positive by Lemma~\ref{lem-joint-moment-A-B} and the fact that $\mathbb E\big[ \prod_{(i,j)\in E} \sigma_i \sigma_j \big] \geq 0$. Similarly, we can show that $\mathbb E_{\Pb}[ \beta_{\pi(S_1)}(A)\beta_{S_2}(A) ] \geq 0$ for all $\pi\in\mathfrak S_n$.

Next we prove \eqref{eq-S,K-intersection-general-A,A} and \eqref{eq-S,K-intersection-general-A,B} (which immediately yields that $\mathbb E_{\P_\pi}\big[ \phi_{S_1} (A)\phi_{S_2} (B)\big] \geq 0$ for all $\pi \in \mathcal A(S_1,S_2)$). We only deal with \eqref{eq-S,K-intersection-general-A,B} and \eqref{eq-S,K-intersection-general-A,A} can be derived in the same manner. Without loss of generality assume $\mathsf{id} \in \mathcal A(S_1,S_2)$ and $\pi = \mathsf{id}$. Define $T=S_1 \cap S_2$. Next, for all $I,J \subset [\iota\aleph]$, define
\begin{align*}
    S_1^I = \mathsf{T}(S_1) \cup \big( \cup_{i \in I} \mathsf L_i(S_1) \big), \quad S_2^J = \mathsf{T}(S_2) \cup \big( \cup_{j \in J} \mathsf L_j(S_2) \big)
\end{align*}
and (recall \eqref{eq-def-beta-S} and \eqref{eq-def-psi-S})
\begin{align}
    \phi_{S_1,S_2;I,J}(A,B) =\ & \beta_{\mathsf T(S_1)}(A) \cdot \beta_{\mathsf T(S_2)}(B) \cdot \Big( \prod_{i \in I} \psi_{\mathsf{L}_i(S_1)}(A)\Big) \cdot \Big( \prod_{j \in J} \psi_{\mathsf{L}_j(S_2)}(B)\Big) \nonumber\\
    &\cdot \Big( \prod_{i \not\in I} \beta_{\mathsf L_i(S_1)}(A)\Big) \cdot \Big( \prod_{j \not\in J} \beta_{\mathsf L_j(S_2)}(B)\Big)\,, \label{eq-def-phi-S1I-S2J}
\end{align}
Also, for all $I,J \subset [\iota\aleph]$ let $e=(i,j)\in E(S_1) \cup E(S_2)$ and define 
\begin{align*}
    & \chi^{I,J}(e)= \big( \chi^I_1(e),\chi^J_2(e)\big ) = \big( E(S^I_1)_{i,j}, E(S^J_2)_{i,j} \big) \,, \\
    & T^{I,J} = S_1^I \cap S_2^J \mbox{ and } S_1^I\cup S_2^J=T^{I,J} \sqcup \big( \sqcup_{i=1}^{|I|+|J|} \mathcal L^{I,J}_i \big) \,,
\end{align*}
where $\{ \mathcal{L}^{I,J}_i:1\leq i \leq |I|+|J|\}$ are disjoint self-avoiding paths with 
\begin{align*}
    V(\mathcal L^{I,J}_i) \cap V(T^{I,J})=\operatorname{EndP}(\mathcal L^{I,J}_i) \,.
\end{align*}
Next, we divide the self-avoiding paths in $\{ \mathsf L_i(S_1)\}_{i \in [\iota\aleph]}$ and $\{ \mathsf L_j(S_2)\}_{j \in [\iota\aleph]}$ into several types and deal with them separately. Define
\begin{align}
    & I_0 = \big\{ i \in [\iota\aleph] : \mathsf L_i(S_1)\cap S_2 = \operatorname{EndP}(\mathsf L_i(S_1)) \big\} \,, I_1 = [\iota\aleph] \setminus I_0\,; \label{eq-def-I0-J0-1} \\
    & J_0 = \big\{ j \in [\iota\aleph] : \mathsf L_j(S_2)\cap S_1 = \operatorname{EndP}(\mathsf L_j(S_2)) \big\} \,, J_1 = [\iota\aleph] \setminus J_0\,. \label{eq-def-I0-J0-2}
\end{align}
Also define
\begin{align}
    & I'_0 = \big\{ i \in [\iota\aleph] : \mathsf L_i(S_1)\cap S_2 = \operatorname{EndP}(\mathsf L_i(S_1)) \not\in \mathsf P(S_2) \big\} \,; \label{eq-def-I'0-J'0-1} \\
    & J'_0 = \big\{ j \in [\iota\aleph] : \mathsf L_j(S_2)\cap S_1 = \operatorname{EndP}(\mathsf L_j(S_2)) \not\in \mathsf P(S_1) \big\} \,. \label{eq-def-I'0-J'0-2}
\end{align}
\begin{figure}[ht]
    \centering
    \vspace{0cm}
    \includegraphics[height=3.39cm,width=13.76cm]{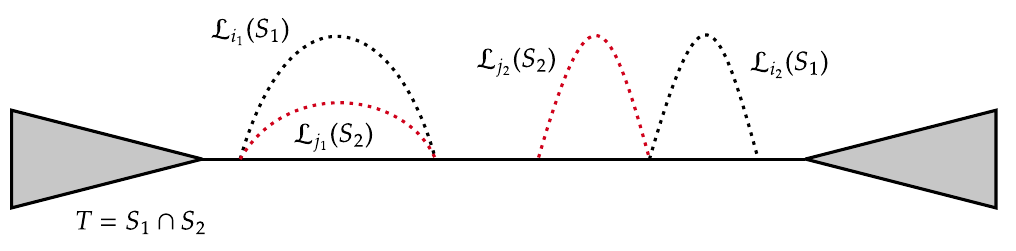}
    \caption{Illustration of $I_0,J_0,I_0',J_0'$: above $i_1 \in I_0,j_1 \in J_0,i_2\in I_0',j_2 \in J_0'$}
    \label{fig:Type_I_0}
\end{figure}
(see Figure~\ref{fig:Type_I_0} for an illustration.) Then by \eqref{eq-def-K1-s1-S2-pi}, \eqref{eq-def-K2-s1-S2-pi} and \eqref{eq-def-Kc-s1-S2-pi} we have $|I_0'|=\mathtt K_1(S_1,S_2;\pi)$, $|J_0'|= \mathtt K_2(S_1,S_2;\pi)$ and $|I_0 \setminus I'_0| = \mathtt K_c(S_1,S_2;\pi)$. We first deal with the paths in $I_0$ and $J_0$ via the following lemma:

\begin{lemma}{\label{lem-phi-S1-S2-reduction}}
    There exists $h_n = O(\tfrac{1}{\sqrt{n}})$ such that 
    \begin{align}
        \mathbb E_{\Pb_{\mathsf{id}}}\big[\phi_{S_1,S_2}\big] \circeq h_n^{|I_0'|+|J_0'|} (1-\epsilon^{2\mathtt m})^{|I_0\setminus I_0'|} \cdot \mathbb E_{\Pb_{\mathsf{id}}}\Big[ \phi_{S_1,S_2;I_1,J_1}(A,B) \Big] \,. \label{eq-phi-S1-S2-reduction-expression}
    \end{align}
\end{lemma}
\begin{proof}
    We will prove \eqref{eq-phi-S1-S2-reduction-expression} by showing that each path in $I_0' \cup J_0'$ contributes a factor of $h_n$ and each path in $(I_0 \setminus I_0') \cup (J_0 \setminus J_0')$ contributes a factor of $1-\epsilon^{\mathtt m}$ via a case-by-case analysis. Without loss of generality, assume $1 \in I_0$. Denote $\operatorname{EndP}(\mathsf L_1(S_1))=\{u,v\}$ and $\mathfrak p$ is the shortest path connecting $u$ and $v$ in $T$. Then by Theorem~\ref{thm-desired-vertex-sets}, there exist no vertex $w \in V(\mathfrak p) \setminus \{ u,v \}$ such that $w \in \mathsf{Vert}(\mathsf P(S_1)) \cup \mathsf{Vert}(\mathsf P(S_2))$, leading to the fact that both $u$ and $v$ are cut points of $S_1 \cup S_2$ (meaning that both $u$ and $v$ divide $S_1\cup S_2$ into two disjoint components). Therefore, we can divide $1\in I_0$ by two cases:

    \noindent{\bf Case 1:} $1 \in I_0'$. Then we have (denote $\mathtt G = (S_1 \cup S_2) \setminus (\mathfrak p \cup \mathsf L_1(S_1))$)
    \begin{align}
        & \mathbb E_{\Pb_{\mathsf{id}}}[\phi_{S_1,S_2}(A,B)]\nonumber\\
        =\ & \mathbb E_{\Pb_{\mathsf{id}}}\Bigg[ \mathbb E_{\Pb_{\mathsf{id}}} \Big[\beta_{\mathsf T(S_1)\setminus \mathfrak p}(A)\beta_{\mathsf T(S_2)\setminus \mathfrak p}(B)\cdot \Big(\prod_{2 \leq i \leq \iota\aleph} \psi_{\mathsf{L}_i(S_1)}(A)\Big)\cdot \Big(\prod_{1 \leq j \leq \iota\aleph} \psi_{\mathsf{L}_j(S_2)}(B)\Big) \mid \sigma_u ,\sigma_v \Big]\nonumber \\
        & \cdot\mathbb E_{\Pb_{\mathsf{id}}} \Big[\psi_{S_1,1}(A)\beta_\mathfrak p(A)\beta_\mathfrak p(B) \mid A_e , B_e, \sigma_w: e \in E(\mathtt G), w \in V(\mathtt G)\Big] \Bigg] \nonumber \\
        =\ & \mathbb E_{\Pb_{\mathsf{id}}}\Bigg[\mathbb E_{\Pb_{\mathsf{id}}}\Big[\beta_{\mathsf T(S_1)\setminus \mathfrak p}(A)\beta_{\mathsf T(S_2)\setminus \mathfrak p}(B)\cdot \Big(\prod_{2 \leq i \leq \iota\aleph} \psi_{S_1,i}(A)\Big)\cdot \Big(\prod_{1 \leq j \leq \iota\aleph} \psi_{S_2,j}(B)\Big) \mid \sigma_u ,\sigma_v \Big]\label{eq-u-v-cut-graph-1} \\
        & \cdot \mathbb E_{\Pb_{\mathsf{id}}}\Big[\psi_{\mathsf{L}_1(S_1)}(A)\beta_\mathfrak p(A)\beta_\mathfrak p(B) \mid \sigma_u , \sigma_v \Big] \Bigg] \,,\label{eq-u-v-cut-graph-2}
    \end{align}
    Since $u,v$ are cut points of $S_1$ and $S_2$, $\eqref{eq-u-v-cut-graph-1}$ does not depend on $\sigma_u$ and $\sigma_v$. Also, \eqref{eq-strong-second-bound-u-v-r,t} yields that when $1 \in I_0'$ we have
    \begin{align*}
        \mathbb E_{\Pb_{\mathsf{id}}} \Big[\psi_{\mathsf{L}_1(S_1)}(A) \beta_\mathfrak p(A)\beta_\mathfrak p(B) \mid \sigma_u , \sigma_v \Big] &= \mathbb E_{\Pb_{\mathsf{id}}}\Big[ (\tfrac{\epsilon^2 \lambda s}{n})^{\ell/2}\big((\sigma_u \sigma_v-\epsilon^\mathtt m )(1 + \epsilon^\mathtt m \sigma_u\sigma_v )+h_n\big) \mid \sigma_u , \sigma_v \Big] \nonumber \\
        &= \big(\tfrac{\epsilon^2 \lambda s}{n} \big)^{\ell/2}\big((1-\epsilon^{2\mathtt m})\sigma_u \sigma_v+h_n\big)  \,,
    \end{align*}
    where $h_n=O(\tfrac{1}{\sqrt{n}})$. Plugging this into \eqref{eq-u-v-cut-graph-2} yields
    \begin{align}
        \mathbb E_{\Pb_{\mathsf{id}}}[\phi_{S_1,S_2}(A,B)] &= \mathbb E_{\Pb_{\mathsf{id}}} \Big[ \eqref{eq-u-v-cut-graph-1} \cdot \big( \tfrac{\epsilon^2 \lambda s}{n} \big)^{\ell/2} \big((1-\epsilon^{2\mathtt m}) \sigma_u \sigma_v+h_n\big) \Big] \nonumber \\
        &= \mathbb E_{\Pb_{\mathsf{id}}} \Big[ \eqref{eq-u-v-cut-graph-1} \Big] \cdot \mathbb E_{\Pb_{\mathsf{id}}} \Big[ \big( \tfrac{\epsilon^2 \lambda s}{n} \big)^{\ell/2} \big((1-\epsilon^{2\mathtt m}) \sigma_u \sigma_v+h_n\big) \Big] \nonumber \\
        &= h_n\mathbb E_{\Pb_{\mathsf{id}}} \Big[ \phi_{
        S_1,S_2;[\iota\aleph]\setminus \{ 1\},[\iota\aleph]}(A,B) \Big]\,. \label{eq-phi-S1-S2-reduction-arguments}
    \end{align}
    \textbf{Case 2: $1 \in  I_0 \setminus I_0'$}. Without loss of generality, we may assume that $\mathsf P_1(S_1) = \mathsf P_{1}(S_2)$. Then similarly to Case 1 we have
    \begin{align}
        &\mathbb E_{\Pb_{\mathsf{id}}}\big[ \phi_{S_1,S_2}(A,B) \big] \nonumber \\
        =\ & \mathbb E_{\Pb_{\mathsf{id}}}\Bigg[ \mathbb E_{\Pb_{\mathsf{id}}} \Big[ \beta_{\mathsf T(S_1)\setminus \mathfrak p}(A)\beta_{\mathsf T(S_2)\setminus \mathfrak p}(B) \prod_{2 \leq i \leq \iota\aleph} \psi_{\mathsf{L}_i(S_1)}(A) \prod_{2 \leq j \leq \iota\aleph} \psi_{\mathsf{L}_j(S_2)}(B) \mid \sigma_u ,\sigma_v \Big] \\
        &\cdot \mathbb E_{\Pb_{\mathsf{id}}}\Big[ \psi_{\mathsf{L}_1(S_1)}(A) \psi_{\mathsf{L}_1(S_2)}(B) \beta_\mathfrak p(A)\beta_\mathfrak p(B) \mid \sigma_u,\sigma_v \Big] \Bigg]
    \end{align}
    and
    \begin{align*}
        &\mathbb E_{\Pb_{\mathsf{id}}}\Big[\psi_{\mathsf{L}_1(S_1)}(A) \psi_{\mathsf{L}_1(S_2)}(B) \beta_\mathfrak p(A)\beta_\mathfrak p(B) \mid \sigma_u,\sigma_v \Big]  \\
        =\ & \big( \tfrac{\epsilon^2 \lambda s}{n} \big)^{\ell/2}\big((\sigma_u \sigma_v-\epsilon^\mathtt m )^2(1 + \epsilon^\mathtt m \sigma_u\sigma_v )+O(n^{-1/2})\big)   \\
        =\ & \big( \tfrac{\epsilon^2 \lambda s}{n} \big)^{\ell/2} \big((1-\epsilon^{2\mathtt m}+O(n^{-1/2}))-(\epsilon^\mathtt m (1-\epsilon^{2\mathtt m})+O(n^{-1/2}))\sigma_u \sigma_v\big) \,.
    \end{align*}
    Similarly as in \eqref{eq-phi-S1-S2-reduction-arguments}, we get that
    \begin{align}
        \mathbb E_{\Pb_{\mathsf{id}}}\big[ \phi_{S_1,S_2}(A,B) \big] =(1-\epsilon^{2\mathtt m}+O(n^{-1/2})) \cdot \mathbb E_{\Pb_{\mathsf{id}}}\big[ \phi_{[\iota\aleph]\setminus \{ 1\},[\iota\aleph] \setminus \{ j_0\}}(A,B) \big] \,. \label{eq-phi-S1-S2-reduction-arguments-2}
    \end{align}
    Repeating our arguments in \eqref{eq-phi-S1-S2-reduction-arguments} yields \eqref{eq-phi-S1-S2-reduction-expression}.
\end{proof}

Next we deal with the right-hand side of \eqref{eq-phi-S1-S2-reduction-expression}. By \eqref{eq-def-psi-S} and \eqref{eq-def-phi-S1I-S2J} we have 
\begin{align}
    \mathbb E_{\Pb_{\mathsf{id}}}[\phi_{S_1,S_2;I_1,J_1}] =\ & \sum_{I_0 \subset I\subset [\iota\aleph] , J_0 \subset J\subset [\iota\aleph]} \E_{\Pb_{\mathsf{id}}} [\beta_{S_1^I , S_2^J}] \big(- \epsilon^\mathtt m (\tfrac{\epsilon^2\lambda s}{n})^{\ell/2}\big)^{|I_1 \setminus I| + |J_1 \setminus J|}\nonumber \\
    =\ & \E_{\Pb_{\mathsf{id}}} [\beta_{S_1 , S_2}] + \sum_{\substack{I_0 \subset I\subset [\iota\aleph] , J_0 \subset J\subset [\iota\aleph] \\(I,J)\neq([\iota\aleph] , [\iota\aleph])}} \big( -\epsilon^\mathtt m (\tfrac{\epsilon^2\lambda s}{n})^{\ell/2} \big)^{|I_1 \setminus I| + |J_1 \setminus J|} \E_{\Pb_{\mathsf{id}}}\big[ \beta_{S_1^I,S_2^J} \big]  \,. \label{eq-expansion-phi-S1-I1-S2-J1-1}
\end{align}
Thus, it suffices to show that
\begin{align}
    \sum_{\substack{I_0 \subset I\subset [\iota\aleph] , J_0 \subset J\subset [\iota\aleph] \\(I,J)\neq([\iota\aleph] , [\iota\aleph])}} \Big( \epsilon^\mathtt m (\tfrac{\epsilon^2\lambda s}{n})^{\ell/2} \Big)^{|I_1 \setminus I|+|J_1 \setminus J|} \cdot \E_{\Pb_{\mathsf{id}}} \big[ \beta_{S_1^I,S_2^J} \big] = o(1) \cdot \E_{\Pb_{\mathsf{id}}} [\beta_{S_1,S_2}] \,.  \label{eq-goal-main-term-phi-S-I,J}
\end{align}
To this end, first observe that 
\begin{align*}
    \mbox{each vertex in } V(T^{I,J}) \mbox{ occurs in at most two elements of } \{ \operatorname{EndP}(\mathcal L^{I,J}_i): 1 \leq i \leq |I|+|J| \} \,.
\end{align*}
Otherwise, suppose $w \in V(T^{I,J})$ occurs in at least three of $\{\operatorname{EndP}(\mathcal L^{I,J}_i): 1 \leq i \leq |I|+|J| \}$. In this case, since each $\mathcal L^{I,J}_i$ belongs to different paths in $\{ \mathsf L_1 (S_1) , \ldots , \mathsf L_{\iota\aleph}(S_1) ; \mathsf L_1 (S_2) , \ldots , \mathsf L_{\iota\aleph}(S_2) \}$, there exist $p, p' \in [\iota\aleph],q \in \{1,2\}$ such that $w \in V(\mathsf L_p(S_q) )\cap V(\mathsf L_{p'}(S_q) )$, which contradicts to Theorem~\ref{thm-desired-vertex-sets}. Next, for any $I,J \subset [\iota\aleph]$ we denote 
\begin{align*}
    \mathtt V^{I,J}  = \triangle_{1\leq i \leq |I|+|J|} \Big\{ \operatorname{EndP}(\mathcal L^{I,J}_i) \Big\} \,,
\end{align*}
as the vertices $w \in V(T^{I,J})$ that occur in exactly one of $\{\operatorname{EndP}(\mathcal L^{I,J}_i)\}_{1 \leq i \leq |I|+|J|}$. We see that $\#\mathtt V^{ I,J}$ is even. Denote (recall the definition of $\mathcal X(\cdot)$ from \eqref{def-mathcal-X-of-sets})
\begin{align}
    \mathtt d^{I,J} = \min \Big\{ \sum_{1 \leq i \leq |V^{I,J}|/2} \mathsf{Dist}_{T^{I,J}}(x_i , y_i) : ( x_1 , \ldots ,x_{|\mathtt V^{I,J}|/2}, y_1 , \ldots ,y_{|\mathtt V^{I,J}|/2} )\in \mathcal X( V^{I,J}) \Big\}\,,\label{eq-def-mathtt-d-I,J}
\end{align}
and define
\begin{align}
    &I_2 = I_0 \cup \{ i \in I_1 : E(\mathsf L_i(S_1)\cap \mathsf T( S_2 ))\neq \emptyset\}\,, I_3 = I_1 \setminus I_2\,;\label{eq-def-I1-J1-1}\\
    &J_2 = J_0 \cup \{ j \in J_1 : E(\mathsf L_j(S_2)\cap \mathsf T( S_1 ))\neq \emptyset\}\,, J_3 = J_1 \setminus J_2\,.\label{eq-def-I1-J1-2}
\end{align}
\begin{figure}[ht]
    \centering
    \vspace{0cm}
    \includegraphics[height=3.84cm,width=13.87cm]{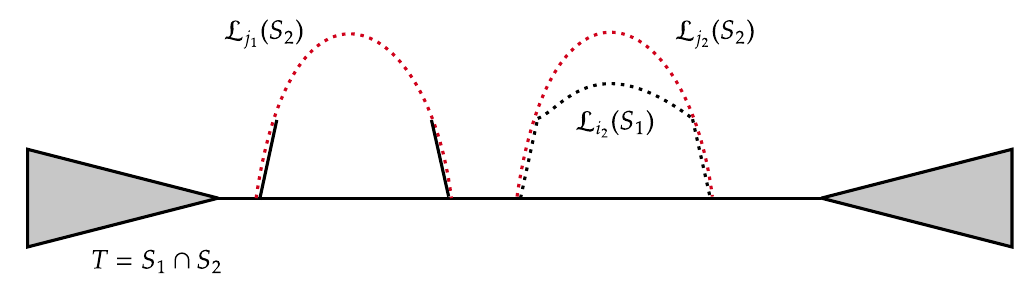}
    \caption{Illustration of $I_2,J_2,I_3,J_3$: above $j_1 \in J_2,i_2\in I_3,j_2 \in J_3$}
    \label{fig:Type_I_2}
\end{figure}
(See Figure~\ref{fig:Type_I_2} for an illustration.)
For $I\subset [\iota\aleph],J \subset [\iota\aleph]$ such that $I_2 \not\subset I$, there exist $i_2 \in I_2 \setminus I$ and $e_2 \in E(\mathsf L_{i_2}(S_1)\cap \mathsf T( S_2 ))$ such that $e_2$ contains an endpoint $v_2 \in \mathcal L(\mathsf T(S_2))$. Then by the definition of $\mathfrak R^*_\mathbf H$ and $\mathcal{A}(S_1,S_2)$, we have
\begin{align*}
    e_2 \not\in E(\cup_{i \in [\iota\aleph]} \mathsf L_{i}(S_2)) \mbox{ and } e_2 \not\in E(\cup_{i \in [\iota\aleph], i \neq i_2} \mathsf L_{i}(S_1)) \,,
\end{align*}
which yields $v_2 \in\mathcal L(S_1^I \cup S_2^J) \neq \emptyset$ and therefore $\E_{\P_{\mathsf{id}}} [ \beta_{S_1^I , S_2^J}(A,B)] = 0$. Therefore, by Lemmas~\ref{lem-est-intersection-istree-P-Q} and \ref{lem-delicate-exp-upper-bound-on-trees} we have
\begin{align}
    \E_{\Pb_{\mathsf{id}}} [\beta_{S_1,S_2}] \circeq s^{|E(T^{ [\iota\aleph], [\iota\aleph]})|}(\tfrac{\epsilon^2\lambda s}{n})^{\tfrac{1}{2}\sum_{1\leq i \leq 2\iota\aleph}|E(\mathcal L^{ [\iota\aleph], [\iota\aleph]}_i)|}\epsilon^{\mathtt d^{ [\iota\aleph], [\iota\aleph]}} \label{eq-expansion-beta-S1-S2-main-term}
\end{align}
and 
\begin{align}
    &\sum_{\substack{I_0 \subset I\subset [\iota\aleph] , J_0 \subset J\subset [\iota\aleph] \\(I,J)\neq([\iota\aleph] , [\iota\aleph])}} \Big( \epsilon^\mathtt m (\tfrac{\epsilon^2\lambda s}{n})^{\ell/2} \Big)^{|I_1 \setminus I|+|J_1 \setminus J|} \cdot \E_{\Pb_{\mathsf{id}}} \big[ \beta_{S_1^I,S_2^J} \big] \nonumber \\
    \circeq\ & \sum_{\substack{I\subset [\iota\aleph] , J\subset [\iota\aleph]  \\(I,J)\neq([\iota\aleph] ,[\iota\aleph] )\\I_2 \subset I , J_2 \subset J}}s^{|E(T^{I,J})|}(\tfrac{\epsilon^2\lambda s}{n})^{\tfrac{1}{2}\sum_{1\leq i \leq |I|+|J|}|E(\mathcal L^{I,J}_i)|}\epsilon^{\mathtt d^{I,J}} \big( \epsilon^\mathtt m (\tfrac{\epsilon^2\lambda s}{n})^{\ell/2}\big)^{| [\iota\aleph] \setminus I| + | [\iota\aleph] \setminus J|}\,. \label{eq-expansion-phi-S1-I1-S2-J1-2}
\end{align}
Next we prove that for any $I\subset  [\iota\aleph] , J\subset  [\iota\aleph]$ such that $(I,J)\neq( [\iota\aleph],  [\iota\aleph])$, $I_2 \subset I$ and $J_2 \subset J$, we have $ \mathtt d^{ [\iota\aleph], [\iota\aleph]} - \mathtt d^{I,J} =\mathtt m (| [\iota\aleph] \setminus I| + | [\iota\aleph] \setminus J|)$. This will be a direct consequence of the following lemma.

\begin{lemma}\label{lem-label-matching-exact}
For any $I,J \subset [\iota\aleph]$ denote 
\begin{align*}
    \mathtt f^{I,J} = \big|E\big((\cup_{i \in I}\mathsf L_i(S_1) )\cap \mathsf T(S_2)\big)\big| + \big|E\big((\cup_{j \in J}\mathsf L_j(S_2) )\cap \mathsf T(S_1)\big)\big| \,.
\end{align*}
Then we have 
\begin{equation}
    \mathtt d^{I,J} = \mathtt f^{I,J} + \mathtt m(|I|+|J|)\,.\label{eq-label-matching-exact}
\end{equation}
\end{lemma}

The proof of this lemma is incorporated in Subsection~\ref{subsec:proof-lem-C-exact}. By \eqref{eq-def-I1-J1-1}, \eqref{eq-def-I1-J1-2} and the definition of $\mathtt f^{I,J}$, we have $\mathtt f^{I,J} = \mathtt f^{I_2,J_2}$ for any $I\subset [\iota\aleph] , J\subset [\iota\aleph]$ such that $(I,J)\neq([\iota\aleph],[\iota\aleph])$, $I_2 \subset I$ and $J_2 \subset J$. In addition,  By Lemma~\ref{lem-label-matching-exact} we have
\begin{align}
    \eqref{eq-expansion-phi-S1-I1-S2-J1-2} =\ & \epsilon^{\mathtt f^{I_2,J_2} + 2\mathtt m \iota\aleph} \sum_{\substack{I\subset [\iota\aleph] , J\subset [\iota\aleph]  \\(I,J)\neq([\iota\aleph] ,[\iota\aleph] )\\I_2 \subset I , J_2 \subset J}}s^{|E(T^{I,J})|}(\tfrac{\epsilon^2\lambda s}{n})^{\tfrac{1}{2}\sum_{1\leq i \leq |I|+|J|}|E(\mathcal L^{I,J}_i)|}\big(  (\tfrac{\epsilon^2\lambda s}{n})^{\ell/2}\big)^{|[\iota\aleph] \setminus I| + |[\iota\aleph] \setminus J|}\Big) \nonumber \\
    =\ & \epsilon^{\mathtt f^{I_2,J_2} + 2\mathtt m \iota\aleph} \sum_{\substack{I\subset [\iota\aleph] , J\subset [\iota\aleph]  \\(I,J)\neq([\iota\aleph] ,[\iota\aleph] )\\I_2 \subset I , J_2 \subset J}} s^{\aleph-1+\ell\iota\aleph} (\tfrac{\epsilon^2\lambda }{n})^{\tfrac{1}{2}\sum_{1\leq i \leq |I|+|J|}|E(\mathcal L^{I,J}_i)|}\big(  (\tfrac{\epsilon^2\lambda}{n})^{\ell/2}\big)^{|[\iota\aleph] \setminus I| + |[\iota\aleph] \setminus J|}\Big) \nonumber \\
    \leq\ & \epsilon^{\mathtt f^{I_2,J_2} + 2\mathtt m \iota\aleph} s^{\aleph-1+\ell\iota\aleph}(\tfrac{\epsilon^2\lambda }{n})^{\tfrac{1}{2}\sum_{1\leq i \leq 2\iota\aleph}|E(\mathcal L^{[\iota\aleph],[\iota\aleph]}_i)|} \cdot 2^{\iota\aleph} (\tfrac{\epsilon^2\lambda }{n})^{1/2} \nonumber \\
    =\ & o(1)\cdot \mathbb E_{\Pb_{\mathsf{id}}}[\beta_{S_1,S_2}(A,B)] \,, \label{eq-expansion-phi-S1-I1-S2-J1-3}
\end{align}
where the second equality follows from
\begin{align*}
    2 |E(T^{I,J})| + \sum_{1 \leq i \leq |I|+|J|} |E(\mathcal L_i^{I,J})| &= |E(S_1^{I,J})| + |E(S_2^{I,J})| \\
    &= 2(\aleph-1+\ell\iota\aleph)- \ell\cdot |[\iota\aleph] \setminus I| - \ell\cdot |[\iota\aleph] \setminus J| 
\end{align*}
and the inequality follows from (note from \eqref{eq-def-I1-J1-1} that $|E(\mathcal L^{[\iota\aleph],[\iota\aleph]}_i)|<\ell$ for all $i \not\in I_2 \cup J_2$)
\begin{align*}
    \sum_{1\leq i \leq 2\iota\aleph} |E(\mathcal L^{[\iota\aleph],[\iota\aleph]}_i)| - \sum_{1\leq i \leq |I|+|J|}|E(\mathcal L^{I,J}_i)| = \sum_{ i \not \in I \cup J } |E(\mathcal L^{[\iota\aleph],[\iota\aleph]}_i)| < \ell( |[\iota\aleph] \setminus I| + |[\iota\aleph] \setminus J| )
\end{align*}
for $(I,J) \neq ([\iota\aleph],[\iota\aleph])$ and $I_2 \subset I, J_2 \subset J$.

\subsection{Proof of Lemma~\ref{lem-character-non-principle}}{\label{subsec:proof-lem-3.10}}

We first deal with the case that $S_1,S_2 \in \mathfrak R_{\mathbf H}^*$ and $\mathsf{id} \not\in\mathcal A(S_1,S_2)$. Recall in this case $S_1,S_2$ are just simple graphs. Denote $H=S_1 \cap S_2$. Note that $\tau(S_1)=\tau(S_2)=\iota\aleph-1$ and $|E(S_1)|=|E(S_2)|=\aleph-1+\ell\iota\aleph$. Then we have 
\begin{align*}
    \tau(G_{\cup}) = 2\iota\aleph-2-\tau(H) \mbox{ and } |E(G_{\cup})| = 2(\aleph-1)+ 2\ell\iota\aleph - |E(H)| \,,
\end{align*}
and thus it suffices to show that in this case we have 
\begin{equation}{\label{eq-goal-lem-3.9-case-I}}
    \tau(H) \leq \tfrac{|E(H)|}{\ell/2}-2 <  \tfrac{|E(H)|+2}{\ell/2}-2 \,.
\end{equation}
To this end, denote $H_0 = H \cap \mathsf T(S_1)$ and $H_i = H_{i-1} \cup (H \cap \mathsf L_i(S_1))$ for $i=1,\dots,\iota \aleph$ (thus $H_{\iota\aleph}=H$). Since $S_1 \in \mathfrak R_{\mathbf H}^*$ implies that $V(\mathsf L_{i+1}(S_1)) \cap V(H_{i}) = \operatorname{EndP}(\mathsf L_{i+1}(S_1))$, we have $\tau(H_{i+1}) \leq \tau(H_i)+1$, with equality holding if and only if $\mathsf L_{i+1}(S_1) \subset H$. In addition, note that
\begin{align*}
    \mathsf L_{i+1}(S_1) \subset H \Longrightarrow |E(H_{i+1})| = |E(H_i)|+\ell \,.
\end{align*}
Thus, we have 
\begin{equation}{\label{eq-lem-3.9-case-I-induction}}
    \tau(H_{i+1})-\tau(H_i) \leq \tfrac{|E(H_{i+1})|-|E(H_i)|}{\ell/2} - \mathbf 1_{ \{ \mathsf L_{i+1}(S_1) \subset H \} } \,.
\end{equation}
Rearranging the above inequality and summing over $i$ yields
\begin{align}\label{eq-tauH-minus-EH-bound}
    \tau(H) - \tfrac{|E(H)|}{\ell/2} = \tau(H_{\iota\aleph}) - \tfrac{|E(H_{\iota\aleph})|}{\ell/2} &\leq \tau(H_{0}) - \tfrac{|E(H_{0})|}{\ell/2} - \#\{ i: \mathsf L_{i}(S_1) \subset H \}\notag \\
    &\leq \tau(H_{0}) - \#\{ i: \mathsf L_{i}(S_1) \subset H \} \,.
\end{align}
By Item (3) of Lemma~\ref{lem-useful-property-trees-and-sets} and the assumption that $\mathcal L(S_1) \cup \mathcal L(S_2) \subset V(S_1 \cap S_2)$, we have that $\mathcal L(\mathsf T(S_1)) = \mathcal L(S_1) \subset V(H_0)$. Also, $\mathsf{id} \not\in\mathcal A(S_1,S_2)$ implies that 
\begin{align}{\label{eq-lem-3.9-exists-j}}
    \exists j\in\{ 1,2\},\mbox{ either } \mathsf T(S_j) \neq H_0 \mbox{ or } \{ i: \mathsf L_{i}(S_j) \subset H \} \neq \emptyset \,; 
\end{align}
this is because otherwise $H = H_{\iota \aleph}$ will be a tree containing $\mathsf T(S_1)$ and $\mathsf T(S_2)$, which contradicts to $\mathsf{id} \not\in\mathcal A(S_1,S_2)$. Without loss of generality, assume $j=1$ in \eqref{eq-lem-3.9-exists-j}. When $\mathsf T(S_1) \neq H_0$, by $H_0 \subset \mathsf T(S_1)$ and $\mathcal L(\mathsf \mathsf T(S_1)) \subset V(H_0)$, all connected components of $H_0$ must be trees and $H_0$ cannot be connected, which yields $\tau(H_0) \leq -2$. When $\mathsf T(S_1)=H_0$, by \eqref{eq-lem-3.9-exists-j}, we have that $\tau(H_0)=-1$ and $\{ i: \mathsf L_{i}(S_1) \subset H \}\neq \emptyset$. Therefore, when $\mathsf{id}\notin \mathcal{A}(S_1,S_2)$ either $\tau(H_0) \leq -2$ or $\{\tau(H_0)=-1\}\cap\{ \{ i: \mathsf L_{i}(S_1) \subset H \} \neq \emptyset \}$ holds, leading to
\begin{equation}{\label{eq-lem-3.9-case-I-induction-0}}
    \tau(H_{0}) - \#\{ i: \mathsf L_{i}(S_1) \subset H \} \leq -2 \,.
\end{equation}
Then \eqref{eq-goal-lem-3.9-case-I} follows by plugging the above estimate into \eqref{eq-tauH-minus-EH-bound}. Now we deal with the case that $S_1 \not\in \mathfrak R_{\mathbf H}^*$ (the case that $S_2 \not\in \mathfrak R_{\mathbf H}^*$ can be treated similarly). In this case we have $\tau(\widetilde S_1) \geq \iota\aleph$. Denoting $H = \widetilde S_1 \cap \widetilde S_2$, we claim the following
\begin{equation}{\label{eq-goal-lem-3.9-case-II}}
    \tau(H) - \tfrac{|E(H)|+\aleph}{\ell} \leq \min\Big\{ \tau(\widetilde S_1) - \tfrac{|E(\widetilde S_1)|}{\ell}, \tau(\widetilde S_2) - \tfrac{|E(\widetilde S_2)|}{\ell} \Big\}  \,.
\end{equation}
To prove \eqref{eq-goal-lem-3.9-case-II}, denote $H_{\mathsf T} = H \cup \mathsf{T}(S_1)$. Since $\mathcal L(\widetilde S_1) \subset V(H)$, we then have 
\begin{equation}{\label{eq-comparison-H_1-H}}
    \tau(H_{\mathsf T}) \geq \tau(H) \mbox{ and } |E(H_{\mathsf T})| \leq |E(H)|+\aleph \,.
\end{equation}
In addition, since $\mathcal L(\widetilde S_1) \subset V(H) \subset V(H_{\mathsf T})$, by using Lemma~\ref{lem-decomposition-H-plus-path}, we see that $\widetilde S_1 \setminus H_{\mathsf T}$ can be decomposed into $\mathtt t= \tau(\widetilde S_1)-\tau(H_{\mathsf T})$ self-avoiding paths $P_{\mathtt 1},\ldots,P_{\mathtt t}$ satisfying Items (1)--(4) of Lemma~\ref{lem-decomposition-H-plus-path}. Clearly we have $|E(P_{\mathtt i})| \leq \ell$, and thus 
\begin{equation}{\label{eq-compaison-H_1-S_1}}
    |E(\widetilde S_1)|-|E(H_{\mathsf T})| \leq \ell \mathtt t \leq \ell(\tau(\widetilde S_1)-\tau(H_{\mathsf T})) \,.
\end{equation}
In conclusion, we have shown that 
\begin{align*}
    \tau(H)- \tfrac{|E(H)|+\aleph}{\ell} \overset{\eqref{eq-comparison-H_1-H}}{\leq} \tau(H_{\mathsf T})- \tfrac{|E(H_{\mathsf T})|}{\ell} \overset{\eqref{eq-compaison-H_1-S_1}}{\leq} \tau(\widetilde S_1)- \tfrac{|E(\widetilde S_1 )|}{\ell} \,.
\end{align*}
Similar results also hold for $\widetilde S_2$, leading to \eqref{eq-goal-lem-3.9-case-II}. Thus,
\begin{align*}
    \tau(G_{\cup}) - \tfrac{|E(G_{\cup})|-\aleph}{\ell} &= \big( \tau(\widetilde S_1) - \tfrac{|E(\widetilde S_1)|}{\ell} \big) + \big( \tau(\widetilde S_2) - \tfrac{|E(\widetilde S_2)|}{\ell} \big) - \big( \tau(H) - \tfrac{|E(H)|+\aleph}{\ell} \big)  \\
    & \overset{\eqref{eq-goal-lem-3.9-case-II}}{\geq} \max\Big\{ \tau(\widetilde S_1) - \tfrac{|E(\widetilde S_1)|}{\ell}, \tau(\widetilde S_2) - \tfrac{|E(\widetilde S_2)|}{\ell} \Big\} \geq \tau(\widetilde S_1) - \tfrac{|E(\widetilde S_1)|}{\ell} \geq - \tfrac{\aleph}{\ell} \,,
\end{align*}
where the second inequality follows from $\tau(\widetilde S_1) \geq \iota\aleph$ and $|E(\widetilde S_1)| \leq \ell\iota\aleph+\aleph$. Using the fact that $|E(G_{\cup})| \leq |E(\widetilde S_1)| + |E(\widetilde S_2)| \leq 2\ell\iota\aleph+2\aleph$, we obtain
\begin{align*}
    \tau(G_{\cup}) - \tfrac{|E(G_{\cup})|}{\ell/2} &= \Big( \tau(G_{\cup}) - \tfrac{|E(G_{\cup})|-\aleph}{\ell} \Big) - \tfrac{|E(G_{\cup})|+\aleph}{\ell} \\
    &\geq -\tfrac{2\aleph}{\ell} - \tfrac{2\ell\iota\aleph+2\aleph}{\ell} = -2\iota\aleph - \tfrac{4\aleph}{\ell} \,,
\end{align*}
as desired.

\subsection{Proof of Lemma~\ref{lem-enu-mathfrak-Q-mathfrak-U}}{\label{subsec:proof-lem-4.3}}

To show \eqref{eq-enum-frakchoose-pairings-Q-H,I}, we shall enumerate $(\mathbf I , K)$ satisfying $K \in \mathfrak Q_{\mathbf H,\mathbf I}(\{ x_u \},y) \cup  \mathfrak U_{\mathbf H,\mathbf I}(\{ x_u \},y)$ given any $\mathbf H \in \mathcal H$ and $S \in \mathfrak R^*_{\mathbf H}$. First, there are at most $\tbinom{\iota\aleph}{y} \leq 2^{\iota\aleph}$ ways to choose 
\begin{align}\label{eq-def-Ind-S,K}
\mathsf{Ind}_{S,K} :=\{ i \in [\iota\aleph] : E(\mathsf L_i (K)\cap \mathsf T(S)) \neq \emptyset\}\,,
\end{align}
and there are at most $\aleph^{2\iota\aleph}$ ways to determine $\{\mathsf P_i(K) : i \in [\iota\aleph] \}$. Then, to choose $\mathsf P_i(K) = (u_i , v_i)$ given $i\in \mathsf{Ind}_{S,K}$, we first prove that $(u_i , v_i) \in \mathfrak{Ch}^{(a,b)}_\mathtt m (\mathsf T(S))$ hold for some $0 \leq a,b \leq 1$. By Theorem~\ref{thm-desired-vertex-sets}, there exists $w_1 , \ldots , w_{\mathtt m-1} \in V(\mathsf T(K))$ such that $\mathfrak p_{\mathsf T(K)}(u_i,v_i) = (u_i , w_1 , \ldots , w_{\mathtt m -1 } , v_i)$, $\mathsf{Branch}_{\mathsf T(K)}(u_i) = \mathsf{Branch}_{\mathsf T(K)}(v_i) = 2$ and $\mathsf{Deg}_{\mathsf T(K)}(w_j) = 2$ for $1 \leq j \leq \mathtt m-1$. Also, since the shortest path connecting two vertices on a tree is unique we have $\mathfrak p_{\mathsf T(S)}(u_i,v_i) = \mathfrak p_{\mathsf T(K)}(u_i,v_i)$. By the fact that $\mathsf T(S) \subset S \cap K$ and $S \cap K$ is a tree, all arm-paths of $\mathsf T(S)$ from $u_i$ or $v_i$ are subgraphs of $\mathsf L_j(K)$ for some $j \in [\iota\aleph]$, and all edges of $\mathsf T(S)$ from $u_i$ or $v_i$ that do not belong to any arm-path of $\mathsf T(S)$ belong to $E(\mathsf T(K))$. Therefore we have 
\begin{align*}
    \mathfrak{k}_{\mathsf T(S)}(u_i)\vee\mathfrak{k}_{\mathsf T(S)}(v_i) &\leq \#\{j : u_i \in \operatorname{EndP}(\mathsf L_j(K)) \} \vee \#\{j : v_i \in \operatorname{EndP}(\mathsf L_j(K)) \} \leq 1
\end{align*}
by the fact that $\operatorname{EndP}(\mathsf L_j(K))$ are pairwise disjoint for $j \in [\iota\aleph]$.
Next, note that
\begin{align*}
    \mathsf{Deg}^{\text{ess}}_{\mathsf T(S)}(u_i)\vee\mathsf{Deg}^{\text{ess}}_{\mathsf T(S)}(v_i) = \mathsf{Deg}^{\text{ess}}_{\mathsf T(K)}(u_i) \vee \mathsf{Deg}^{\text{ess}}_{\mathsf T(K)}(v_i) = 3\,.
\end{align*}
Also, for any $1 \leq j \leq \mathtt m-1$ we have $\mathsf{Deg}_{\mathsf T(S)}(w_j)\leq \mathsf{Deg}_{K}(w_j)=2$ by $\mathsf{T}(S) \subset S\cap K$, and $\mathsf{Deg}_{\mathsf T(S)}(w_j) \geq \mathsf{Deg}_{\mathfrak p_{\mathsf T(S)}(u_i , v_i)}(w_j) = 2$. Therefore we have $\mathsf{Deg}_{\mathsf T(S)}(w_j)=2$. Therefore, by definition of $\mathfrak{Ch}_\mathtt m^{(a,b)}$ in \eqref{def-chained-signal-pair-in-T_iota} we have
\begin{align}\label{eq-pairing-must-match-chain-candidate-with-args-leq3}
    \mathsf P_i(K) \in \cup_{0 \leq a,b \leq 1} \mathfrak{Ch}^{(a,b)}_\mathtt m(\mathsf T(S)) \,.
\end{align}
Therefore, given any $\mathsf P_i(K) = (u_i , v_i)$ where $i \in \mathsf{Ind}_{S,K}$ there is at most $1$ arm-path starting from $u_i$ or $v_i$ in $\mathsf T(S)$ by \eqref{eq-pairing-must-match-chain-candidate-with-args-leq3}. In addition, the union of (at most $2$) arm-paths starting from $u_i$ or $v_i$ in $\mathsf T(S)$ is equal to $\mathsf L_i(K) \cap \mathsf T(S)$; this is because there is no arm-path starting from $u_i$ or $v_i$ in $\mathsf T(K)$. As a result, given $i \in \mathsf{Ind}_{S,K}$, $\mathsf L_i(K) \cap \mathsf T(S)$ is determined by $\mathsf P_i(K) = (u_i , v_i)$.

Therefore, it remains to determine $\mathsf L_i(K) \setminus S$ in order to determine $\mathsf L_i(K)$, which yields that there are at most $n^{\ell - 1 - |E(\mathcal L_{u_i})| - |E(\mathcal L_{v_i})|}$ ways to choose $\mathsf L_i(K)$ given any $\mathsf P_i(K)$ for $i \in \mathsf{Ind}_{S,K}$. Also, there are at most $n^{\ell - 1}$ ways to choose $\mathsf L_i(K)$ given any $\mathsf P_i(K)$ where $i \in [\iota\aleph]\setminus \mathsf{Ind}_{S,K}$. Combining the above bounds, there are at most
\begin{align*}
\aleph^{2\iota\aleph}n^{(\ell-1)\iota\aleph- \sum x_u} =n^{(\ell-1)\iota\aleph- \sum x_u+o(1)}
\end{align*}
ways to determine $\mathsf P(K)$ and $\mathsf L_i(K)$. Note that $\mathsf T(K) \setminus \mathsf T(S) = \bigcup_{v \in \mathtt V}\mathcal L_v$ is determined by $S$ and $\{ x_v : v\in \mathtt V\}$, and $\mathsf T(S)\setminus \mathsf T(K) = \mathsf T(S)\cap \big(\bigcup_{i \in \mathsf{Ind}_{S,K}} \mathsf L_i(K)\big)$ is determined by $S$ and $\mathsf P(K)$ (recall that we have shown given $i \in \mathsf{Ind}_{S,K}$, $\mathsf L_i(K) \cap \mathsf T(S)$ is determined by $\mathsf P_i(K) = (u_i , v_i)$), we can successfully determine $\mathsf T(K)$ given $\mathsf P(K)$, $\mathsf L_i(K)$ and $\{ x_v : v\in \mathtt V\}$, which yields \eqref{eq-enum-frakchoose-pairings-Q-H,I}.

We then show \eqref{eq-enum-frakchoose-pairings-U-valid-I-enum}. Given an arbitrary $S \in \mathfrak R^*_{\mathbf H}$, $y \geq 0$ and $\{ x_u : u \in \mathtt V\,, x_u \geq 0\}$, we perform the following steps to determine $K \in \mathfrak \cup_{\mathbf I \in\mathcal H} \mathfrak U_{\mathbf H,\mathbf I}(S,\{ x_u \},y)$ up to isomorphism: First, applying the arguments in the previous paragraph, given $S$ it suffices to determine $\mathsf P(K)$ in order to determine $K \in \mathfrak \cup_{\mathbf I \in\mathcal H} \mathfrak U_{\mathbf H,\mathbf I}(S,\{ x_u \},y)$ up to isomorphism. Then, there are at most $\tbinom{\iota\aleph}{y} \leq 2^{\iota\aleph}$ ways to determine the indices $i \in [\iota\aleph]$ such that $\mathsf P_i(K) \not\in \mathsf P(S)$, and at most $\tbinom{\iota\aleph}{\iota\aleph-y} \leq 2^{\iota\aleph}$ ways to determine the pairing set $\{ \mathsf P_i(K): \mathsf P_i(K)\in \mathsf P(S) \}$. Recall the definition of $\mathsf{Ind}_{S,K}$ in \eqref{eq-def-Ind-S,K}. By definition of $\mathfrak U_{\mathbf H,\mathbf I}$ in \eqref{eq-def-V-H-I}, we have $\mathsf{Ind}_{S,K} = \{i : \mathsf P_i(K) \in \mathsf P(S)\}$. Note that \eqref{eq-pairing-must-match-chain-candidate-with-args-leq3} holds for $i \in \mathsf{Ind}_{S,K}$, the number of ways to choose $\{ \mathsf P_i(K): \mathsf P_i(K)\not\in \mathsf P(S) \}$ is therefore bounded by $\tbinom{100\mathfrak c^{9}\mathfrak C_\aleph}{y}$, where we used
\begin{align*}
    \#\Big( \cup_{0 \leq a,b \leq 1} \mathfrak{Ch}^{(a,b)}_\mathtt m(\mathsf T(S)) \Big) \leq 100\mathfrak c^{9}\mathfrak C_\aleph
\end{align*}
by Item~(6) in Definition~\ref{def-tilde-T-K}. In conclusion, there are at most $2^{2\iota\aleph} \tbinom{100\mathfrak c^{9}\mathfrak C_\aleph}{y}$ ways to determine $\mathsf P(K)$, which yields \eqref{eq-enum-frakchoose-pairings-U-valid-I-enum}.

Finally we show \eqref{eq-enum-frakchoose-pairings-U-valid-I-enum-+}. Given $\mathbf H , \mathbf I \in \mathcal H$, an arbitrary $S \in \mathfrak R^*_{\mathbf H}$, $y \geq 0$ and $\{ x_u : u \in \mathtt V\,, x_u \geq 0\}$, we assume $\mathfrak U_{\mathbf H,\mathbf I}(S,\{ x_u \},y) \neq \emptyset$ without loss of generality and perform the following steps to determine $K \in \mathfrak U_{\mathbf H,\mathbf I}(S,\{ x_u \},y)$: First, we determine $\big( \mathsf T(K), \mathsf P(K)\big)$, which can be done by determining $\mathsf T(K)$ alone (note that for all $K \Vdash \mathbf I$ and all the isomorphism $\varsigma:\mathsf T(K) \to \mathsf T(\mathbf I)$, since $\mathsf P(\mathbf I) \subset \mathsf{Fix}(\mathsf T(\mathbf I))$ we see that $\varsigma^{-1}( \mathsf{P}(\mathbf I) )$ does not depend on the choice of $\varsigma$, so the choice of $\mathsf P(K)$ is unique given $\mathsf T(K)$). Since for all $K \in\mathfrak U_{\mathbf H,\mathbf I}(S,\{ x_u \},y)$ the tree
\begin{align*}
    T_{\mathbf H , \mathbf I; \{ x_u \} , y}:=S \cap K = \mathsf T(S) \cup \big( \cup_{u \in \mathsf{Vert} (\mathsf P(S))} \mathcal L^{(x_u)}_u\big) \mbox{ where }\cup_{u \in \mathsf{Vert} (\mathsf P(S))} \mathcal L^{(x_u)}_u \subset \cup_{i \in [\iota\aleph]}\mathsf L_i(S)
\end{align*}
is fixed, we can determine $\mathsf T(K)$ by choosing a subtree of $T_{\mathbf H , \mathbf I; \{ x_u \} , y}$ isomorphic to $\mathsf T(\mathbf I)$. Also, since $\mathfrak U_{\mathbf H,\mathbf I}(\{ x_u \},y) \neq 0$, there exists $K_* \in\mathfrak U_{\mathbf H,\mathbf I}(\{ x_u \},y)$ and $\{ z_v \geq 0 : v \in \mathsf{Vert} (\mathsf P(K_*))\}$ such that 
\begin{align*}
    T_{\mathbf H , \mathbf I; \{ x_u \} , y} = \mathsf T(K_*) \cup \big(\cup_{v \in \mathsf{Vert} (\mathsf P (K_*))} \mathcal L^{(z_v)}_v\big) \,.
\end{align*}
Therefore, by Item (4) and Item (5) of Lemma~\ref{lem-useful-property-trees-and-sets} the number of choices of $\mathsf{T}(K)$ (and thus that of $(\mathsf{T}(K),\mathsf P(K))$) is bounded by
\begin{align*}
    \frac{ \operatorname{Aut}( \mathsf T(K_*) \cup (\cup_{v \in \mathsf{Vert} (\mathsf P (K_*))} \mathcal L^{(z_v)}_v) ) }{ \mathsf T(K_*) } \leq e^{4\iota\aleph} \,.
\end{align*}
Then, by determining $(\mathsf{T}(K),\mathsf P(K))$, we can determine $\cup_{1 \leq i \leq \iota\aleph}(S \cap \mathsf L_i(K)) = \mathsf T(S) \setminus \mathsf T(K)$. Last, we determine $\cup_{1 \leq i \leq \iota\aleph}(\mathsf L_i(K) \setminus S)$ where we have at most $n^{(\ell - 1)\iota\aleph - \sum x_u}$ choices. Combining the above bounds yields \eqref{eq-enum-frakchoose-pairings-U-valid-I-enum-+}.

\subsection{Proof of Lemma~\ref{lem-characterize-setminus-R-H-I}}{\label{subsec:proof-lem-4.6}}

Define  
\begin{align*}
    G_{\mathsf T}=G_{\geq 2} \cup \mathsf T(S_1) \cup \mathsf T(S_2) \cup \mathsf T(K_1) \cup \mathsf T(K_2) \,.
\end{align*}
Note that we have assumed $\mathcal L(G_{\cup}) \subset V(G_{\geq 2})$. Define $\Gamma$ to be the number of elements in
\begin{align*}
    \{ \mathsf L_i(S_1), \mathsf L_i(S_2), \mathsf L_i(K_1), \mathsf L_i(K_2): 1 \leq i \leq \iota\aleph \}
\end{align*}
that are included in $G_{\mathsf T}$. It is straightforward to check that $\tau(G_{\cup})-\tau(G_{\mathsf T}) \geq 4\iota\aleph-\Gamma$ and $|E(G_{\mathsf T})| \geq \ell\Gamma$. Thus, 
\begin{align*}
    \tau(G_{\cup}) - \tau(G_{\mathsf T}) + \tfrac{|E(G_{\mathsf T})|}{\ell} \geq 4\iota\aleph \,.
\end{align*}
In addition, we have 
\begin{align*}
    |E(G_{\cup})| \leq 4\ell\iota\aleph+4\aleph -|E(G_{\geq 2})| \leq 4\ell\iota\aleph+4\aleph - (|E(G_{\mathsf T})|-4\aleph) = 4\ell\iota\aleph+8\aleph -|E(G_{\mathsf T})| \,.
\end{align*}
Thus, if $\tau(G_{\mathsf T}) \geq -1$ we immediately get that \eqref{eq-characterize-setminus-R-H-I} holds.

Now we assume that $\tau(G_{\mathsf T}) \leq -2$. Note that in this case $G_{\mathsf T}$ has at least two connected components. If one of the connected components of $G_{\mathsf T}$ contains only vertices in $S_1$ (or $S_2$/$K_1$/$K_2$), we see that $S_1 \in \mathfrak R_{\mathbf H} \setminus \mathfrak R_{\mathbf H}^*$ (or similarly for $S_2/K_1/K_2$). Denote $1\leq\kappa\leq 4$ the number of such components, and in this case it is straightforward to verify that 
\begin{align*}
    \tau(G_{\cup}) - \tau(G_{\mathsf T}) + \tfrac{|E(G_{\mathsf T})|}{\ell} \geq 4\iota\aleph+2\kappa \,,
\end{align*}
which combined with the fact that $\tau(G_\mathsf T) \geq -1-\kappa$ yields \eqref{eq-characterize-setminus-R-H-I}.

The only remaining case is that $\tau(G_{\mathsf T})=-2$ and $G_{\mathsf T}$ has exactly two connected tree components and each intersects with at least two of $S_1,S_2,K_1,K_2$. Suppose that one of the components intersects with $S_1,S_2$ and the other component intersects with $K_1,K_2$. Also, denote 
\begin{equation*}
G_S = \widetilde S_1 \cup \widetilde S_2 \,,\quad G_K = \widetilde K_1 \cup \widetilde K_2\,.
\end{equation*}
If $(S_1,S_2) \in \mathfrak R_{\mathbf H}^*$ and $S_1 \cap S_2$ is a tree containing $\mathsf T(S_1) \cup \mathsf T(S_2)$, it is straightforward to verify that 
\begin{align*}
    2\iota\aleph-1-\tau(G_S) \leq \tfrac{2\ell\iota\aleph+2\aleph-|E(G_S)|}{\ell}  \,;
\end{align*}
if the aforementioned condition does not hold, then recall \eqref{eq-def-mathcal-A-diamond} we have $\mathsf{id}\not\in \mathcal A(S_1,S_2)$. In addition, from $\mathcal L(G_{\cup}) \subset V(G_{\geq 2})$ we have $\mathcal L(S_1),\mathcal L(S_2) \subset V(S_1) \cap V(S_2)$. Thus, using Lemma~\ref{lem-character-non-principle} we see that  
\begin{align*}
    2\iota\aleph-\tau(G_S) \leq \tfrac{2\ell\iota\aleph+2\aleph-|E(G_S)|}{\ell}\,.
\end{align*}
The above inequality also holds by replacing $S$ with $K$. Thus, since $\mathsf{id}\not\in\mathcal A_{\star}$ we see that 
\begin{align*}
    \mbox{either } \Bigg\{
    \begin{split}
        & 2\iota\aleph-1-\tau(G_S) \leq \tfrac{2\ell\iota\aleph+2\aleph-|E(G_S)|}{\ell} \\
        & 2\iota\aleph-\tau(G_K) \leq \tfrac{2\ell\iota\aleph+2\aleph-|E(G_K)|}{\ell}
    \end{split}
    \mbox{ or }
    \Bigg\{ 
    \begin{split}
        & 2\iota\aleph-\tau(G_S) \leq \tfrac{2\ell\iota\aleph+2\aleph-|E(G_S)|}{\ell} \\
        & 2\iota\aleph-1-\tau(G_K) \leq \tfrac{2\ell\iota\aleph+2\aleph-|E(G_K)|}{\ell}
    \end{split}
    \,.
\end{align*}
This leads to \eqref{eq-characterize-setminus-R-H-I}, since $\tau(G_\cup)=\tau(G_S)+\tau(G_K)$ and $|E(G_{\cup})|=|E(G_S)|+|E(G_K)|$.

\subsection{Proof of Lemma~\ref{lem-replace-by-non-backtracking-path}}{\label{subsec:proof-lem-5.2}}

Define $\mathcal B(\mathbf H)$ to be the set of $(S_1,S_2)$ such that $S_1,S_2 \vdash \mathbf H$ but $S_1,S_2 \not\Vdash \mathbf H$. Also define
\begin{equation}{\label{eq-def-f-bad}}
    f_{\operatorname{bad}} = \sum_{ \mathbf H \in \mathcal H } \frac{ \operatorname{Aut}(\mathsf T(\mathbf H)) s^{\aleph-1} (\epsilon^2\lambda s)^{\ell\iota\aleph} }{ n^{\aleph+\ell\iota\aleph} } \sum_{S_1,S_2 \in \mathcal B(\mathbf H)} \phi_{S_1,S_2} \,.
\end{equation}
By Proposition~\ref{prop-first-moment} and Lemma~\ref{lem-enu-decorated-trees}, it suffices to show
\begin{equation}{\label{eq-goal-lem-5.2}}
    \frac{f_{\operatorname{bad}}}{\mathbb E_{\Pb}[f]} = \sum_{ \mathbf H \in \mathcal H } \frac{ \operatorname{Aut}(\mathsf T(\mathbf H)) s^{\aleph-1} }{ n^{\aleph+(\ell-2)\iota\aleph+o(1)}(\epsilon^2\lambda s)^{\ell\iota\aleph} }\sum_{S_1,S_2 \in \mathcal B(\mathbf H)} \phi_{S_1,S_2} \to 0
\end{equation}
in probability under both $\Pb$ and $\Qb$. We will only prove \eqref{eq-goal-lem-5.2} under $\Pb$, and the argument for $\Qb$ is similar. Since $f_{\operatorname{bad}}$ is symmetric under $\pi_*$, it suffices to show \eqref{eq-goal-lem-5.2} under $\Pb_{\mathsf{id}}$. To this end, we require additional definitions from \cite{MNS18}.

\begin{DEF}{\label{def-tangle}}
    We say a connected graph $H$ is a tangle, if $H$ contains at least two cycles and the diameter of $H$ is at most $\sqrt{\log n}$. Given a multigraph $S \vdash \mathbf H \in \mathcal H$, we say that $S$ contains $t$ tangles if we need to delete at least $t$ edges (counting multiplicity) so that the remaining subgraph of $S$ contains no tangle.
\end{DEF}

\begin{DEF}{\label{def-new-old-returning}}
    For a path $\gamma=(\gamma_1,\ldots,\gamma_l)$ and a graph $G$, we say an edge $(\gamma_i,\gamma_{i+1})$ is 
    \begin{itemize}
        \item \emph{new} with respect to $G$, if $\gamma_{i+1} \not \in V(G) \cup \{ \gamma_1,\ldots,\gamma_{i} \}$.
        \item \emph{old} with respect to $G$, if $(\gamma_i,\gamma_{i+1}) \in E(G) \cup \{ (\gamma_1,\gamma_2), \ldots, (\gamma_{i-1},\gamma_i) \}$.
        \item \emph{returning} with respect to $G$, if $(\gamma_i,\gamma_{i+1})$ is neither new nor old with respect to $G$.
    \end{itemize}
\end{DEF}

We now divide $f_{\operatorname{bad}}$ into the following cases. For $S_1,S_2 \vdash \mathbf H$ and $S_1,S_2 \in \mathcal B(\mathbf H)$, denote $G_{S}=\widetilde S_1\cup \widetilde S_2$. Define
\begin{align}
    &\mathcal B_1(\mathbf H) = \Big\{ (S_1,S_2) \in \mathcal B(\mathbf H): \tau(G_S) > 20\aleph^2 \Big\} \,; \label{eq-def-mathcal-B-1} \\
    &\mathcal B_2(\mathbf H) = \Big\{ (S_1,S_2) \in \mathcal B(\mathbf H): \tau(G_S) \leq 20\aleph^2, G_S \mbox{ contains at least } 40\aleph^3 \mbox{ tangles} \Big\} \,; \label{eq-def-mathcal-B-2} \\
    &\mathcal B_3(\mathbf H) = \Big\{ (S_1,S_2) \in \mathcal B(\mathbf H): \tau(G_S) \leq 20\aleph^2, G_S \mbox{ contains at most } 40\aleph^3 \mbox{ tangles} \Big\} \,. \label{eq-def-mathcal-B-3}
\end{align}
Then define
\begin{equation}{\label{eq-def-f-B-i}}
    f_{\mathcal B_i} = \sum_{\mathbf H \in \mathcal H } \frac{ \operatorname{Aut}(\mathsf T(\mathbf H)) s^{\aleph-1} }{ n^{\aleph+(\ell-2)\iota\aleph}(\epsilon^2\lambda s)^{\ell\iota\aleph}} \sum_{ (S_1,S_2) \in \mathcal B_i(\mathbf H) } \phi_{S_1,S_2} \mbox{ for } i=1,2,3 \,.
\end{equation}

\begin{lemma}{\label{lem-control-f-B-1}}
    We have $\mathbb E_{\Pb_{\mathsf{id}}}\big[ |f_{\mathcal B_1}| \big] = o(1)$.
\end{lemma}
The proof of Lemma~\ref{lem-control-f-B-1} is incorporated in Section~\ref{subsec:proof-lem-C.3}. Having established a bound for $f_{\mathcal B_1}$, we deal with $f_{\mathcal B_2}$. Define $\Xi$ to be the event that the parent graph $G$ does not contain any tangle. It was shown in \cite[Lemma~6.1]{MNS18} that $\mathbb P(\Xi)=1-o(1)$. We show the following estimation.
\begin{lemma}{\label{lem-control-f-B-2}}
    We have $\mathbb E_{\Pb_{\mathsf{id}}}\big[ \mathbf 1_{\Xi} |f_{\mathcal B_2}| \big]=o(1)$.
\end{lemma}
The proof of Lemma~\ref{lem-control-f-B-2} is postponed to Section~\ref{subsec:proof-lem-C.4}. Finally we bound $f_{\mathcal B_3}$.
\begin{lemma}{\label{lem-control-f-B-3}}
    We have $\mathbb E_{\Pb_{\mathsf{id}}}\big[ f_{\mathcal B_3}^2 \big]=o(1)$.
\end{lemma}
The proof of Lemma~\ref{lem-control-f-B-3} is postponed to Section~\ref{subsec:proof-lem-C.5}. We can now complete the proof of Lemma~\ref{lem-replace-by-non-backtracking-path}.
\begin{proof}[Proof of Lemma~\ref{lem-replace-by-non-backtracking-path}]
    Using Lemmas~\ref{lem-control-f-B-1}, \ref{lem-control-f-B-2} and \ref{lem-control-f-B-3} respectively, we see that
    \begin{align*}
         f_{\mathcal B_1}, f_{\mathcal B_2}, f_{\mathcal B_3} \to 0 \mbox{ in probability under } \Pb_{\mathsf{id}}, \mbox{ respectively} \,. 
    \end{align*}
    Plugging this into \eqref{eq-def-f-B-i} and \eqref{eq-goal-lem-5.2}, we see that
    \begin{equation*}
        \frac{f_{\operatorname{bad}}}{\mathbb E_{\Pb}[f]} = f_{\mathcal B_1} + f_{\mathcal B_2} + f_{\mathcal B_3} \to 0
    \end{equation*}
    in probability under $\Pb_{\mathsf{id}}$, as desired.
\end{proof}

\subsection{Proof of Lemma~\ref{lem-label-matching-exact}}{\label{subsec:proof-lem-C-exact}}

We first show that $\mathtt d^{I,J} \leq \mathtt m(|I|+|J|) + \mathtt f^{I,J}$. Denote $S_1^I$ and $S_2^J$ by $S_1$ and $S_2$ for simplicity. For $1 \leq i \leq |I|+|J|, a \in \{1,2\}$ and $b \in [\iota\aleph]$, denote $i \lhd (a,b)$ if and only if $\mathcal L^{I,J}_i \subset \mathsf L_b(S_a)$. By definition of $\mathcal L^{I,J}_i$, for all $1 \leq i \leq |I|+|J|$ there exists a unique pair $(a,b)$ such that $i \lhd (a,b)$. Next, denote $(a_1,b_1;c_1) \sim (a_2,b_2;c_2)$ if and only if $c_1 \lhd (a_1,b_1) $, $c_2 \lhd (a_2,b_2)$ such that $|\operatorname{EndP}(\mathcal L^{I,J}_{c_1}) \cap \operatorname{EndP}(\mathcal L^{I,J}_{c_2})| = 1$. Then
\begin{align*}
    \{ (a,b;c): 1 \leq c \leq |I| + |J| \,, c \lhd (a,b)\}
\end{align*}
is a set with $|I| + |J|$ elements by definition of $\mathcal L^{I,J}_i$, and is divided into connected components $\mathfrak{CO}_1 , \ldots , \mathfrak{CO}_{\mathtt x}$ under the binary relation $\sim$. Also, by the fact that each vertex in $V(T^{I,J})$ occurs in at most two of $\{\operatorname{EndP}(\mathcal L^{I,J}_i)\}_{1 \leq i \leq |I|+|J|}$, each component can be denoted by
\begin{align*}
    \mathfrak{CO}_i = \{ (a_{i,1} , b_{i,1}; c_{i,1}) \sim \ldots \sim (a_{i,\mathtt n_i} , b_{i,\mathtt n_i};c_{i,\mathtt n_i}) \}\,.
\end{align*}
Given $1 \leq i \leq \mathtt x$ and $1 \leq j \leq \mathtt n_i - 1$ we always have $\mathsf P_{b_{i,j}}(S_{a_{i,j}} ) \cap \mathsf P_{b_{i,j+1}}(S_{a_{i,j+1}} )$ is a singleton and we denote its element by $w_{i,j}$, and $\mathsf L_{b_{i,j}}(S_{a_{i,j}} ) \cap \mathsf L_{b_{i,j+1}}(S_{a_{i,j+1}} )$ is an arm-path of $T^{I,J}$ starting from $w_{i,j}$. Also, by Theorem~\ref{thm-desired-vertex-sets}, we have
\begin{align*}
    a_{i,j} \neq a_{i,j+1} \mbox{ and } E\Big(\mathsf L_{b_{i,j}}(S_{a_{i,j}} ) \cap\mathsf L_{b_{i,j+1}}(S_{a_{i,j+1}} ) \cap \big(\mathsf T(S_{a_{i,j}} )\cup \mathsf T(S_{a_{i,j+1}} ) \big)\Big)= \emptyset
\end{align*}
by the fact that all arm-paths in $\mathsf T(S_1)$ or $\mathsf T(S_2)$ does not contain $w_{i,j}$. Therefore, if we denote 
\begin{align*}
    \mathsf P_{b_{i,j}} (S_{a_{i,j}}) &= \{ w_{i,j-1} , w_{i,j}\} (1 \leq j \leq \mathtt n_i)\,, x_i = w_{i,0}\,, y_i = w_{i,\mathtt n_i}\,,\\
    \{ x_i' , x_i\} &= \operatorname{EndP}(\mathsf T(S_{3-a_{i,1}}) \cap \mathsf L_{b_{i,1}}(S_{a_{i,1}}))\,,\\
    \{ y_i' , y_i\} &= \operatorname{EndP}(\mathsf T(S_{3-a_{i,\mathtt n_i}}) \cap \mathsf L_{b_{i,\mathtt n_i}}(S_{a_{i,\mathtt n_i}}))
\end{align*}
(where $x_i' = x_i$ or $y_i' = y_i$ is allowed), then $(x'_1 , y'_1 , \ldots , x'_\mathtt x , y'_\mathtt x ) \in \mathcal X(\mathtt V^{I,J})$, which gives $\mathtt x = |\mathtt V^{I,J}|$. Therefore, we have
\begin{align}
    \mathtt d^{I,J} &\leq \sum_{i=1}^{\mathtt x}\mathsf{Dist}_{T^{I,J}}(x'_i , y'_i)\nonumber \\
    &= \sum_{i=1}^{\mathtt x}\Big(\mathsf{Dist}_{T^{I,J}}(x'_i , x_i) + \sum_{j=1}^{\mathtt n_i}\mathsf{Dist}_{T^{I,J}}(w_{i,j-1} , w_{i,j}) + \mathsf{Dist}_{T^{I,J}}(y_i , y'_i)\Big)\nonumber \\
    &= \mathtt m(|I|+|J|) +\sum_{i=1}^{\mathtt x}\Big(\mathsf{Dist}_{T^{I,J}}(x'_i , x_i) + \mathsf{Dist}_{T^{I,J}}(y_i , y'_i)\Big)\nonumber \\
    &= \mathtt m(|I|+|J|) + \mathtt f^{I,J}\,,\label{eq-lem-exact-distance-construct}
\end{align}
where the first equality holds by the fact that the shortest path between any two vertices on a tree is unique.

To further show that $\mathtt d^{I,J} \geq \mathtt m(|I|+|J|) + \mathtt f^{I,J}$, for $u,v \in V(T^{I,J})$ we denote $\mathfrak p_{T^{I,J}}(u,v)$ the shortest path in $T^{I,J}$ connecting $u$ and $v$. We first show that $\{\mathfrak p_{T^{I,J}}(x_i',y_i')\}_{1 \leq i \leq \mathtt x}$ are pairwise disjoint. Since $\{ \mathfrak p_{T^{I,J}}(x_i',x_i)\}_{1 \leq i \leq \mathtt x}$ and $\{ \mathfrak p_{T^{I,J}}(y_i',y_i)\}_{1 \leq i \leq \mathtt x}$ are arm-paths, it suffices to show that $\{\mathfrak p_{T^{I,J}}(x_i,y_i)\}_{1 \leq i \leq \mathtt x}$ are pairwise disjoint. Otherwise, suppose $1 \leq i_1 < i_2 \leq \mathtt x$ such that $v \in V(\mathfrak p_{T^{I,J}}(x_{i_1},y_{i_1}) \cap \mathfrak p_{T^{I,J}}(x_{i_2},y_{i_2}))$. Then since $\mathfrak{CO}_{i_1}$ and $\mathfrak{CO}_{i_2}$
are distinct connected components, we have $v \not\in \{ w_{i_1 , j}: 0 \leq j \leq \mathtt n_{i_1}\}$ and there exists $0 \leq j < \mathtt n_{i_1}$ such that $v \in V(\mathfrak p_{T^{I,J}}(w_{i_1 ,j } , w_{i_1 ,j + 1}))$. By the fact that $\mathsf{Deg}_{T^{I,J}}(u) = 2$ for any $u \in V(\mathfrak p_{T^{I,J}}(w_{i_1 ,j } , w_{i_1 ,j + 1})) \setminus \operatorname{EndP}(\mathfrak p_{T^{I,J}}(w_{i_1 ,j } , w_{i_1 ,j + 1}))$ and $\mathsf{Deg}_{T^{I,J}}(u) \geq 3$ for any $u \in \operatorname{EndP}(\mathfrak p_{T^{I,J}}(w_{i_1 ,j } , w_{i_1 ,j + 1}))$, we have $\mathfrak p_{T^{I,J}}(w_{i_1 ,j } , w_{i_1 ,j + 1}) \subset \mathfrak p_{T^{I,J}}(x_{i_2} , y_{i_2})$ which again contradicts that $\mathfrak{CO}_{i_1}$ and $\mathfrak{CO}_{i_2}$ are distinct connected components.

Next, suppose $(x''_1 , \ldots , x''_\mathtt x , y''_1 , \ldots , y''_\mathtt x )$ is an arbitrary element of $\mathcal X(\mathtt V^{I,J})$. It suffices to show 
\begin{equation}\label{eq-lem-exact-distance-ineq}
    \sum_{i=1}^{\mathtt x} \mathsf{Dist}_{T^{I,J}}(x''_{i} , y''_{i}) \geq \sum_{i=1}^{\mathtt x} \mathsf{Dist}_{T^{I,J}}(x'_{i} , y'_{i})
\end{equation}
by showing that $\bigcup_{1 \leq i \leq \mathtt x}\mathfrak{p}_{T^{I,J}}(x'_{i} , y'_{i})\subset \bigcup_{1 \leq i \leq \mathtt x}\mathfrak{p}_{T^{I,J}}(x''_{i} , y''_{i})$. Otherwise, suppose $e \in \mathfrak{p}_{T^{I,J}}(x'_{i_0} , y'_{i_0})\setminus \big( \bigcup_{1 \leq j \leq \mathtt x}\mathfrak{p}_{T^{I,J}}(x''_{j} , y''_{j})\big)$ for some $1 \leq i_0 \leq \mathtt x$. Then both connected components of $T^{I,J}\setminus \{ e\}$, say $T^{I,J}_1$ and $T^{I,J}_2$, should contain an even number of vertices in $\mathtt V^{I,J}$. However, for $i \neq i_0$, $x'_i$ and $y'_i$ are both in $T^{I,J}_1$ or both in $T^{I,J}_2$; in addition, exactly one of $x'_{i_0}$ and $y'_{i_0}$ will be in $T^{I,J}_1$. These two properties imply that both $T_1^{I, J}$ and $T_2^{I, J}$ contain odd numbers of vertices in $\mathtt V^{I, J}$, which is a contradiction. Therefore we have shown \eqref{eq-lem-exact-distance-ineq}, which gives $\mathtt d^{I,J} \geq \mathtt m(|I|+|J|) + \mathtt f^{I,J}$. Combining this with \eqref{eq-lem-exact-distance-construct} yields our desired result.

\subsection{Proof of Lemma~\ref{lem-control-f-B-1}}{\label{subsec:proof-lem-C.3}}

Note that we have the estimation 
\begin{align*}
    \mathbb E_{\Pb_{\mathsf{id}}}\big[ |A_{i,j}-\tfrac{\lambda s}{n}| \big], \mathbb E_{\Pb_{\mathsf{id}}}\big[ |B_{i,j}-\tfrac{\lambda s}{n}| \big], \mathbb E_{\Pb_{\mathsf{id}}}\big[ |(A_{i,j}-\tfrac{\lambda s}{n})(B_{i,j}-\tfrac{\lambda s}{n})| \big] \leq \tfrac{2\lambda s}{n} \,.
\end{align*}
Combined with \eqref{eq-def-f-S} and \eqref{eq-phi-as-linear-combination-of-beta}, it yields that
\begin{align*}
    & \mathbb E_{\Pb_{\mathsf{id}}}\big[ |f_{\mathcal B_1}| \big] \leq \frac{ \operatorname{Aut}(\mathsf T(\mathbf H)) s^{\aleph-1} }{ n^{\aleph+(\ell-2)\iota\aleph+o(1)}(\epsilon^2\lambda s)^{\ell\iota\aleph} } \sum_{(S_1,S_2) \in \mathcal B_1(\mathbf H)} \sum_{I,J \subset [\iota\aleph]} n^{-\ell(2\iota\aleph-|I|-|J|)/2} \mathbb E_{\Pb_{\mathsf{id}}}\big[ |\beta_{S_1^{I},S_2^{J}}| \big]  \\
    \leq\ & \sum_{(S_1,S_2) \in \mathcal B_1(\mathbf H)} \sum_{I,J \subset [\iota\aleph]} \frac{ \operatorname{Aut}(\mathsf T(\mathbf H)) s^{\aleph-1} (2\lambda s)^{2\ell\iota\aleph+2\aleph-2} }{ n^{\aleph+(\ell-2)\iota\aleph+o(1)}(\epsilon^2\lambda s)^{\ell\iota\aleph} } n^{ \frac{|E(S_1^{I})|+|E(S_2^{J})|-2|E(S_1^{I} \cap S_2^{J})|-\ell(2\iota\aleph-|I|-|J|)}{2} } \\ 
    \leq\ & \sum_{(S_1,S_2) \in \mathcal B_1(\mathbf H)} \frac{ \operatorname{Aut}(\mathsf T(\mathbf H)) s^{\aleph-1} }{ n^{\aleph+(\ell-2)\iota\aleph+o(1)}(\epsilon^2\lambda s)^{\ell\iota\aleph} } (2\lambda s)^{2\ell\iota\aleph+2\aleph-2} n^{(\aleph-1+\ell\iota\aleph)-|E(G_{S})| } \\
    \leq\ & n^{o(1)} \cdot \sum_{(S_1,S_2) \in \mathcal B_1(\mathbf H)} (2\lambda s)^{2\ell\iota\aleph+2\aleph} n^{\aleph-|E(G_S)|} \,,
\end{align*}
where the second inequality follows from (note that $E(S_1 \cap S_2) \leq |E(S_1' \cap S_2')|+ \min\{ |E(S_1 \setminus S_1')|, |E(S_2 \setminus S_2')| \}$ for all multigraphs $S_1,S_2,S_1',S_2'$)
\begin{align*}
    & |E(S_1^{I})|+|E(S_2^{J})|-2|E(S_1^{I}) \cap E(S_2^{J})|-(2\iota\aleph-|I|-|J|) \\
    =\ & (\aleph-1+\ell\iota\aleph) - 2\ell (2\iota\aleph-|I|-|J|) -2|E(S_1^{I} \cap S_2^{J}))| \\
    \leq\ & 2(\aleph-1+\ell\iota\aleph) - 2|E(S_1 \cap S_2)|= 2(\aleph-1+\ell\iota\aleph) - 2|E(G_S)| \,.
\end{align*}
In order to prove Lemma~\ref{lem-control-f-B-1}, we first show the following lemma.
\begin{lemma}{\label{lem-coarse-enumeration-B-1}}
    Recall that $G_S = \widetilde{S}_1 \cup \widetilde{S}_2$. We have 
    \begin{align*}
        \#\Big\{ (S_1,S_2)\in \mathcal B_1(\mathbf H): |V(G_S)| = v, \tau(G_S) = m \Big\} \leq n^{v+2\aleph+\max\{10\aleph^2,m/2\}+o(1)} \,.
    \end{align*}
\end{lemma}
\begin{proof}
    Clearly, the enumeration of $(\mathsf T(S_1),\mathsf T(S_2))$ is bounded by $n^{2\aleph}$ and given $(\mathsf T(S_1),\mathsf T(S_2))$, the enumeration of $(\mathsf P(S_1),\mathsf P(S_2))$ is bounded by $(2\aleph)^{4\iota\aleph} = n^{o(1)}$. Now we give an upper bound on the enumeration of 
    \begin{align*}
        ( \mathsf L_1(S_1), \ldots, \mathsf L_{\iota\aleph}(S_1); \mathsf L_1(S_2), \ldots, \mathsf L_{\iota\aleph}(S_2) ) \,.
    \end{align*}
    Denote $\mathcal L_{i+(j-1)\iota\aleph} = \mathsf L_i (S_j)\, (1\leq i \leq \iota\aleph ,1 \leq j \leq 2 )$ and suppose that $\mathcal L_i$ has $n_i$ new edges and $r_i$ returning edges with respect to $\mathsf T(S_1) \cup \mathsf T(S_2) \cup (\cup_{j \leq i-1} \mathcal L_j)$. Clearly, by $-2 \leq \tau(\mathsf T(S_1)\cup\mathsf T(S_2)) \leq \aleph$ we have
    \begin{align*}
        \sum_{i=1}^{2\iota\aleph} n_i \leq v \mbox{ and }  m+2 \geq \sum_{i=1}^{2\iota\aleph} r_i \geq m-\aleph \,.
    \end{align*}
    We now control the enumeration of $\{ \mathcal L_i \}$. Recall that $\iota < \tfrac{1}{100}$. Therefore, we have $|V(G_{S})|\leq|E(G_{S})|\leq 2(\ell \iota+1)\aleph \leq \ell\aleph$. For each $\mathcal L_i$ with $n_i$ new edges and $r_i$ returning edges, there are $\frac{\ell!}{n_i!r_i!(\ell-n_i-r_i)!}$ choices to decide the order of old, new and returning edges. Also, we may consider the choices for vertices along the path sequentially and observe the following: for each returning edge $(v_j,v_{j+1})$ there are at most $|V(S_1\cup S_2)| \leq 2(\ell-1)\iota\aleph+2\aleph \leq 3\ell\aleph$ choices for $v_{j+1}$; for each new edge $(v_j,v_{j+1})$ there are at most $n$ choices for $v_{j+1}$; for each old edge $(v_j,v_{j+1})$ there are at most $\operatorname{deg}_{G_S}(v_j)\leq \operatorname{deg}_{\mathsf T(S_1) \cup \mathsf T(S_2)}(v_j) + \sum_{i=1}^{2\iota\aleph}r_i \leq 2\aleph+m $ choices for $v_{j+1}$. Thus, the enumeration of $\mathcal L_i$ is bounded by
    \begin{align*}
        & \frac{\ell!}{n_i!r_i!(\ell-n_i-r_i)!} \cdot n^{n_i} (3\ell\aleph)^{r_i} (2\aleph+m)^{\ell-n_i-r_i} 
        \leq \frac{\ell^{\ell-n_i}n^{n_i} (3\ell\aleph)^{r_i} (2\aleph+m)^{\ell-n_i-r_i}}{(\ell-n_i-r_i)!} \\
        =\ & n^{n_i} (3\ell^2 \aleph)^{r_i} \frac{ (\ell(\aleph+m))^{\ell-n_i-r_i}}{(\ell-n_i-r_i)!} \leq n^{n_i} (3\ell^2 \aleph)^{r_i} (20(\aleph+m))^{\ell} \,,
    \end{align*}
    where the first inequality follows from $\frac{\ell!}{n_i!} \leq \ell^{\ell-n_i}$, and the last inequality follows from the fact that the function $\frac{y^x}{x!}$ is increasing in $x$ when $x\leq y$. Thus, the enumeration of $\{\mathcal L_i\}$ is bounded by 
    \begin{align*}
        & \sum_{\substack{n_1 + \dots + n_{2\iota\aleph} \leq v \\ m-\aleph \leq r_1 + \dots + r_{2\iota\aleph} \leq m+2}} \prod_{i=1}^{2\iota\aleph} \Big( n^{n_i} (3\ell^2 \aleph)^{r_i} (20(\aleph+m))^{\ell} \Big) \\
        =\ & \sum_{\substack{n_1 + \dots + n_{2\iota\aleph} \leq v \\ m-\aleph \leq r_1 + \dots + r_{2\iota\aleph} \leq m+2}} n^{v} (3\ell^2\aleph)^{m+2} (20(\aleph+m))^{2\iota\aleph\ell}  \\
        \leq\ & n^{v+o(1)} (3\ell^2\aleph)^{m} \aleph^{2\iota\aleph\ell} m^{2\iota\aleph\ell} \leq e^{2\iota\ell\aleph\log(\aleph)} n^{v+ \frac{\log(3\ell^2\aleph)}{\log n} m + \aleph\log m} \leq n^{v+\max\{10\aleph^2,m/2\}+o(1)} \,,
    \end{align*}
    where the first inequality follows from the fact that
    \begin{align*}
        & \#\big\{ (n_1,\ldots,n_{2\iota\aleph}): n_1 + \dots + n_{2\iota\aleph} \leq v \big\} \leq v^{2\iota\aleph} \overset{v=|V(G_S)|}{\leq} (\ell\aleph)^{2\iota\aleph} = n^{o(1)} \,; \\
        & \#\big\{ (r_1,\ldots,r_{2\iota\aleph}): r_1 + \dots + r_{2\iota\aleph} \leq m+2 \big\} \leq (m+2)^{2\iota\aleph} \overset{m=\tau(G_S)}{\leq} (\ell\aleph)^{2\iota\aleph} = n^{o(1)} \,,
    \end{align*}
    the second inequality follows from the fact that $\ell\leq \log n$ (see \eqref{eq-choice-aleph}), and the last inequality follows from $\aleph \log m \leq 10\aleph^2 + 0.1(\log m)^2 \leq 10\aleph^2 + 0.1m$ for $m$ larger than a sufficiently large constant. Thus, the enumeration of $(S_1,S_2)$ is bounded by
    \begin{align*}
        n^{2\aleph+o(1)} \cdot n^{v+\max\{10\aleph^2,m/2\}+o(1)} = n^{v+2\aleph+\max\{10\aleph^2,m/2\}+o(1)} \,,
    \end{align*}
    as desired.
\end{proof}
Now we complete the proof of Lemma~\ref{lem-control-f-B-1}.
\begin{proof}[Proof of Lemma~\ref{lem-control-f-B-1}]
    Note that we have 
    \begin{align*}
        & \sum_{(S_1,S_2) \in \mathcal B_1(\mathbf H)}  (2\lambda s)^{2\ell\iota\aleph+2\aleph} n^{\aleph-|E(G_S)|} \\
        \leq\ & \sum_{0 \leq v \leq 2\ell\aleph, m \geq 20\aleph^2}   (2\lambda s)^{2\ell\iota\aleph} n^{\aleph-(v+m)+o(1)} \#\big\{ (S_1,S_2): |V(G_S)|=v, \tau(G_S) = m \big\} \\
        \overset{\text{Lemma~\ref{lem-coarse-enumeration-B-1}}}{\leq}\ & \sum_{0 \leq v \leq 2\ell\aleph, m \geq 20\aleph^2}(2\lambda s)^{2\ell\iota\aleph} n^{\aleph-(v+m)} n^{v+m/2+2\aleph+o(1)} \\
        \leq\ & \sum_{0 \leq v \leq 2\ell\aleph, m \geq 20\aleph^2} n^{-m/2+\aleph^2+o(1)} = o(1) \,,
    \end{align*}
    leading to the desired result.
\end{proof}

\subsection{Proof of Lemma~\ref{lem-control-f-B-2}}{\label{subsec:proof-lem-C.4}}

Recall that for each $(S_1,S_2) \in \mathcal B_2(\mathbf H)$, we denote $G_S=\widetilde S_1 \cup \widetilde S_2$. We say a subset of unordered pairs $F\subset \mathrm{U}_n$ is \emph{pivotal} to $G_S$, if $G_S \setminus F$ is tangle-free (a graph is tangle-free if it does not contain any tangle) but $G_S \setminus F'$ contains at least one tangle for any $F' \subsetneq F$. In case that $F$ is pivotal to $G_S$, at most two connected components of $G_S \setminus F$ can be trees (otherwise since $G_S = \widetilde S_1 \cup \widetilde S_2$ has at most two connected components, an edge could be added in $G_S \setminus F$ between two of the tree components without introducing any tangle, which contradicts to our definition). Therefore, supposing that $\mathsf{Comp}_1 ,\ldots,\mathsf{Comp}_i$ are all connected components of $G_S \setminus F$, we have $|E(G_S)|-|E(F)|=\sum_{j=1}^i |E(\mathsf{Comp}_j)| \geq \sum_{j=1}^i (|V(\mathsf{Comp}_j)|-\mathbf{1}_{\{\mathsf{Comp}_j\text{ is a tree}\}})\geq |V(G_S)|-2$, which implies $|F|\leq \tau (G_S) + 2$. Denote $\mathcal F(S)$ to be the collection of pivotal sets of $G_S$. Defining
\begin{equation}{\label{eq-def-Omega-F(S)}}
    \Omega_{\mathcal F(S)} = \cup_{F \in \mathcal F(S)} \Omega_F \mbox{ where } \Omega_F = \Big\{ G_{i,j}=0 : (i,j) \in F \Big\} \,,
\end{equation}
it is clear that $\Xi \subset \Omega_{\mathcal F(S)}$: since on the event $\cap_{F \in \mathcal F(S)} \Omega_F^{c}$ there exists a tangle $H \subset S$ with $G_{i,j}=1$ for all $(i,j)\in E(H)$, contradicting to the definition of $\Xi$ (which means that the edge set $\{ (i,j): G_{i,j}=1 \}$ is tangle-free). Thus, from \eqref{eq-phi-as-linear-combination-of-beta} we have 
\begin{align}
    \mathbb E_{\Pb_{\mathsf{id}}} \big[ \mathbf 1_{\Xi} |f_{\mathcal B_2}| \big] &\leq \frac{1}{n^{\aleph+(\ell-2)\iota\aleph}} \sum_{(S_1,S_2) \in \mathcal B_2(\mathbf H)} \sum_{I,J \subset [\iota\aleph]} n^{ \frac{-\ell(2\iota\aleph-|I|-|J|)}{2} } \mathbb E_{\Pb_{\mathsf{id}}} \big[ \mathbf 1_{\Omega_{\mathcal F(S)}} |\beta_{S_1^{I},S_2^{J}}| \big] \nonumber \\
    &\leq \frac{1}{n^{\aleph+(\ell-2)\iota\aleph}} \sum_{(S_1,S_2) \in \mathcal B_2(\mathbf H)} \sum_{I,J \subset [\iota\aleph]} \sum_{F \in \mathcal F(S)} n^{ \frac{-\ell(2\iota\aleph-|I|-|J|)}{2} }\mathbb E_{\Pb_{\mathsf{id}}}  \big[ \mathbf 1_{\Omega_{F}} |\beta_{S_1^{I},S_2^{J}}| \big] \,. \label{eq-f-B-2-relax-1}
\end{align}
Since $\tau(G_S) \leq 20\aleph^2$, for all $F \in \mathcal F(S)$ we have $|F| \leq 20\aleph^2 + 2 \leq 100\aleph^2$. In addition, since $G_S$ has at least $40\aleph^3$ tangles we see that
\begin{align*}
    \sum_{(i,j)\in F} \big( E(S_1)_{i,j} + E(S_2)_{i,j} - 2 \big) \geq 40\aleph^3-200\aleph^2\geq 35\aleph^3 \mbox{ for all } F \in \mathcal F(S) \,.
\end{align*}
Thus, we see by the definition of $\beta_{S_1,S_2}$ in \eqref{eq-def-beta-S} that
\begin{align*}
    &\ \quad n^{ \frac{-\ell(2\iota\aleph-|I|-|J|)}{2} } \mathbb E_{\Pb_{\mathsf{id}}} \big[ \mathbf 1_{\Omega_{F}} |\beta_{S_1^{I},S_2^{J}}| \big] \\
    &\leq n^{ \frac{-\ell(2\iota\aleph-|I|-|J|)}{2} } \big(\tfrac{2\lambda s}{n}\big)^{\tfrac{1}{2}(\sum_{(i,j) \in F}(E(S_1^{I})_{i,j} +E(S_2^{J})_{i,j}))} \cdot \big(\tfrac{\epsilon^2\lambda s}{n}\big)^{\tfrac{1}{2}(\sum_{(i,j) \not\in F}(2-E(S_1^{I})_{i,j} -E(S_2^{J})_{i,j}))}\cdot 2^{|E(G_S)|}\\
    &\leq n^{ \frac{-\ell(2\iota\aleph-|I|-|J|)}{2} } \big(\tfrac{2\lambda s}{n}\big)^{|F|+\sum_{(i,j) \in F}(E(S_1^{I})_{i,j} +E(S_2^{J})_{i,j}-2)} \cdot \big(\tfrac{2\lambda s}{n}\big)^{\tfrac{1}{2}(\sum_{(i,j)}(2-E(S_1^{I})_{i,j} -E(S_2^{J})_{i,j}))} \cdot \big(\tfrac{2}{\epsilon^2}\big)^{4\ell\aleph}\\
    &\leq \big(\tfrac{2\lambda s}{n}\big)^{|F|+\sum_{(i,j) \in F}(E(S_1^{I})_{i,j} +E(S_2^{J})_{i,j}-2) + \ell(2\iota\aleph-|I|-|J|) + |E(G_S)|-\aleph+1-\ell\iota\aleph } \cdot \big(\tfrac{2}{\epsilon^2}\big)^{4\ell\aleph} \\
    &\leq \big(\tfrac{2\lambda s}{n}\big)^{|E(G_S)|+35\aleph^3-\aleph+1-\ell\iota\aleph} \cdot n^{\aleph^2} \leq n^{-30\aleph^3-|E(G_S)|+\ell\iota\aleph} \,,
\end{align*}
for all $F \in \mathcal F(S)$, where the first inequality holds by Lemma~\ref{lem-joint-moment-A-B}, the third inequality follows from
\begin{align*}
    \sum_{(i,j) \in E(G_S)} (2-E(S_1^{I})_{i,j} -E(S_2^{J})_{i,j}) = 2|E(G_S)| -2(\ell\iota\aleph+\aleph-1)+\ell(2\iota\aleph-|I|-|J|) \,,
\end{align*}
and the last inequality holds since for all $F \in \mathcal F(S)$ we have that
\begin{align*}
    &\ell(2\iota\aleph-|I|-|J|) + \sum_{(i,j)\in F} \big( E(S_1^{I})_{i,j} + E(S_2^{J})_{i,j} - 2 \big) \\
    \geq\ & \sum_{(i,j)\in F} \big( E(S_1)_{i,j} + E(S_2)_{i,j} - 2 \big) \geq 35\aleph^3  \,.
\end{align*}
Plugging this estimation into \eqref{eq-f-B-2-relax-1}, we get that (note that $0 \leq \tau(G_S) \leq 20\aleph^2$)
\begin{align}
    \eqref{eq-f-B-2-relax-1} &\leq n^{-25\aleph^3} \sum_{(S_1,S_2) \in \mathcal B_2} \sum_{F \in \mathcal F(S)} n^{-|E(G_S)|} \nonumber \\
    &\leq n^{-25\aleph^3} \sum_{k \leq 3\ell\iota\aleph} \sum_{m \leq 20\aleph^2} n^{-k} \cdot \#\{ (S_1,S_2): |E(\widetilde S_1 \cup \widetilde S_2)|=k, \tau(\widetilde S_1 \cup \widetilde S_2)=m \} \nonumber \\
    &\overset{\text{Lemma~\ref{lem-coarse-enumeration-B-1}}}{\leq} n^{-25\aleph^3} \sum_{k \leq 3\ell\iota\aleph} \sum_{m \leq 20\aleph^2} n^{-k} \cdot n^{ (k-m) + 12\aleph^2 + o(1) } = o(1) \,. \label{eq-f-B-2-relax-2}
\end{align}
This leads to the desired result.

\subsection{Proof of Lemma~\ref{lem-control-f-B-3}}{\label{subsec:proof-lem-C.5}}

Denote $G_{\cup}=\widetilde S_1\cup \widetilde S_2\cup \widetilde K_1 \cup \widetilde K_2$. By Lemma~\ref{lem-Pb-var-bad-part-2-expectation-estimate-general} and the fact that 
\begin{align*}
    \sum_{(i,j) \in E(G_{\cup})} ( E(S_1)_{i,j} + E(S_2)_{i,j} + E(K_1)_{i,j} + E(K_2)_{i,j} ) = 4(\aleph-1+\ell\iota\aleph)\,,
\end{align*}
we have (similarly as how we derive \eqref{eq-Pb-var-bad-part-2-relax-I}) 
\begin{align}
    \mathbb E_{\Pb_{\mathsf{id}}}\big[ f_{\mathcal B_3}^2 \big] &= \sum_{\mathbf H \in \mathcal H } \frac{ (\operatorname{Aut}(\mathsf T(\mathbf H)))^2 s^{2\aleph-2} 2^{2\iota\aleph} }{ n^{2\aleph+2(\ell-2)\iota\aleph}(\epsilon^2\lambda s)^{2\ell\iota\aleph}} \sum_{ (S_1,S_2),(K_1,K_2) \in \mathcal B_3(\mathbf H) } \E_{\P_{\mathsf{id}}} [\phi_{S_1,S_2} \phi_{K_1,K_2}] \nonumber \\
    &\leq \sum_{\mathbf H \in \mathcal H} \frac{ n^{o(1)} s^{2\aleph-2} }{ n^{2\aleph+2(\ell-2)\iota\aleph}(\epsilon^2\lambda s)^{2\ell\iota\aleph}} \sum_{ (S_1,S_2),(K_1,K_2) \in \mathcal B_3(\mathbf H) } \big( \frac{n}{\epsilon^2 \lambda s} \big)^{ 2(\aleph-1+\ell\iota\aleph)-|E(G_{\cup})| } \nonumber \\
    &\leq \frac{ n^{4\iota\aleph-2+o(1)}s^{2\aleph-2} |\mathcal H| }{ (\epsilon^2 \lambda s)^{4\ell\iota\aleph +2(\aleph - 1)} } \sum_{ (S_1,S_2),(K_1,K_2) \in \mathcal B_3(\mathbf H) } \big( \tfrac{\epsilon^2\lambda s}{n} \big)^{|E(G_{\cup})|} \nonumber \\ 
    &= \frac{ n^{4\iota\aleph-2+o(1)} }{ (\epsilon^2 \lambda s)^{4\ell\iota\aleph } } \sum_{ (S_1,S_2),(K_1,K_2) \in \mathcal B_3(\mathbf H) } \big( \tfrac{\epsilon^2\lambda s}{n} \big)^{|E(G_{\cup})|} \label{eq-f-B-3-relax-1}\,.
\end{align}
In addition, the following holds by Remark~\ref{cor-of-lem-4.4}: for $G_\cup \subset \mathsf K_n$, if
\begin{align}
    \exists (S_1,S_2),(K_1,K_2) \in \mathcal B_3(\mathbf H) \mbox{ such that } G_{\cup}= \widetilde S_1\cup \widetilde S_2\cup \widetilde K_1 \cup \widetilde K_2\mbox{ and }\mathcal L(G_{\cup}) \subset V(G_{\geq 2}) \,, \label{eq-characerize-G-cup-appendix-pre}
\end{align}
then we have
\begin{align}
    4\iota\aleph-1-\tau(G_\cup) \leq \tfrac{4\ell\iota\aleph+4\aleph-|E(G_\cup)|}{\ell/2} \,. \label{eq-characerize-G-cup-appendix}
\end{align}
We say a simple graph $\mathbf G$ is with \eqref{eq-characerize-G-cup-appendix-pre} and \eqref{eq-characerize-G-cup-appendix} if $\mathbf G \cong G_\cup$ for some $G_\cup$ satisfying \eqref{eq-characerize-G-cup-appendix-pre} and \eqref{eq-characerize-G-cup-appendix}. Using this terminology, we get that
\begin{align}
    \eqref{eq-f-B-3-relax-1} \leq\ & \frac{ n^{4\iota\aleph-2+o(1)} }{ (\epsilon^2 \lambda s)^{4\ell\iota\aleph} } \sum_{ \mathbf G \text{ with } \eqref{eq-characerize-G-cup-appendix-pre},\eqref{eq-characerize-G-cup-appendix} } \big( \tfrac{\epsilon^2\lambda s}{n} \big)^{|E(\mathbf G)|} \nonumber \\
    & \cdot \sum_{G_{\cup} \cong \mathbf G} \#\{ (S_1,S_2),(K_1,K_2) \in \mathcal B_3: \widetilde S_1\cup \widetilde S_2\cup \widetilde K_1 \cup \widetilde K_2 = G_{\cup} \} \,. \label{eq-f-B-3-relax-2}
\end{align}
We first need the following lemma. 
\begin{lemma}{\label{lem-enu-decorated-trees-with-few-tangles}}
    Given $J \subset \mathsf K_n$ with at most $\mathtt z = O(1)$ connected components, we have
    \begin{align}
        \#\big\{ S \subset J:\exists S'\mbox{ such that }(S,S') \in \mathcal B_3\big\} \leq \binom{|E(J)|}{\aleph} n^{o(1)}\aleph^{2\iota\aleph} e^{ \sqrt{\log n} (\log \log n)^4 (\tau(J)+\mathtt z) } \,. \label{eq-enu-decorated-trees-with-few-tangles}
    \end{align}
\end{lemma}
\begin{proof}
    Note that the enumeration of $\mathsf T(S)$ is bounded by $\binom{|E(J)|}{\aleph}$, and given $\mathsf T(S)$ the enumeration of $\mathsf P(S)$ is bounded by $\aleph^{2\iota\aleph}$. Given $\mathsf T(S)$ and $\mathsf P(S)$, we now explain how to choose $\mathsf L_{i}(S)$. By Lemma~\ref{lem-revised-decomposition-H-Subset-S}, $J \setminus \mathsf T(S)$ can be decomposed into at most $5(\tau(J)+1)$ self-avoiding paths such that they only overlap at endpoints. Now define the trace of $\mathsf L_i(S)$ by $\mathsf{Tr}_i = \widetilde{\mathsf L_i(S)}$. Since $\mathsf{Tr}_i$ is the union of a subset of these $5(\tau(J)+1)$ self-avoiding paths, the enumeration for $\mathsf{Tr}_i$ is bounded by $2^{5(\tau(J)+1)}$. 
    
    Given $\mathsf{Tr}_i$, we will bound the enumeration of $\mathsf L_i(S)$ by constructing a surjection from certain tuples of variables (see \eqref{lemma-C-long-tuple} below) to realizations of $\mathsf L_i(S)$, and then bound the enumeration of such tuples. The following proof is similar to the arguments in \cite{MNS18}. We will introduce some definitions first to make our motivations in the following proof clearer. Recall Definition~\ref{def-tangle}. Define $\mathcal T = \mathcal T(\mathsf L_i(S))$ to be the family of minimal multigraphs $T$ such that $\mathsf L_i(S)\setminus T$ is tangle-free. By definition of $\mathcal B_3$, $\mathsf L_i(S)$ has at most $40\aleph^3$ tangles and therefore $|E(T)| \leq 40\aleph^3$ for $T \in \mathcal T(\mathsf L_i(S))$. Given any $T \in \mathcal T$, any vertex $v$ in $\mathsf{Tr}_i$ and any neighbor $u$ of $v$ in $\mathsf{Tr}_i$, we say $u$ is a \emph{short} neighbor of $v$ if it is in a cycle in $\mathsf L_i(S) \setminus T$ with length bounded by $\sqrt{\log n}$; by definition of $\mathcal T$, every vertex has at most two short neighbors. We say $u$ is a \emph{tangled} neighbor of $v$ if it is not a short neighbor of $v$, but it is in a cycle in $\mathsf L_i(S)$ with length bounded by $\sqrt{\log n}$. We say $u$ is a \emph{long} neighbor of $v$ if it is neither a short nor a tangled neighbor of $v$. Define $V^{(i)}_{\geq 3} = \{ v \in V(\mathsf{Tr}_i ) : d_{\mathsf{Tr}_i}(v)\geq 3\}$, which is determined by $\mathsf{Tr}_i$ (here $d_{\mathsf{Tr}_i}(v)$ is the degree of $v$ in $\mathsf{Tr}_i$). Also, for each $\mathsf L_i(S)$ such that $\mathsf L_i(S) =(v_0 ,v_1, \ldots ,v_\ell)$ where $(v_0 ,v_\ell) \in \mathsf P(S)$, and given any $T \in \mathcal T$ we define a list $\mathsf{List}^{(i)}_T(v)$ with $\ell$ lines $\mathsf{List}^{(i)}_T(v)_1 , \ldots , \mathsf{List}^{(i)}_T(v)_\ell$ on each $v \in V^{(i)}_{\geq 3}$ such that the following hold: (i) each line has three possible states: blank, non-short-coming and non-short-going; (ii) the state of $\mathsf{List}^{(i)}_T(v)_j$ is non-short-coming (respectively, non-short-going) if and only if $v_{j-1}$ (respectively, $v_{j+1}$) exists, $v=v_j$ and $v_{j-1}$ (respectively, $v_{j+1}$) is not a short neighbor of $v_j$. Next we prove the following claim: given $\mathsf L_i(S) = (v_0 ,v_1, \ldots ,v_\ell)$, $T \in \mathcal T(\mathsf L_i(S))$ and $v \in V^{(i)}_{\geq 3}$ (which is determined by $\mathsf L_i(S)$), there are at most $\tfrac{\ell}{\sqrt{\log n}} +40\aleph^3$ lines in $\mathsf{List}^{(i)}_T(v)$ with a non-short-coming or non-short-going state. To this end, (by symmetry) it suffices to deal with the non-short-going case. By definition of long and tangled neighbors, every time a step of $\mathsf L_i(S)$ goes to a long neighbor (of some vertex $v$) it takes at least $\sqrt{\log n}$ steps to go back to $v$, and every time a step of $\mathsf L_i(S)$ goes to a tangled neighbor (of some vertex $v$) it either takes at least $\sqrt{\log n}$ steps to go back to $v$, or traverses an edge in $T$ at least once. Therefore, our claim follows from $|E(T)|\leq 40\aleph^3$.

    Now we construct a mapping 
    \begin{align}
        \mathcal{X}: &\Big( \mathsf{Tr}_i,(v_0,v_1,v_\ell), \big(s(v),\mathsf{H}(v)\big)_{v \in V^{(i)}_{\geq 3}}, \big(w_{v,j}:1\leq j \leq \tfrac{2\ell}{\sqrt{\log n}} + 80\aleph^3,v\in V^{(i)}_{\geq 3} \big) \Big) \nonumber\\
        &\longrightarrow (v_0,v_1,\ldots,v_\ell) \,, \label{lemma-C-long-tuple} 
    \end{align}
    where the notations in \ref{lemma-C-long-tuple} are specified as follows: for any $v \in V^{(i)}_{\geq 3}$, $s(v) = \{ v',v'' \}$ where $v',v'' \in \mathsf{Nei}_{\mathsf{Tr}_i}(v)$; $\mathsf{H}(v)$ is a list with $\ell$ lines $\mathsf H(v)_1 ,\ldots,\mathsf H(v)_\ell$ such that each line has three possible states: blank, non-short coming and non-short-going, and at most $\tfrac{\ell}{\sqrt{\log n}} + 40\aleph^3$ lines of $\mathsf{H}(v)$ are in the state of non-short-coming (going); for $v\in V^{(i)}_{\geq 3}$ each $(w_{v, \cdot})$ is a vertex tuple where $w_{v,j} \in V(\mathsf{Tr}_i)$ for $1\leq j \leq \tfrac{2\ell}{\sqrt{\log n}} + 80\aleph^3$. Also, the image $(v_0,v_1,\ldots,v_\ell)$ serves as the non-backtracking path $\mathsf L_i(S)$ we hope for. The mapping is constructed as follows: First, we start from $v_0$ and choose the steps of $\mathsf L_i(S)$ one by one. For the $x$th step starting at $v_{x-1}$, if $x=1$ then $v_1$ is already in the input; else if $d_{\mathsf{Tr}_i}(v_{x-1})\leq 2$, the choice of $v_x$ is unique by the fact that $(v_0 ,v_1 ,\ldots ,v_\ell)$ is non-backtracking; else if the state of the $x$th line of $\mathsf H(v_{x-1})$ is blank, we set $v_x$ the element in $s(v_{x-1})$ such that $v_x \neq v_{x-2}$; else we set $v_x =w_{v_{x-1},y}$ if $x = \mathsf{ind}(y, \mathsf H(v))$, where $\mathsf{ind}(y, \mathsf H(v))$ is the $y$th smallest index $z$ such that $\mathsf H(v)_z$ is not in the blank state. Then we have obtained $(v_0 ,v_1 ,\ldots ,v_\ell) =\mathsf L_i (S)$. 
    \begin{table}[!ht]
    \centering    
    \begin{tabular}{|l|l|l|l|l|l|l|l|l|l|}
    \hline
        \multirow{6}*{
        \begin{tikzpicture}
        \node (A) [above=0.01cm] at (0,0.9){1};
        \node (B) [above=0.01cm, right=0.015cm] at (0.8,0.3){2};
        \node (C) [below=0.15cm, right=0.015cm] at (0.45,-0.68){3};
        \node (D) [below=0.15cm, left=0.001cm] at (-0.45,-0.68){4};
        \node (E) [above=0.01cm, left=0.015cm] at (-0.8,0.3){5};
        \node (F) [above=0.01cm] at (0.7,0.9){6};
        \draw [very thin, black] (0,0.9) -- (0.7,0.9);
        \draw [very thin, black] (0,0.9) -- (0.8,0.3);
        \draw [very thin, black] (0.45,-0.68) -- (0.8,0.3);
        \draw [very thin, red] (0.45,-0.68) -- (0,0.9);
        \draw [very thin, red] (-0.45,-0.68) -- (0,0.9);
        \draw [very thin, black] (0.45,-0.68) -- (-0.45,-0.68);
        \draw [very thin, black] (-0.45,-0.68) -- (-0.8,0.3);
        \draw [very thin, black] (0,0.9) -- (-0.8,0.3);
        \draw [very thin, red] (0.45,-0.68) -- (-0.8,0.3);
        \draw [very thin, red] (-0.45,-0.68) -- (0.8,0.3);
        \draw [very thin, red] (0.8,0.3) -- (-0.8,0.3);
        \end{tikzpicture}
        } & $v$ & $S(v)$ & $H(v)$ & $w_v$ \\ \cline{2-5}
         ~ & 1 & $\{5,2\}$ & NSG at 5th line; NSC at 10th, 13th line & [3,2,6,arbitrary] \\ \cline{2-5}
         ~ & 2 & $\{1,3\}$ & NSC at 8th line; NSG at 11th line & [4,4,arbitrary] \\ \cline{2-5}
         ~ & 3 & $\{2,4\}$ & NSC at 6th line & [5,arbitrary] \\ \cline{2-5}
         ~ & 4 & $\{3,5\}$ & NSC at 9th, 12th line & [1,1,arbitrary] \\ \cline{2-5}
         ~ & 5 & $\{4,1\}$ & NSC at 7th line & [1,2,arbitrary] \\ \hline
    \end{tabular}
    \caption{\noindent An example of a non-backtracking path and its preimage. Left: an example of $\mathsf{Tr}_i$; Right: a preimage of the non-backtracking path $\mathsf{L}_i(S) = (2,3,4,5,1,3,5,2,4,1,2,4,1,6)$, where $V^{(i)}_{\geq 3} = \{ 1,2,3,4,5\}$. We omitted $(v_0,v_\ell) = (2,6)$ and $v_1 = 3$ in the preimage table, and NSC (NSG) in the table refers to non-short-coming (non-short-going). The edges in $\mathsf{Tr}_i$ which are in a possible $T$ are colored red. The lines of $H(v)$'s which are not mentioned are in the blank state. The notation $w_v$ stands for the sequence $\{ w_{v,j} \}_{ 1 \leq j \leq \tfrac{2\ell}{\sqrt{\log n}} + 80\aleph^3  }$, and for any $1 \leq v \leq 5$, an ``arbitrary'' in the end of $w_v$ means an arbitrary sequence of vertices with appropriate length.}
    \label{table:NBP-bijection}
\end{table}
    
    Then we claim that given any $\mathsf{Tr}_i$ and $(v_0 ,v_\ell)$, $\mathcal X$ is a surjection to all non-backtracking paths $\mathsf L_i(S)$ with at most $40\aleph^3$ tangles and with $\widetilde{\mathsf L_i(S)} = \mathsf{Tr}_i$. It suffices to show that for each such $\mathsf L_i(S)$ we can construct a preimage of $\mathcal X$ with image $\mathsf L_i(S)$. To this end, consider an instance $\mathsf L_i (S) = (v_0 ,v_1 ,\ldots ,v_\ell)$ with trace $\mathsf{Tr}_i$, and an arbitrary $T \in \mathcal T(\mathsf L_i(S))$ with $|T| \leq 40\aleph^3$. Then for any $v\in V^{(i)}_{\geq 3}$, we determine the corresponding variables in our preimage by
    \begin{align*}
        s(v) &= \{ u \in \mathsf{Nei}_{\mathsf{Tr}_i}(v):u\text{ is a short neighbor of }v\text{ with respect to }T\}\,,\\
        \mathsf H(v) &= \mathsf{List}^{(i)}_T (v)\,,\, w_{v,y}=v_{\mathsf{ind}(y, \mathsf H(v))}\,.
    \end{align*}
    Now we argue that $\mathcal X$ indeed maps
    \begin{align*}
    \Big(\mathsf{Tr}_i ,(v_0 ,v_\ell ) ,v_1 , \big(s(v),\mathsf{H}(v)\big)_{v \in V^{(i)}_{\geq 3}} , \big(w_{v,j} :1\leq j \leq \tfrac{2\ell}{\sqrt{\log n}} + 80\aleph^3 ,v\in V^{(i)}_{\geq 3} \big)\Big)
    \end{align*}
    to $(v_0,v_1,\ldots,v_\ell) =\mathsf L_i (S)$. Otherwise, suppose that the image is $(v_0',v_1',\ldots,v_\ell') \neq (v_0,v_1,\ldots,v_\ell)$ and there exists a minimal $x$ such that $v_x' \neq v_x$. Then by definition of our map we have $x\geq 2$; we also have $d_{\mathsf{Tr}_i}(v_{x-1}) \geq 3$ since both $(v_0',v_1',\ldots,v_\ell')$ and $(v_0,v_1,\ldots,v_\ell)$ are non-backtracking, making it impossible for $v_{x-2} = v_{x-2}'$, $v_{x-1} = v_{x-1}'$ and $v_x \neq v_x'$ when $d_{\mathsf{Tr}_i}(v_{x-1}) \leq 2$. If the state of the $x$th line of $\mathsf H(v_{x-1}) = \mathsf{List}^{(i)}_T (v_{x-1})$ is blank, then by definition of $\mathsf{List}^{(i)}_T$ we have that both $v_{x-2}$ and $v_{x}$ are short neighbors of $v_{x-1}$, and since $(v_0 ,v_1 ,\ldots ,v_\ell)$ is non-backtracking we have $s(v_{x-1}) = \{ v_{x-2} ,v_{x}\}$. By the definition of $\mathcal X$, we have $v_x' \in s(v_{x-1}') = s(v_{x-1})$ and $v_x' \neq v_{x-2}' = v_{x-2}$, which gives $v_x' = v_x$, contradicting $v_x' \neq v_x$. Therefore, the state of the $x$th line of $\mathsf H(v_{x-1}) = \mathsf{List}^{(i)}_T (v_{x-1})$ is non-short-coming or non-short-going. However, in both cases by definition of $\mathcal X$ we have $v'_{x} = w_{v,y} = v_{\mathsf{ind}(y, \mathsf H(v))}$ where $\mathsf{ind}(y, \mathsf H(v)) =x$, which is a contradiction. Therefore, we have shown that $\mathcal X$ is a surjection, which gives that given $\mathsf{Tr}_i$ and $(v_0 ,v_\ell)$, the enumeration of non-backtracking paths $\mathsf L_i(S)$ with at most $40\aleph^3$ tangles and with $\widetilde{\mathsf L_i(S)} = \mathsf{Tr}_i$ is bounded by the enumeration of \eqref{lemma-C-long-tuple}.
    
    Finally, we bound the enumeration of \eqref{lemma-C-long-tuple} given $\mathsf{Tr}_i$ and $(v_0 ,v_\ell)$. Since $|V(\mathsf{Tr}_i)| \leq |E(\mathsf{Tr}_i)| +1 \leq \ell +1$, the enumeration of $v_1$ and $\big(w_{v,j} :1\leq j \leq \tfrac{2\ell}{\sqrt{\log n}} + 80\aleph^3 ,v\in V^{(i)}_{\geq 3} \big)$ is bounded by $(\ell +1)^{1+(2\ell/\sqrt{\log n}+80\aleph^3)|V^{(i)}_{\geq 3}|}$. Then, since for each $v\in V_{\geq 3}^{(i)}$ there are at most $\tfrac{2\ell}{\sqrt{\log n}}+80\aleph^3$ lines with states that are not blank, the enumeration of $\big(\mathsf{H}(v)\big)_{v \in V^{(i)}_{\geq 3}}$ is bounded by $(\ell +1)^{(2\ell/\sqrt{\log n}+80\aleph^3)|V^{(i)}_{\geq 3}|}$. Also, by a direct upper bound the enumeration of $\big(s(v)\big)_{v \in V^{(i)}_{\geq 3}}$ is bounded by $(2\ell +2)^{2|V^{(i)}_{\geq 3}|}$. Therefore, by the fact that $|V^{(i)}_{\geq 3}| \leq \tau(\mathsf{L}_i(S)) \leq \tau (S) + 1 \leq \tau(J) + \mathtt z$ we obtain that given $\mathsf{Tr}_i$ and $(v_0 ,v_\ell)$ the enumeration of \eqref{lemma-C-long-tuple} is bounded by
    \begin{align*}
        (2\ell+2)^{ 1+(4\ell/\sqrt{\log n}+2+160\aleph^3)(\tau(J)+\mathtt z) } \leq e^{ \sqrt{\log n} (\log\log n)^3 (\tau(J)+\mathtt z) } \,.
    \end{align*}
    Thus, combining all the enumerations above, the enumeration of $S$ is bounded by 
    \begin{align*}
        \binom{|E(J)|}{\aleph}\cdot\aleph^{2\iota\aleph}\cdot 2^{\iota \aleph \cdot 5 (\tau (J)+1)}\cdot e^{ \aleph\sqrt{\log n} (\log \log n)^3 (\tau(J)+\mathtt z) }\,,
    \end{align*}
    and the desired result follows from $\aleph = o(\log \log n)$ and $e^{ \sqrt{\log n} (\log \log n)^4 \mathtt z } = n^{o(1)}$.
\end{proof}

Now we see that
\begin{align}
    \eqref{eq-f-B-3-relax-2} \leq\ & \frac{ n^{4\iota\aleph-2+o(1)} }{ (\epsilon^2 \lambda s)^{4\ell\iota\aleph} } \sum_{ \mathbf G \text{ with } \eqref{eq-characerize-G-cup-appendix-pre},\eqref{eq-characerize-G-cup-appendix} } \big( \tfrac{\epsilon^2\lambda s}{n} \big)^{|E(\mathbf G)|} \sum_{G \cong \mathbf G}  \binom{|E(G)|}{\aleph}^4 \aleph^{8\iota\aleph} e^{ 4\sqrt{\log n} (\log \log n)^4 \tau(G) } \nonumber \\
    =\ & \frac{ n^{4\iota\aleph-2+o(1)} }{ (\epsilon^2 \lambda s)^{4\ell\iota\aleph} } \sum_{ \mathbf G \text{ with } \eqref{eq-characerize-G-cup-appendix-pre},\eqref{eq-characerize-G-cup-appendix} } \big( \tfrac{\epsilon^2\lambda s}{n} \big)^{|E(\mathbf G)|} e^{ 4\sqrt{\log n} (\log \log n)^4 \tau(\mathbf G) } \#\{ G \subset \mathsf K_n : G \cong \mathbf G \} \nonumber \\
    \leq\ & n^{4\iota\aleph-2+o(1)} \sum_{ \substack{0\leq x \leq 4\iota\ell\aleph+4\aleph\\ \frac{2x-8\aleph-4\ell\iota\aleph-\ell}{\ell}\leq y \leq 2\aleph^2} } (\epsilon^2\lambda s)^{-(4\ell\iota\aleph+4\aleph-x)} n^{-x} e^{ 4\sqrt{\log n} (\log \log n)^4 y} n^{-y+x}  \nonumber \\
    &\cdot \#\{ \mathbf G: \mathbf G \text{ with } \eqref{eq-characerize-G-cup-appendix-pre},\eqref{eq-characerize-G-cup-appendix} , |E(\mathbf G)| = x, \tau(\mathbf G) = y \} \nonumber \\
    \leq\ & n^{4\iota\aleph-2+o(1)} \sum_{0\leq x \leq 4\iota\ell\aleph+4\aleph} n^{-2\tfrac{4\ell\iota\aleph+4\aleph-x}{l}} \big(\tfrac{n}{(\ell\aleph )^2}\big)^{-\tfrac{2x-8\aleph-4\ell\iota\aleph-\ell}{\ell}} \nonumber \\
    =\ & n^{-1+o(1)} \,,
\end{align}
where in the first inequality we apply Lemma~\ref{lem-enu-decorated-trees-with-few-tangles}, the third inequality holds by \eqref{eq-characerize-G-cup-appendix} and $(\epsilon^2 \lambda s)^{4\aleph} = n^{o(1)}$, and the last inequality holds by $e^{ 4\sqrt{\log n} (\log \log n)^4 y} \leq e^{ 8\sqrt{\log n} (\log \log n)^4 \aleph^2} = n^{o(1)}$ and Lemma~\ref{lem-enu-union-of-decorated-trees}. This leads to the desired result.

\bibliographystyle{alpha}

\small
\begin{thebibliography}{10} 

\bibitem[AS16]{AS16}
Emmanuel Abbe and Colin Sandon. 
\newblock Achieving the KS threshold in the general stochastic block model with linearized acyclic belief propagation. 
\newblock In {\em Advances in Neural Information Processing Systems (NIPS)}, volume~29, pages 1334--1342. Curran Associates, Inc., 2016.

\bibitem[AS18]{AS18}
Emmanuel Abbe and Colin Sandon. 
\newblock Proof of the achievability conjectures for the general stochastic block model. 
\newblock {\em Communications on Pure and Applied Mathematics}, 71(7):1334--1406, 2018.

\bibitem[AYZ95]{AYZ95}
Noga Alon, Raphael Yuster, and Uri Zwick. 
\newblock Color-coding. 
\newblock {\em Journal of the ACM}, 42(4):844--856, 1995.

\bibitem[Bab16]{Babai16}
L\'{a}szl\'{o} Babai. 
\newblock Graph isomorphism in quasipolynomial time. 
\newblock In {\em Proceedings of the 48th Annual ACM Symposium on Theory of Computing (STOC)}, pages 684--697. ACM, 2016. 

\bibitem[BAH+22]{BAH+22}
Afonso S. Bandeira, Ahmed El Alaoui, Samuel B. Hopkins, Tselil Schramm, Alexander S. Wein, and Ilias Zadik.
\newblock The Franz-Parisi criterion and computational trade-offs in high dimensional statistics.
\newblock In {\em Advances in Neural Information Processing Systems (NIPS)}, volume~35, pages 33831--33844. Curran Associates, Inc., 2022.

\bibitem[BKW19]{BKW19}
Afonso S. Bandeira, Dmitriy Kunisky, and Alexander S. Wein.
\newblock Computational hardness of certifying bounds on constrained PCA problems.
\newblock In {\em Proceedings of the 11th Innovations in Theoretical Computer Science Conference (ITCS)}, Schloss Dagstuhl-Leibniz-Zentrumf\H{u}r Informatik, 2020.

\bibitem[BMNN16]{BMNN16}
Jess Banks, Cristopher Moore, Joe Neeman, and Praneeth Netrapalli.
\newblock Information-theoretic thresholds for community detection in sparse networks.
\newblock In {\em Proceedings of the 29th Conference on Learning Theory (COLT)}, pages 383--416. PMLR, 2016.

\bibitem[BCL+19]{BCL+19}
Boaz Barak, Chi-Ning Chou, Zhixian Lei, Tselil Schramm, and Yueqi Sheng. 
\newblock (Nearly) efficient algorithms for the graph matching problem on correlated random graphs. 
\newblock In {\em Advances in Neural Information Processing Systems (NIPS)}, volume~32, pages 9190--9198. Curran Associates, Inc., 2019.

\bibitem[BHK+19]{BHK+19}
Boaz Barak, Samuel B. Hopkins, Jonathan Kelner, Pravesh K. Kothari, Ankur Moitra, and Aaron Potechin.
\newblock A nearly tight sum-of-squares lower bound for the planted clique problem. 
\newblock {\em SIAM Journal on Computing}, 48(2):687--735, 2019.

\bibitem[Bop87]{Boppana87}
Ravi B. Boppana. 
\newblock Eigenvalues and graph bisection: An average-case analysis. 
\newblock In {\em Proceedings of the IEEE 28th Annual Symposium on Foundations of Computer Science (FOCS)}, pages 280--285. IEEE, 1987.

\bibitem[BLM15]{BLM15}
Charles Bordenave, Marc Lelarge, and Laurent Massouli\'{e}.
\newblock Non-backtracking spectrum of random graphs: Community detection and non-regular Ramanujan graphs.
\newblock In {\em Proceedings of the IEEE 56th Annual Symposium on Foundations of Computer Science (FOCS)}, pages 1347--1357. IEEE, 2015.

\bibitem[BSH19+]{BSH19}
Mahdi Bozorg, Saber Salehkaleybar and Matin Hashemi.
\newblock Seedless graph matching via tail of degree distribution for correlated \ER graphs.
\newblock arXiv preprint, arXiv:1907.06334.

\bibitem[BH22]{BH22}
Guy Bresler and Brice Huang.
\newblock The algorithmic phase transition of random $k$-SAT for low degree polynomials.
\newblock In \emph{Proceedings of the IEEE 62nd Annual Symposium on Foundations of Computer Science (FOCS)}, pages 298--309. IEEE, 2022.

\bibitem[BDG+16]{BDG+16}
Gerandy Brito, Ioana Dumitriu, Shirshendu Ganguly, Christopher Hoffman, and Linh V. Tran. 
\newblock Recovery and rigidity in a regular stochastic block model. 
\newblock In {\em Proceedings of the 27th Annual ACM-SIAM Symposium on Discrete Algorithms (SODA)}, pages 1589--1601. SIAM, 2016.

\bibitem[BCL+84]{BCL+84}
Thang N. Bui, Soma Chaudhuri, Frank T. Leighton, and Michael Sipser. 
\newblock Graph bisection algorithms with good average case behavior. 
\newblock In {\em Proceedings of the IEEE 25th Annual Symposium on Foundations of Computer Science (FOCS)}, pages 181--192. IEEE, 1984.

\bibitem[BCPP98]{BCPP98}
Rainer E. Burkard, Eranda Cela, Panos M. Pardalos, and Leonidas S. Pitsoulis.
\newblock The quadratic assignment problem. 
\newblock {\em Handbook of Combinatorial Optimization}, pages 1713--1809. Springer, 1998.

\bibitem[CR24]{CR24}
Shuwen Chai and Mikl\'{o}s Z. R\'{a}cz. 
\newblock Efficient graph matching for correlated stochastic block models. 
\newblock In {\em Advances in Neural Information Processing Systems (NIPS)}, volume~37, pages 116388--116461. Curran Associates, Inc., 2024.

\bibitem[CDGL24+]{CDGL24+}
Guanyi Chen, Jian Ding, Shuyang Gong, and Zhangsong Li.
\newblock A computational transition for detecting correlated stochastic block models by low-degree polynomials.
\newblock to appear in {\em Annals of Statistics}.

\bibitem[CS21+]{CS21+}
Byron Chin and Allan Sly.
\newblock Optimal reconstruction of general sparse stochastic block models.
\newblock to appear in {\em Annales de l'IHP Probabilit\'{e}s et Statistiques}.

\bibitem[CB81]{CB81}
Charles J. Colbourn and Kellogg S. Booth. 
\newblock Linear time automorphism algorithms for trees, interval graphs, and planar graphs. 
\newblock {\em SIAM Journal on Computing}, 10(1):203--225, 1981.

\bibitem[CK16]{CK16}
Daniel Cullina and Negar Kiyavash.
\newblock Improved achievability and converse bounds for \ER graph matching.
\newblock In {\em Proceedings of the 2016 ACM SIGMETRICS International Conference on Measurement and Modeling of Computer Science}, page 63--72. ACM, 2016. 

\bibitem[CK17]{CK17}
Daniel Cullina and Negar Kiyavash.
\newblock Exact alignment recovery for correlated \ER graphs.
\newblock arXiv preprint, arXiv:1711.06783.

\bibitem[DCK+19]{DCK+19}
Osman E. Dai, Daniel Cullina, Negar Kiyavash, and Matthias Grossglauser.
\newblock Analysis of a canonical labeling algorithm for the alignment of correlated \ER graphs.
\newblock In {\em Proceedings of the ACM on Measurement and Analysis of Computing Systems}, pages 1--25. ACM, 2019.

\bibitem[DKMZ11]{DKMZ11}
Aurelien Decelle, Florent Krzakala, Cristopher Moore, and Lenka Zdeborov\'{a}.
\newblock Asymptotic analysis of the stochastic block model for modular networks and its algorithmic applications.
\newblock {\em Physics Review E}, 84:066106, 2011.

\bibitem[DMW23+]{DMW23+}
Abhishek Dhawan, Cheng Mao, and Alexander S. Wein.
\newblock Detection of dense subhypergraphs by low-degree polynomials.
\newblock to appear in {\em Random Structures and Algorithms}.

\bibitem[DD23a]{DD23a}
Jian Ding and Hang Du.
\newblock Detection threshold for correlated \ER graphs via densest subgraph.
\newblock \emph{IEEE Transactions on Information Theory}, 69(8):5289--5298, 2023.

\bibitem[DD23b]{DD23b}
Jian Ding and Hang Du.
\newblock Matching recovery threshold for correlated random graphs.
\newblock {\em Annals of Statistics}, 51(4): 1718--1743, 2023.

\bibitem[DDL23+]{DDL23+}
Jian Ding, Hang Du, and Zhangsong Li.
\newblock Low-degree hardness of detection for correlated \ER Graphs.
\newblock to appear in {\em Annals of Statistics}.

\bibitem[DL25]{DL22+}
Jian Ding and Zhangsong Li.
\newblock A polynomial time iterative algorithm for matching Gaussian matrices with non-vanishing correlation.
\newblock {\em Foundations of Computational Mathematics}, 25(4):1287--1344, 2025. 

\bibitem[DL23+]{DL23+}
Jian Ding and Zhangsong Li.
\newblock A polynomial-time iterative algorithm for random graph matching with non-vanishing correlation.
\newblock arXiv preprint, arXiv:2306.00266.

\bibitem[DMWX21]{DMWX21}
Jian Ding, Zongming Ma, Yihong Wu, and Jiaming~Xu.
\newblock Efficient random graph matching via degree profiles.
\newblock {\em Probability Theory and Related Fields}, 179(1-2):29--115, 2021.

\bibitem[DKW22]{DKW22}
Yunzi Ding, Dmitriy Kunisky, Alexander S. Wein, and Afonso S. Bandeira.
\newblock Subexponential-time algorithms for sparse PCA.
\newblock {\em Foundations of Computational Mathematics}, 22(1):1--50, 2022.

\bibitem[Din15]{Dinneen15}
Michael J. Dinneen. 
\newblock Constant time generation of free trees. 
\newblock University of Auckland Lecture, 2015.

\bibitem[Du25+]{Du25+}
Hang Du.
\newblock Optimal recovery of correlated \ER graphs.
\newblock arXiv preprint, arXiv:2502.12077.

\bibitem[DF89]{DF89}
Martin E. Dyer and Alan M. Frieze. 
\newblock The solution of some random NP-hard problems in polynomial expected time.
\newblock {\em Journal of Algorithms}, 10(4):451--489, 1989.

\bibitem[FMWX23a]{FMWX23a}
Zhou Fan, Cheng Mao, Yihong Wu, and Jiaming Xu.
\newblock Spectral graph matching and regularized quadratic relaxations I: Algorithm and Gaussian analysis.
\newblock {\em Foundations of Computational Mathematics}, 23(5):1511--1565, 2023.

\bibitem[FMWX23b]{FMWX23b}
Zhou Fan, Cheng Mao, Yihong Wu, and Jiaming Xu.
\newblock Spectral graph matching and regularized quadratic relaxations {II}: \ER graphs and universality.
\newblock {\em Foundations of Computational Mathematics}, 23(5):1567--1617, 2023.

\bibitem[GJW20]{GJW20}
David Gamarnik, Aukosh Jagannath, and Alexander S. Wein.
\newblock Hardness of random optimization problems for Boolean circuits, low-degree polynomials, and Langevin dynamics.
\newblock \emph{SIAM Journal on Computing}, 53(1):1--46, 2024. 

\bibitem[GM20]{GM20}
Luca Ganassali and Laurent Massouli\'e.
\newblock From tree matching to sparse graph alignment.
\newblock In {\em Proceedings of the 33rd Conference on Learning Theory (COLT)}, pages 1633--1665. PMLR, 2020.

\bibitem[GML21]{GML21}
Luca Ganassali, Laurent Massouli\'{e}, and Marc Lelarge. 
\newblock Impossibility of partial recovery in the graph alignment problem. 
\newblock In {\em Proceedings of the 34th Conference on Learning Theory (COLT)}, pages 2080--2102. PMLR, 2021.

\bibitem[GML24]{GML20+}
Luca Ganassali, Laurent Massouli\'{e}, and Marc Lelarge.
\newblock Correlation detection in trees for planted graph alignment.
\newblock {\em Annals of Applied Probability}, 34(3):2799--2843, 2024.

\bibitem[GMS24]{GMS22+}
Luca Ganassali, Laurent Massouli\'{e}, and Guilhem Semerjian.
\newblock Statistical limits of correlation detection in trees.
\newblock {\em Annals of Applied Probability}, 34(4):3701--3734, 2024.

\bibitem[GRS22]{GRS22}
Julia Gaudio, Mikl\'{o}s Z. R\'{a}cz, and Anirudh Sridhar.
\newblock Exact community recovery in correlated stochastic block models.
\newblock In {\em Proceedings of the 35th Conference on Learning Theory (COLT)}, pages 2183--2241. PMLR, 2022.

\bibitem[GS94]{GS94}
William M. Y. Goh and Eric Schmutz. 
\newblock Unlabeled trees: distribution of the maximum degree.
\newblock {\em Random Structures and Algorithms}, 5(3):411--440, 1994.

\bibitem[GL24+]{GL24+}
Shuyang Gong and Zhangsong Li.
\newblock The Umeyama algorithm for matching correlated Gaussian geometric models in the low-dimensional regime.
\newblock arXiv preprint, arXiv:2402.15095.

\bibitem[HNM05]{HNM05}
Aria D. Haghighi, Andrew Y. Ng, and Christopher D. Manning.
\newblock Robust textual inference via graph matching.
\newblock In {\em Proceedings of Human Language Technology Conference and Conference on Empirical Methods in Natural Language Processing (EMNLP)}, pages 387--394. ACL, 2005.

\bibitem[HM23]{HM23}
Georgina Hall and Laurent Massouli\'{e}.
\newblock Partial recovery in the graph alignment problem.
\newblock {\em Operations Research}, 71(1):259--272, 2023.

\bibitem[HLL83]{HLL83}
Paul W. Holland, Kathryn B. Laskey, and Samuel Leinhardt. 
\newblock Stochastic block models: first steps.
\newblock {\em Social Networks}, 5(2):109--137, 1983. 

\bibitem[Hop18]{Hopkins18}
Samuel B. Hopkins.
\newblock Statistical inference and the sum of squares method.
\newblock PhD thesis, Cornell University, 2018.

\bibitem[HKP+17]{HKP+17}
Samuel B. Hopkins, Pravesh K. Kothari, Aaron Potechin, Prasad Raghavendra, Tselil Schramm, and David Steurer.
\newblock The power of sum-of-squares for detecting hidden structures.
\newblock In {\em Proceedings of the IEEE 58th Annual Symposium on Foundations of Computer Science (FOCS)}, pages 720--731. IEEE, 2017.

\bibitem[HR17]{HR17}
Godfrey H. Hardy and Srinivasa Ramanujan.
\newblock Asymptotic formulae in combinatory analysis.
\newblock {\em Proceedings of the London Mathematical Society}, 17(1):75--115, 1917.

\bibitem[HS17]{HS17}
Samuel B. Hopkins and David Steurer.
\newblock Efficient Bayesian estimation from few samples: community detection and related problems.
\newblock In {\em Proceedings of the IEEE 58th Annual Symposium on Foundations of Computer Science (FOCS)}, pages 379--390. IEEE, 2017.

\bibitem[KS66]{KS66}
Harry Kesten and Bernt P. Stigum. 
\newblock Additional limit theorems for indecomposable multidimensional Galton-Watson processes. 
\newblock {\em Annals of Mathematical Statistics}, 37:1463--1481, 1966.

\bibitem[KMW24]{KMW24+}
Dmitriy Kunisky, Cristopher Moore, and Alexander S. Wein.
\newblock Tensor cumulants for statistical inference on invariant distributions.
\newblock In {\em Proceedings of the IEEE 65th Annual Symposium on Foundations of Computer Science (FOCS)}, pages 1007--1026. IEEE, 2024.

\bibitem[KWB22]{KWB22}
Dmitriy Kunisky, Alexander S. Wein, and Afonso S. Bandeira.
\newblock Notes on computational hardness of hypothesis testing: Predictions using the low-degree likelihood ratio.
\newblock In {\em Mathematical Analysis, its Applications and Computation: ISAAC 2019}, pages 1--50. Springer, 2022.

\bibitem[MMS10]{MMS10}
Konstantin Makarychev, Rajsekar Manokaran, and Maxim Sviridenko.  
\newblock Maximum quadratic assignment problem: Reduction from maximum label cover and lp-based approximation algorithm. 
\newblock In {\em International Colloquium on Automata, Languages, and Programming (ICALP)}, pages 594--604. Springer, 2010.

\bibitem[MRT21]{MRT21}
Cheng Mao, Mark Rudelson, and Konstantin Tikhomirov.
\newblock Random graph matching with improved noise robustness.
\newblock In {\em Proceedings of the 34th Conference on Learning Theory (COLT)}, pages 3296--3329. PMLR, 2021.

\bibitem[MRT23]{MRT23}
Cheng Mao, Mark Rudelson, and Konstantin Tikhomirov.
\newblock Exact matching of random graphs with constant correlation.
\newblock {\em Probability Theory and Related Fields}, 186(2):327--389, 2023.

\bibitem[MW25]{MW21+}
Cheng Mao and Alexander S. Wein.
\newblock Optimal spectral recovery of a planted vector in a subspace.
\newblock {\em Bernoulli}, 31(2):1114--1139, 2025.

\bibitem[MWXY24]{MWXY21+}
Cheng Mao, Yihong Wu, Jiaming Xu, and Sophie H. Yu.
\newblock Testing network correlation efficiently via counting trees.
\newblock {\em Annals of Statistics}, 52(6):2483--2505, 2024.

\bibitem[MWXY23]{MWXY23}
Cheng Mao, Yihong Wu, Jiaming Xu, and Sophie H. Yu. 
\newblock Random graph matching at Otter's threshold via counting chandeliers.
\newblock In {\em Proceedings of the 55th Annual ACM Symposium on Theory of Computing (STOC)}, pages 1345--1356. ACM, 2023.

\bibitem[Mas14]{Massoulie14}
Laurent Massouli\'{e}.
\newblock Community detection thresholds and the weak Ramanujan property. 
\newblock In {\em Proceedings of the 46th Annual ACM Symposium on Theory of Computing (STOC)}, pages 694--703. ACM, 2014.

\bibitem[MW25]{MW22}
Andrea Montanari and Alexander S. Wein.
\newblock Equivalence of approximate message passing and low-degree polynomials in rank-one matrix estimation.
\newblock {\em Probability Theory and Related Fields}, 191(1-2):181--233, 2025.

\bibitem[MNS15]{MNS15}
Elchanan Mossel, Joe Neeman, and Allan Sly.
\newblock Reconstruction and estimation in the planted partition model.
\newblock {\em Probability Theory and Related Fields}, 162(3-4):431--461, 2015.

\bibitem[MNS16]{MNS16}
Elchanan Mossel, Joe Neeman, and Allan Sly.
\newblock Belief propagation, robust reconstruction and optimal recovery of block models.
\newblock {\em Annals of Applied Probability}, 26(4):2211--2256, 2016.

\bibitem[MNS18]{MNS18}
Elchanan Mossel, Joe Neeman, and Allan Sly. 
\newblock A proof of the block model threshold conjecture.
\newblock {\em Combinatorica}, 38(3):665--708, 2018.

\bibitem[MSS25a]{MSS23}
Elchanan Mossel, Allan Sly, and Youngtak Sohn.
\newblock Exact phase transitions for stochastic block models and reconstruction on trees.
\newblock {\em Annals of Probability}, 53(3):967--1018, 2025.

\bibitem[MSS25b]{MSS24+}
Elchanan Mossel, Allan Sly, and Youngtak Sohn.
\newblock Weak recovery, hypothesis testing, and mutual information in stochastic block models and planted factor graphs.
\newblock In {\em Proceedings of the 57th Annual ACM Symposium on Theory of Computing (STOC)}, pages 2062--2073, 2025.

\bibitem[MX20]{MX20}
Elchanan Mossel and Jiaming Xu.
\newblock Seeded graph matching via large neighborhood statistics.
\newblock {\em Random Structures and Algorithms}, 57(3):570--611, 2020.

\bibitem[NS08]{NS08}
Arvind Narayanan and Vitaly Shmatikov.
\newblock Robust de-anonymization of large sparse datasets.
\newblock In {\em IEEE Symposium on Security and Privacy}, pages 111--125. IEEE, 2008.

\bibitem[NS09]{NS09}
Arvind Narayanan and Vitaly Shmatikov.
\newblock De-anonymizing social networks.
\newblock In {\em IEEE Symposium on Security and Privacy}, pages 173--187. IEEE, 2009.

\bibitem[OGE16]{OGE16}
Efe Onaran, Siddharth Garg, and Elza Erkip. 
\newblock Optimal de-anonymization in random graphs with community structure. 
\newblock In {\em Proceedings of the 50th Asilomar Conference on Signals, Systems and Computers}, pages 709--713. IEEE, 2016.

\bibitem[Ott48]{Otter48}
Richard Otter.
\newblock The number of trees.
\newblock {\em Annals of Mathematics}, 49(3):583--599, 1948.

\bibitem[PG11]{PG11}
Pedram Pedarsani and Matthias Grossglauser.
\newblock On the privacy of anonymized networks.
\newblock In {\em Proceedings of the 17th ACM SIGKDD International Conference on Knowledge Discovery and Data Mining}, pages 1235--1243. ACM, 2011. 

\bibitem[RS21]{RS21}
Mikl\'{o}s Z. R\'{a}cz and Anirudh Sridhar.
\newblock Correlated stochastic block models: exact graph matching with applications to recovering communities.
\newblock In {\em Advances in Neural Information Processing Systems (NIPS)}, volume~34, pages 22259--22273. Curran Associates, Inc., 2021.

\bibitem[RS22]{RS22}
Mikl\'{o}s Z. R\'{a}cz and Anirudh Sridhar. 
\newblock Correlated randomly growing graphs. 
\newblock {\em Annals of Applied Probability}, 32(2):1058--1111, 2022.

\bibitem[SW22]{SW22}
Tselil Schramm and Alexander S. Wein.
\newblock Computational barriers to estimation from low-degree polynomials.
\newblock {\em Annals of Statistics}, 50(3):1833--1858, 2022.

\bibitem[SXB08]{SXB08}
Rohit Singh, Jinbo Xu, and Bonnie Berger.
\newblock Global alignment of multiple protein interaction networks with application to functional orthology detection.
\newblock {\em Proceedings of the National Academy of Sciences of the United States of America}, 105:12763--12768, 2008.

\bibitem[WWXY22]{WWXY22}
Haoyu Wang, Yihong Wu, Jiaming Xu, and Israel Yoloh. 
\newblock Random graph matching in geometric models: the case of complete graphs. 
\newblock In {\em Proceedings of the 35th Conference on Learning Theory (COLT)}, pages 3441--3488. PMLR, 2022.

\bibitem[Wein22]{Wein22}
Alexander S. Wein. 
\newblock Optimal low-degree hardness of maximum independent set. 
\newblock \emph{Mathematical Statistics and Learning}, pages 221--251, 2022.

\bibitem[WROM86]{WROM86}
Robert A. Wright, Bruce Richmond, Andrew Odlyzko, and Brendan D. McKay.
\newblock Constant time generation of free trees. 
\newblock {\em SIAM Journal on Computing}, 15(2):540--548, 1986.

\bibitem[WXY23]{WXY23}
Yihong Wu, Jiaming Xu, and Sophie H. Yu.
\newblock Testing correlation of unlabeled random graphs.
\newblock {\em Annals of Applied Probability}, 33(4):2519--2558, 2023.

\bibitem[WXY22]{WXY22}
Yihong Wu, Jiaming Xu, and Sophie H. Yu. 
\newblock Settling the sharp reconstruction thresholds of random graph matching.
\newblock {\em IEEE Transactions on Information Theory}, 68(8):5391--5417, 2022.

\bibitem[YSC23]{YSC23}
Joonhyuk Yang, Dongpil Shin, and Hye Won Chung.
\newblock Efficient algorithms for exact graph matching on correlated stochastic block models with constant correlation.
\newblock In {\em Proceedings of the 40th International Conference on Machine Learning (ICML)}, pages 39416--39452. PMLR, 2023.







\end{thebibliography}

\end{document}